\documentclass[12pt]{report}

\usepackage[utf8]{inputenc}
\usepackage{graphicx}
\graphicspath{{images/}}
\usepackage[T1]{fontenc}
\usepackage{braket}
\usepackage{float}
\usepackage{times}
\usepackage{amsmath}
\usepackage{amssymb}
\DeclareMathOperator{\Tr}{Tr}
\usepackage[dvipsnames]{xcolor}
\usepackage{bbold}
\allowdisplaybreaks
\usepackage{MnSymbol}
\usepackage{tikz-network}
\usepackage{tikz}
\usetikzlibrary{quantikz}
\usepackage[font={small,it}]{caption}
\usepackage{mathtools}
\usepackage{physics}
\usepackage{phoenician}
\usepackage[makeroom]{cancel}
\usepackage{subcaption}
\usepackage{mathrsfs}
\usepackage{textgreek}
\usepackage{pdfpages}
\usepackage{cite}
\usepackage{hyperref}

\usepackage{url}
	
\newtheorem{theorem}{Theorem}[section]
\newtheorem{corollary}{Corollary}[theorem]
\newtheorem{lemma}[theorem]{Lemma}
\newtheorem{definition}{Definition}[section]
\newtheorem{example}{Example}[section]
\newtheorem{remark}{Remark}[section]
\newtheorem{proposition}{Proposition}[section]

\newenvironment{proof}{\subparagraph{Proof:}}{\hfill$\textbf{QED}$\\}

\newcommand\numberthis{\addtocounter{equation}{1}\tag{\theequation}}

\newcommand{\dket}[1]{\vert #1\rrangle}
\newcommand{\dbra}[1]{\llangle #1\vert}

\DeclareMathOperator{\supp}{Supp}

\begin{document}

\begin{center}
	\begin{large}
		{\bf COMENIUS UNIVERSITY IN BRATISLAVA}  \\ 
	\end{large}
	\begin{large}
		{\bf FACULTY OF MATHEMATICS, PHYSICS AND INFORMATICS} \\
	\end{large}
	
	\vspace*{0.3cm}
\end{center}

\vspace*{5.0cm}

\begin{large}
	\begin{center}
		{\bf PROGRAMMABLE QUANTUM PROCESSORS: \\EQUIVALENCE AND LEARNING}
	\end{center}
\end{large}
\begin{center}
	\begin{Large}
		{\bf Dissertation}
	\end{Large}
\end{center}

\vfill

\begin{large}
	{\bf 2024} \hspace*{9.0cm}{\bf Mgr. Jaroslav Pavličko}
\end{large}

\newpage


\begin{center}
	\begin{large}
		{\bf COMENIUS UNIVERSITY IN BRATISLAVA}  \\ 
	\end{large}
	\begin{large}
		{\bf FACULTY OF MATHEMATICS, PHYSICS AND INFORMATICS} \\
	\end{large}
	
	\vspace*{4.8cm}
\end{center}

\begin{large}
	\begin{center}
		{\bf PROGRAMMABLE QUANTUM PROCESSORS: \\EQUIVALENCE AND LEARNING}
	\end{center}
\end{large}

\begin{center}
	\begin{Large}
		{\bf Dissertation}
	\end{Large}
\end{center}

\vspace*{5.5cm}

\begin{large}
	Study program: \quad Theoretical and Mathematical Physics
\end{large}

\begin{large}
	Study field: \hspace*{1.2cm} 13. - Physics
\end{large}

\begin{large}
	Training center: \quad Institute of Physics, Slovak Academy of Sciences
\end{large}

\begin{large}
	Supervisor: \hspace*{1.2cm} Doc. Mgr. Mário Ziman, PhD.
\end{large}

\vfill

\begin{large}
	{\bf Bratislava 2024} \hspace*{7.5cm}{\bf Mgr. Jaroslav Pavličko}
\end{large}
\thispagestyle{empty}

\includepdf[pages=-]{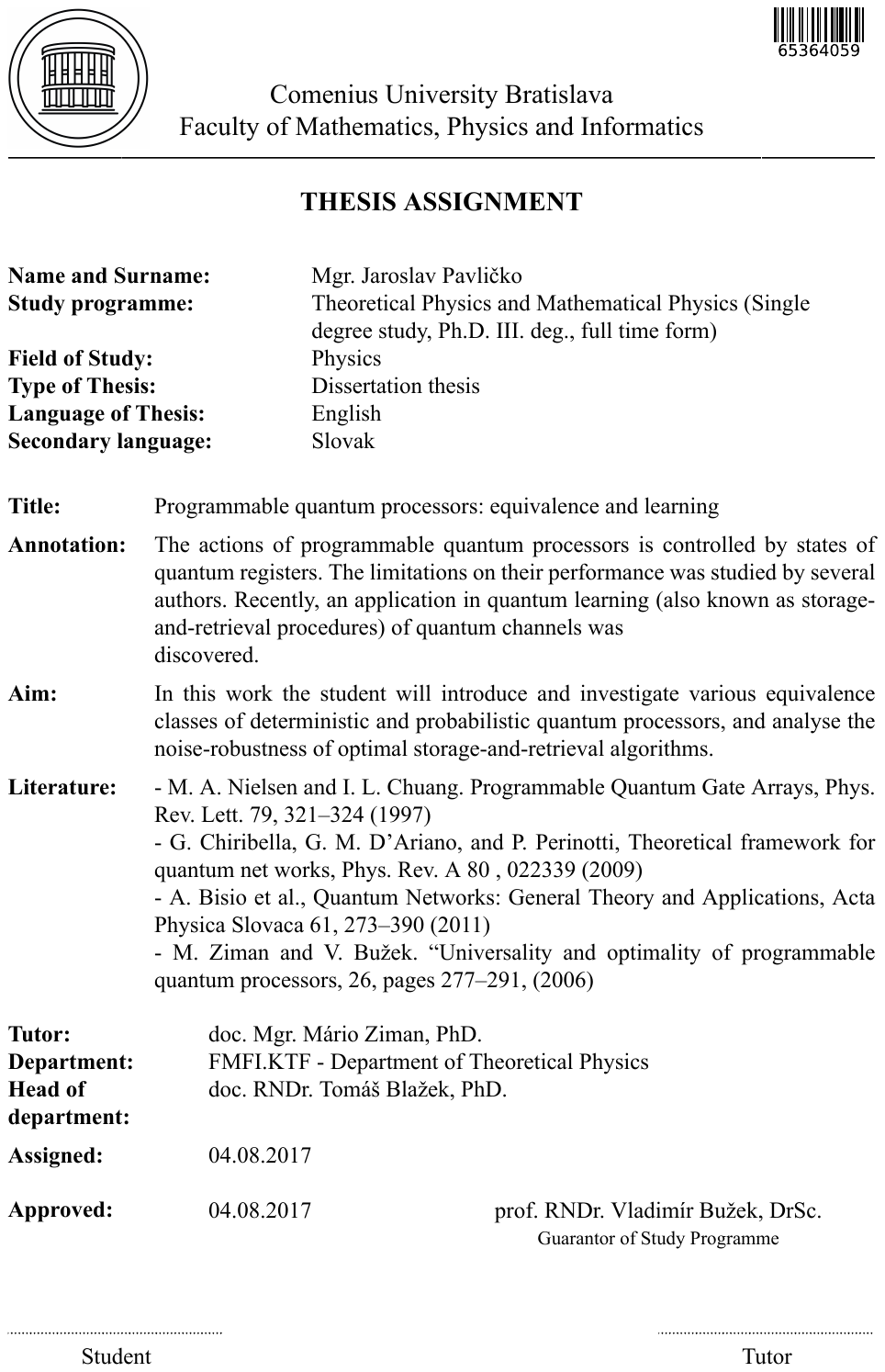}

\newpage
\vspace*{\fill}
\begin{center}
	I want to express my eternal gratitude for love and support \\I have received from my family - my mom, Miro, Vierka, Martin and Anička, \\who is little condensed joy of the purest of forms
\end{center}
\vspace*{\fill}

\chapter*{Abstrakt}
V prvej časti práce bola skúmaná ekvivalencia kvantových deterministických a pravdepodobnostných procesorov. Programovateľný kvantový procesor je zariadenie, ktoré je schopné zmeniť vstupný dátový stav želaným spôsobom. Bola definovaná deterministická a tri typy pravdepodobnostnej (silná, slabá a štruktúrna) ekvivalencie. Boli objavené nevyhnutné a postačujúce podmienky pre deterministickú a štruktúrnu ekvivalenciu unitárne zviazaných procesorov. Ekvivalencia deterministického SWAP procesora pre dvojdimenzionálny dátový a dvojdimenzionálny programový priestor bola kompletne vyriešená. Bolo zistené, že spany operátorov štruktúrne ekvivalentných procesorov sú identické. Vzťahy medzi rôznymi typmi ekvivalencií boli takisto preskúmané.

V druhej časti bola preskúmaná odolnosť pravdepodobnostného úložného a získavacieho zariadenia (PÚAZ), ktoré bolo pôvodne optimalizované pre implementáciu fázového hradla, voči šumu. Konkrétne voči depolarizácii a fázovému tlmeniu. V prípade depolarizačného kanála zmiešaného s unitárnym kanálom, zariadenie implementuje zašumený kanál s klesajúcou pravdepodobnosťou vzhľadom k rastúcemu počtu použití daného kanálu. V prípade fázového tlmenia, zariadenie implementuje zašumený kanál s rovnakou pravdepodobnosťou ako originálne PÚAZ optimalizované pre fázového hradlo. Konkrétne implementácie - cez Vidalovu-Masanesovu-Cirac-ovu schému a virtuálny qudit - boli tiež preskúmané. Vidal-Masanes-Cirac dáva rovnaké výsledky pre oba zašumené kanály, ktoré sú zároveň lepšie ako výsledky z PÚAZ. Implementácia depolarizačného kanála pomocou virtuálneho quditu prináša zhoršenú pravdepodobnosť úspešného merania v porovnaní s Vidalom-Masanesom-Ciracom. Avšak, je stále lepšia ako v prípade PÚAZ. Pravdepodobnosť úspešného merania pre fázové tlmenie implementované pomocou virtuálneho quditu je rovnaká ako pre Vidalovu-Masanesovu-Ciracovu schému a PÚAZ.
\newline
\textbf{Kľúčové slová:} kvantový procesor, ekvivalencia kvantových procesorov, pravdepodobnostné úložné a získavacie zariadenie

\chapter*{Abstract}
In the first part of the work, the equivalence of quantum deterministic and probabilistic processors was investigated. A programmable quantum processor is a device able to transform input data states in a desired way. Deterministic equivalence as well as three types of probabilistic equivalences - strong, weak, and structural - were defined. Necessary and sufficient conditions for deterministic and structural equivalence of unitarily related processors were discovered. Equivalence of deterministic SWAP processor for two-dimensional data and two-dimensional program space was completely solved. It was found that spans of operators of structurally equivalent processors are identical. Relations between types of individual equivalences were also examined.

In the second part, robustness of probabilistic storing and retrieval device (PSAR), originally optimized for implementing a phase gate, to noise was examined - specifically to depolarization and phase damping. In the case of a depolarizing channel mixed with a unitary channel, the device implements noisy channel with the probability that decreases with an increasing number of times the given channel is applied. In the case of the phase damping channel, the device implements noisy channel with the same probability as the original PSAR device optimized for phase gate. Concrete implementations - through the Vidal-Masanes-Cirac scheme and virtual qudit - were examined. Vidal-Masanes-Cirac gives the same result for both noisy channels which is better than the result from PSAR. Implementation through virtual qudit for depolarization yields worse probability of successful measurement than Vidal-Masanes-Cirac. However, it is still better than the probability for PSAR. Probability of successful measurement obtained for phase damping implemented through virtual qudit is the same as for Vidal-Masanes-Cirac and PSAR.
\newline
\textbf{Keywords:} quantum processor, equivalence of quantum processors, probabilistic storing and retrieval device

\chapter*{Foreword}
My desire was to prepare as complete and as clear a work as I had been able to accomplish. Especially, original derivations might be described in excruciating detail. However, I had myself on my mind, while writing in this way. I know that I despise when vast parts of calculations are skipped (or better yet, left as an exercise for a reader), therefore I had opted for this kind of writing style. However, it is clear to me, that I might not have upheld my ambition in every part of the thesis. And also, it is clear to me, that this kind of writing comes with the risk of causing more confusion than clarity. It is also my longing and wish that the English language, vocabulary, and syntax would come to me in a more ordered, nuanced, rich, and clear flow than it did. However, I express my wish that the work is comprehensible enough.

Quantum processors themselves are very intriguing and peculiar subject as there exist close connection to Stinespring dilation and quantum instruments. Thus, studying these devices can bring about a profound physical epiphanies with deep consequences.

In the university's description of the foreword, one of the suggestions or requirements was to describe methods used in the thesis. I have decided to humor this particular demand by listing my two favorite methods - lying on a bed and thinking and lying in a hot bathtub and reading.

Let me conclude with a poem that has none whatsoever to do with physics, but I like it and the fact it was written more than 25 centuries ago by $\Sigma\alpha\pi\varphi\acute{\omega}$ enhances its haunting beauty and ambition (translated by Aaron Poochigian):
\begin{verse}
	\emph{\hspace*{4cm}I declare\\
		\hspace*{4cm}That later on,\\
		\hspace*{4cm}Even in an age unlike our own,\\
		\hspace*{4cm}Someone will remember who we are.}
\end{verse}

\tableofcontents

\chapter{Introduction}
Quantum computing and information theory is still a relatively young field of science. Its roots can be traced back to Feynman and his question whether it is possible to simulate physics on a computer back in 1982 \cite{SimulatingPhysicsWithComputers}. More attention came with first quantum algorithms promising speed-ups compared to their classical counterparts such as quantum factorization \cite{DiscreteLogarithmsAndFactoring} or quantum search \cite{AFastQuantumMechanicalAlgorithmForDatabaseSearch} as well as important results in quantum cryptography \cite{QuantumCryptographyPublicKeyDistributionAndCoinTossing} and quantum teleportation \cite{TeleportingAnUnknownQuantumStateViaDualClassicalAndEinsteinPodolskyRosenChannels}. Quantum mechanics is also shaped by no-go theorems putting boundaries on what is and what is not possible in quantum physics. Examples of such theorems include no-cloning \cite{ASingleQuantumCannotBeCloned} and no-deleting \cite{ImpossibilityOfDeletingAnUnknownQuantumState}, no-broadcasting \cite{NoncommutingMixedStatesCannotBeBroadcast}, no-hiding theorem \cite{QuantumInformationCannotBeCompletelyHiddenInCorrelationsImplicationsForTheBlackHoleInformationParadox}, Bell's theorem \cite{OnTheEinsteinPodolskyRosenParadox}, or Bell-Kochen-Specker theorem \cite{TheProblemOfHiddenVariablesInQuantumMechanics}.

The heart of a ''classical'' computer is a processor - a device able to manipulate input data. Analogous device for quantum computers was first proposed by Nielsen and Chuang \cite{ProgrammableQuantumGateArrays}. They proved a no-go theorem called no-programming, which states that there exist no quantum processor that would be able to perfectly implement every possible unitary transformation. A way to deal with no-programming theorem is to either implement the desired transformation perfectly, but only with certain probability, or to implement the desired transformation only approximately. In this work, we ask the question when quantum processors are able to implement the same transformations, or to put it in other words, when are they equivalent.

An especially useful tool for optimizing quantum devices is a higher-order formalism called quantum networks that encapsulates description of quantum operations, states, and measurements into a common way of describing them \cite{TheoreticalFrameworkForQuantumNetworks, TheoreticalFrameworkForHigher-orderQuantumTheory, QuantumNetworks:GeneralTheoryAndApplications, QuantumCircuitArchitecture}. Therefore, simplifying manipulation and optimization of mentioned concepts. After all, quantum processors can also be described by this formalism. One of the uses of quantum networks is the investigation of how quantum dynamics can be stored in quantum states via a probabilistic storage and retrieval device \cite{OptimalCloningOfUnitaryTransformation, OptimalProbabilisticStorageAndRetrievalOfUnitaryChannels, OptimalQuantumLearningOfAUnitaryTransformation}. This process is sometimes referred to as quantum learning.  We examine closer the ability of the optimal device for probabilistic storage and retrieval of phase gates to resist the noise.

\chapter{Mathematical Formalism}\label{MF}
Mathematical language is a fundamental stone of every physical theory. Mathematics is the way to describe physics. It has the ability to simplify understanding and work with the theory. Mathematical formalism of quantum information theory is based on linear operators inhabiting Hilbert spaces. It moves description of quantum states from vectors to density operators, evolutions from unitary operators to quantum channels and measurements from von Neumann measurements to positive operator-valued measures. All these objects describe of quantum systems and its transformations into a case where one does not have full understanding of the entire system.

All these operators are used to describe quantum circuits. Their description is further generalized in chapter \ref{QN} into operators describing quantum networks that are created by composition of quantum circuits.

In the last section of this chapter, we also provide a brief introduction into the group theory, especially its representation theory with emphasis on the simplest unitary group $U(1)$.

\section{Hilbert Space}

Quantum system is a physical system that has to be described using quantum mechanics. To each quantum system, there is a corresponding Hilbert space.

Hilbert space is a vector space that generalizes Euclidean vector spaces into infinite dimensions \cite{ACourseInFunctionalAnalysis, OperationalQuantumPhysics}.

\begin{definition}
	Hilbert Space
	\\
	Hilbert space ${\cal{H}}$ is a vector space over complex numbers $\mathbb{C}$ with inner product $\bra{.}\ket{.}: \cal{H} \times \cal{H} \rightarrow \mathbb{C}$, for which the following conditions hold:
	\begin{itemize}
		\item \textcolor{gray}{Positivity}: every vector $\ket{\Psi} \in \cal{H}$ is positive $\bra{\Psi}\ket{\Psi} > 0$, if $\ket{\Psi} \neq 0$.
		\item \textcolor{gray}{Conjugate symmetry}: inner product is conjugate symmetric, i.e., for every pair $\ket{\Psi}, \ket{\Phi} \in \cal{H}$, the following $\bra{\Psi}\ket{\Phi} = \bra{\Phi}\ket{\Psi}^{\ast}$ holds, where $\bra{\Phi}\ket{\Psi}^{\ast}$ denotes complex conjugation of $\braket{\Phi}{\Psi}$.
		\item \textcolor{gray}{Linearity}: inner product is linear in its second argument, i.e., for all $a, b \in \mathbb{C}$ and all $\ket{\Psi}, \ket{\Phi}, \ket{\Xi} \in \cal{H}$, the following $\bra{\Xi}\ket{a\Psi + b\Phi} = a\bra{\Xi}\ket{\Psi} + b\bra{\Xi}\ket{\Phi}$ is true.
	\end{itemize}
\end{definition}
By combining conjugate symmetry property with linearity in the second argument, we obtain anti-linearity in the first argument $\bra{a\Psi + b\Phi}\ket{\Xi} = a^{\ast}\bra{\Psi}\ket{\Xi} + b^{\ast}\bra{\Phi}\ket{\Xi}$.

Two finite-dimensional Hilbert spaces are isomorphic if they have the same dimension. Let us have $\ket{\Psi}, \ket{\Phi} \in {\cal{H}}$ and $\ket{\Psi^{\prime}}, \ket{\Phi^{\prime}} \in {\cal{H}}^{\prime}$, where dimensions of their respective Hilbert spaces are the same. Then, there always exists a bijection such that $\bra{U\Psi}\ket{U\Phi} = \bra{\Psi^{\prime}}\ket{\Phi^{\prime}}$, where $U$ is a unitary operator, which means that $U^{\dagger}U = UU^{\dagger} = \mathbb{1}$.


Allow us to also briefly introduce the set of bounded linear operators on Hilbert space ${\cal{H}}$ denoted by $\cal{L}(\cal{H})$. Norm of a bounded operator $A:\cal{L}(\cal{H}) \rightarrow \cal{L}(\cal{H})$ is given by the following definition $\norm{A} = \{\sup(A\Psi) \mid \Psi \in {\cal{H}}, \norm{\Psi} = 1 \} < \infty$. 

\section{Quantum States}

Experiments can be divided in two parts - preparation and measurement. There can be numerous ways how to prepare the same state. Therefore, quantum state can be viewed as an equivalence class of preparations. Quantum state provides probability distribution for every possible measurement.

Mathematically, quantum states are described by density matrices - Hermitian, positive semi-definite operators with trace equal to one \cite{TheMathematicalLanguageOfQuantumTheoryFromUncertaintyToEntanglement, QuantumComputationAndQuantumInformation}. These operators form state space ${\cal{S}}({\cal{H}}) = \{\varrho \in {\cal{T}}({\cal{H}}) \mid \varrho \geq 0, \Tr(\varrho) = 1 \}$, where $\cal{T}({\cal{H}})$ denotes the set of trace-class operators, i.e., bounded operators with finite trace. Space of states is convex, which means that all states can be expressed by convex combination of extremal elements. Extremal elements are called pure states. States $\varrho$ that are formed by convex combination of other states $\varrho = \lambda\varrho_{1} + (1-\lambda)\varrho_{2}$ are called mixed. Mixed state, composed of two other states, can be interpreted as having two distinct preparation devices between which we switch during preparation.

Pure states also exhibit one peculiar property called quantum superposition, which has no classical counterpart. Quantum superposition means that pure state $\ket{\Psi}$ can be expressed through linear combination of other pure states $\ket{\Psi} = \sum_{i} c_{i} \ket{\varphi_{i}}$ with $\sum_{i} \abs{c_{i}}^{2} = 1$.

Let us state the relation between decompositions of mixed states, which shall be used later in proof of proposition \ref{equalityOfChannels}.
\begin{proposition}\label{relOfDecom}
	If the decomposition of a mixed state into pure states is $\varrho = \sum_{j=1}^{N} p_{j}\dyad{\Phi_{j}}$, then all its decompositions into pure states are of the form $\varrho = \sum_{k=1}^{M} q_{k}\dyad{\varphi_{k}}$, where for all vectors $\ket{\varphi_{1}}, \cdots, \ket{\varphi_{M}}$ and $q_{1}, \cdots, q_{M}$, the relations $\sqrt{p_{j}} \ket{\Phi_{j}} = \sum_{k=1}^{M} u_{jk}\sqrt{q_{k}}\ket{\varphi_{k}}$ are satisfied and for complex numbers $u_{jk}$, the following $\sum_{j}u_{jk}u_{jk^{\prime}}^{\ast} = \delta_{kk^{\prime}}$ holds.
\end{proposition}
Proof of the previous proposition can be found in \cite{TheMathematicalLanguageOfQuantumTheoryFromUncertaintyToEntanglement}.

\section{Effects}

Effects describe simplest measuring devices that have outcomes \textit{yes} or \textit{no}. Outcomes of measurements are called events. Therefore, effect describes whether the event (associated with a given effect) has happened during the measurement or not. Whether the device has detected the outcome \textit{yes} or \textit{no}. Effect is a collection of all \textit{yes} or \textit{no} events that can happen during various experiments, where the probabilities corresponding to these events remain constant for all quantum states \cite{TheMathematicalLanguageOfQuantumTheoryFromUncertaintyToEntanglement}.

Effects $E$ are affine maps from state space into the interval of real numbers $E(\varrho): {\cal{S}} \rightarrow \left[0,1\right]$. They can be represented by bounded Hermitian operators $\widehat{E}$. In the following discussion we shall drop the hat and denote operator corresponding to a given effect simply by $E$. Effects attribute probabilities to quantum states through Born rule $\Tr(\varrho E)$, which gives probability of an event, corresponding to effect $E$, happening during measurement, after the preparation of state $\varrho$. It is an affine transformation, because it must preserve attributed probability for convex state $E(\varrho) = E(\lambda \varrho_{1} + (1 - \lambda) \varrho_{2}) = \lambda E(\varrho_{1}) + (1-\lambda)E(\varrho_{2})$. Because effects attribute probabilities to states, certain restrictions are put on them. If effects $E_{i}$ for $i = 1, \cdots, N$ describe all possible events that can occur in a given measurement, then $\sum_{i} E_{i} = \mathbb{1}$, because probability of some event happening must be exactly $1$. And because probability of an individual event to happen cannot exceed $1$, restriction on every individual effect is the following $\mathbf{0} \leq E_{i} \leq \mathbb{1}$, where $\mathbf{0}$ denotes zero effect assigning $0$ probability of happening to every state and $\mathbb{1}$ being identity effect assigning probability $1$ to every state.

Effects form a set on Hilbert space ${\cal{E}}({\cal{H}}) = \{E \in {\cal{L}}_{s}({\cal{H}}) \mid \mathbf{0} \leq E \leq \mathbb{1} \}$, where ${\cal{L}}_{s}({\cal{H}})$ denotes set of bounded Hermitian operators.


\section{Observables}

Observables are properties of a quantum system that can be measured. Measurement is described by a collection of all possible events, which are represented by effects, which can happen during measurement.
Therefore, observables are collections of effects.

Observables can be described by positive operator-valued measures (POVMs) \cite{OperationalQuantumPhysics, TheMathematicalLanguageOfQuantumTheoryFromUncertaintyToEntanglement, QuantumComputationAndQuantumInformation}. Let us now define POVMs for finite dimensions \cite{PreskillLectureNotes}:
\begin{definition}
	Positive Operator-Valued Measure
	\\
	Positive operator-valued measure $M$ is formed by measurement operators $M_{i}$, for which the following conditions are fulfilled:
	\begin{itemize}
		\item Hermiticity: $M_{i} = M_{i}^{\dagger}$,
		\item positivity: $M_{i} \geq 0$,
		\item completeness: $\sum_{i}M_{i} = \mathbb{1}$.
	\end{itemize}
\end{definition}

Probability of measuring an outcome $i$ for a state $\varrho$ is given by $\Tr(\varrho M_{i})$. Probabilities must be real and positive numbers; therefore, one asks for positivity and Hermiticity of elements of POVMs $\{M_{i}\}$. Completeness means that probability of observing an observable during measurement is one.

Let us also introduce special case of measurement.
\begin{definition}
	Projection-valued Measure
	\\
	POVM $A(x)$ is a projection-valued measure (PVM) if for all $x \in X$, where $X$ is countable set of all possible outcomes, $A(x)$ is a projection.
\end{definition}
From Neumark theorem, we know that POVMs can be obtained from projective measurements on a larger Hilbert space \cite{OperationalQuantumPhysics, ProbabilisticAndStatisticalAspectsOfQuantumTheory}.
\begin{theorem}\label{neumark}
	Neumark
	\\
	For every POVM $\{F_{j}\}$ on Hilbert space ${\cal{H}}$, there exists a Hilbert space ${\cal{H}}_{env}$ with state $\xi$ and a PVM $\{E_{j}\}$ such that
	\begin{align*}
		\Tr(\varrho F_{j}) = \Tr\left[(\varrho \otimes \xi) E_{j}\right],
	\end{align*}
	for every $\varrho \in {\cal{H}}$.
\end{theorem}
PVM $E_{j}$ can be chosen to be in the following form $E_{j} = U^{\dagger} (\mathbb{1} \otimes P_{j}) U$, where $U$ is unitary operator and $P_{j}$ is projector.

\section{Channels}

Dynamics of quantum system is described by linear trace preserving completely positive map called quantum channel. It can be imagined as a device to transfer quantum information. Quantum information is encoded in quantum states, that means, that quantum channel takes as input a quantum state and also outputs (changed) quantum state. Quantum channels are special case of transformations called quantum operations, which are linear trace non-increasing completely positive maps.

Linearity is required, because of convex decomposition of state $\varrho = \sum_{i}\lambda_{i}\varrho_{i}$. Then the evolution described by channel ${\cal{C}}$ of this state is ${\cal{C}}(\varrho) = {\cal{C}}(\sum_{i}\lambda_{i}\varrho_{i})$. But channel ${\cal{C}}$ can also be applied to every individual state $\varrho_{i}$ from convex decomposition. Therefore, we require $\sum_{i}\lambda_{i}{\cal{C}}(\varrho_{i}) = {\cal{C}}(\sum_{i}\lambda_{i}\varrho_{i})$ to be true and the requirement for linearity follows. Trace cannot increase after transformation, because resulting probability has to always be $1$. Quantum states are positive operators; therefore, we require channel to be also positive and preserve the positivity of states. But it is not enough, because channel can act only on part of the whole Hilbert space. Let us have channel $\cal{C}$ acting on Hilbert space ${\cal{H}}_{a}$. Let us expand given space by another space ${\cal{H}}_{b}$, while requiring that the expanded channel ${\cal{C}} \otimes \mathbb{1}_{b}$ acts non-trivially only on the space ${\cal{H}}_{a}$. Therefore, we desire for this expanded operator to also preserve positivity. But certain positive operators, such as partial transpose, fail to remain positive on the whole expanded space. Therefore, we require channel to be completely positive.

Channels can be realized by isometries on a larger Hilbert space through Stinespring dilation \ref{Stinespring}, as will be shown in the \ref{QN}th chapter.

\subsection{Operator-sum Representation}

We shall introduce representation of channels that shall be used in later chapters. Let us have completely positive linear operators $A_{1}, \cdots, A_{N}$, for which $\sum_{i} A_{i}^{\dagger}A_{i} = \mathbb{1}$ holds. Therefore, we are allowed to combine such operators into a channel \cite{TheMathematicalLanguageOfQuantumTheoryFromUncertaintyToEntanglement}.
\begin{proposition} \label{opSumRepr}
	Operator-sum Representation
	\\
	Let us have finite set of bounded operators $A_{1}, \cdots, A_{N}$. Then, they form quantum channel $\cal{C}$ if:
	\begin{align}\label{Kraus}
		{\cal{C}}(\varrho) = \sum_{i} A_{i}\varrho A_{i}^{\dagger}, \qquad
		\sum_{i} A_{i}^{\dagger} A_{i} = \mathbb{1}.
	\end{align}
\end{proposition}
Operators, that describe channel as in equation (\ref{Kraus}), are called Kraus operators.

Further, in chapter \ref{QP}, when we discuss equivalence of quantum processors, the following lemma giving condition on when two distinct sets of Kraus operators define the same channel, shall prove itself to be particularly useful.
\begin{proposition}\label{equalityOfChannels}
	Let us have two finite sets of operators $\{A_{1}, \cdots, A_{N}\}$ and $\{B_{1}, \cdots, B_{M}\}$. They describe the same operation if and only if for every $i$, the following $A_{i} = \sum_{j} u_{ij}B_{j}$ and $\sum_{i} u_{ij}u_{ij^{\prime}}^{\ast} = \delta_{jj^{\prime}}$ hold.
\end{proposition}
\begin{proof}
	At first, let us assume that $A_{i} = \sum_{j}u_{ij}B_{j}$ and $\sum_{i} u_{ij}u_{ij^{\prime}}^{\ast} = \delta_{jj^{\prime}}$. By direct calculation we can verify that $B$'s describe the same channel: $\sum\limits_{i} A_{i}\varrho A_{i}^{\dagger} = \sum\limits_{i,j,j^{\prime}}u_{ij} B_{j} \varrho u_{ij^{\prime}}^{\ast} B_{j^{\prime}}^{\dagger} = \sum\limits_{i,j,j^{\prime}} u_{ij}u_{ij^{\prime}}^{\ast} B_{j} \varrho B_{j^{\prime}}^{\dagger} = \sum\limits_{j} B_{j} \varrho B_{j}^{\dagger}$.
	
	Now, let us assume that $\sum_{i} A_{i} \varrho A_{i}^{\dagger} = \sum_{j} B_{j} \varrho B_{j}^{\dagger}$. Let us choose $\varrho =  \dyad{\Psi}$, where $\ket{\Psi}$ is a unit vector. Then we obtain that $\sum_{i}A_{i}\dyad{\Psi}A_{i}^{\dagger} = \sum_{j}B_{j}\dyad{\Psi}B_{j}^{\dagger}$. We can view this expression as a convex decomposition of a mixed state $\varrho^{\prime} = \sum_{i}A_{i}\dyad{\Psi}A_{i}^{\dagger} = \sum_{j}B_{j}\dyad{\Psi}B_{j}^{\dagger}$ and from lemma \ref{relOfDecom}, we obtain that $A_{i}\ket{\Psi} = \sum_{j}u_{ij}B_{j}\ket{\Psi}$ with $\sum_{i} u_{ij} u_{ij^{\prime}}^{\ast} = \delta_{jj^{\prime}}$. And because this holds for all the unit vectors, we arrive at the expression from the current lemma $A_{i} = \sum_{j}u_{ij}B_{j}$ with $\sum_{i}u_{ij}u_{ij^{\prime}}^{\ast} = \delta_{jj^{\prime}}$.
\end{proof}

The time to introduce two channels used in chapter \ref{QN}, when discussing robustness of PSAR regarding noise, has arrived.

\subsection{Depolarizing Channel}

Depolarizing channel maps input state into convex combination of itself and total mixture:
\begin{align}\label{Dchannel}
	{\cal{D}} = p \frac{\mathbb{1}}{d} + (1-p) \varrho,
\end{align}
where $0 \leq p \leq 1$ is a probability of depolarization of a state $\varrho$.

\subsection{Phase Damping Channel}

Phase damping channel causes the loss of quantum information, the decoherence from initial quantum superposition into a total mixture of orthogonal states over time. Kraus operators defining such a channel are the following \cite{QuantumComputationAndQuantumInformation}:
\begin{align*}
	E_{0} = \left(\begin{matrix}
		1 & 0\\
		0 & \sqrt{1 - \lambda}
	\end{matrix}\right) \qquad
	E_{1} = \left(\begin{matrix}
		0 & 0\\
		0 & \sqrt{\lambda}
	\end{matrix}\right),
\end{align*}
where $\lambda \in \left[0,1\right]$. Let us have quantum state $\varrho = a\ket{0} + b\ket{1}$, then the resulting state after applying phase damping channel ${\cal{N}}_{P}$ is:
\begin{align}\label{PDchannel}
	{\cal{P}} (\varrho) = \left(\begin{matrix}
		\abs{a}^{2} & ab^{\ast}\sqrt{1-\lambda}\\
		a^{\ast}b\sqrt{1-\lambda} & \abs{b}^{2}
	\end{matrix}\right).
\end{align}

\section{Instruments}

Quantum instruments ${\cal{I}}$ are devices that can have both quantum and classical outcome. They are formed by a set of quantum operations ${\cal{O}}_{i}$ such that these operations sum up to quantum channel $\sum_{i} {\cal{O}}_{i} = \cal{C}$. These devices take as input a quantum state and have two outputs - quantum state after the measurement and the outcome of the measurement. Let us give definition for discrete instrument, as we have only encountered discrete instruments in our work.
\begin{definition}
	Quantum Instrument
	\\
	Let $X$ be a countable set of all possible outcomes of measurement. Then, quantum instrument ${\cal{I}}$ is a set of quantum operations $\{{\cal{O}}_{i}\}_{i \in X}$ such that these sum up to linear completely positive trace preserving map $\sum_{i} {\cal{O}}_{i} = {\cal{C}}$.
\end{definition}

From point of view of an experiment, instrument describes a measuring device with multiple distinguishable outcomes. Let us imagine, that the outcome of measurement is $i$, then the state after measurement is given by $\varrho_{i}^{\prime} = \frac{{\cal{O}}_{i}(\varrho)}{\Tr\left[{\cal{O}}_{i}(\varrho)\right]}$, where $\varrho$ is the original state before measurement. Probability of measuring outcome $i$ is $\Tr\left[{\cal{O}}_{i}(\varrho)\right]$.

\section{Group Theory}

Symmetries in physics are mirrored in transformations of a system. Groups are mathematical objects used to describe symmetric transformations and provide tools for working with symmetries \cite{DifferentialGeometryAndLieGroupsForPhysicists, RepresentationTheoryAFirstCourse}. Quantum-mechanical system can be changed by a unitary transformation.
\begin{definition}
	Group
	\\
	Group $G$ is a set with binary operation $\circ$ on the set, such that the following conditions hold:
	\begin{itemize}
		\item \textcolor{gray}{closure}: for all $g_{1}, g_{2} \in G$, it holds that $g_{1} \circ g_{2} \in G$,
		\item \textcolor{gray}{associativity}: $(g_{1} \circ g_{2}) \circ g_{3} = g_{1} \circ (g_{2} \circ g_{3})$ holds for all $g_{1}, g_{2}, g_{3} \in G$,
		\item \textcolor{gray}{identity}: there exists a unique identity element $e \in G$ such that for all $g \in G$ the following is true $g \circ e = e \circ g = g$,
		\item \textcolor{gray}{inverse}: for every element $g \in G$, there exists a unique inverse element $g^{-1} \in G$ such that $g \circ g^{-1} = g^{-1} \circ g = e$.
	\end{itemize}
\end{definition}

\subsection{Representation Theory}

Representation of group means that to every element of group, a linear operator from some vector space $V$ is assigned. In other words, every element of group is represented by linear operator.
\begin{definition}
	Group Representation
	\\
	A representation of group $G$ in a vector space $V$ is given by homomorphism $\rho:G \rightarrow GL(V)$.
\end{definition}
Group $GL(V)$ is a group of all bijective linear transformations $V \rightarrow V$.  Homomorphism is a map between two groups $\rho:G \rightarrow GL(V)$ such that $\rho(g_{1} \circ g_{2}) = \rho(g_{1}) \bullet \rho(g_{2})$ holds for every $g_{1}, g_{2} \in G$, where $\circ$ is group operation in group $G$ and $\bullet$ in $GL(V)$.

Let us have representation $\rho: G \rightarrow GL(V)$ and subspace $W \subseteq V$. Then $W$ is called \emph{invariant} subspace if for every element $w \in W$, $\rho(g)w \in W$ is from the same subspace. That means, that representation does not "take" any element out of a subspace. Subspaces that are empty $W = \emptyset$ or are equal to the original space $W = V$, are called trivial. \emph{Irreducible} representation (irrep) only has trivial subspaces. Unitary group can always be represented by irreducible representations \cite{RepresentationTheoryAFirstCourse}.

Let us have two different representations $\rho_{1}$ and $\rho_{2}$ of the same group $G$ in vector spaces $V_{1}$ and $V_{2}$, respectively. Linear transformation $A:V_{1} \rightarrow V_{2}$ is called \emph{intertwining operator} if $\rho_{1}(g) A = A \rho_{2}(g)$ holds. Two representations are \emph{equivalent} if $\rho_{2}(g) = A\rho_{1}(g)A^{-1}$.

We shall state Schur's lemma (proof can be found in \cite{RepresentationTheoryAFirstCourse}), as it is used further in description of irreps of group $U(1)$.
\begin{lemma}\label{Schur}
	Schur's Lemma
	\\
	Let $\rho_{1}$ and $\rho_{2}$ be two irreps of group $G$, i.e., $\rho_{1}:G \rightarrow GL(V_{1})$ and $\rho_{2}:G \rightarrow GL(V_{2})$, where $V_{1}$ and $V_{2}$ are complex vector spaces. Let $A:V_{1} \rightarrow V_{2}$ be an intertwining operator and if:
	\begin{itemize}
		\item $\rho_{1}$ and $\rho_{2}$ are not equivalent, then $A = 0$,
		\item $\rho_{1}$ and $\rho_{2}$ are equivalent, then $A = \lambda \mathbb{1}$ for some $\lambda \in \mathbb{C}$.
	\end{itemize}
\end{lemma}
Let us give corollary of Shur's lemma about commutative (i.e., abelian) groups. Commutativity in group $G$ means that $g_{1} \circ g_{2} = g_{2} \circ g_{1}$ for all $g_{1}, g_{2} \in G$.
\begin{corollary}\label{schurCor}
	If $G$ is an abelian group, then all its irreps are one-dimensional.
\end{corollary}
\begin{proof}
	Let $\rho$ be a representation of $G$, then from commutativity we have $\rho(g_{1})\circ\rho(g_{2}) = \rho(g_{2})\circ\rho(g_{1})$. Here, we can imagine, that $\rho(g_{2})$ is an intertwining operator and because commutativity holds for every element from $G$, using Schur's lemma \ref{Schur}, we get that $\rho(g) = \lambda_{g}\mathbb{1}$. For $\rho(g)$ to be irrep, it has to be one-dimensional. (Let us imagine that $\rho(g)$ is two-dimensional. Then $\rho(g) = 
	\mqty(\dmat[0]{\lambda_{1},\lambda_{2}})$, where clearly, there are two invariant subspaces and $\rho$ can be further reduced to $\rho = \lambda_{1} \bigoplus \lambda_{2}$, where $\bigoplus$ denotes direct sum.)
\end{proof}

\subsection{\texorpdfstring{$U(1)$}{U(1)} Group}

As was already mentioned, quantum systems are transformed by unitary operators. Let us illustrate that unitaries, indeed, preserve inner product in Hilbert space: $\braket{U\Psi}{U\Phi} = \braket{\Psi}{U^{\dagger}U\Phi} = \braket{\Psi}{\Phi}$. Specifically, $U(1)$ group is a group of complex numbers $e^{i\alpha}$ with operation being multiplication, i.e., it is a group of rotations on unit circle.

Let us show that irreps of $U(1)$ group are of form $e^{ik\Phi}$, where $k$ is an integer \cite{RepresentationTheoryAndQuantumMechanics}. Group $U(1)$ is abelian, therefore, from corollary of Schur's lemma \ref{schurCor}, we know that all irreducible representations $\rho$ of $U(1)$ must be one-dimensional. Let us now differentiate an irrep:
\begin{align}
	\frac{d}{d\Phi} \rho\left(e^{i\Phi}\right) &= \lim_{\Delta \rightarrow 0} \frac{\rho\left( e^{i(\Phi + \Delta\Phi)} \right) - \rho(e^{i\Phi}) }{\Delta\Phi} = \rho(e^{i\Phi}) \lim_{\Delta \rightarrow 0} \frac{\rho(e^{i\Delta\Phi}) - \rho(1)}{\Delta\Phi} \\&= \rho(e^{i\Phi}) \frac{d}{d\Phi}\left[\rho(e^{i\Phi})\right]_{\Phi = 0} = k \rho(e^{i\Phi}), \numberthis
\end{align}
where $k = \frac{d}{d\Phi}\left[\rho(e^{i\Phi})\right]_{\Phi = 0}$ and $\phi \in \left[0, 2\pi\right)$. Therefore, irrep of $U(1)$ must be $\rho(e^{i\Phi}) = e^{ik\Phi}$. Because representation is a homomorphism, it must hold that $1 = \varrho(1) = \varrho(e^{i\pi}e^{i\pi}) = \varrho(e^{i\pi})\varrho(e^{i\pi}) = e^{ik2\pi}$. Therefore, $k$ must be an integer.

\chapter{Quantum Processors}\label{QP}
Inspiration for quantum processors can be drawn from their classical counterparts. The role of processor in "classical" computers is to transform and manipulate the input data. Therefore, one can envision a device with similar function also for quantum computers. Quantum processors were firstly introduced by Nielsen and Chuang \cite{ProgrammableQuantumGateArrays}.

Quantum processor has two input registers as can be seen in figure \ref{QPimage} \cite{ProgrammableQuantumProcessors}. Data register takes data state, that one wishes to transform, as input. Whereas, program state, that selects the transformation applied on the data state, serves as input to the program register. One is able to choose program state and this act is called \emph{quantum programming}. Processor $G$ itself is formed by an array of quantum gates and as a whole it is a unitary operator, because evolution in closed quantum systems is governed by unitary transformations. And we can always expand Hilbert spaces of a processor to be encapsulating the entire considered system.
\begin{figure}[H]
	\begin{center}
		\includegraphics[scale=0.5]{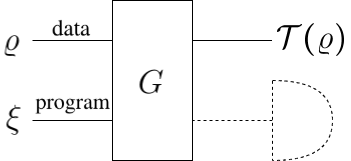}
	\end{center}
	\caption[font=small]{Quantum processor has two input registers - data register with input data state $\varrho$ that one desires to transform, and program register with input program state $\xi$ that controls which transformation $\cal{T}$ is applied on data state $\varrho$. Probabilistic processor has measurement at the end of program register that decides whether the device is successful in implementing the desired transformation on input data state or not.}
	\label{QPimage}
\end{figure}
Deterministic quantum processor implements quantum channel (linear completely positive trace preserving map) on data \cite{ImplementationOfQuantumMapsByProgrammableQuantumProcessors, ImplementabilityOfObservablesChannelsAndInstrumentsInQuantumTheoryOfMeasurement}. Nielsen and Chuang, in their seminal work, proved an important no-go theorem that forbids an existence of deterministic quantum processor able to implement all possible unitary transformations on data \cite{ProgrammableQuantumGateArrays}. Therefore, the theorem forbids the existence of truly universal quantum processor able to change the data completely according to our liking. The reason is that dimension of program space grows with every new unitary operator one wishes to implement. Let us repeat their proof for the completeness of our work.
\begin{theorem}\label{no-programming}
	No-programming
	\\
	There exists no universal deterministic quantum processor.
\end{theorem}
\begin{proof}
	Let us have two program states $\ket{P}$ and $\ket{Q}$ implementing two different unitaries $U_{p}$ and $U_{q}$:
	\begin{align*}
		G(\ket{D} \otimes \ket{P}) = U_{p}\ket{D} \otimes \ket{P^{\prime}} \qquad G(\ket{D} \otimes \ket{Q}) = U_{q}\ket{D} \otimes \ket{Q^{\prime}},
	\end{align*}
	where $\ket{D}$ is input state into data register. Now, let us calculate inner product of the previous equations:
	\begin{align}\label{noprog}
		\braket{Q}{P} = \braket{Q^{\prime}}{P^{\prime}} \mel{D}{U_{q}^{\dagger}U_{p}}{D},
	\end{align}
	where we have used the unitarity of processor $G$. By dividing the previous equation by $\braket{Q^{\prime}}{P^{\prime}} \neq 0$, we obtain:
	\begin{align*}
		\frac{\braket{Q}{P}}{\braket{Q^{\prime}}{P^{\prime}}} = \mel{D}{U_{q}^{\dagger}U_{p}}{D}.
	\end{align*}
	Left hand side does not depend on the state $\ket{D}$, therefore nor the right-hand side can depend on it. In consequence, $U_{q}^{\dagger}U_{p} = c\mathbb{1}$, where $c \in \mathbb{C}$, which in turn means that $U_{q}$ and $U_{p}$ can only differ by a global phase. But that is in contradiction with the assumption. Therefore, $\braket{Q^{\prime}}{P^{\prime}}$ must be equal to $0$. In that case, from equation (\ref{noprog}), we can see that also $\braket{Q}{P} = 0$, i.e., program states $\ket{P}$ and $\ket{Q}$ implementing different unitaries must be orthogonal. That means that for every new unitary we want to implement on the data state, program space of processor grows. And because there are uncountably many unitaries, universal deterministic processor, that would be able to implement any unitary at will, would have to have an uncountable number of dimensions in program space, which is not possible. 
\end{proof}
One can now choose two paths - either to introduce approximate quantum processor, which implements channels approximately \cite{ApproximateProgrammableQuantumProcessors,  OptimalUniversalProgrammingOfUnitaryGates, ResourceQuantificationForTheNoProgramingTheorem, AsymptoticTeleportationSchemeAsAUniversalProgrammableQuantumProcessor, AsymptoticPerformanceOfPortBasedTeleportation}, or go with probabilistic quantum processor, which is able to implement every unitary, albeit only with certain probability \cite{ProbabilisticImplementationOfUniversalQuantumProcessors, ProbabilisticProgrammableQuantumProcessors, ImprovingThePerformanceOfProbabilisticProgrammableQuantumProcessors, ProbabilisticProgrammableQuantumProcessorsWithMultipleCopiesOfProgram}. Probabilistic processor has measurement at the output of program register and implements quantum operation \cite{UniversalityAndOptimalityOfProgrammableQuantumProcessors, RealizationOfPositive-operator-valuedMeasuresUsingMeasurement-assistedProgrammableQuantumProcessors}. In probabilistic processor one cannot be certain that the implementation of the desired transformation shall be successful, but one always knows if it failed. Further on, we shall examine more closely deterministic and probabilistic quantum processors, for which we shall also investigate the notion of equivalence. Equivalence of approximate processors might not be so easily defined due to ambivalence in which measure to use to quantify fidelity of approximately transformed data in comparison to ideally transformed data.

In the following sections we shall denote dimension of program space ${\cal{H}}_{p}$ of processor $G$ with capital letter $P$ and dimension of data space ${\cal{H}}_{d}$ with capital letter $D$.

\section{Deterministic Processor}

In general, deterministic quantum processor can be expressed as follows:
\begin{align}\label{QPeq}
	G = \sum_{jk}^{P} A_{jk} \otimes \dyad{j}{k},
\end{align}
where states $\{\ket{j}\}$ form an orthonormal basis of program space ${\cal{H}}_{p}$ and operators $A_{jk}$ are applied on the data state. Because processor must be unitary:
\begin{align*}
	&GG^{\dagger} = \left(\sum^{P}\limits_{jk} A_{jk} \otimes \dyad{j}{k}\right) \left(\sum^{P}\limits_{j^{\prime}k^{\prime}} A_{j^{\prime}k^{\prime}}^{\dagger} \otimes \dyad{k^{\prime}}{j^{\prime}}\right) = \sum^{P}\limits_{jkj^{\prime}k^{\prime}} A_{jk}A_{j^{\prime}k^{\prime}}^{\dagger} \otimes \ket{j}\hspace*{-1.5mm}\braket{k}{k^{\prime}}\hspace*{-1.5mm}\bra{j^{\prime}} \\&= \sum^{P}\limits_{jj^{\prime}k} A_{jk}A_{j^{\prime}k}^{\dagger} \otimes \dyad{j}{j^{\prime}} \overset{!}{=} \mathbb{1},
\end{align*}
we obtain conditions on operators applied to data register:
\begin{align}\label{condA}
	\sum_{k}^{P} A_{jk}A_{j^{\prime}k}^{\dagger} = \mathbb{1}\delta_{jj^{\prime}}, \quad \sum_{j}^{P} A_{jk}^{\dagger}A_{jk^{\prime}} = \mathbb{1}\delta_{kk^{\prime}},
\end{align}
where the second condition was derived analogously to the first one, just now from the requirement $G^{\dagger}G \overset{!}{=} \mathbb{1}$. Set of channels that is implemented by deterministic quantum processor $G$ is the following:
\begin{align}\label{detQPimpl}
	\mathfrak{C}_{G}^{det} = \{\Tr_{p} \left[G (\varrho \otimes \xi) G^{\dagger}\right] \mid \xi \in {\cal{S}}({\cal{H}})\},
\end{align}
where ${\cal{S}}({\cal{H}})$ denotes the set of all quantum states and $\Tr_{p}$ denotes trace over program space. One traces out program space because the desired transformation is applied only on data. Let us evaluate an element from this set corresponding to a concrete program state $\xi$ denoted by $\mathfrak{C}_{G,\xi}^{det}$:
\begin{align*}
	\mathfrak{C}_{G,\xi}^{det}(\varrho) &= \Tr_{p} \left[G(\varrho \otimes \xi)G^{\dagger}\right] = \Tr_{p} \left[\left(\sum^{P}\limits_{jk} A_{jk} \otimes \dyad{j}{k}\right) \left(\varrho \otimes \xi\right) \left(\sum^{P}\limits_{j^{\prime}k^{\prime}} A_{j^{\prime}k^{\prime}}^{\dagger} \otimes \dyad{k^{\prime}}{j^{\prime}}\right)\right] \\&= \sum^{P}\limits_{jkj^{\prime}k^{\prime}} A_{jk}\varrho A_{j^{\prime}k^{\prime}}^{\dagger} \Tr\left(\dyad{j}{k} \xi \dyad{k^{\prime}}{j^{\prime}}\right) = \sum^{P}\limits_{jkk^{\prime}} \xi_{kk^{\prime}} A_{jk} \varrho A_{jk^{\prime}}^{\dagger},
\end{align*}
where $\xi_{kk^{\prime}} = \mel{k}{\xi}{k^{\prime}}$ and the dimension of program space is $\dim({\cal{H}}_{p}) = P$. Therefore, one element $\mathfrak{C}_{G,\xi}^{det}$ of set of all channels that are implementable by a processor $G$, governed by a program state $\xi$, can be expressed as:
\begin{align}\label{elemOfImplDet}
	\mathfrak{C}_{G,\xi}^{det}(\varrho) = \sum\limits_{jkk^{\prime}}^{P} \xi_{kk^{\prime}} A_{jk} \varrho A_{jk^{\prime}}^{\dagger}.
\end{align}

\section{Probabilistic Processor}\label{probProc}

Probabilistic processor can be written in the same manner as a deterministic one in equation (\ref{QPeq}) and the conditions put on operators in equation (\ref{condA}) remain unchanged. What distinguishes probabilistic processor from a deterministic one is a measurement at the output of program register. If we obtain the proper result, i.e., a state of program space corresponding to a successful measurement, the processor implements the desired transformation. On the other hand, if we obtain a different result, i.e., a state of program space that is orthogonal to states corresponding to a successful measurement, processor failed to apply the desired transformation on data state. Therefore, set of transformations a probabilistic processor $G$ implements is:
\begin{align}\label{prQPimpl}
	\mathfrak{G}_{G}^{pr} = \left\{\frac{1}{p} \Tr_{p} \left[G (\varrho \otimes \xi) G^{\dagger} \left(\mathbb{1} \otimes \dyad{\chi}\right)\right] \mid \xi \in \cal{S} ({\cal{H}}) \right\},
\end{align}
where we have chosen $\mathbb{1} \otimes \dyad{\chi} = \mathbb{1} \otimes \frac{1}{P} \sum_{nn^{\prime}}^{P} \dyad{n}{n^{\prime}}$ to be successful measurement expressed in a basis spanned by states $\{\ket{n}\}$ and $p$ is the probability of successfully implementing the desired transformation.

Let us now calculate concrete element $\mathfrak{G}_{G,\xi}^{pr}$ of set of possible transformations implemented on the data state $\varrho$ by a probabilistic processor $G$ that corresponds to program state $\xi$:
\begin{align*}
	\mathfrak{G}_{G,\xi}^{pr} &= \frac{1}{p} \Tr_{p} \left[G(\varrho \otimes \xi)G^{\dagger} \left(\mathbb{1} \otimes \dyad{\chi}\right)\right] \\&= \frac{1}{p} \Tr_{p} \left[\left(\sum\limits_{jk}^{P} A_{jk} \otimes \dyad{j}{k}\right) \left(\varrho \otimes \xi\right) \left(\sum\limits_{j^{\prime}k^{\prime}}^{P} A_{j^{\prime}k^{\prime}}^{\dagger} \otimes \dyad{k^{\prime}}{j^{\prime}}\right) \left(\mathbb{1} \otimes \frac{1}{P} \sum_{nn^{\prime}}^{P} \dyad{n}{n^{\prime}}\right)\right] \\&= \frac{1}{pP}\sum\limits_{jkj^{\prime}k^{\prime}}^{P} A_{jk}\varrho A_{j^{\prime}k^{\prime}}^{\dagger} \sum\limits_{nn^{\prime}}^{P} \Tr_{p}\left(\dyad{j}{k} \xi \dyad{k^{\prime}}{j^{\prime}}\hspace*{-1.5mm}\dyad{n}{n^{\prime}}\right) = \frac{1}{pP} \sum\limits_{jj^{\prime}kk^{\prime}}^{P} \xi_{kk^{\prime}} A_{jk} \varrho A_{j^{\prime}k^{\prime}}^{\dagger},
\end{align*}
where $\dim{{\cal{H}}_{p}} = P$ and $p$ is probability of implementing the desired transformation. Let us now, for clarity and future use, explicitly write the result of the previous calculation:
\begin{align}\label{elemOfImplPr}
		\mathfrak{G}_{G,\xi}^{pr} = \frac{1}{pP} \sum\limits_{jj^{\prime}kk^{\prime}}^{P} \xi_{kk^{\prime}} A_{jk} \varrho A_{j^{\prime}k^{\prime}}^{\dagger}.
\end{align}
We shall proceed by calculating probability with which a processor $G$ implements a given transformation on data state $\varrho$:
\begin{align*}
	p &= \Tr \left[G (\varrho \otimes \xi) G^{\dagger} \left(\mathbb{1} \otimes \dyad{\chi}\right)\right] \\
	&= \Tr \left[\left(\sum\limits_{jk}^{P} A_{jk} \otimes \dyad{j}{k}\right) \left(\varrho \otimes \xi\right) \left(\sum\limits_{j^{\prime}k^{\prime}}^{P} A_{j^{\prime}k^{\prime}}^{\dagger} \otimes \dyad{k^{\prime}}{j^{\prime}}\right) \left(\mathbb{1} \otimes \frac{1}{P} \sum_{nn^{\prime}}^{P} \dyad{n}{n^{\prime}}\right)\right] \\
	&= \frac{1}{P}\sum_{jkj^{\prime}k^{\prime}}^{P} \sum_{nn^{\prime}}^{P} \Tr \left(A_{jk} \varrho A_{j^{\prime}k^{\prime}}^{\dagger}\right) \Tr \left(\dyad{j}{k} \xi \dyad{k^{\prime}}{j^{\prime}}\hspace*{-1.5mm}\dyad{n}{n^{\prime}}\right) = \frac{1}{P} \sum_{jkj^{\prime}k^{\prime}}^{P} \xi_{kk^{\prime}} \Tr\left(A_{jk} \varrho A_{j^{\prime}k^{\prime}}^{\dagger}\right).
\end{align*}
Again, for clarity, let us express the probability of successful implementation:
\begin{align}\label{probQP}
	p = \frac{1}{P} \sum\limits_{jj^{\prime}kk^{\prime}}^{P} \xi_{kk^{\prime}} \Tr \left(A_{jk} \varrho A_{j^{\prime}k^{\prime}}^{\dagger}\right).
\end{align}
From this equation, we can clearly see that, in the general case, the probability $p$ depends on the input state $\varrho$. Therefore, probabilistic processor can implement both channels and measurements. Channels are linear operations; thus they are implemented only in case when the probability is not a function of data state $\varrho$ and can only depend on program state:
\begin{align*}
	p = \frac{1}{P} \sum\limits_{jj^{\prime}kk^{\prime}}^{P} \xi_{kk^{\prime}} \Tr \left(A_{jk} \varrho A_{j^{\prime}k^{\prime}}^{\dagger}\right) \overset{!}{=} \frac{k_{\xi}}{P},
\end{align*}
where $k_{\xi} \in \mathbb{R}$ and $\mathbb{R}$ denotes the set of real numbers. Let us take a closer look at the previous condition:
\begin{align*}
	&\sum\limits_{jj^{\prime}kk^{\prime}}^{P} \xi_{kk^{\prime}} \Tr \left(A_{jk} \varrho A_{j^{\prime}k^{\prime}}^{\dagger}\right) = \sum\limits_{jkk^{\prime}}^{P} \xi_{kk^{\prime}} \Tr \left(A_{jk} \varrho A_{jk^{\prime}}^{\dagger}\right) + \sum\limits_{j\neq j^{\prime}}^{P} \sum\limits_{kk^{\prime}}^{P} \xi_{kk^{\prime}} \Tr \left(A_{jk} \varrho A_{j^{\prime}k^{\prime}}^{\dagger}\right) \\
	&= \sum\limits_{jkk^{\prime}}^{P} \xi_{kk^{\prime}} \Tr \left(\varrho A_{jk^{\prime}}^{\dagger}A_{jk}\right) + \sum\limits_{j\neq j^{\prime}}^{P} \sum\limits_{kk^{\prime}}^{P} \xi_{kk^{\prime}} \Tr \left(\varrho A_{j^{\prime}k^{\prime}}^{\dagger}A_{jk}\right) \\
	&\overset{(\ref{condA})}{=} \sum\limits_{kk^{\prime}}^{P} \xi_{kk^{\prime}} \delta_{kk^{\prime}} \Tr \left(\varrho\right) + \sum\limits_{j\neq j^{\prime}}^{P} \sum\limits_{kk^{\prime}}^{P} \xi_{kk^{\prime}} \Tr \left(\varrho A_{j^{\prime}k^{\prime}}^{\dagger}A_{jk}\right) = 1 + \Tr \left(\varrho \sum\limits_{j\neq j^{\prime}}^{P} \sum\limits_{kk^{\prime}}^{P} \xi_{kk^{\prime}} A_{j^{\prime}k^{\prime}}^{\dagger}A_{jk}\right) \overset{!}{=} 1 + l_{\xi},
\end{align*}
where $\sum\limits_{j\neq j^{\prime}}$ denotes summation over both $j$ and $j^{\prime}$ but only for cases when $j\neq j^{\prime}$ and $l_{\xi} = k_{\xi} - 1$. We arrive at the condition, where both the program and the processor itself are instrumental:
\begin{align*}
	\sum\limits_{j\neq j^{\prime}}^{P} \sum\limits_{kk^{\prime}}^{P} \xi_{kk^{\prime}} A_{j^{\prime}k^{\prime}}^{\dagger}A_{jk} = l_{\xi} \mathbb{1}.
\end{align*}
Only if this condition holds, processor implements quantum channel. We shall define the set of channels implemented by probabilistic processor as:
\begin{align*}
	\mathfrak{C}_{G}^{pr} = \left\{\frac{1}{p} \Tr_{p} \left[G (\varrho \otimes \xi) G^{\dagger} \left(\mathbb{1} \otimes \dyad{\chi}\right)\right] \mid \xi \in {\cal{S}} ({\cal{H}}) \wedge \sum\limits_{j\neq j^{\prime}}^{P} \sum\limits_{kk^{\prime}}^{P} \xi_{kk^{\prime}} A_{j^{\prime}k^{\prime}}^{\dagger}A_{jk} = l_{\xi} \mathbb{1} \right\},
\end{align*}
where $l_{\xi} \in \mathbb{R}$. Elements of this set are expressed in the same way as in equation (\ref{elemOfImplPr}), only the choice of programs is limited. Let us now give an example of a processor unable to apply channel on a data state.
\begin{example}\label{noChannelProcessor}
	Processor under consideration:
	\begin{align*}
		G &= \frac{1}{\sqrt{2}}\left(\mathbb{1}\otimes\dyad{0}{0} + \sigma_{z}\otimes\dyad{1}{0} + \sigma_{x}\otimes\dyad{0}{1} - i\sigma_{y}\otimes\dyad{1}{1}\right),
	\end{align*}
	which implements:
	\begin{align*}
		\mathfrak{G}_{G,\xi}^{pr} &\overset{(\ref{elemOfImplPr})}{=} \frac{1}{2p} \left[\xi_{00} \sum_{j}^{2} A_{j0} \varrho \sum_{j^{\prime}}^{2} A_{j^{\prime}0}^{\dagger} + \xi_{01} \sum_{j}^{2} A_{j0} \varrho \sum_{j^{\prime}}^{2} A_{j^{\prime}1}^{\dagger} + \xi_{10} \sum_{j}^{2} A_{j1} \varrho \sum_{j^{\prime}}^{2} A_{j^{\prime}0}^{\dagger} + \xi_{11} \sum_{j}^{2} A_{j1} \varrho \sum_{j^{\prime}}^{2} A_{j^{\prime}1}^{\dagger} \right] \\
		&= \frac{1}{4p} [\xi_{00} \left(\mathbb{1} + \sigma_{z}\right) \varrho \left(\mathbb{1} + \sigma_{z}\right) + \xi_{01} \left(\mathbb{1} + \sigma_{z}\right) \varrho \left(\sigma_{x} + i\sigma_{y}\right) \\
		&\quad + \xi_{10} \left(\sigma_{x} - i\sigma_{y}\right) \varrho \left(\mathbb{1} + \sigma_{z}\right) + \xi_{11} \left(\sigma_{x} - i\sigma_{y}\right) \varrho \left(\sigma_{x} + i\sigma_{y}\right)] \\
		&= \frac{1}{2p} \left(\xi_{00} \dyad{0}{0}\varrho\dyad{0}{0} + \xi_{01} \dyad{0}{0}\varrho\dyad{0}{1} + \xi_{10} \dyad{1}{0}\varrho\dyad{0}{0} + \xi_{11} \dyad{1}{0}\varrho\dyad{0}{1}\right) = \frac{\varrho_{00}}{2p} \xi.
	\end{align*}
	If we calculate probability, we can see that it always depends on data state.
	\begin{align*}
		p = \frac{1}{2} \Tr \left(\xi \varrho_{00}\right) = \frac{1}{2} \varrho_{00}.
	\end{align*}
	Therefore, there exists no program state for which processor $G$ is able to implement a channel.
\end{example}

From a different point of view a processor implements a set of quantum operations.
\begin{align}\label{ProbQOp}
	\mathfrak{O}_{G}^{pr} = \left\{\Tr_{p} \left[G (\varrho \otimes \xi) G^{\dagger} \left(\mathbb{1} \otimes \dyad{\chi}\right)\right] \mid \xi \in \cal{S} ({\cal{H}}) \right\}.
\end{align}
We do not normalize the result with probability $p$, therefore we recover subnormalized quantum states on the output. Elementwise this translates to:
\begin{align*} \label{SprobQPelem}
	\mathfrak{O}_{G,\xi}^{pr} = \frac{1}{P}\sum\limits_{jkj^{\prime}k^{\prime}}^{P} \xi_{kk^{\prime}} A_{jk} \varrho A_{j{\prime}k^{\prime}}^{\dagger}. \numberthis
\end{align*}

\section{Equivalence of Quantum Processors}

Basic idea while considering equivalence of quantum processors is to find different processors able to transform input data in the same manner. Therefore, equivalent processors should be able to implement the same transformations on input states. However, for probabilistic case one must consider different scenarios and thus we are using multiple definitions of probabilistic equivalence. Nonetheless, let us firstly begin with a simpler task of defining deterministic equivalence.
\begin{definition}\label{detEqv}
	Equivalence of Deterministic Processors
	\\
	Two quantum deterministic processors $G$ and $\widetilde{G}$ are equivalent if $\mathfrak{C}_{G}^{det} = \mathfrak{C}_{\widetilde{G}}^{det}$. We denote deterministic equivalence with $G \sim_{det} \widetilde{G}$.
\end{definition}
In order for deterministic quantum processors to be equivalent, they must be able to implement the same set of channels. 

In case of probabilistic processor, one also has to take into account the probability of implementing individual channels. In case of equal probabilities, we define the strong equivalence.
\begin{definition}\label{strongEqv}
	Strong Equivalence of Probabilistic Processors
	\\
	Two quantum probabilistic processors $G$ and $\widetilde{G}$ with probabilities of successfully implementing the desired channels are $p$ and $\widetilde{p}$ respectively, are strongly equivalent if $\mathfrak{C}_{G}^{pr} = \mathfrak{C}_{\widetilde{G}}^{pr}$ and $p = \widetilde{p}$ holds for all possible channels implemented by given processors. We denote strong probabilistic equivalence with $G \sim_{pr} \widetilde{G}$.
\end{definition}
Here, we are not considering measurements that processors are able to execute, only channels they are able to implement. Therefore, strongly equivalent probabilistic processors implement the same channels with the same probabilities. For the differing probabilities, we define weak equivalence.
\begin{definition} \label{weakEqv}
	Weak Equivalence of Probabilistic Processors
	\\
	Two quantum probabilistic processors $G$ and $\widetilde{G}$ with probabilities of successfully implementing the desired channels are $p$ and $\widetilde{p}$ respectively, are weakly equivalent if $\mathfrak{C}_{G}^{pr} = \mathfrak{C}_{\widetilde{G}}^{pr}$ and probabilities $p$ and $\widetilde{p}$ are not always equal. We denote weak probabilistic equivalence with $G \approx_{pr} \widetilde{G}$.
\end{definition}
This means that both processors implement the same channels, but they might do it with different probabilities $p$ and $\widetilde{p}$. It is worth noting, that probabilities sometimes might be equal for certain programs, while they might differ for others, all the while still implementing the same channel (as might be noted in the example \ref{diffProbEx} further down). These probabilities depend not only on processors, but also on input program states as can be seen from equation (\ref{probQP}). We shall also define structural equivalence where we are not looking at the probabilities, but we are only interested in quantum operations that processors are able to realize.
\begin{definition}\label{structEqv}
	Structural Equivalence of Probabilistic Processors
	\\
	Two quantum probabilistic processors $G$ and $\widetilde{G}$ are structurally equivalent if $\mathfrak{O}_{G}^{pr} = K_{\xi,\widetilde{\xi}} \mathfrak{O}_{\widetilde{G}}^{pr}$, where $K_{\xi,\widetilde{\xi}} \in \mathbb{R}_{>0}$ depends on the particular program states $\xi$ and $\widetilde{\xi}$ of the respective processors $G$ and $\widetilde{G}$. We denote structural equivalence with $G \sim_{st} \widetilde{G}$.
\end{definition}
Equation $\mathfrak{O}_{G}^{pr} = K_{\xi,\widetilde{\xi}} \mathfrak{O}_{\widetilde{G}}^{pr}$ conveys that for every element from the set $\mathfrak{O}_{G}^{pr}$, there exists an element from $\mathfrak{O}_{\widetilde{G}}^{pr}$ that is equal, except for the multiplicative parameter $K_{\xi,\widetilde{\xi}}$ and vice versa. 

One can also consider having flexibility in measurements and allowing for two different processors to have different success measurements. Then, one can consider not only equivalence between processors but also between couples consisting of processor and measurement. However, in present work we shall not concern with such thoughts.

\subsection{Equivalence of Deterministic Processors}

Let us begin with the investigation of equivalence of deterministic processors as defined in the definition \ref{detEqv}. Due to decomposition of mixed quantum states into pure states, it is enough that we only consider pure program states (usually denoted with a Greek letter $\xi$). We state necessary and sufficient condition for the equivalence of deterministic processors for specific relations between the processors:
\begin{theorem}\label{SufNecCondDetProc}
	Sufficient and Necessary Condition
	\\
	Let us have quantum deterministic processor $G = \sum_{jk}^{P} A_{jk} \otimes \dyad{j}{k}$. Let us also assume that:
	\begin{itemize}
		\item[$1)$] a different processor $G_{L}$ can be expressed as $G_{L} = UG$, where $U = \sum_{rq}^{P} U_{rq} \otimes \dyad{r}{q}$ is a unitary operator. Then processors $G$ and $G_{L}$ are deterministically equivalent $G \sim_{det} G_{L}$ if and only if the following equation holds true:
		\begin{align*}
			\sum_{jk}^{P} \xi_{k} U_{rj} A_{jk} = \sum_{ikq}^{P} w_{ri} y_{kq} \xi_{k} A_{ik},
		\end{align*}
		where $\xi$ is a program state of $G_{L}$ and $\{w_{ri}\}$ and $\{y_{kq}\}$ form unitary operators,
		
		\item[$2)$] a different processor $G_{R}$ can be expressed as $G_{R} = GV$, where $V = \sum_{rq}^{P} V_{rq} \otimes \dyad{r}{q}$ is a unitary operator. Then processors $G$ and $G_{R}$ are deterministically equivalent $G \sim_{det} G_{R}$ if and only if the following equation holds true:
		\begin{align*}
			\sum_{kq}^{P} \xi_{q} A_{jk} V_{kq} = \sum_{jkq}^{P} w_{ji} y_{kq} \xi_{q} A_{ik},
		\end{align*}
		where $\xi$ is a program state of $G_{R}$ and $\{w_{ri}\}$ and $\{y_{kq}\}$ form unitary operators.
	\end{itemize}
\end{theorem}
\begin{proof}
	\textbf{Freedom in Program States}
	
	As was already mentioned, it is enough to limit ourselves to pure states due to the decomposition of mixed states into pure ones. Relation between arbitrary two pure states $\xi$ and $\widetilde{\xi}$ is always given through unitary transformation $\widetilde{\xi} = Y^{p} \xi Y^{p\dagger}$, where $Y^{p} = \sum_{rq}^{P} y_{rq} \dyad{r}{q}$ is a unitary operator. Let us calculate what a processor $G$ with program $\widetilde{\xi}$ implements:
	\begin{align*} \label{prg}
		\mathfrak{C}^{det}_{G, \widetilde{\xi}} 
		&= \Tr_{p} \left[G \left(\varrho \otimes \widetilde{\xi}\right) G^{\dagger}\right] \\
		&= \Tr_{p} \left[G \left(\varrho \otimes Y^{p} \xi Y^{p\dagger}\right) G^{\dagger}\right] \\
		&= \Tr_{p} \left[\left(\sum_{jk}^{P} A_{jk} \otimes \dyad{j}{k}\right) \left(\varrho \otimes \sum_{rq}^{P} y_{rq} \dyad{r}{q} \xi \sum_{r^{\prime}q^{\prime}}^{P} y_{r^{\prime}q^{\prime}}^{\ast} \dyad{q^{\prime}}{r^{\prime}}\right) \left(\sum_{j^{\prime}k^{\prime}}^{P} A_{j^{\prime}k^{\prime}}^{\dagger} \otimes \dyad{k^{\prime}}{j^{\prime}}\right)\right] \\
		&= \sum_{jj^{\prime}kk^{\prime}}^{P} \sum_{rr^{\prime}qq^{\prime}}^{P} A_{jk} \varrho A_{j^{\prime}k^{\prime}}^{\dagger} y_{rq} y_{r^{\prime}q^{\prime}}^{\ast} \Tr\left(\dyad{j}{k}\hspace*{-1.5mm}\dyad{r}{q} \xi \dyad{q^{\prime}}{r^{\prime}}\hspace*{-1.5mm}\dyad{k^{\prime}}{j^{\prime}}\right) \\
		&= \sum_{jkk^{\prime}}^{P} \sum_{qq^{\prime}}^{P} \xi_{qq^{\prime}} y_{kq} y_{k^{\prime}q^{\prime}}^{\ast} A_{jk} \varrho A_{jk^{\prime}}^{\dagger} \\
		&= \sum_{j}^{P} \left(\sum_{kq}^{P} \xi_{q} y_{kq} A_{jk}\right) \varrho \left(\sum_{k^{\prime}q^{\prime}}^{P} \xi_{q^{\prime}}^{\ast} y_{k^{\prime}q^{\prime}}^{\ast} A_{j^{\prime}k^{\prime}}^{\dagger}\right) \\
		&= \sum_{j}^{P} \textphnc{\Aaleph}_{j} \varrho \textphnc{\Aaleph}_{j}^{\dagger}. \numberthis
	\end{align*}
	
	Let us show that $\{\textphnc{\Aaleph}_{j}\}$ are, indeed, Kraus operators, which will prove very useful on many occasions due to proposition \ref{equalityOfChannels}:
	\begin{align*}
		\sum_{j}^{P} \textphnc{\Aaleph}_{j}^{\dagger} \textphnc{\Aaleph}_{j} &= \sum_{j}^{P} \left(\sum_{kq}^{P} \xi_{q}^{\ast} y_{kq}^{\ast} A_{jk}^{\dagger}\right) \left(\sum_{k^{\prime}q^{\prime}}^{P} \xi_{q^{\prime}} y_{k^{\prime}q^{\prime}} A_{jk^{\prime}}\right) = \sum_{jkk^{\prime}}^{P} \sum_{qq^{\prime}}^{P} \xi_{q^{\prime}q} y_{kq}^{\ast} y_{k^{\prime}q^{\prime}} A_{jk}^{\dagger} A_{jk^{\prime}} \\ &\overset{(\ref{condA})}{=} \sum_{kqq^{\prime}}^{P} \xi_{q^{\prime}q} y_{kq}^{\ast} y_{kq^{\prime}} \mathbb{1} \overset{(i)}{=} \sum_{q}^{P} \xi_{qq} \mathbb{1} = \mathbb{1},
	\end{align*}
	where in $(i)$ we are using unitarity of $Y^{p}$.
	
	\textbf{Left Transformation}
	\\
	Let us begin with the first case - transformation $U = \sum_{rq}^{P} U_{rq} \otimes \dyad{r}{q}$ applied on processor $G$ from the left. Conditions on operators forming $U$ are the same as conditions on operators forming processor, explicitly $\sum_{r}^{P} U_{rq}^{\dagger} U_{rq^{\prime}} = \mathbb{1} \delta_{qq^{\prime}}$ and $\sum_{q}^{P} U_{rq} U_{r^{\prime}q}^{\dagger} = \mathbb{1} \delta{rr^{\prime}}$. Now we shall proceed with the calculation of channel implemented by the transformed processor $G_{L} = UG$:
	\begin{align*} \label{leftTransformation}
		\mathfrak{C}_{G_{L}, \xi}^{det} 
		&= \Tr_{p} \left[UG \left(\varrho \otimes \xi\right) G^{\dagger}U^{\dagger}\right] \\
		&= \Tr_{p} \left[\left(\sum_{rq}^{P} U_{rq} \otimes \dyad{r}{q}\right) \left(\sum_{jk}^{P} A_{jk} \otimes \dyad{j}{k}\right) \left(\varrho \otimes \xi\right) \left(\sum_{j^{\prime}k^{\prime}}^{P} A_{j^{\prime}k^{\prime}}^{\dagger} \otimes \dyad{k^{\prime}}{j^{\prime}}\right) \left(\sum_{r^{\prime}q^{\prime}}^{P} U_{r^{\prime}q^{\prime}}^{\dagger} \otimes \dyad{q^{\prime}}{r^{\prime}}\right)\right] \\
		&= \sum_{jj^{\prime}kk^{\prime}}^{P} \sum_{rr^{\prime}qq^{\prime}}^{P} U_{rq} A_{jk} \varrho A_{j^{\prime}k^{\prime}}^{\dagger} U_{r^{\prime}q^{\prime}}^{\dagger} \Tr\left(\dyad{r}{q}\hspace*{-1.5mm}\dyad{j}{k} \xi \dyad{k^{\prime}}{j^{\prime}}\hspace*{-1.5mm}\dyad{q^{\prime}}{r^{\prime}}\right) \\
		&= \sum_{r}^{P} \left( \sum_{jk}^{P} \xi_{k} U_{rj} A_{jk}\right) \varrho \left( \sum_{j^{\prime}k^{\prime}}^{P} \xi_{k^{\prime}}^{\ast} A_{j^{\prime}k^{\prime}}^{\dagger} U_{rj^{\prime}}^{\dagger} \right) \\
		&= \sum_{r}^{P} \textphnc{\ARdaleth}_{r} \varrho \textphnc{\ARdaleth}_{r}^{\dagger}. \numberthis
	\end{align*}
	Let us show that operators $\{\textphnc{\ARdaleth}_{r}\}$ are, in reality, Kraus operators:
	\begin{align*}
		\sum_{r}^{P} \textphnc{\ARdaleth}_{r}^{\dagger} \textphnc{\ARdaleth}_{r} &= \sum_{r}^{P} \left(\sum_{jk}^{P} \xi_{k}^{\ast} A_{jk}^{\dagger} U_{rj}^{\dagger}\right) \left(\sum_{j^{\prime}k^{\prime}}^{P} \xi_{k^{\prime}} U_{rj^{\prime}} A_{j^{\prime}k^{\prime}}\right) = \sum_{r}^{P} \sum_{jj^{\prime}kk^{\prime}}^{P} \xi_{k^{\prime}k} A_{jk}^{\dagger} U_{rj}^{\dagger} U_{rj^{\prime}} A_{j^{\prime}k^{\prime}} \\ &\overset{(i)}{=} \sum_{jkk^{\prime}}^{P} \xi_{k^{\prime}k} A_{jk}^{\dagger} A_{jk^{\prime}} \overset{(\ref{condA})}{=} \sum_{k}^{P} \xi_{kk} \mathbb{1} = \mathbb{1},
	\end{align*}
	where in $(i)$ we have used conditions on operators $U_{rq}$ derived from unitarity of $U$. 
	
	By comparing (\ref{leftTransformation}) with (\ref{prg}) we arrive at the equation:
	\begin{align*}
		\sum_{r}^{P} \textphnc{\ARdaleth}_{r} \varrho \textphnc{\ARdaleth}_{r}^{\dagger} \overset{!}{=} \sum_{j}^{P} \textphnc{\Aaleph}_{j} \varrho \textphnc{\Aaleph}_{j}^{\dagger}.
	\end{align*}
	From proposition (\ref{equalityOfChannels}), we know that two sets of Kraus operators $\{\textphnc{\ARdaleth}_{r}\}$ and $\{\textphnc{\Aaleph}_{j}\}$ define the same quantum channel if and only if the following condition holds:
	\begin{align*}
		\textphnc{\ARdaleth}_{r} = \sum_{j}^{P} w_{rj} \textphnc{\Aaleph}_{j}, \qquad \sum_{r}^{P} w_{rj}^{\ast} w_{rj^{\prime}} = \delta_{jj^{\prime}}.
	\end{align*}
	From this equation we recover the condition given in the theorem:
	\begin{align*}
		\sum_{jk}^{P} \xi_{k} U_{rj} A_{jk} = \sum_{jkq}^{P} w_{rj} y_{kq} \xi_{q} A_{jk}.
	\end{align*}
	
	\textbf{Right Transformation}
	By applying transformation $V = \sum_{rq}^{P} v_{rq} \otimes \dyad{r}{q}$, with $\sum_{r}^{P} V_{rq}^{\dagger} V_{rq^{\dagger}} = \delta_{qq^{\prime}} \mathbb{1}$ and $\sum_{q}^{P} V_{rq} V_{r^{\prime}q}^{\dagger} = \delta_{rr^{\prime}} \mathbb{1}$, from the right side of the processor $G$, we arrive at the following implemented channel:
	\begin{align*} \label{rightTransformation}
		\mathfrak{C}_{G_{R}, \xi}^{\det} &= \Tr_{p} \left[GV \left(\varrho \otimes \xi\right) V^{\dagger} G^{\dagger}\right] \\
		&= \Tr_{p} \left[\left(\sum_{jk}^{P} A_{jk} \otimes \dyad{j}{k}\right) \left(\sum_{rq}^{P} V_{rq} \otimes \dyad{r}{q}\right) \left(\varrho \otimes \xi\right) \left(\sum_{r^{\prime}q^{\prime}}^{P} V_{r^{\prime}q^{\prime}}^{\dagger} \otimes \dyad{q^{\prime}}{r^{\prime}}\right) \left(\sum_{j^{\prime}k^{\prime}}^{P} A_{j^{\prime}k^{\prime}}^{\dagger} \otimes \dyad{k^{\prime}}{j^{\prime}}\right)\right] \\
		&= \sum_{jj^{\prime}kk^{\prime}}^{P} \sum_{rr^{\prime}qq^{\prime}}^{P} A_{jk} V_{rq} \varrho V_{r^{\prime} q^{\prime}}^{\dagger} A_{j^{\prime} k^{\prime}}^{\dagger} \Tr\left(\dyad{j}{k}\hspace*{-1.5mm}\dyad{r}{q} \xi \dyad{q^{\prime}}{r^{\prime}}\hspace*{-1.5mm}\dyad{k^{\prime}}{j^{\prime}}\right) \\
		&= \sum_{j}^{P} \left(\sum_{kq}^{P} \xi_{q} A_{jk} V_{kq}\right) \varrho \left(\sum_{k^{\prime}q^{\prime}}^{P} \xi_{q^{\prime}}^{\ast} V_{k^{\prime}q^{\prime}}^{\dagger} A_{jk^{\prime}}^{\dagger}\right) \\
		&= \sum_{j}^{P} \textphnc{\ARgimel}_{j} \varrho \textphnc{\ARgimel}_{j}^{\dagger}. \numberthis
	\end{align*}
	Let us show that the operators $\{\textphnc{\ARgimel}_{j}\}$ are Kraus operators:
	\begin{align*}
		\sum_{j}^{P} \textphnc{\ARgimel}_{j}^{\dagger} \textphnc{\ARgimel}_{j} &= \sum_{j}^{P} \left(\sum_{kq}^{P} \xi_{q}^{\ast} V_{kq}^{\dagger} A_{jk}^{\dagger}\right) \left(\sum_{k^{\prime}q^{\prime}}^{P} \xi_{q^{\prime}} A_{jk^{\prime}} V_{k^{\prime}q^{\prime}}\right) = \sum_{jkk^{\prime}}^{P} \sum_{qq^{\prime}}^{P} \xi_{q^{\prime}q} V_{kq}^{\dagger} A_{jk}^{\dagger} A_{jk^{\prime}} V_{k^{\prime}q^{\prime}} \\
		&\overset{(\ref{condA})}{=} \sum_{kqq^{\prime}}^{P} \xi_{qq^{\prime}} V_{kq}^{\dagger} V_{kq^{\prime}} \overset{(i)}{=} \sum_{q}^{P} \xi_{qq} \mathbb{1} = \mathbb{1},
	\end{align*}
	where in $(i)$ we have used unitarity of $V$.
	
	Comparing equations (\ref{rightTransformation}) and (\ref{prg}), we receive the following equation:
	\begin{align*}
		\sum_{j}^{P} \textphnc{\ARgimel}_{j} \varrho \textphnc{\ARgimel}_{j}^{\dagger} \overset{!}{=} \sum_{i}^{P} \textphnc{\Aaleph}_{i} \varrho \textphnc{\Aaleph}_{i}^{\dagger}.
	\end{align*}
	Considering that both sides of the equation are expressed using Kraus operators, we can use proposition \ref{equalityOfChannels} and eventually recover condition used in the theorem:
	\begin{align*}
		\textphnc{\ARgimel}_{j} &= \sum_{i}^{p} w_{ji} \textphnc{\Aaleph}_{i} \\
		\sum_{kq}^{P} \xi_{q} A_{jk} V_{kq} &= \sum_{ikq}^{P} w_{ji} y_{kq} \xi_{q} A_{ik},
	\end{align*}
	with additional restriction being $\sum_{j}^{P} w_{ji}^{\ast} w_{ji^{\prime}} = \delta_{ii^{\prime}}$.
\end{proof}

\subparagraph{Concrete Solutions}
In the current section we shall reveal some more concrete expressions for unitaries $U$ and $V$. Let us start with the transformation from the left-hand side of the processor.
\begin{theorem}\label{concreteSolutionDetProcLeft}
	Sufficient Condition for Transformation from Left
	\\
	Let us have quantum deterministic processor $G = \sum_{jk}^{P} A_{jk} \otimes \dyad{j}{k}$. Then processor $G_{L} = UG$, where $U$ is unitary operator, and $G$ are deterministically equivalent $G \sim_{det} G_{L}$ if:
	\begin{itemize}
		\item[1)] $U = \mathbb{1} \otimes U^{p}$,
		\item[2)] $U = \sum_{r}^{P} U_{r} \otimes \dyad{r}{r}$, where $U_{r} = \left(\sum_{jkq}^{P} w_{rj} y_{kq} \xi_{k} A_{jk}\right) \left(\sum_{l}^{P} \xi_{l} A_{rl}\right)^{-1}$.
	\end{itemize}
\end{theorem}
\begin{proof}
	Let us calculate the first case:
	\begin{align*}
		\mathfrak{C}_{G_{L},\xi}^{\det} &= \Tr_{p} \left[UG \left(\varrho \otimes \xi\right) G^{\dagger}U^{\dagger}\right] \\
		&= \Tr_{p} \left[\left(\mathbb{1} \otimes U^{p}\right) G \left(\varrho \otimes \xi\right) G^{\dagger} \left(\mathbb{1} \otimes U^{p\dagger}\right)\right] \\
		&\overset{(i)}{=} \Tr_{p} \left[G \left(\varrho \otimes \xi\right) G^{\dagger}\right] = \mathfrak{C}_{G,\xi}^{\det}.
	\end{align*}
	In $(i)$ we have used cyclic property of trace.
	
	Now we shall continue with the second solution. As is tradition, we begin by calculating the channel that is implemented by a processor $G_{L} = UG$:
	\begin{align*} \label{leftTransformationDiagonal}
		\mathfrak{C}_{G_{L}, \xi}^{\det} &= \Tr_{p} \left[\left(\sum_{r}^{P} U_{r} \otimes \dyad{r}{r}\right) \left(\sum_{jk}^{P} A_{jk} \otimes \dyad{j}{k}\right) \left(\varrho \otimes \xi\right) \left(\sum_{j^{\prime}k^{\prime}}^{P} A_{j^{\prime}k^{\prime}}^{\dagger} \otimes \dyad{k^{\prime}}{j^{\prime}}\right) \left(\sum_{r^{\prime}}^{P} U_{r^{\prime}}^{\dagger} \otimes \dyad{r^{\prime}}{r^{\prime}}\right)\right] \\
		&= \sum_{rr^{\prime}}^{P} \sum_{jj^{\prime}kk^{\prime}}^{P} U_{r} A_{jk} \varrho A_{j^{\prime}k^{\prime}}^{\dagger} U_{r^{\prime}}^{\dagger} \Tr\left(\dyad{r}{r}\hspace*{-1.5mm}\dyad{j}{k} \xi \dyad{k^{\prime}}{j^{\prime}}\hspace*{-1.5mm}\dyad{r^{\prime}}{r^{\prime}}\right) \\
		&= \sum_{r}^{P} \left(U_{r} \sum_{k}^{P}\xi_{k} A_{rk}\right) \varrho \left(\sum_{k^{\prime}}^{P} \xi_{k^{\prime}}^{\ast} A_{rk^{\prime}}^{\dagger} U_{r}^{\dagger}\right) \\
		&= \sum_{r}^{P} \textphnc{\ARvav}_{r} \varrho \textphnc{\ARvav}_{r}^{\dagger}. \numberthis
 	\end{align*}
 	Let us show that $\{\textphnc{\ARvav}_{r}\}$ are Kraus operators:
 	\begin{align*}
 		\sum_{r}^{P} \textphnc{\ARvav}_{r}^{\dagger} \textphnc{\ARvav}_{r} = \sum_{r}^{P} \left(\sum_{k}^{P} \xi_{k} A_{rk}^{\dagger} U_{r}^{\dagger}\right) \left(U_{r} \sum_{k^{\prime}}^{P} \xi_{k^{\prime}} A_{rk^{\prime}}\right) \overset{(i)}{=} \sum_{rkk^{\prime}}^{P} \xi_{k^{\prime}k} A_{rk}^{\dagger} A_{rk^{\prime}} \overset{(\ref{condA})}{=} \sum_{k}^{P} \xi_{kk} \mathbb{1} = \mathbb{1}.
 	\end{align*}
 	In $(i)$ we have used the unitarity of $U_{r}$ which stems from unitarity of $U$. Let us now compare equations (\ref{leftTransformationDiagonal}) and (\ref{prg}):
 	\begin{align*}
 		\sum_{r}^{P} \textphnc{\ARvav}_{r} \varrho \textphnc{\ARvav}_{r}^{\dagger} = \sum_{j}^{P} \textphnc{\Aaleph}_{j} \varrho \textphnc{\Aaleph}_{j}.
 	\end{align*}
 	We are once again going to use proposition \ref{equalityOfChannels} and discover the condition:
 	\begin{align*} \label{concreteSolutionLeft}
 		\textphnc{\ARvav}_{r} &= \sum_{j}^{P} w_{rj} \textphnc{\Aaleph}_{j} \\
 		\sum_{l}^{P} \xi_{l} U_{r} A_{rl} &= \sum_{jkq}^{P} w_{rj} y_{kq} \xi_{q} A_{jk} \\
 		U_{r} &= \left(\sum_{jkq}^{P} w_{rj} y_{kq} \xi_{q} A_{jk}\right) \left(\sum_{l}^{P} \xi_{l} A_{rl}\right)^{-1}. \numberthis
 	\end{align*}
 	For completeness, let us also write conditions arising from unitarity of $U_{r}$:
 	\begin{align*}
 		U_{r}^{\dagger} U_{r} &= \left(\sum_{l}^{P} \xi_{l}^{\ast} A_{rl}^{\dagger}\right)^{-1} \left(\sum_{j}^{P} w_{rj}^{\ast} \textphnc{\Aaleph}_{j}^{\dagger}\right) \left(\sum_{j^{\prime}}^{P} w_{rj^{\prime}} \textphnc{\Aaleph}_{j^{\prime}}\right) \left(\sum_{l^{\prime}}^{P} \xi_{l^{\prime}} A_{rl^{\prime}}\right)^{-1} \overset{!}{=} \mathbb{1}, 
 		\\
 		U_{r} U_{r}^{\dagger} &=\left(\sum_{j}^{P} w_{rj} \textphnc{\Aaleph}_{j}\right) \left(\sum_{l}^{P} \xi_{l} A_{rl}\right)^{-1} \left(\sum_{l^{\prime}}^{P} \xi_{l^{\prime}}^{\ast} A_{rl^{\prime}}^{\dagger}\right)^{-1} \left(\sum_{j^{\prime}}^{P} w_{rj^{\prime}}^{\ast} \textphnc{\Aaleph}_{j^{\prime}}^{\dagger}\right) \overset{!}{=} \mathbb{1}.
 	\end{align*}
\end{proof}

And now we shall continue with the transformation from the right-hand side of the processor.
\begin{theorem} \label{concreteSolutionDetProcRight}
	Sufficient Condition for Transformation from Right
	\\
	Let us have deterministic quantum processor $G = \sum_{jk}^{P} A_{jk} \otimes \dyad{j}{k}$. Then processor $G_{R} = GV$, where $V$ is unitary operator, and $G$ are deterministically equivalent $G \sim_{det} G_{R}$ if:
	\begin{itemize}
		\item[1)] $V = \left(\mathbb{1} \otimes V^{p}\right)$,
		\item[2)] $V = V^{d} \otimes \mathbb{1}$, where $V^{d} = \left(\sum_{l}^{P} \xi_{l} A_{jl}\right)^{-1} \left(\sum_{ikq}^{P} w_{ji} y_{kq} \xi_{q} A_{ik}\right)$ for all $j$.
	\end{itemize}
\end{theorem}
\begin{proof}
	We shall start with the first case and calculate the implemented channel.
	\begin{align*}
		\mathbb{C}_{G_{R},\xi}^{\det} &= \Tr_{p} \left[GV \left(\varrho \otimes \xi\right) V^{\dagger}G^{\dagger}\right] \\
		&= \Tr_{p} \left[G \left(\mathbb{1} \otimes V^{p}\right) \left(\varrho \otimes \xi\right) \left(\mathbb{1} \otimes V^{p\dagger}\right) G^{\dagger}\right] \\
		&= \Tr_{p} \left[G \left(\varrho \otimes V^{p}\xi V^{p\dagger}\right) G^{\dagger}\right] = \mathbb{C}_{G, \widetilde{\xi}}^{\det}.
	\end{align*}
	Processor $G_{R} = GV$ with program state $\xi$ implements precisely the same channel as processor $G$ with program $\widetilde{\xi} = V^{p} \xi V^{p\dagger}$.
	
	In the second case, we shall again, rather unsurprisingly, calculate the implemented channel by processor $G_{R} = GV = G(V^{d} \otimes \mathbb{1})$:
	\begin{align*} \label{rightTransformationDiagonal}
		\mathfrak{C}_{G_{R}, \xi}^{det} &= \Tr_{p} \left[\left(\sum_{jk}^{P} A_{jk} \otimes \dyad{j}{k}\right) \left(V^{d} \otimes \mathbb{1}\right) \left(\varrho \otimes \xi\right) \left(V^{d\dagger} \otimes \mathbb{1}\right) \left(\sum_{j^{\prime}k^{\prime}}^{P} A_{j^{\prime}k^{\prime}}^{\dagger} \otimes \dyad{k^{\prime}}{j^{\prime}}\right)\right] \\
		&= \sum_{j}^{P} \left(\sum_{k}^{P} \xi_{k} A_{jk} V^{d}\right) \varrho \left(V^{d\dagger} \sum_{k^{\prime}}^{P} \xi_{k^{\prime}}^{\ast} A_{jk^{\prime}}^{\dagger}\right) \\
		&= \sum_{j}^{P} \textphnc{\ARhe}_{j} \varrho \textphnc{\ARhe}_{j}^{\dagger}. \numberthis
	\end{align*}
	Once again, let us show that these are Kraus operators:
	\begin{align*}
		\sum_{j}^{P} \textphnc{\ARhe}_{j}^{\dagger} \textphnc{\ARhe}_{j} = \sum_{j}^{P} \left(V^{d\dagger} \sum_{k^{\prime}}^{P} \xi_{k}^{\ast} A_{jk^{\prime}}^{\dagger}\right) \left(\sum_{k}^{P} \xi_{k} A_{jk} V^{d}\right) = \sum_{j}^{P} \sum_{kk^{\prime}}^{P} \xi_{kk^{\prime}} V^{d\dagger} A_{jk}^{\dagger} A_{jk^{\prime}}^{\dagger} V^{d} = \sum_{k}^{P} \xi_{kk} V^{d\dagger} V^{d} \overset{(\ref{condA})}{=} \mathbb{1},
	\end{align*}
	where we have used unitarity of $V^{d}$. And now is the time to compare equation (\ref{rightTransformationDiagonal}) with (\ref{prg}):
	\begin{align*}
		\sum_{j}^{P} \textphnc{\ARhe}_{j} \varrho \textphnc{\ARhe}_{j}^{\dagger} \overset{!}{=} \sum_{i}^{P} \textphnc{\Aaleph}_{i} \varrho \textphnc{\Aaleph}_{i}^{\dagger}.
	\end{align*}
	Again, using proposition \ref{equalityOfChannels}, we recover:
	\begin{align*}
		\textphnc{\ARhe}_{j} &= \sum_{i}^{P} w_{ji} \textphnc{\Aaleph}_{i} \\
		\sum_{l} \xi_{l} A_{jl} V^{d} &= \sum_{ikq}^{P} w_{ji} y_{kq} \xi_{q} A_{ik} \\
		V^{d} &= \left(\sum_{l}^{P} \xi_{l} A_{jl}\right)^{-1} \left(\sum_{ikq}^{P} w_{ji} y_{kq} \xi_{q} A_{ik}\right).
	\end{align*}
	Previous equation has to be valid for all $j$'s. For the sake of completness, we shall also explicitly write the condition of unitarity:
	\begin{align*}
		V^{d\dagger}V^{d} &= \left(\sum_{i}^{P} w_{ji}^{\ast} \textphnc{\Aaleph}_{i}^{\dagger}\right) \left(\sum_{l}^{P} \xi_{l}^{\ast} A_{jl}^{\dagger}\right)^{-1} \left(\sum_{l^{\prime}}^{P} \xi_{l^{\prime}} A_{jl^{\prime}}\right)^{-1} \left(\sum_{i^{\prime}}^{P} w_{ji^{\prime}} \textphnc{\Aaleph}_{i^{\prime}}
		\right) \overset{!}{=} \mathbb{1}\\
		V^{d}V^{d\dagger} &= \left(\sum_{l^{\prime}}^{P} \xi_{l^{\prime}} A_{jl^{\prime}}\right)^{-1} \left(\sum_{i^{\prime}}^{P} w_{ji^{\prime}} \textphnc{\Aaleph}_{i^{\prime}}
		\right) \left(\sum_{i}^{P} w_{ji}^{\ast} \textphnc{\Aaleph}_{i}^{\dagger}\right)\left(\sum_{l^{\prime}}^{P} \xi_{l^{\prime}} A_{jl^{\prime}}\right)^{-1} \overset{!}{=} \mathbb{1}.
	\end{align*}
\end{proof}

We shall continue with simple example for the second solution from theorem \ref{concreteSolutionDetProcLeft}:
\begin{example}
	Let us have two processors:
	\begin{align*}
		G = \mathbb{1} \otimes \dyad{0}{0} + \sigma_{z} \otimes \dyad{1}{1},\qquad
		\widetilde{G} = \sigma_{z} \otimes \dyad{0}{0} + \mathbb{1} \otimes \dyad{1}{1}.
	\end{align*}
	These processors are clearly equivalent as they implement the following channels:
	\begin{align*}
		\mathfrak{C}_{G,\xi}^{det} \overset{(\ref{elemOfImplDet})}{=} \xi_{00} \mathbb{1} \varrho \mathbb{1} + \xi_{11} \sigma_{z} \varrho \sigma_{z},\qquad
		\mathfrak{C}_{\widetilde{G},\widetilde{\xi}}^{det} \overset{(\ref{elemOfImplDet})}{=} \widetilde{\xi}_{00} \sigma_{z} \varrho \sigma_{z} + \widetilde{\xi}_{11} \mathbb{1} \varrho \mathbb{1}.
	\end{align*}
	Relation between the processors is given in the following manner:
	\begin{align*}
		\widetilde{G} = \left(\sigma_{z} \otimes \dyad{0}{0} + \sigma_{z} \otimes \dyad{1}{1}\right) G = \left(\sigma_{z} \otimes \mathbb{1}\right) G.
	\end{align*}
	From the mentioned theorem \ref{concreteSolutionDetProcLeft}, we see that the transformation is given by equations:
	\begin{align*}
		U &= \sum_{r}^{P} U_{r} \otimes \dyad{r}{r} \\
		U_{r} &= \left(\sum_{jkq}^{P} w_{rj} y_{kq} \xi_{k} A_{jk}\right) \left(\sum_{l}^{P} \widetilde{\xi}_{l} A_{rl}\right)^{-1}.
	\end{align*}
	Therefore $U_{0} = U_{1} = \sigma_{z}$. Let us retroactively calculate operator $U_{0}$. Firstly, relation between programs is given by unitary transformation $\widetilde{\xi} = Y \xi Y^{\dagger} = \sigma_{x} \xi \sigma_{x}$.For matrix $Y$, we can also explicitly write particular elements $y_{01} = y_{10} = 1$ and $y_{00} = y_{11} = 0$. Elements $w_{rj}$ shall be obtained from relation between Kraus operators of the processors. Kraus operators for processor $G$ are $\textphnc{\Aaleph}_{j} = \sum_{k}^{2} \xi_{k} A_{jk}$, where $j = \{0,1\}$ and Kraus operators for processor $\widetilde{G}$ are $\textphnc{\Abeth}_{m} = \sum_{n}^{2} \widetilde{\xi}_{n} B_{mn}$, where $m = \{0,1\}$. Relation between such operators is given by equation $\textphnc{\Abeth}_{m} = \sum_{j}^{2} w_{mj} \textphnc{\Aaleph}_{j}$ with $\sum_{m}^{2} w_{mj}^{\ast}  w_{mj^{\prime}}= \delta_{jj^{\prime}}$. We shall only need first operator $\textphnc{\Abeth}_{0} = w_{00} \textphnc{\Aaleph}_{0} + w_{01} \textphnc{\Aaleph}_{1}$, or equivalently $\widetilde{\xi}_{0} \sigma_{z} = w_{00} \xi_{0} \mathbb{1} + w_{01} \xi_{1} \sigma_{z}$. Therefore $w_{00} = 0$ and $w_{01} = 1$. Now, we can put all these things together:
	\begin{align*}
		U_{0} = \left(\sum_{jkq}^{2} w_{0j} y_{kq} \xi_{k} A_{jk}\right) \left(\sum_{l}^{2} \widetilde{\xi}_{l} A_{0l}\right)^{-1} = \widetilde{\xi}_{0}^{-1} w_{01} y_{10} \xi_{1} A_{11} = \widetilde{\xi}_{0}^{-1} w_{01} \xi_{1} \sigma_{z} = \sigma_{z}.
	\end{align*}
	We have also used relation between Kraus operators, particularly $\widetilde{\xi}_{0} \sigma_{z} = w_{01} \xi_{1} \sigma_{z}$. Matrix $U_{1} = \sigma_{z}$ can be calculated in a similar fashion.
\end{example}

\subsubsection{Two-Dimensional Case}

We shall investigate processors with dimensions of both program and data space being $P = D = 2$. Quantum processor is a unitary matrix and any unitary matrix of dimension $2 \times 2$ can be expressed as \cite{OptimalCreationOfEntanglementUsingATwo-QubitGate}:
\begin{align*}
	G = (U \otimes V)W(U^{\prime} \otimes V^{\prime}),
\end{align*}
where $U$, $V$, $U^{\prime}$ and $V^{\prime}$ are two dimensional unitary matrices and $W = \exp[i(\sum\limits_{\alpha} \alpha \sigma_{\alpha} \otimes \sigma_{\alpha})]$ with $\alpha = \{x, y, z\}$. Let us use the following notation $\sigma_{\alpha} \otimes \sigma_{\alpha} = \varsigma_{\alpha}$ and $\mathbb{1} \otimes \mathbb{1} = \mathbb{I}$. Due to commutation between matrices $\varsigma_{\alpha}$, we can factor the exponential and obtain $W = \exp(ix\varsigma_{x})\exp(iy\varsigma_{y})\exp(iz\varsigma_{z})$. We shall use Taylor expansion of exponential function:
\begin{align*}
	\exp(ix\varsigma_{x}) &= \sum\limits_{n=0}^{\infty} \frac{(ix\varsigma_{x})^{n}}{n!} \\
	&= (ix\varsigma_{x})^{0} + (ix\varsigma_{x})^{1} + (ix\varsigma_{x})^{2} + (ix\varsigma_{x})^{3} + (ix\varsigma_{x})^{4} + (ix\varsigma_{x})^{5} + (ix\varsigma_{x})^{6} + (ix\varsigma_{x})^{7} + \ldots \\
	&= \mathbb{I} + ix\varsigma_{x} - \frac{x^{2}}{2!} \mathbb{I} - i\frac{x^{3}}{3!}\varsigma_{x} + \frac{x^{4}}{4!} \mathbb{I} + i\frac{x^{5}}{5!}\varsigma_{x} - \frac{x^{6}}{6!}\mathbb{I} - i\frac{x^{7}}{7!}\varsigma_{x} + \ldots \\
	&= \left(1-\frac{1}{2!}x^{2}+\frac{1}{4!}x^{4}-\frac{1}{6!}x^{6}+\ldots\right) \mathbb{I} + i \left(x-\frac{1}{3!}x^{3}+\frac{1}{5!}x^{5}-\frac{1}{7!}x^{7}+\ldots\right) \varsigma_{x} \\
	&= \cos(x)\mathbb{I} + i \sin(x)\varsigma_{x}.
\end{align*}
With this knowledge, we can reconstruct the entire matrix $W$:
\begin{align*}\label{W}
	W &= \exp(ix\varsigma_{x})\exp(iy\varsigma_{y})\exp(iz\varsigma_{z}) \\
	&= \left[\cos(x)\mathbb{I} + i \sin(x)\varsigma_{x}\right]\left[\cos(y)\mathbb{I} + i \sin(y)\varsigma_{y}\right]\left[\cos(z)\mathbb{I} + i \sin(z)\varsigma_{z}\right] \\
	&= \left[\cos(x)\cos(y)\cos(z) + i\sin(x)\sin(y)\sin(z)\right] \mathbb{I} \\
	&+ \left[\cos(x)\sin(y)\sin(z) + i\sin(x)\cos(y)\cos(z)\right] \mathbb{\varsigma_{x}} \\ &+ \left[\sin(x)\cos(y)\sin(z) + i\cos(x)\sin(y)\cos(z)\right] \mathbb{\varsigma_{y}}\\ 
	&+ \left[\sin(x)\sin(y)\cos(z) + i\cos(x)\cos(y)\sin(z)\right] \mathbb{\varsigma_{z}} \\
	&= \left[\cos(x-y)e^{iz}\right] \left(\dyad{00}{00} + \dyad{11}{11}\right) + \left[\cos(x+y)e^{-iz}\right] \left(\dyad{01}{01} + \dyad{10}{10}\right) \\
	&+ \left[i\sin(x-y)e^{iz}\right] \left(\dyad{00}{11} + \dyad{11}{00}\right) + \left[i\sin(x+y)e^{-iz}\right] \left(\dyad{01}{10} + \dyad{10}{01}\right). \numberthis
\end{align*}
As $W$ is a unitary matrix, we can view it also as a processor:
\begin{align*} \label{GW}
	W &\overset{(\ref{QPeq})}{=} A^{W}_{00} \otimes \dyad{0}{0} + A^{W}_{01} \otimes \dyad{0}{1} + A^{W}_{10} \otimes \dyad{1}{0} + A^{W}_{11} \otimes \dyad{1}{1} \\ 
	&= \left[\cos(x-y)e^{iz}\dyad{0}{0} + \cos(x+y)e^{-iz}\dyad{1}{1}\right] \otimes \dyad{0}{0} \\
	&+ i\left[\sin(x-y)e^{iz}\dyad{0}{1} + \sin(x+y)e^{-iz}\dyad{1}{0}\right] \otimes \dyad{0}{1} \\
	&+i\left[\sin(x-y)e^{iz}\dyad{1}{0} + \sin(x+y)e^{-iz}\dyad{0}{1}\right] \otimes \dyad{1}{0} \\
	&+ \left[\cos(x-y)e^{iz}\dyad{1}{1} + \cos(x+y)e^{-iz}\dyad{0}{0}\right] \otimes \dyad{1}{1}. \numberthis
\end{align*}
Furthermore, we shall derive equations, which fulfillment would mean that two processors are deterministically equivalent. For this reason, we shall calculate Choi matrix for processor $W$. Firstly, let us denote coefficients next to the individual matrices $\mathbb{I}$ and $\varsigma_{\alpha}$ from equation (\ref{W}) as follows:
\begin{align*}\label{ks}
	k_{i} &=  \cos(x)\cos(y)\cos(z) + i\sin(x)\sin(y)\sin(z), \\ k_{x} &= \cos(x)\sin(y)\sin(z) + i\sin(x)\cos(y)\cos(z), \\
	k_{y} &= \sin(x)\cos(y)\sin(z) + i\cos(x)\sin(y)\cos(z), \\ k_{z} &= \sin(x)\sin(y)\cos(z) + i\cos(x)\cos(y)\sin(z). \numberthis
\end{align*}
Let us now proceed with the Choi matrix itself, that is calculated as given in equation (\ref{choi}).
\begin{align*}
	&\left(\mathfrak{C}_{W,\xi}^{det} \otimes \mathcal{I}\right) \left(\sum_{i,i^{\prime}}^{D} \dyad{ii}{i^{\prime}i^{\prime}}\right) = \sum_{ii^{\prime}}^{D} \Tr_{p} \left[G \left(\dyad{i}{i^{\prime}} \otimes \xi\right) G^{\dagger}\right] \otimes \dyad{i}{i^{\prime}} \\
	&= \sum_{i,i^{\prime}}^{D} \Tr_{p} \left[\left(k_{i} \mathbb{I} + k_{x} \mathbb{\varsigma}_{x} + k_{y} \mathbb{\varsigma}_{y} + k_{z} \mathbb{\varsigma}_{z} \right)  \left(\dyad{i}{i^{\prime}} \otimes \xi\right) \left(k_{i}^{\ast} \mathbb{I} + k_{x}^{\ast} \mathbb{\varsigma}_{x} + k_{y}^{\ast} \mathbb{\varsigma}_{y} + k_{z}^{\ast} \mathbb{\varsigma}_{z} \right) \right] \otimes \dyad{i}{i^{\prime}} \\
	&= \sum_{ii^{\prime}}^{D} [k_{i}k_{i}^{\ast} \Tr\left(\mathbb{1} \xi \mathbb{1}\right) \mathbb{1} \dyad{i}{i^{\prime}} \mathbb{1} + k_{i}k_{x}^{\ast} \Tr\left(\mathbb{1} \xi \sigma_{x}\right) \mathbb{1}  \dyad{i}{i^{\prime}} \sigma_{x} \\
	&+ k_{i}k_{y}^{\ast} \Tr\left(\mathbb{1} \xi \sigma_{y}\right) \mathbb{1}  \dyad{i}{i^{\prime}} \sigma_{y} + k_{i}k_{z}^{\ast} \Tr\left(\mathbb{1} \xi \sigma_{z}\right) \mathbb{1}  \dyad{i}{i^{\prime}} \sigma_{z} \\
	&+ k_{x}k_{i}^{\ast} \Tr\left(\sigma_{x} \xi \mathbb{1}\right) \sigma_{x} \dyad{i}{i^{\prime}} \mathbb{1}
	+ k_{x}k_{x}^{\ast} \Tr\left(\sigma_{x} \xi \sigma_{x}\right) \sigma_{x} \dyad{i}{i^{\prime}} \sigma_{x} \\
	&+ k_{x}k_{y}^{\ast} \Tr\left(\sigma_{x} \xi \sigma_{y}\right) \sigma_{x} \dyad{i}{i^{\prime}} \sigma_{y}
	+ k_{x}k_{z}^{\ast} \Tr\left(\sigma_{x} \xi \sigma_{z}\right) \sigma_{x} \dyad{i}{i^{\prime}} \sigma_{z} \\
	&+ k_{y}k_{i}^{\ast} \Tr\left(\sigma_{y} \xi \mathbb{1}\right) \sigma_{y} \dyad{i}{i^{\prime}} \mathbb{1}
	+ k_{y}k_{x}^{\ast} \Tr\left(\sigma_{y} \xi \sigma_{x}\right) \sigma_{y} \dyad{i}{i^{\prime}} \sigma_{x} \\
	&+ k_{y}k_{y}^{\ast} \Tr\left(\sigma_{y} \xi \sigma_{y}\right) \sigma_{y} \dyad{i}{i^{\prime}} \sigma_{y}
	+ k_{y}k_{z}^{\ast} \Tr\left(\sigma_{y} \xi \sigma_{z}\right) \sigma_{y} \dyad{i}{i^{\prime}} \sigma_{z} \\
	&+ k_{z}k_{i}^{\ast} \Tr\left(\sigma_{z} \xi \mathbb{1}\right) \sigma_{z} \dyad{i}{i^{\prime}} \mathbb{1}
	+ k_{z}k_{x}^{\ast} \Tr\left(\sigma_{z} \xi \sigma_{x}\right) \sigma_{z} \dyad{i}{i^{\prime}} \sigma_{x} \\
	&+ k_{z}k_{y}^{\ast} \Tr\left(\sigma_{z} \xi \sigma_{y}\right) \sigma_{z} \dyad{i}{i^{\prime}} \sigma_{y}
	+ k_{z}k_{z}^{\ast} \Tr\left(\sigma_{z} \xi \sigma_{z}\right) \sigma_{z} \dyad{i}{i^{\prime}} \sigma_{z}] \otimes \dyad{i}{i^{\prime}},
\end{align*}
where $\mathcal{I}$ denotes identity channel. We can express two-dimensional quantum state in the following way $\xi = \frac{1}{2} (\mathbb{1} + \vec{e}\cdot\vec{\sigma}) = \frac{1}{2} (\mathbb{1} + e_{x}\sigma_{x} + e_{y}\sigma_{y} + e_{z}\sigma_{z})$, where $e_{x}, e_{y}, e_{z} \in \mathbb{R}$ and $e_{x}^{2} + e_{y}^{2} + e_{z}^{2} = 1$, and use this to calculate traces from the previous equation:
\begin{align*}
	\Tr\left(\mathbb{1}\xi\mathbb{1}\right) = \Tr\left(\sigma_{\alpha}\xi\sigma_{\alpha}\right) = 1, \qquad \Tr\left(\xi\sigma_{\alpha}\right) = \Tr\left(\sigma_{\alpha}\xi\right) = e_{\alpha},
\end{align*}
where $\alpha = \{x,y,z\}$ and we are also taking advantage of the fact that Pauli matrices are traceless. Additionally, by using following relations between Pauli matrices:
\begin{align*}
	\sigma_{x}\sigma_{y} = -\sigma_{y}\sigma_{x} = i\sigma_{z}, \qquad \sigma_{y}\sigma_{z} = -\sigma_{z}\sigma_{y} = i\sigma_{x}, \qquad \sigma_{z}\sigma_{x} = -\sigma_{x}\sigma_{z} = i\sigma_{y},
\end{align*}
we can arrive at the final expression for our Choi matrix.
\begin{align*}\label{ChoiOfG}
	&\left(\mathfrak{C}_{W,\xi}^{det} \otimes \mathcal{I}\right) \left(\sum_{i,i^{\prime}}^{D} \dyad{ii}{i^{\prime}i^{\prime}}\right) = \sum_{ii^{\prime}}^{D} [k_{i}k_{i}^{\ast} \mathbb{1}\dyad{i}{i^{\prime}}\mathbb{1} + k_{x}k_{x}^{\ast} \sigma_{x}\dyad{i}{i^{\prime}}\sigma_{x} + k_{y}k_{y}^{\ast} \sigma_{y}\dyad{i}{i^{\prime}}\sigma_{y} + k_{z}k_{z}^{\ast} \sigma_{z}\dyad{i}{i^{\prime}}\sigma_{z} \\ 
	&+ e_{x} \left(k_{i}k_{x}^{\ast} \mathbb{1}\dyad{i}{i^{\prime}}\sigma_{x} + k_{x}k_{i}^{\ast} \sigma_{x}\dyad{i}{i^{\prime}}\mathbb{1} + ik_{z}k_{y}^{\ast} \sigma_{z}\dyad{i}{i^{\prime}}\sigma_{y} - ik_{y}k_{z}^{\ast} \sigma_{y}\dyad{i}{i^{\prime}}\sigma_{z}\right) \\
	&+ e_{y} \left(k_{i}k_{y}^{\ast} \mathbb{1}\dyad{i}{i^{\prime}}\sigma_{y} + k_{y}k_{i}^{\ast} \sigma_{y}\dyad{i}{i^{\prime}}\mathbb{1} + ik_{x}k_{z}^{\ast} \sigma_{x}\dyad{i}{i^{\prime}}\sigma_{z} - ik_{z}k_{x}^{\ast} \sigma_{z}\dyad{i}{i^{\prime}}\sigma_{x}\right) \\
	&+ e_{z} \left(k_{i}k_{z}^{\ast} \mathbb{1}\dyad{i}{i^{\prime}}\sigma_{z} + k_{z}k_{i}^{\ast} \sigma_{z}\dyad{i}{i^{\prime}}\mathbb{1} + ik_{y}k_{x}^{\ast} \sigma_{y}\dyad{i}{i^{\prime}}\sigma_{x} - ik_{x}k_{y}^{\ast} \sigma_{x}\dyad{i}{i^{\prime}}\sigma_{y}\right)] \otimes \dyad{i}{i^{\prime}}. \numberthis
\end{align*}
We can now take a different processor $W^{\prime}$ with different parameters $x^{\prime}, y^{\prime}, z^{\prime}$ and derive equations for which two processors are equivalent. This is possible since two quantum channels are the same if and only if they have identical Choi matrices. Firstly, let us show, on two examples, the calculation of coefficients next to individual states from Choi matrix. We shall use the following notation $\cos = c$ and $\sin = s$.
\begin{align*}
	k_{i}k_{i}^{\ast} &\overset{(\ref{ks})}{=} \left[c(x)c(y)c(z) + is(x)s(y)s(z)\right] \left[c(x)c(y)c(z) - is(x)s(y)s(z)\right] \\ 
	&= c^{2}(x)c^{2}(y)c^{2}(z) + s^{2}(x)s^{2}(y)s^{2}(z) \\
	k_{i}k_{x}^{\ast} &\overset{(\ref{ks})}{=} \left[c(x)c(y)c(z) + is(x)s(y)s(z)\right] \left[c(x)s(y)s(z) - is(x)c(y)c(z)\right] \\
	&\overset{(i)}{=} c^{2}(x)\frac{1}{2}s(2y)\frac{1}{2}s(2z) + \frac{i}{2} s(2x)s^{2}(y)s^{2}(z) - \frac{i}{2}s(2x)c^{2}(y)c^{2}(z) + s^{2}(x)\frac{1}{2}s(2y)\frac{1}{2}s(2z) \\ 
	&= \frac{1}{4} \left[c^{2}(x) + s^{2}(x)\right] s(2y)s(2z) + \frac{i}{2} s(2x) \left[s^{2}(y)s^{2}(z) - c^{2}(y)c^{2}(z)\right] \\
	&\overset{(ii)}{=} \frac{1}{4} s(2y)s(2z) - \frac{i}{2} s(2x)c(y-z)c(y+z),
\end{align*}
where in $(i)$ we have used that $\cos(x)\sin(x)=\frac{1}{2}\sin(2x)$ and in $(ii)$ we have used $\sin^{2}(y)\sin^{2}(z) - \cos^{2}(y)\cos^{2}(z) = - \cos(y-z)\cos(y+z)$. Other coefficients can be calculated in a similar fashion. Now, we can take different processor $W^{\prime}$, characterized by parameters $x^{\prime}, y^{\prime}, z^{\prime}$, with different program state $\xi^{\prime} = \frac{1}{2} (\mathbb{1} + e_{x}^{\prime} \sigma_{x} + e_{y}^{\prime} \sigma_{y} + e_{z}^{\prime} \sigma_{z})$ and compare it with Choi matrix of processor $W$ expressed in equation (\ref{ChoiOfG}). We recover $10$ equations for the case when Choi matrices of $\mathfrak{C}_{W,\xi}^{det}$ and $\mathfrak{C}_{W^{\prime},\xi^{\prime}}^{det}$ are equal.
\begin{align*}
	c^{2}(x)c^{2}(y)c^{2}(z) + s^{2}(x)s^{2}(y)s^{2}(z) &\overset{!}{=} c^{2}(x^{\prime})c^{2}(y^{\prime})c^{2}(z^{\prime}) + s^{2}(x^{\prime})s^{2}(y^{\prime})s^{2}(z^{\prime}), \\
	c^{2}(x)s^{2}(y)s^{2}(z) + s^{2}(x)c^{2}(y)c^{2}(z) &\overset{!}{=} c^{2}(x^{\prime})s^{2}(y^{\prime})s^{2}(z^{\prime}) + s^{2}(x^{\prime})c^{2}(y^{\prime})c^{2}(z^{\prime}), \\
	s^{2}(x)c^{2}(y)s^{2}(z) + c^{2}(x)s^{2}(y)c^{2}(z) &\overset{!}{=} s^{2}(x^{\prime})c^{2}(y^{\prime})s^{2}(z^{\prime}) + c^{2}(x^{\prime})s^{2}(y^{\prime})c^{2}(z^{\prime}), \\
	s^{2}(x)s^{2}(y)c^{2}(z) + c^{2}(x)c^{2}(y)s^{2}(z) &\overset{!}{=} s^{2}(x^{\prime})s^{2}(y^{\prime})c^{2}(z^{\prime}) + c^{2}(x^{\prime})c^{2}(y^{\prime})s^{2}(z^{\prime}), \\
	e_{x} \left[\frac{1}{4} s(2y)s(2z) - \frac{i}{2} s(2x)c(y-z)c(y+z)\right] &\overset{!}{=} e_{x}^{\prime} \left[\frac{1}{4} s(2y^{\prime})s(2z^{\prime}) - \frac{i}{2} s(2x^{\prime})c(y^{\prime}-z^{\prime})c(y^{\prime}+z^{\prime})\right], \\
	e_{x} \left[\frac{1}{4} s(2y)s(2z) - \frac{i}{2} s(2x)s(y-z)s(y+z)\right] &\overset{!}{=} e_{x}^{\prime} \left[\frac{1}{4} s(2y^{\prime})s(2z^{\prime}) - \frac{i}{2} s(2x^{\prime})s(y^{\prime}-z^{\prime})s(y^{\prime}+z^{\prime})\right], \\
	e_{y} \left[\frac{1}{4} s(2z)s(2x) - \frac{i}{2} s(2y)c(z-x)c(z+x)\right] &\overset{!}{=} e_{y}^{\prime} \left[\frac{1}{4} s(2z^{\prime})s(2x^{\prime}) - \frac{i}{2} s(2y^{\prime})c(z^{\prime}-x^{\prime})c(z^{\prime}+x^{\prime})\right], \\
	e_{y} \left[\frac{1}{4} s(2z)s(2x) - \frac{i}{2} s(2y)s(z-x)s(z+x)\right] &\overset{!}{=} e_{y}^{\prime} \left[\frac{1}{4} s(2z^{\prime})s(2x^{\prime}) - \frac{i}{2} s(2y^{\prime})s(z^{\prime}-x^{\prime})s(z^{\prime}+x^{\prime})\right], \\
	e_{z} \left[\frac{1}{4} s(2x)s(2y) - \frac{i}{2} s(2z)c(x-y)c(x+y)\right] &\overset{!}{=} e_{z}^{\prime} \left[\frac{1}{4} s(2x^{\prime})s(2y^{\prime}) - \frac{i}{2} s(2z^{\prime})c(x^{\prime}-y^{\prime})c(x^{\prime}+y^{\prime})\right], \\
	e_{z} \left[\frac{1}{4} s(2x)s(2y) - \frac{i}{2} s(2z)s(x-y)s(x+y)\right] &\overset{!}{=} e_{z}^{\prime} \left[\frac{1}{4} s(2x^{\prime})s(2y^{\prime}) - \frac{i}{2} s(2z^{\prime})s(x^{\prime}-y^{\prime})s(x^{\prime}+y^{\prime})\right].
\end{align*}
Solving these equations would reveal when processors $W$ and $W^{\prime}$ are equivalent $W \sim_{det} W^{\prime}$.

\paragraph{SWAP Processor} SWAP gate, denoted with $S$, acting on two-dimensional quantum system swaps quantum state of the first subsystem with the quantum state of the second subsystem: $S(\ket{\psi} \otimes \ket{\Xi}) = \ket{\Xi} \otimes \ket{\psi}$. In computational basis, SWAP can be expressed as follows $S = \dyad{00}{00} + \dyad{01}{10} + \dyad{10}{01} + \dyad{11}{11}$. SWAP used as a deterministic processor results in program state on the output:
\begin{align*}
	\mathfrak{C}_{S,\xi}^{det} = \Tr_{p} \left[S\left(\varrho \otimes \xi\right)S^{\dagger}\right] = \Tr_{p} \left[\xi \otimes \varrho\right] = \xi \Tr\left(\varrho\right) = \xi.
\end{align*}
We shall strive to find all processors equivalent to SWAP processor for $P=D=2$. Let us begin by investigation of matrix $W = \exp(ix\varsigma_{x})\exp(iy\varsigma_{y})\exp(iz\varsigma_{z})$. Generally, $W$ used as a processor implements:
\begin{align*}\label{WimplementationA}
	\mathfrak{C}_{W,\xi}^{det} &\overset{(\ref{elemOfImplDet})}{=} \xi_{00} \sum_{j}^{2} \left(A_{j0}^{W} \varrho A_{j0}^{W\dagger}\right) + \xi_{01}  \sum_{j}^{2} \left(A_{j0}^{W} \varrho A_{j1}^{W\dagger}\right) \\
	&+ \xi_{10}  \sum_{j}^{2} \left(A_{j1}^{W} \varrho A_{j0}^{W\dagger}\right) + \xi_{11}  \sum_{j}^{2} \left(A_{j1}^{W} \varrho A_{j1}^{W\dagger}
	\right). \numberthis
\end{align*}
We know that the output should be program state $\xi$. Therefore, if we take program to be $\xi = \dyad{0}{0}$, we can check when the output state is be pure. We can use this assumption as it is weaker than to expect the exact state $\dyad{0}{0}$ on the output. Using this fact, we shall, hopefully, derive conditions on $x$, $y$ and $z$. Let us therefore explicitly calculate expression $\tau = A_{00}^{W} \varrho A_{00}^{W\dagger} + A_{10}^{W} \varrho A_{10}^{W\dagger}$ and figure out when state $\tau$ is pure. We shall, again, use notation $c = \cos$ and $s = \sin$. Moreover, before we begin with the calculation, let us remind that $A_{00}^{W} = c(x-y)e^{iz} \dyad{0}{0} + c(x+y)e^{-iz} \dyad{1}{1}$ and $A_{10}^{W} = s(x-y)e^{iz} \dyad{1}{0} + s(x+y)e^{-iz} \dyad{0}{1}$.
\begin{align*}
	\tau &= \left[c(x-y)e^{iz} \dyad{0}{0} + c(x+y)e^{-iz} \dyad{1}{1}\right] \varrho \left[c(x-y)e^{-iz} \dyad{0}{0} + c(x+y)e^{iz} \dyad{1}{1}\right] \\
	&+ \left[s(x-y)e^{iz} \dyad{1}{0} + s(x+y)e^{-iz} \dyad{0}{1}\right] \varrho \left[s(x-y)e^{-iz} \dyad{0}{1} + s(x+y)e^{iz} \dyad{1}{0}\right] \\
	&= c^{2}(x-y)\varrho_{00}\dyad{0}{0} + c(x-y)c(x+y)e^{2iz}\varrho_{01}\dyad{0}{1} \\
	&+ c(x+y)c(x-y)e^{-2iz}\varrho_{10}\dyad{1}{0} + c^{2}(x+y)\varrho_{11}\dyad{1}{1} \\
	&+ s^{2}(x-y)\varrho_{00}\dyad{1}{1} + s(x-y)s(x+y)e^{2iz}\varrho_{01}\dyad{1}{0} \\
	&+ s(x+y)s(x-y)e^{-2iz}\varrho_{10}\dyad{0}{1} + s^{2}(x+y)\varrho_{11}\dyad{0}{0} \\
	&= \left[c^{2}(x-y)\varrho_{00} + s^{2}(x+y)\varrho_{11}\right] \dyad{0}{0} \\
	&+ \left[c(x-y)c(x+y)e^{2iz}\varrho_{01} + s(x+y)s(x-y)e^{-2iz}\varrho_{10}\right] \dyad{0}{1} \\
	&+ \left[c(x+y)c(x-y)e^{-2iz}\varrho_{10} + s(x-y)s(x+y)e^{2iz}\varrho_{01}\right] \dyad{1}{0} \\
	&+ \left[c^{2}(x+y)\varrho_{11} + s^{2}(x-y)\varrho_{00}\right] \dyad{1}{1}.
\end{align*}
To check when the state $\tau$ is pure we shall use purity $\Tr(\tau^{2}) \overset{!}{=} 1$. Thus, we shall calculate $\tau^{2}$. However, it is enough to only calculate entries on the diagonal, as we are taking trace of the state $\tau^{2}$. Therefore, let us calculate matrix element $\bra{0}\tau^{2}\ket{0}$:
\begin{align*}
	&\left[c^{2}(x-y)\varrho_{00} + s^{2}(x+y)\varrho_{11}\right]^{2} + c^{2}(x-y)c^{2}(x+y)\varrho_{01}\varrho_{10} + s^{2}(x-y)s^{2}(x+y)\varrho_{10}\varrho_{01} \\
	&+ c(x-y)c(x+y)s(x-y)s(x+y)e^{4iz}\varrho_{01}^{2} + c(x-y)c(x+y)s(x-y)s(x+y)e^{-4iz}\varrho_{10}^{2} \\
	&= \left[c^{2}(x-y)\varrho_{00} + s^{2}(x+y)\varrho_{11}\right]^{2} + \left[c^{2}(x-y)c^{2}(x+y) + s^{2}(x-y)s^{2}(x+y)\right] \varrho_{01}\varrho_{10} \\
	&+ c(x-y)c(x+y)s(x-y)s(x+y) \left(e^{4iz}\varrho_{01}^{2} + e^{-4iz}\varrho_{10}^{2}\right) \\
	&\overset{(i)}{=} \left[c^{2}(x-y)\varrho_{00} + s^{2}(x+y)\varrho_{11}\right]^{2} + \frac{1}{4} \left[c(4x) + c(4y) + 2\right] \varrho_{01}\varrho_{10} \\
	&+ \frac{1}{8} \left[c(4y) - c(4x)\right] \left(e^{4iz}\varrho_{01}^{2} + e^{-4iz}\varrho_{10}^{2}\right).
\end{align*}
In $(i)$, we have used help from WolframAlpha to arrive at the final expression. We have calculated $\bra{1}\tau^{2}\ket{1}$ similarly and therefore, expression for trace is:
\begin{align*}
	\Tr(\tau^{2}) &= \left[c^{2}(x-y)\varrho_{00} + s^{2}(x+y)\varrho_{11}\right]^{2} + \frac{1}{4} \left[c(4x) + c(4y) + 2\right] \varrho_{01}\varrho_{10} \\
	&+ \frac{1}{8} \left[c(4y) - c(4x)\right] \left(e^{4iz}\varrho_{01}^{2} + e^{-4iz}\varrho_{10}^{2}\right) \\
	&+ \left[c^{2}(x+y)\varrho_{00} + s^{2}(x-y)\varrho_{11}\right]^{2} + \frac{1}{4} \left[c(4x) + c(4y) + 2\right] \varrho_{01}\varrho_{10} \\
	&+ \frac{1}{8} \left[c(4y) - c(4x)\right] \left(e^{4iz}\varrho_{01}^{2} + e^{-4iz}\varrho_{10}^{2}\right) \\
	&= \left[c^{2}(x+y)\varrho_{00} + s^{2}(x-y)\varrho_{11}\right]^{2} + \left[c^{2}(x-y)\varrho_{00} + s^{2}(x+y)\varrho_{11}\right]^{2} \\
	&+ \frac{1}{2} \left[c(4x) + c(4y) + 2\right] \varrho_{01}\varrho_{10} 
	+ \frac{1}{4} \left[c(4y) - c(4x)\right] \left(e^{4iz}\varrho_{01}^{2} + e^{-4iz}\varrho_{10}^{2}\right) \overset{!}{=} 1.
\end{align*}
We can have a look at $c(4x) + c(4y) + 2 \overset{!}{=} 0$ which needs to be $0$ as the result cannot depend on $\varrho$ and also has to be valid for any quantum state $\varrho$. Therefore:
\begin{align*}\label{thisLabelIsXY}
	x = \frac{\pi}{4} + \frac{n_{1}\pi}{2}, \qquad y = \frac{\pi}{4} + \frac{n_{2}\pi}{2}, \numberthis
\end{align*}
where $n_{1}, n_{2} \in \mathbb{Z}$ and $\mathbb{Z}$ denotes set of integers. Let us substitute our result back into trace, specifically for $x = \frac{\pi}{4}$ and $y = \frac{\pi}{4}$:
\begin{align*}
	\Tr(\tau^{2}) &= \left[c^{2}\left(\frac{\pi}{2}\right)\varrho_{00} + s^{2}(0)\varrho_{1}\right]^{2} + \left[c^{2}\left(0\right)\varrho_{00} + s^{2}(\frac{\pi}{2})\varrho_{1}\right]^{2} \\
	&+ \frac{1}{2} \left[c\left(\pi\right) + c\left(\pi\right) + 2\right] \varrho_{01}\varrho_{10} + \frac{1}{4} \left[c\left(\pi\right) - c\left(\pi\right)\right]\left(e^{4iz}\varrho_{01}^{2} + e^{-4iz}\varrho_{10}^{2}\right) \\
	&= \left(\varrho_{00} + \varrho_{11}\right)^{2} + \frac{1}{2} \left(-1-1+2\right)\varrho_{01}\varrho_{10} + \frac{1}{4}\left(-1+1\right)\left(e^{4iz}\varrho_{01}^{2} + e^{-4iz}\varrho_{10}^{2}\right) = 1
\end{align*}
We can also check our conclusion for a concrete quantum state $\varrho = \frac{1}{2} (\dyad{0}{0} + \dyad{1}{1})$. This reduces trace of $\tau^{2}$ to:
\begin{align*}
	\Trace(\tau^{2}) &= \frac{1}{4} \left\{\left[c^{2}(x+y) + s^{2}(x-y)\right]^{2} + \left[c^{2}(x-y) + s^{2}(x+y)\right]^{2}\right\} \\
	&= \frac{1}{4} [c^{4}(x+y) + s^{4}(x-y) + 2c^{2}(x+y)s^{2}(x-y) + c^{4}(x-y) + s^{4}(x+y) + 2c^{2}(x-y)s^{2}(x+y)] \\
	&\hspace*{-1mm}\overset{(i)}{=} \frac{1}{4} [c^{4}(x-y)+s^{4}(x-y) + + c^{4}(x+y) + s^{4}(x+y) - \frac{1}{2}c(4x) - \frac{1}{2}c(4y) + 1] \\
	&\hspace*{-1mm}\overset{(ii)}{=} \frac{1}{8} \left[-2s^{2}(2x)c(4y) - c(4x) + 5\right] \overset{!}{=} 1.
\end{align*}
In $(i)$ and $(ii)$ we have used WolframAlpha. We can see that $2s^{2}(2x)c(4y) + c(4x) \overset{!}{=} -3$ in order for the entire expression in square brackets to be $8$, so the trace can be equal to $1$. This is achieved only for case when $x = \frac{\pi}{4} + \frac{n_{1}\pi}{2}$ and $y = \frac{\pi}{4} + \frac{n_{2}\pi}{2}$, which confirms equation (\ref{thisLabelIsXY}). Now we shall substitute our solutions in equation (\ref{GW}) and obtain:
\begin{align*}
	W = e^{iz} \dyad{00}{00} + ie^{-iz} \dyad{10}{01} + ie^{-iz} \dyad{01}{10} + e^{iz} \dyad{11}{11}.
\end{align*}
Therefore, our new operators forming processor $W$ are $A_{00}^{W} = e^{iz}\dyad{0}{0}, A_{01}^{W} = ie^{-iz}\dyad{1}{0}, A_{10}^{W} = ie^{-iz}\dyad{0}{1}$ and $A_{11}^{W} = e^{iz}\dyad{1}{1}$. And we can substitute these operators into equation (\ref{WimplementationA}):
\begin{align*}\label{CWXI}
	\mathfrak{C}_{W,\xi}^{det} &= \xi_{00} \left(\dyad{0}{0}\varrho\dyad{0}{0} + \dyad{0}{1}\varrho\dyad{1}{0}\right) + \xi_{01} \left[e^{iz}\dyad{0}{0}\varrho\dyad{0}{1}(-i)e^{iz} + ie^{-iz}\dyad{0}{1}\varrho\dyad{1}{1}e^{-iz}\right] \\
	&+ \xi_{10} \left[ie^{-iz} \dyad{1}{0}\varrho\dyad{0}{0} e^{-iz} + e^{iz} \dyad{1}{1}\varrho\dyad{1}{0} (-i)e^{iz}\right] + \xi_{11} \left(\dyad{1}{0}\varrho\dyad{0}{1} + \dyad{1}{1}\varrho\dyad{1}{1}\right) \\
	&= \xi_{00} \dyad{0}{0} + i\xi_{01} \left(-e^{2iz}\varrho_{00} + e^{-2iz}\varrho_{11}\right) \dyad{0}{1} + i\xi_{10} \left(e^{-2iz}\varrho_{00} - e^{2iz}\varrho_{11}\right) \dyad{1}{0} + \xi_{11} \dyad{1}{1}. \numberthis
\end{align*}
For processor $W$ to be equivalent to SWAP processor $S$, we cannot have any dependence in $\mathfrak{C}_{W,\xi}^{det}$ on data state $\varrho$. Therefore, $e^{-2iz} \overset{!}{=} -e^{2iz}$ or considering that $e^{-2iz} = \cos(2z) - i\sin(2z)$ and $-e^{2iz} = -\cos(2z) - i\sin(2z)$ we arrive at the equation:
\begin{align*}
	\cos(2z) - i\sin(2z) = -\cos(2z) - i\sin(2z),
\end{align*}
and thus $2\cos(2z) \overset{!}{=} 0$, which is true for $z = \frac{\pi}{4} + \frac{n_{3}\pi}{2}$ with $n_{3} \in \mathbb{Z}$ being integer. By substituting this value for $z$ in equation $(\ref{CWXI})$, we can see that $\mathfrak{C}_{W,\xi}^{det}$ gives program state $\xi$ as a result, because $-e^{i\frac{\pi}{2}} = e^{-i\frac{\pi}{2}} = -i$ and thus:
\begin{align*}
	\mathfrak{C}_{W,\xi}^{det} = \xi_{00}\dyad{0}{0} + i\xi_{01} \left(-i\varrho_{00} - i\varrho_{11}\right) \dyad{0}{1} + i\xi_{10} \left(-i\varrho_{00} - i\varrho_{11}\right) \dyad{1}{0} + \xi_{11} \dyad{1}{1} = \xi.
\end{align*}
Therefore, only solutions where $W$ forms equivalent processor to the SWAP processor are:
\begin{align*}\label{xyzxyzxyz}
	x = \frac{\pi}{4} + \frac{n_{1}\pi}{2}, \qquad y = \frac{\pi}{4} + \frac{n_{2}\pi}{2}, \qquad z = \frac{\pi}{4} + \frac{n_{3}\pi}{2}, \numberthis
\end{align*}
where $n_{1}, n_{2}, n_{3} \in \mathbb{Z}$. Let us substitute $x$, $y$ and $z$ for $n_{1} = n_{2} = n_{3} = 0$ back in $W$:
\begin{align*}
	W &\overset{(\ref{W})}{=} \cos(0) e^{i\frac{\pi}{4}} (\dyad{00}{00} + \dyad{11}{11}) + \cos(\frac{\pi}{2}) e^{-i\frac{\pi}{4}} (\dyad{01}{01} + \dyad{10}{10}) \\
	&\hspace*{3mm}+ i\sin(0)e^{i\frac{\pi}{4}} (\dyad{00}{11} + \dyad{11}{00}) + i\sin(\frac{\pi}{2}) e^{-i\frac{\pi}{4}} (\dyad{01}{10} + \dyad{10}{01}) \\
	&= e^{i\frac{\pi}{4}} (\dyad{00}{00} + \dyad{11}{11}) + ie^{-i\frac{\pi}{4}} (\dyad{01}{10} + \dyad{10}{01}) = e^{i\frac{\pi}{4}} S.
\end{align*}

In the final paragraph of this section, let us only consider $W = e^{i\frac{\pi}{4}} S$. However, arbitrary unitary matrix $2 \times 2$ is given by $\left(U \otimes V\right) W \left(U^{\prime} \otimes V^{\prime}\right)$. From theorems \ref{concreteSolutionDetProcLeft} and \ref{concreteSolutionDetProcRight} we know that deterministic processors are equivalent under local unitary transformations applied on program space, therefore $S \sim_{det} \left(\mathbb{1} \otimes V\right) W \left(\mathbb{1} \otimes V^{\prime}\right)$. However, processor $W(U^{\prime} \otimes \mathbb{1})$ is also equivalent to SWAP processor $S$, because:
\begin{align*}
	\mathfrak{C}_{W(U^{\prime}\otimes\mathbb{1}),\xi}^{det} &= \Tr_{p} \left[W(U^{\prime}\otimes\mathbb{1}) (\varrho \otimes \xi) (U^{\prime\dagger}\otimes\mathbb{1})W^{\dagger}\right] = \Tr_{p} \left[W(U^{\prime}\varrho U^{\prime\dagger} \otimes \xi) W^{\dagger}\right] \\
	&= \Tr_{p} \left[W (\varrho^{\prime} \otimes \xi) W^{\dagger}\right] \overset{(i)}{=} \Tr_{p} \left(\xi \otimes \varrho^{\prime}\right) = \xi,
\end{align*}
where in $(i)$ we have used $W = e^{i\frac{\pi}{4}} S$. Result is still only program state, exactly as for SWAP processor $S$. And finally, considering $(U \otimes \mathbb{1}) W$, we are only applying unitary transformation on program state $U\xi U^{\dagger}$, which can be superseded by using different program state $\widetilde{\xi} = U^{\dagger} \xi U$:
\begin{align*}
	\mathfrak{C}_{(U \otimes \mathbb{1})W,\widetilde{\xi}}^{det} &= \Tr_{p} \left[(U \otimes \mathbb{1}) W (\varrho \otimes \widetilde{\xi}) W^{\dagger} (U^{\dagger} \otimes \mathbb{1})\right] = \Tr_{p} \left[(U \otimes \mathbb{1}) (\widetilde{\xi} \otimes \varrho) (U^{\dagger} \otimes \mathbb{1})\right] \\
	&= U\widetilde{\xi}U^{\dagger} = UU^{\dagger} \xi UU^{\dagger} = \xi.
\end{align*}
Thus, SWAP processor $S$ is equivalent to $S \sim_{det} \left(U \otimes V\right) W \left(U^{\prime} \otimes V^{\prime}\right)$ with values for $x, y$ and $z$ given in equation (\ref{xyzxyzxyz}).


\subsection{Equivalence of Probabilistic Processors}

Let us remind, that we shall consider measurement $M$ at the end of program register to be successful, if the outcome is measured by an element $\mathbb{1} \otimes \dyad{\chi}{\chi} = \mathbb{1} \otimes \frac{1}{P} \sum_{nn^{\prime}}^{P} \dyad{n}{n^{\prime}}$, where the non-trivial measurement is only on the program space. If the outcome corresponds to different element, the implementation of the desired transformation failed. Let us denote probability of successful implementation by processor $G$ with $p$ and by processor $\widetilde{G}$ with $\widetilde{p}$.
\begin{theorem}\label{necAndSucStruct}
	Necessary and Sufficient Condition
	\\
	Let us have a quantum processor $G = \sum\limits_{jk}^{P} A_{jk} \otimes \dyad{j}{k}$. Let us also assume that:
	\begin{itemize}
		\item[$1)$] a different processor $G_{L}$ can be expressed as $G_{L} = UG$, where $U = \sum_{rq}^{P} U_{rq} \otimes \dyad{r}{q}$ is a unitary operator. Then processors $G$ and $G_{L}$ are structurally probabilistically equivalent $G \sim_{st} UG$ if and only if the following equation is true:
		\begin{align*}
			\sum_{jkq}^{P} v_{kq} \xi_{q} A_{jk} \overset{!}{=} e^{i\phi} \sqrt{K_{\xi, \widetilde{\xi}}} \sum_{r}^{P} \sum_{jk}^{P} \xi_{k} U_{rj} A_{jk},
		\end{align*}
		where $\widetilde{\xi}$ and $\xi$ are pure program states of $G$ and $UG$ respectively, $K_{\xi, \widetilde{\xi}} \in \mathbb{R}_{> 0}$ is a real positive number depending on programs and $\phi \in \mathbb{R}$.
		
		\item[$2)$] a different processor $G_{R}$ can be expressed as $G_{R} = GV$, where $V = \sum_{rq}^{P} V_{rq} \otimes \dyad{r}{q}$ is a unitary operator. Then processors $G$ and $G_{R}$ are structurally probabilistically equivalent $G \sim_{st} GV$ if and only if the following equation is true:
		\begin{align*}
			\sum_{jkq}^{P} v_{kq} \xi_{q} A_{jk} = e^{i\psi} \sqrt{C_{\xi, \widetilde{\xi}}} \sum_{jkq}^{P} \xi_{q} A_{jk} V_{kq},
		\end{align*}
		where $\widetilde{\xi}$ and $\xi$ are pure program states of $G$ and $GV$respectively, $C_{\xi, \widetilde{\xi}} \in \mathbb{R}_{> 0}$ is a real positive number depending on programs and $\psi \in \mathbb{R}$.
	\end{itemize}
\end{theorem}
\begin{proof}
	We shall calculate Choi matrix of what a processor $G$ with pure program state $\widetilde{\xi} = V^{p} \xi V^{p\dagger}$, where $V^{p} = \sum_{rq}^{P} v_{rq} \dyad{r}{q}$ is unitary operator, implements:
	\begin{align*} \label{NandSCondEq1}
		&\left(\mathfrak{O}_{G, \widetilde{\xi}}^{pr} \otimes \mathcal{I}\right)\left(\sum_{ii^{\prime}}^{D} \dyad{ii}{i^{\prime}i^{\prime}}\right) \underset{(\ref{choi})}{\overset{(\ref{SprobQPelem})}{=}} \sum_{ii^{\prime}}^{D} \Tr_{p}\left[G \left(\dyad{i}{i^{\prime}} \otimes \widetilde{\xi}\right) G^{\dagger} \frac{1}{P}\left(\mathbb{1} \otimes \sum_{nn^{\prime}}^{P} \dyad{n}{n^{\prime}}\right)\right] \otimes \dyad{i}{i^{\prime}} \\
		&= \frac{1}{P} \sum_{ii^{\prime}}^{D} \Tr_{p} \bigg[\left(\sum_{jk}^{P} A_{jk} \otimes \dyad{j}{k}\right) \left(\dyad{i}{i^{\prime}} \otimes \widetilde{\xi}\right) \left(\sum_{j^{\prime}k^{\prime}}^{P} A_{j^{\prime}k^{\prime}}^{\dagger} \otimes \dyad{k^{\prime}}{j^{\prime}}\right) \left(\mathbb{1} \otimes \sum_{n,n^{\prime}}^{P} \dyad{n}{n^{\prime}}\right)\bigg] \otimes \dyad{i}{i^{\prime}} \\
		&= \frac{1}{P} \sum_{ii^{\prime}}^{D} \sum_{rr^{\prime}}^{P} \sum_{qq^{\prime}}^{P} \sum_{jj^{\prime}}^{P} \sum_{kk^{\prime}}^{P} \sum_{nn^{\prime}}^{P} v_{rq} v_{r^{\prime}q^{\prime}}^{\ast} A_{jk} \dyad{i}{i^{\prime}} A_{j^{\prime}k^{\prime}}^{\dagger} \Tr\left(\dyad{j}{k} \hspace*{-1.5mm} \dyad{r}{q} \xi \dyad{q^{\prime}}{r^{\prime}} \hspace*{-1.5mm} \dyad{k^{\prime}}{j^{\prime}} \hspace*{-1.5mm} \dyad{n}{n^{\prime}}\right) \otimes \dyad{i}{i^{\prime}} \\
		&\overset{(i)}{=} \frac{1}{P} \sum_{ii^{\prime}}^{D} \sum_{qq^{\prime}}^{P} \sum_{jj^{\prime}}^{P} \sum_{kk^{\prime}}^{P} v_{kq} \xi_{q} A_{jk} \dyad{i}{i^{\prime}} v_{k^{\prime}q^{\prime}}^{\ast} \xi_{q^{\prime}}^{\ast} A_{j^{\prime}k^{\prime}}^{\dagger} \otimes \dyad{i}{i^{\prime}}, \numberthis
	\end{align*}
	where in $(i)$ we have used that the program state is pure.	Similarly, we can calculate Choi matrix for processors $UG$ and $GV$ with pure program states:
	\begin{align*}
		&\left(\mathfrak{O}_{UG,\xi}^{pr} \otimes \mathcal{I}\right)\left(\sum_{i,i^{\prime}}^{D} \dyad{ii}{i^{\prime}i^{\prime}}\right) \underset{(\ref{choi})}{\overset{(\ref{SprobQPelem})}{=}} \frac{1}{P} \sum_{ii^{\prime}}^{D} \sum_{rr^{\prime}}^{P} \sum_{jj^{\prime}}^{P} \sum_{kk^{\prime}}^{P} \xi_{k} U_{rj} A_{jk} \dyad{i}{i^{\prime}} \xi_{k^{\prime}}^{\ast} A_{j^{\prime}k^{\prime}}^{\dagger} U_{r^{\prime}j^{\prime}}^{\dagger} \otimes \dyad{i}{i^{\prime}} \\
		&\left(\mathfrak{O}_{GV,\xi}^{pr} \otimes \mathcal{I}\right)\left(\sum_{i,i^{\prime}}^{D} \dyad{ii}{i^{\prime}i^{\prime}}\right) \underset{(\ref{choi})}{\overset{(\ref{SprobQPelem})}{=}} \frac{1}{P} \sum_{ii^{\prime}}^{D} \sum_{qq^{\prime}}^{P} \sum_{jj^{\prime}}^{P} \sum_{kk^{\prime}}^{P} \xi_{q} A_{jk} V_{kq} \dyad{i}{i^{\prime}} \xi_{q^{\prime}}^{\ast} V_{k^{\prime}q^{\prime}}^{\dagger} A_{j^{\prime}k^{\prime}}^{\dagger} \otimes \dyad{i}{i^{\prime}}.
	\end{align*}
	Because two Choi matrices must be the same to describe the same quantum channel, we can compare previous equations and recover conditions similar to those in our theorem:
	\begin{align*}
		\sum_{jkq}^{P} v_{kq} \xi_{q} A_{jk} &\overset{!}{=} e^{i\phi} \sum_{jkr}^{P} \xi_{k} U_{rj} A_{jk} \\
		\sum_{jkq}^{P} v_{kq} \xi_{q} A_{jk} &\overset{!}{=} e^{i\psi} \sum_{jkq}^{P} \xi_{q} A_{jk} V_{kq}.
	\end{align*}
	Due to the definition of structural equivalence, we can add a positive real number to the previous conditions and arrive at their final form:
	\begin{align*} \label{NandSCondEq2}
		\sum_{jkq}^{P} v_{kq} \xi_{q} A_{jk} &\overset{!}{=} e^{i\phi} \sqrt{K_{\xi, \widetilde{\xi}}} \sum_{jkr}^{P} \xi_{k} U_{rj} A_{jk} \numberthis \\ 
		\sum_{jkq}^{P} v_{kq} \xi_{q} A_{jk} &\overset{!}{=} e^{i\psi} \sqrt{C_{\xi, \widetilde{\xi}}} \sum_{jkq}^{P} \xi_{q} A_{jk} V_{kq}.
	\end{align*}
	Let us substitute equation (\ref{NandSCondEq2}) back into (\ref{NandSCondEq1}):
	\begin{align*}
		&\left(\mathfrak{O}_{G, \widetilde{\xi}}^{pr} \otimes \mathcal{I}\right)\left(\sum_{i,i^{\prime}}^{D} \dyad{ii}{i^{\prime}i^{\prime}}\right) = \frac{1}{P} \sum_{ii^{\prime}}^{D} \sum_{jkr}^{P} \sum_{j^{\prime}k^{\prime}r^{\prime}}^{P} e^{i\phi} \sqrt{K_{\xi, \widetilde{\xi}}} \xi_{k} U_{rj} A_{jk} \dyad{i}{i^{\prime}} e^{-i\phi} \sqrt{K_{\xi, \widetilde{\xi}}} \xi_{k^{\prime}}^{\ast} A_{j^{\prime}k^{\prime}}^{\dagger} U_{r^{\prime}j^{\prime}}^{\dagger} \otimes \dyad{i}{i^{\prime}} \\
		&= \frac{K_{\xi, \widetilde{\xi}}}{P} \sum_{ii^{\prime}}^{D} \sum_{jkr}^{P} \sum_{j^{\prime}k^{\prime}r^{\prime}}^{P} \xi_{k} U_{rj} A_{jk} \dyad{i}{i^{\prime}} \xi_{k^{\prime}}^{\ast} A_{j^{\prime}k^{\prime}}^{\dagger} U_{r^{\prime}j^{\prime}}^{\dagger} \otimes \dyad{i}{i^{\prime}} = K_{\xi, \widetilde{\xi}} \left(\mathfrak{O}_{UG,\xi}^{pr} \otimes \mathcal{I}\right)\left(\sum_{ii^{\prime}}^{D} \dyad{ii}{i^{\prime}i^{\prime}}\right).
	\end{align*}	
	Therefore, processors $G$ and $UG$ are, indeed, structurally equivalent $G \sim_{st} UG$. Similarly, one can show that also processors $G$ and $GV$ are structurally equivalent.
\end{proof}

In what follows, we shall make use of the following notation for operators $\sum\limits_{j}^{P} A_{jk} \equiv \mathscr{A}_{k}$ and $\sum\limits_{m}^{\widetilde{P}} B_{mn} = \mathscr{B}_{n}$. Therefore, set of implemented operations is $\mathfrak{G}^{pr}_{G,\xi} \overset{(\ref{prQPimpl})}{=} \frac{1}{P} \sum\limits_{jkj^{\prime}k^{\prime}}^{P} \xi_{kk^{\prime}} A_{jk} \varrho A_{j^{\prime}k^{\prime}}^{\dagger} = \frac{1}{P} \sum\limits_{kk^{\prime}}^{P} \xi_{kk^{\prime}} \mathscr{A}_{k} \varrho \mathscr{A}_{k^{\prime}}^{\dagger}$.
\begin{theorem} \label{spansOfOperators}
	Spans of Operators
	\\
	Let us have two quantum probabilistic processors $G = \sum\limits_{jk}^{P} A_{jk} \otimes \dyad{j}{k}$ and $\widetilde{G} = \sum\limits_{mn}^{\widetilde{P}} B_{mn} \otimes \dyad{m}{n}$. Processors are structurally equivalent $G \sim_{st} \widetilde{G}$ if and only if $\mathscr{B}_{n}  = \sum\limits_{k}^{\widetilde{P}} a_{nk} \mathscr{A}_{k}$ and $\mathscr{A}_{k} = \sum\limits_{n}^{P} b_{kn} \mathscr{B}_{n}$, where $a_{nk}, b_{kn} \in \mathbb{C}$ for all $n$ and $k$.
\end{theorem}
Symbol $\mathbb{C}$ denotes set of complex numbers.
\begin{proof}
	Firstly, let us assume that $\mathscr{B}_{n}  = \sum\limits_{k}^{P} a_{nk} \mathscr{A}_{k}$ and show that processors are structurally equivalent $G \sim_{st} \widetilde{G}$. Let us calculate what $\widetilde{G}$ implements:
	\begin{align*}
		\mathfrak{O}^{pr}_{\widetilde{G},\widetilde{\xi}} = \frac{1}{\widetilde{P}} \sum\limits_{nn^{\prime}}^{\widetilde{P}} \widetilde{\xi}_{nn^{\prime}} \mathscr{B}_{n} \varrho \mathscr{B}_{n^{\prime}}^{\dagger} = \frac{1}{\widetilde{P}} \sum\limits_{nn^{\prime}}^{\widetilde{P}} \sum\limits_{kk^{\prime}}^{P} \widetilde{\xi}_{nn^{\prime}} a_{nk} a_{n^{\prime}k^{\prime}}^{\ast} \mathscr{A}_{k} \varrho \mathscr{A}_{k^{\prime}} = \frac{K_{\widetilde{\xi}}}{P} \sum\limits_{kk^{\prime}}^{P} \xi_{kk^{\prime}} \mathscr{A}_{k} \varrho \mathscr{A}_{k^{\prime}}^{\dagger} = K_{\widetilde{\xi}} \mathfrak{O}_{G,\xi}^{pr},
	\end{align*}
	where $K_{\widetilde{\xi}} \xi_{kk^{\prime}} = \frac{P}{\widetilde{P}} \sum_{nn^{\prime}}^{\widetilde{P}} \xi_{nn^{\prime}} a_{nk} a_{n^{\prime}k^{\prime}}^{\ast}$ and $K_{\widetilde{\xi}} \in \mathbb{R}_{>0}$. Let us use normalization condition of a quantum state:
	\begin{align*}
		\sum_{k}^{P} \xi_{kk} = \frac{1}{K_{\widetilde{\xi}}} \frac{P}{\widetilde{P}} \sum_{k}^{P} \sum_{nn^{\prime}}^{\widetilde{P}} \widetilde{\xi}_{nn^{\prime}} a_{nk} a_{n^{\prime}k}^{\ast} \overset{!}{=} 1.
	\end{align*}
	Thus, relation for $K_{\widetilde{\xi}}$ is:
	\begin{align*}
		K_{\widetilde{\xi}} = \frac{P}{\widetilde{P}} \sum_{k}^{P} \sum_{nn^{\prime}}^{\widetilde{P}} \widetilde{\xi}_{nn^{\prime}} a_{nk} a_{n^{\prime}k}^{\ast}.
	\end{align*}
	However, this only shows that for every element from set $\mathfrak{O}_{\widetilde{G}}^{pr}$, it is possible to find corresponding element in $\mathfrak{O}_{G}^{pr}$. Let us now turn the situation around and use that also operators $\{\mathscr{A}_{k}\}$ must be linear combinations of $\{\mathscr{B}_{n}\}$. We shall not repeat the same calculation here, but only write the result:
	\begin{align*}
		\mathfrak{O}_{G,\xi}^{pr} = K_{\xi} \mathfrak{O}_{\widetilde{G}, \widetilde{\xi}}^{pr},
	\end{align*}
	where $K_{\xi} = \frac{\widetilde{P}}{P} \sum_{n}^{\widetilde{P}}\sum_{kk^{\prime}}^{P} \xi_{kk^{\prime}} b_{kn} b_{k^{\prime}n^{\prime}}^{\ast}$. Therefore, number $K_{\xi, \widetilde{\xi}}$ from the definition \ref{structEqv} of structural equivalence, can be sometimes equal to $K_{\xi}$ and sometimes to $\frac{1}{K_{\widetilde{\xi}}}$. We can conclude that processors are operationally equivalent $G \sim_{st} \widetilde{G}$.
	
	Now, we shall assume that processors are equivalent $G \sim_{st} \widetilde{G}$. Let us assume that there exists such an operator $\mathscr{B}_{v}$ that cannot be expressed as linear combination of operators from processor $G$, i.e., $\mathscr{B}_{v} \neq \sum\limits_{k}^{P} a_{vk} \mathscr{A}_{k}$. We shall choose program state to be $\widetilde{\xi} = \dyad{v}{v}$, therefore:
	\begin{align*}
		\mathfrak{O}_{\widetilde{G}, \widetilde{\xi}}^{pr} = \frac{1}{\widetilde{P}}  \mathscr{B}_{v} \varrho \mathscr{B}_{v}^{\dagger}.
	\end{align*}
	However, no such quantum operation and nor any multiple of this quantum operation can be implemented by processor $G$ as:
	\begin{align*}
		\frac{1}{\widetilde{P}} \mathscr{B}_{v} \varrho \mathscr{B}_{v}^{\dagger} \neq \frac{1}{\widetilde{P}} \sum\limits_{kk^{\prime}}^{P} \widetilde{\xi}_{vv} a_{vk} a_{vk^{\prime}}^{\ast} \mathscr{A}_{k} \varrho \mathscr{A}_{k^{\prime}}^{\dagger},
	\end{align*}
	where on the right-hand side there is $K_{\widetilde{\xi}} \mathfrak{O}_{G, \xi}^{pr}$. We could also repeat similar calculation for operators $\{\mathscr{A}_{k}\}$. Therefore, operators $\{\mathscr{B}_{n}\}$ must be linear combinations of $\{\mathscr{A}_{k}\}$ and vice versa.
\end{proof}
Previous theorem says that the spans of operators $\{\mathscr{A}_{k}\}$ and $\{\mathscr{B}_{n}\}$ must be the same. Let us proceed with a corollary, where we put restriction on operators $\mathscr{A}_{k}$ and $\mathscr{B}_{n}$.
\begin{corollary}
	Structural Equivalence with Orthogonal Operators
	\\
	Let us have processors $G = \sum_{jk}^{P} A_{jk} \otimes \dyad{j}{k}$ and $\widetilde{G} = \sum_{mn}^{\widetilde{P}} B_{mn} \otimes \dyad{m}{n}$ for which the following conditions hold: $\Tr\left(\mathscr{A}_{k}^{\dagger} \mathscr{A}_{k^{\prime}}\right) = D\delta_{kk^{\prime}}$ and $\Tr\left(\mathscr{B}_{n}^{\dagger} \mathscr{B}_{n^{\prime}}\right) = D \delta_{nn^{\prime}}$, where ${D}$ is dimension of data spaces of both processors $G$ and $\widetilde{G}$. These processors are structurally equivalent $G \sim_{st} \widetilde{G}$ if and only if $\mathscr{B}_{n} = \sum_{k}^{P} a_{nk} \mathscr{A}_{k}$ and $\mathscr{A}_{k} = \sum\limits_{n}^{\widetilde{P}} b_{kn} \mathscr{B}_{n}$, where $a_{nk}, b_{kn} \in \mathbb{C}$ for all $n$ and $k$. And also, $\{a_{nk}\}$ and $\{b_{kn}\}$ form co-isometries.
\end{corollary}
Co-isometric matrix $M$ fulfills the following equation: $MM^{\dagger} = \mathbb{1}$.
\begin{proof}
	Proposition that if $\mathscr{B}_{n} = \sum_{k}^{P} a_{nk} \mathscr{A}_{k}$ and $\mathscr{A}_{k} = \sum\limits_{n}^{P} b_{kn} \mathscr{B}_{n}$, then $G \sim_{st} \widetilde{G}$, was already proven in the previous theorem \ref{spansOfOperators}.
	
	Let us now assume that the processors are equivalent. From previous theorem it follows that in such a case, the relation between operators is $\mathscr{B}_{n} = \sum\limits_{k}^{P} a_{nk} \mathscr{A}_{k}$. Let us also assume that the following conditions on operators hold:
	\begin{align*}
		\Tr\left(\mathscr{A}_{k}^{\dagger} \mathscr{A}_{k^{\prime}}\right) = D\delta_{kk^{\prime}}, \qquad
		\Tr\left(\mathscr{B}_{n}^{\dagger} \mathscr{B}_{n^{\prime}}\right) = D \delta_{nn^{\prime}}.
	\end{align*}
	Let us now use previous restrictions:
	\begin{align*}
		\Tr\left(\mathscr{B}_{n}^{\dagger} \mathscr{B}_{n^{\prime}}\right) &= \Tr\left[\left(\sum_{k}^{P} a_{nk}^{\ast} \mathscr{A}_{k}^{\dagger}\right) \left(\sum_{k^{\prime}}^{P} a_{n^{\prime}k^{\prime}} \mathscr{A}_{k^{\prime}}\right)\right] = \sum_{kk^{\prime}}^{P} a_{nk}^{\ast} a_{n^{\prime}k^{\prime}} \Tr\left(\mathscr{A}_{k}^{\dagger} \mathscr{A}_{k^{\prime}}\right) \overset{(i)}{=} \sum_{kk^{\prime}}^{P} a_{nk}^{\ast} a_{n^{\prime}k^{\prime}} \delta_{kk^{\prime}} D \\&= \sum_{k}^{P} a_{nk}^{\ast} a_{n^{\prime}k} D \overset{(ii)}{=} \delta_{nn^{\prime}} D,
	\end{align*}
	where in $(i)$ we have used the assumption and from $(ii)$ it follows that $\sum_{k}^{P} a_{nk}^{\ast} a_{n^{\prime}k} = \delta_{nn^{\prime}}$ and therefore, $\{a_{nk}\}$ form co-isometry. Calculation for $\{b_{kn}\}$ forming co-isometry is, basically, identical. 
\end{proof}

Furthermore, let us examine equivalence of two probabilistic processors which can be found in literature. Nielsen and Chuang proposed probabilistic processor $G_{NC}$ based on quantum teleportation able to implement arbitrary unitary \cite{ProgrammableQuantumGateArrays} channel, while in the book written by Teiko Heinosaari and M\'{a}rio Ziman \cite{TheMathematicalLanguageOfQuantumTheoryFromUncertaintyToEntanglement} one can found processor $G_{S}$ using SWAP gate that is able to implement arbitrary channel. However, these processors make use of different success measurements compared to the one chosen in present work. Their success measurement is $M = \mathbb{1} \otimes \frac{1}{L} \sum_{mn}^{L} \dyad{mm}{nn}$ with $L = \log_{2}\left[\dim\left({\cal{H}}_{p}\right)\right]$. Let us rewrite general expression of processor from equation (\ref{QPeq}) to better mirror notation for chosen measurement: $G = \sum_{jk}^{P} A_{jk} \otimes \dyad{j}{k} = \sum_{ijkl}^{L} A_{ij,kl} \otimes \dyad{ij}{kl}$. Now we shall proceed with calculation of what a processor with such a measurement implements.
\begin{align*}\label{QPImplDiffM}
	\mathfrak{G}_{G,\xi}^{pr} &= \frac{1}{p} \Tr_{p} \left[G \left(\varrho \otimes \xi\right) G^{\dagger} M\right] \\
	&= \frac{1}{p} \Tr_{p} \left[\left(\sum\limits_{ijkl}^{L} A_{ij,kl} \otimes \dyad{ij}{kl}\right) \left(\varrho \otimes \xi\right) \left(\sum\limits_{i^{\prime}j^{\prime}k^{\prime}l^{\prime}}^{L} A_{i^{\prime}j^{\prime},k^{\prime}l^{\prime}}^{\dagger} \otimes \dyad{i^{\prime}j^{\prime}}{k^{\prime}l^{\prime}}\right) \left(\mathbb{1} \otimes \frac{1}{L} \sum\limits_{mn}^{L} \dyad{mm}{nn}\right)\right] \\
	&= \frac{1}{pL} \sum\limits_{ijkl}^{L} \sum\limits_{i^{\prime}j^{\prime}k^{\prime}l^{\prime}}^{L} \sum\limits_{mn}^{L} A_{ij,kl} \varrho A_{i^{\prime}j^{\prime},k^{\prime}l^{\prime}}^{\dagger} \Tr\left(\dyad{ij}{kl} \xi \dyad{k^{\prime}l^{\prime}}{i^{\prime}j^{\prime}} \hspace*{-1.5mm} \dyad{mm}{nn}\right) \\
	&= \frac{1}{pL} \sum\limits_{klk^{\prime}l^{\prime}}^{L} \sum\limits_{jj^{\prime}}^{L} \xi_{kl,k^{\prime}l^{\prime}} A_{jj,kl} \varrho A_{j^{\prime}j^{\prime},k^{\prime}l^{\prime}}^{\dagger}, \numberthis
\end{align*}
where $p$ denotes probability of successful implementation. We can finally devote our attention to the processors themselves.
\begin{example}
	Universal Processors
	\\
	Processor $G_{NC}$ is formed by three CNOT gates which we shall denote with $C_{ij}$ where $i$ is control and $j$ denotes target qubit.
	\begin{align*}
		G_{NC} &= C_{02}C_{20}C_{02} \\
		&= \dyad{0}{0} \otimes \dyad{00}{00} + \dyad{0}{1} \otimes \dyad{01}{00} + \dyad{0}{0} \otimes \dyad{10}{10} + \dyad{0}{1} \otimes \dyad{11}{10} \\
		&+ \dyad{1}{0} \otimes \dyad{00}{01} + \dyad{1}{1} \otimes \dyad{01}{01} + \dyad{1}{0} \otimes \dyad{10}{11} + \dyad{1}{1} \otimes \dyad{11}{11} \\
		&= A_{00,00}^{NC} \otimes \dyad{00}{00}  + A_{01,00}^{NC} \otimes \dyad{01}{00} + A_{10,10}^{NC} \otimes \dyad{10}{10} + A_{11,10}^{NC} \otimes \dyad{11}{10} \\
		&+ A_{00,01}^{NC} \otimes \dyad{00}{01} + A_{01,01}^{NC} \otimes \dyad{01}{01} + A_{10,11}^{NC} \otimes \dyad{10}{11} + A_{11,11}^{NC} \otimes \dyad{11}{11}.
	\end{align*}

	The second processor $G_{S}$ is based on SWAP gate, which we shall denote with $S$.
	\begin{align*}
		G_{S} &= S \otimes \mathbb{1} \\
		&= \dyad{0}{0} \otimes \dyad{00}{00} + \dyad{0}{1} \otimes \dyad{10}{00} + \dyad{1}{0} \otimes \dyad{00}{10} + \dyad{1}{1} \otimes \dyad{10}{10} \\
		&+ \dyad{0}{0} \otimes \dyad{01}{01} + \dyad{0}{1} \otimes \dyad{11}{01} + \dyad{1}{0} \otimes \dyad{01}{11} + \dyad{1}{1} \otimes \dyad{11}{11} \\
		&= A_{00,00}^{S} \otimes \dyad{00}{00}  + A_{10,00}^{S} \otimes \dyad{10}{00} + A_{00,10}^{S} \otimes \dyad{00}{10} + A_{10,10}^{S} \otimes \dyad{10}{10} \\
		&+ A_{01,01}^{S} \otimes \dyad{01}{01} + A_{11,01}^{S} \otimes \dyad{11}{01} + A_{01,11}^{S} \otimes \dyad{01}{11} + A_{11,11}^{S} \otimes \dyad{11}{11}.
	\end{align*}

	Nielsen-Chuang processor implements the following transformations:
	\begin{align*}
		\mathfrak{G}_{G_{NC},\xi}^{pr} &\overset{(\ref{QPImplDiffM})}{=} \frac{1}{2p} \sum\limits_{klk^{\prime}l^{\prime}}^{2} \xi_{kl,k^{\prime}l^{\prime}} \left(A_{00,kl}^{NC} \varrho A_{00,k^{\prime}l^{\prime}}^{NC\dagger} + A_{00,kl}^{NC} \varrho A_{11,k^{\prime}l^{\prime}}^{NC\dagger} + A_{11,kl}^{NC} \varrho A_{00,k^{\prime}l^{\prime}}^{NC\dagger} + A_{11,kl}^{NC} \varrho A_{11,k^{\prime}l^{\prime}}^{NC\dagger}\right) \\
		&\hspace*{3mm}= \frac{1}{2p} (\xi_{00,00} \dyad{0}{0}\varrho\dyad{0}{0} + \xi_{00,01} \dyad{0}{0}\varrho\dyad{0}{1} + \xi_{01,00} \dyad{1}{0}\varrho\dyad{0}{0} + \xi_{01,01} \dyad{1}{0}\varrho\dyad{0}{1} \\
		&\qquad + \xi_{00,10} \dyad{0}{0}\varrho\dyad{1}{0} + \xi_{00,11} \dyad{0}{0}\varrho\dyad{1}{1} + \xi_{01,10} \dyad{1}{0}\varrho\dyad{1}{0} + \xi_{01,11} \dyad{1}{0}\varrho\dyad{1}{1} \\
		&\qquad + \xi_{10,00} \dyad{0}{1}\varrho\dyad{0}{0} + \xi_{10,01} \dyad{0}{1}\varrho\dyad{0}{1} + \xi_{11,00} \dyad{1}{1}\varrho\dyad{0}{0} + \xi_{11,01} \dyad{1}{1}\varrho\dyad{0}{1} \\
		&\qquad + \xi_{10,10} \dyad{0}{1}\varrho\dyad{1}{0} + \xi_{10,11} \dyad{0}{1}\varrho\dyad{1}{1} + \xi_{11,10} \dyad{1}{1}\varrho\dyad{1}{0} + \xi_{11,11} \dyad{1}{1}\varrho\dyad{1}{1}).
	\end{align*}
	And SWAP processor implements the following transformations:
	\begin{align*}
		\mathfrak{G}_{G_{S},\widetilde{\xi}}^{pr} &\overset{(\ref{QPImplDiffM})}{=} \frac{1}{2p} \sum\limits_{klk^{\prime}l^{\prime}}^{2} \xi_{kl,k^{\prime}l^{\prime}} \left(A_{00,kl}^{S} \varrho A_{00,k^{\prime}l^{\prime}}^{S\dagger} + A_{00,kl}^{S} \varrho A_{11,k^{\prime}l^{\prime}}^{S\dagger} + A_{11,kl}^{S} \varrho A_{00,k^{\prime}l^{\prime}}^{S\dagger} + A_{11,kl}^{S} \varrho A_{11,k^{\prime}l^{\prime}}^{S\dagger}\right) \\
		&\hspace*{3mm}= \frac{1}{2p} (\widetilde{\xi}_{00,00} \dyad{0}{0}\varrho\dyad{0}{0} + \widetilde{\xi}_{00,10} \dyad{0}{0}\varrho\dyad{0}{1} + \widetilde{\xi}_{10,00} \dyad{1}{0}\varrho\dyad{0}{0} + \widetilde{\xi}_{10,10} \dyad{1}{0}\varrho\dyad{0}{1} \\
		&\qquad + \widetilde{\xi}_{00,01} \dyad{0}{0}\varrho\dyad{1}{0} + \widetilde{\xi}_{00,11} \dyad{0}{0}\varrho\dyad{1}{1} + \widetilde{\xi}_{10,01} \dyad{1}{0}\varrho\dyad{1}{0} + \widetilde{\xi}_{10,11} \dyad{1}{0}\varrho\dyad{1}{1} \\
		&\qquad + \widetilde{\xi}_{01,00} \dyad{0}{1}\varrho\dyad{0}{0} + \widetilde{\xi}_{01,10} \dyad{0}{1}\varrho\dyad{0}{1} + \widetilde{\xi}_{11,00} \dyad{1}{1}\varrho\dyad{0}{0} + \widetilde{\xi}_{11,10} \dyad{1}{1}\varrho\dyad{0}{1} \\
		&\qquad + \widetilde{\xi}_{01,01} \dyad{0}{1}\varrho\dyad{1}{0} + \widetilde{\xi}_{01,11} \dyad{0}{1}\varrho\dyad{1}{1} + \widetilde{\xi}_{11,01} \dyad{1}{1}\varrho\dyad{1}{0} + \widetilde{\xi}_{11,11} \dyad{1}{1}\varrho\dyad{1}{1}).
	\end{align*}
	Relation between program states is given by SWAP $\widetilde{\xi} = S \xi S$. Therefore, all possible implemented channels are implemented with the same probabilities. These processors are strongly probabilistically equivalent $G_{NC} \sim_{pr} G_{S}$. On top of that, these processors are also structurally equivalent $G_{NC} \sim_{st} G_{S}$.
	
	Relation between processors $G_{NC}$ and $G_{S}$ is given by unitary $U$:
	\begin{align*}
		U &= \dyad{000}{000} + \dyad{001}{010} + \dyad{010}{001} + \dyad{011}{011}\\
		&+ \dyad{100}{100} + \dyad{101}{110} + \dyad{110}{101} + \dyad{111}{111} \\
		&= \left(\dyad{0}{0}\right) \otimes \left(\dyad{00}{00} + \dyad{01}{10} + \dyad{10}{01} + \dyad{11}{11}\right) = \mathbb{1} \otimes S,
	\end{align*}
	where $UG_{NC}U^{\dagger} = G_{S}$.
	
	In conclusion, we realize that not only processor $G_{S}$, but also $G_{NC}$ is able to implement arbitrary quantum channel.
\end{example}

We shall proceed with investigation of equivalence of probabilistic processors when the probabilities of successful measurements might differ as per definition \ref{weakEqv}. In the following theorem, we provide sufficient condition on weak equivalence of processors with different dimensions of program spaces. 
\begin{theorem}\label{sufNecDimProbWeak}
	Relation between Probabilistically Equivalent Processors
	\\
	Let us have processors $G = \sum_{jk}^{P} A_{jk} \otimes \dyad{j}{k}$ and $\widetilde{G} = \left(\sum_{r}^{D} \dyad{r}{r} \otimes U_{r}\right) G \left(\mathbb{1} \otimes V^{\dagger}\right)$ with pure program states $\xi$ and $\widetilde{\xi}$ respectively and where $U_{r} = \sum_{x}^{\widetilde{P}} \sum_{j}^{P} u_{xj}^{r} \dyad{x}{j}$ are isometric operators for all $r$ and operator $V$ is co-isometry. Then, these processors are weakly probabilistically equivalent $G \approx_{pr} \widetilde{G}$, if $\sum_{xx^{\prime}}^{\widetilde{P}} u_{xj}^{r} u_{x^{\prime}j^{\prime}}^{r^{\prime}\star} = 1$. Furthermore, these processors are strongly probabilistically equivalent $G \sim_{pr} \widetilde{G}$, if also $P = \widetilde{P}$.
\end{theorem}
\begin{proof}
	Relation between pure program states is always chaperoned by a unitary matrix $\widetilde{\xi} = V \xi V^{\dagger}$. Let us calculate what $\widetilde{G}$ implements:
	\begin{align*}
		\mathfrak{C}_{\widetilde{G}, \widetilde{\xi}}^{pr}
		&= \frac{1}{\widetilde{p}} \Tr_{p} \left[ \left(\sum_{r}^{D} \dyad{r}{r} \otimes U_{r}\right) G \left(\mathbb{1} \otimes V^{\dagger}\right) \left(\varrho \otimes \widetilde{\xi}\right) \left(\mathbb{1} \otimes V\right) G^{\dagger} \left(\sum_{r^{\prime}}^{D} \dyad{r^{\prime}}{r^{\prime}} \otimes U_{r^{\prime}}^{\dagger}\right) \left(\mathbb{1} \otimes \frac{1}{\widetilde{P}} \sum_{nn^{\prime}}^{\widetilde{P}} \dyad{n}{n^{\prime}}\right) \right]\\
		&= \frac{1}{\widetilde{p}\widetilde{P}} \Tr_{p} \Bigg[ \left(\sum_{r}^{D} \dyad{r}{r} \otimes U_{r}\right) \left(\sum_{jk}^{P} A_{jk} \otimes \dyad{j}{k}\right) \left(\mathbb{1} \otimes V^{\dagger}\right) \left(\varrho \otimes V \xi V^{\dagger}\right) \left(\mathbb{1} \otimes V\right) \\ &\quad\left(\sum_{j^{\prime}k^{\prime}}^{P} A_{j^{\prime}k^{\prime}}^{\dagger} \otimes \dyad{k^{\prime}}{j^{\prime}}\right) \left(\sum_{r^{\prime}}^{D} \dyad{r^{\prime}}{r^{\prime}}  \otimes U_{r^{\prime}}^{\dagger}\right) \left(\mathbb{1} \otimes \sum_{n,n^{\prime}}^{\widetilde{P}} \dyad{n}{n^{\prime}}\right) \Bigg]\\
		&= \frac{1}{\widetilde{p}\widetilde{P}} \sum_{rr^{\prime}}^{D}\sum_{jk}^{P}\sum_{j^{\prime}k^{\prime}}^{P} \dyad{r}{r} A_{jk} \varrho A_{j^{\prime}k^{\prime}}^{\dagger} \dyad{r^{\prime}}{r^{\prime}} \sum_{nn^{\prime}}^{\widetilde{P}} \Tr \left[U_{r} \dyad{j}{k} \xi \dyad{k^{\prime}}{j^{\prime}} U_{r^{\prime}}^{\dagger} \dyad{n}{n^{\prime}}\right]\\
		&= \frac{1}{\widetilde{p}\widetilde{P}} \sum_{rr^{\prime}}^{D}\sum_{jk}^{P}\sum_{j^{\prime}k^{\prime}}^{P} \dyad{r}{r} A_{jk} \varrho A_{j^{\prime}k^{\prime}}^{\dagger} \dyad{r^{\prime}}{r^{\prime}}  \sum_{nn^{\prime}}^{\widetilde{P}} \sum_{xx^{\prime}}^{\widetilde{P}} \sum_{yy^{\prime}}^{P} u_{xy}^{r} u_{x^{\prime}y^{\prime}}^{r^{\prime}\star} \Tr \left[\dyad{x}{y} \hspace*{-1.5mm} \dyad{j}{k} \xi \dyad{k^{\prime}}{j^{\prime}} \hspace*{-1.5mm} \dyad{y^{\prime}}{x^{\prime}} \hspace*{-1.5mm} \dyad{n}{n^{\prime}}\right]\\
		&= \frac{1}{\widetilde{p}\widetilde{P}} \sum_{rr^{\prime}}^{D}\sum_{jk}^{P}\sum_{j^{\prime}k^{\prime}}^{P} \xi_{kk^{\prime}} \dyad{r}{r} A_{jk} \varrho A_{j^{\prime}k^{\prime}}^{\dagger} \dyad{r^{\prime}}{r^{\prime}} \sum_{xx^{\prime}}^{\widetilde{P}} u_{xj}^{r} u_{x^{\prime}j^{\prime}}^{r^{\prime}\star}\\
		&\overset{(i)}{=} \frac{1}{\widetilde{p}\widetilde{P}} \sum_{jk}^{P} \sum_{j^{\prime}k^{\prime}}^{P} \xi_{kk^{\prime}} \sum_{r}^{P} \dyad{r}{r} A_{jk} \varrho A_{j^{\prime}k^{\prime}}^{\dagger} \sum_{r^{\prime}} \dyad{r^{\prime}}{r^{\prime}} = \frac{1}{\widetilde{p}\widetilde{P}} \sum_{jk}^{P} \sum_{j^{\prime}k^{\prime}}^{P} \xi_{kk^{\prime}} A_{jk} \varrho A_{j^{\prime}k^{\prime}}^{\dagger}\\
		&\overset{(ii)}{=} \frac{1}{pP} \sum_{jkj^{\prime}k^{\prime}}^{P} \xi_{kk^{\prime}} A_{jk} \varrho A_{j^{\prime}k^{\prime}}^{\dagger} = \mathfrak{C}_{G, \xi}^{pr}.
	\end{align*}
	In $(i)$ we have used the assumption that $\sum_{xx^{\prime}}^{\widetilde{P}} u_{xj}^{r} u_{x^{\prime}j^{\prime}}^{r^{\prime}\star} = 1$ and in $(ii)$ we are using the relation between probabilities:
	\begin{align*}
		\widetilde{p} \overset{(\ref{probQP})}{=} \Tr \left[\frac{1}{\widetilde{P}} \sum_{jkj^{\prime}k^{\prime}}^{P} \xi_{kk^{\prime}} A_{jk} \varrho A_{k^{\prime}j^{\prime}}^{\dagger}\right] = \frac{1}{\widetilde{P}} P \Tr \left[\frac{1}{P} \sum_{jkj^{\prime}k^{\prime}}^{P} \xi_{kk^{\prime}} A_{jk} \varrho A_{k^{\prime}j^{\prime}}^{\dagger}\right] = \frac{P}{\widetilde{P}} p.
	\end{align*}
	We can see that if $P = \widetilde{P}$, then probabilities are equal $p = \widetilde{p}$ and thus in such a case, processors are strongly probabilistically equivalent $G \sim_{pr} \widetilde{G}$. 
	Moreover, $U_{r}$ must be isometry $U_{r}^{\dagger} U_{r} = \mathbb{1}$ and $V$ must be co-isometry $VV^{\dagger} = \mathbb{1}$ to accommodate different dimensions of program registers of $G$ and $\widetilde{G}$.
\end{proof}

Probability of implementing the desired operation depends not only on processors $G$ and $\widetilde{G}$, but also on chosen program states. Let us present an example, where one time $p > \widetilde{p}$ and then the next time $\widetilde{p} > p$.
\begin{example}\label{diffProbEx}
	Processors $G = \mathbb{1} \otimes \dyad{0}{0} + \sigma_{z} \otimes \dyad{1}{1}$ and $\widetilde{G} = \mathbb{1} \otimes \dyad{0}{0} + \mathbb{1} \otimes \dyad{1}{1} + \sigma_{z} \otimes \dyad{2}{2}$ are weakly probabilistically equivalent as they implement the following transformations:
	\begin{align*}
		&\mathfrak{C}_{G,\xi}^{pr} \overset{(\ref{elemOfImplPr})}{=} \frac{1}{2p} (\xi_{00} \mathbb{1} \varrho \mathbb{1} + \xi_{01} \mathbb{1} \varrho \sigma_{z} + \xi_{10} \sigma_{z} \varrho \mathbb{1} + \xi_{11} \sigma_{z} \varrho \sigma_{z}) \\
		& \mathfrak{C}_{G,\widetilde{\xi}}^{pr} \overset{(\ref{elemOfImplPr})}{=} \frac{1}{3\widetilde{p}} \left[(\widetilde{\xi}_{00} + \widetilde{\xi}_{01} + \widetilde{\xi}_{10} + \widetilde{\xi}_{11}) \mathbb{1} \varrho \mathbb{1} + (\widetilde{\xi}_{02} + \widetilde{\xi}_{12}) \mathbb{1} \varrho \sigma_{z} + (\widetilde{\xi}_{20} + \widetilde{\xi}_{21}) \sigma_{z} \varrho \mathbb{1} + \widetilde{\xi}_{22} \sigma_{z} \varrho \sigma_{z}\right].
	\end{align*}
	As our first option, let us choose $\widetilde{\xi}_{00} = \widetilde{\xi}_{11} = \frac{1}{2}$. Corresponding program for processor $G$ is $\xi_{00} = 1$. Then $\mathfrak{C}_{G,\xi}^{pr} = \frac{1}{2p} \mathbb{1} \varrho \mathbb{1}$ with $p = \frac{1}{2}$ and $\mathfrak{C}_{\widetilde{G},\widetilde{\xi}}^{pr} = \frac{1}{3\widetilde{p}} \mathbb{1} \varrho \mathbb{1}$ with $\widetilde{p} = \frac{1}{3}$. In this case $p > \widetilde{p}$.
	
	Let us now choose $\widetilde{\xi}_{00} = \widetilde{\xi}_{01} = \widetilde{\xi}_{10} = \widetilde{\xi}_{11} = \frac{1}{2}$, while again $\xi_{00} = 1$. Then $\mathfrak{C}_{\widetilde{G},\widetilde{\xi}}^{pr} = \frac{2}{3\widetilde{p}} \mathbb{1} \varrho \mathbb{1}$ with $\widetilde{p} = \frac{2}{3}$. And because $p$ is the same as previously calculated, now $p < \widetilde{p}$, because non-diagonal elements of $\widetilde{\xi}$ also contribute to implementation of an identity channel.
\end{example}

\subsection{Relations between Types of Equivalences}

At first, we shall direct our focus to the examination of relations between deterministic and probabilistic equivalences. Let us start with two deterministically equivalent processors and show that deterministic equivalence does not imply any of the probabilistic types of equivalence.
\begin{example}
	Let us have two processors:
	\begin{align*}
		G_{1} &= \frac{1}{\sqrt{2}}\left(\mathbb{1}\otimes\dyad{0}{0} + \sigma_{z}\otimes\dyad{1}{0} + \sigma_{x}\otimes\dyad{0}{1} - i\sigma_{y}\otimes\dyad{1}{1}\right) \\
		G_{2} &= \dyad{0}{0}\otimes\dyad{0}{0} + \dyad{1}{1}\otimes\dyad{1}{0} +  \dyad{1}{0}\otimes\dyad{0}{1} + \dyad{0}{1}\otimes\dyad{1}{1}.
	\end{align*}
	Given processors are deterministically equivalent:
	\begin{align*}
		\mathfrak{C}_{G_{1},\xi}^{det} = \mathfrak{C}_{G_{2},\xi}^{det} \overset{(\ref{detQPimpl})}{=} &\xi_{00} \left(\dyad{0}{0}\varrho\dyad{0}{0} + \dyad{1}{1}\varrho\dyad{1}{1}\right) + \xi_{01} \left(\dyad{0}{0}\varrho\dyad{0}{1} + \dyad{1}{1}\varrho\dyad{1}{0}\right) + \\&\xi_{10} \left(\dyad{0}{1}\varrho\dyad{1}{1} + \dyad{1}{0}\varrho\dyad{0}{0}\right) + \xi_{11} \left(\dyad{0}{1}\varrho\dyad{1}{0} + \dyad{1}{0}\varrho\dyad{0}{1}\right).
	\end{align*}
	However, by calculating what processors $G_{1}$ and $G_{2}$ implement probabilistically, we discover that they are not probabilistically equivalent, regardless of the definition in use:
	\begin{align*}
		\mathfrak{G}_{G_{1},\xi}^{pr} &\overset{(\ref{elemOfImplPr})}{=} \frac{1}{p} \left(\xi_{00} \dyad{0}{0}\varrho\dyad{0}{0} + \xi_{01} \dyad{0}{0}\varrho\dyad{0}{1} + \xi_{10} \dyad{1}{0}\varrho\dyad{0}{0} + \xi_{11} \dyad{1}{0}\varrho\dyad{0}{1}\right) \\
		\mathfrak{G}_{G_{2},\widetilde{\xi}}^{pr} &\overset{(\ref{elemOfImplPr})}{=} \frac{1}{p} [\widetilde{\xi}_{00} \left(\dyad{0}{0}\varrho\dyad{0}{0} + \dyad{0}{0}\varrho\dyad{1}{1} + \dyad{1}{1}\varrho\dyad{0}{0} + \dyad{1}{1}\varrho\dyad{1}{1}\right) 
		\\&+ \widetilde{\xi}_{01} \left(\dyad{0}{0}\varrho\dyad{0}{1} + \dyad{0}{0}\varrho\dyad{1}{0} + \dyad{1}{1}\varrho\dyad{0}{1} + \dyad{1}{1}\varrho\dyad{1}{0}\right) 
		\\&+ \widetilde{\xi}_{10} \left(\dyad{1}{0}\varrho\dyad{0}{0} + \dyad{1}{0}\varrho\dyad{1}{1} + \dyad{0}{1}\varrho\dyad{0}{0} + \dyad{0}{1}\varrho\dyad{1}{1}\right) 
		\\&+ \widetilde{\xi}_{11} \left(\dyad{1}{0}\varrho\dyad{0}{1} + \dyad{1}{0}\varrho\dyad{1}{0} + \dyad{0}{1}\varrho\dyad{0}{1} + \dyad{0}{1}\varrho\dyad{1}{0}\right)].
	\end{align*}
	We have already encountered processor $G_{1}$ in example \ref{noChannelProcessor} as a probabilistic processor not able to implement any quantum channel. On the other hand, processor $G_{2}$ is able to apply quantum channel on data state $\varrho$. Therefore, they are not probabilistically equivalent in any of the three defined ways. And therefore, deterministic equivalence of processors does not guarantee any probabilistic equivalence.
\end{example}
Now, let us provide counterexample for statement that probabilistic equivalence implies deterministic one.
\begin{example}
	Let us have two processors:
	\begin{align*}
		G_{1} &= \dyad{0}{0}\otimes\dyad{0}{0} + \dyad{1}{1}\otimes\dyad{1}{0} +  \dyad{1}{1}\otimes\dyad{0}{1} + \dyad{0}{0}\otimes\dyad{1}{1}, \\
		G_{2} &= \mathbb{1}\otimes\dyad{0}{0} + \mathbb{1}\otimes\dyad{1}{1}.
	\end{align*}
	Probabilistically, they both implement identity channel with the same probability, meaning that they are strongly probabilistically $G \sim_{pr} \widetilde{G}$, as well as structurally $G \sim_{st} \widetilde{G}$ equivalent:
	\begin{align*}
		\mathfrak{G}_{G_{1},\xi}^{pr} = \mathfrak{G}_{G_{2},\xi}^{pr} \overset{(\ref{elemOfImplPr})}{=} \frac{1}{p} \left(\xi_{00}+\xi_{01}+\xi_{10}+\xi_{11}\right) \mathbb{1}\varrho\mathbb{1}.
	\end{align*}
	However deterministically, they implement two distinct sets of channels:
	\begin{align*}
		\mathfrak{C}_{G_{1},\xi}^{det} &\overset{(\ref{detQPimpl})}{=} \left(\xi_{00}+\xi_{11}\right)\left(\dyad{0}{0}\varrho\dyad{0}{0} + \dyad{1}{1}\varrho\dyad{1}{1}\right) + \left(\xi_{01}+\xi_{10}\right)\left(\dyad{0}{0}\varrho\dyad{1}{1} + \dyad{1}{1}\varrho\dyad{0}{0}\right) \\
		\mathfrak{C}_{G_{2},\xi}^{det} &\overset{(\ref{detQPimpl})}{=} \left(\xi_{00}+\xi_{11}\right) \mathbb{1}\varrho\mathbb{1}.
	\end{align*}
	Processor $G_{1}$ with the right program state is able to implement the same transformation as processor $G_{2}$, however it is not true the other way around. We can also consider another processor $G_{3} = \mathbb{1} \otimes \left(\dyad{0}{0} + \dyad{1}{1} + \dyad{2}{2}\right)$ that is weakly, but not deterministically equivalent with $G_{1}$ (it is deterministically equivalent with $G_{2}$). Therefore, probabilistic equivalence of any type does not imply deterministic equivalence.
\end{example}

Now, we shall continue with investigation of relations between probabilistic equivalences. Already from the definition \ref{strongEqv} of strong and the definition \ref{weakEqv} of weak equivalence, we can see that they are mutually exclusive. Either processors always implement every channel with the same probability (strong equivalence) or they do not (weak equivalence). Furthermore, structural equivalence does not imply that processors are definitively strongly or definitively weakly equivalent.
\begin{example}
	Let us consider three processors that are structurally equivalent.
	\begin{align*}
		G_{1} &= \mathbb{1} \otimes \dyad{0}{0} + \sigma_{z} \otimes \dyad{1}{1}, \\
		G_{2} &= \sigma_{z} \otimes \dyad{0}{0} + \mathbb{1} \otimes \dyad{1}{1}, \\
		G_{3} &= \mathbb{1} \otimes \left(\dyad{0}{0} + \dyad{2}{2}\right) + \sigma_{z} \otimes \left(\dyad{1}{1} + \dyad{3}{3}\right).
	\end{align*}
	We can show that these processors implement the following transformations:
	\begin{align*}
		\mathfrak{O}_{G_{1},\xi^{1}}^{pr} &\overset{(\ref{SprobQPelem})}{=} \frac{1}{2} \left(\xi_{00}^{1} \mathbb{1} \varrho \mathbb{1} + \xi_{01}^{1} \mathbb{1} \varrho \sigma_{z} + \xi_{10}^{1} \sigma_{z} \varrho \mathbb{1} + \xi_{11}^{1} \sigma_{z} \varrho \sigma_{z}\right), \\
		\mathfrak{O}_{G_{2},\xi^{2}}^{pr} &\overset{(\ref{SprobQPelem})}{=} \frac{1}{2} \left(\xi_{00}^{2} \sigma_{z} \varrho \sigma_{z} + \xi_{01}^{2} \sigma_{z} \varrho \mathbb{1} + \xi_{10}^{2} \mathbb{1} \varrho \sigma_{z} + \xi_{11}^{2} \mathbb{1} \varrho \mathbb{1}\right), \\
		\mathfrak{O}_{G_{3},\xi^{3}}^{pr} &\overset{(\ref{SprobQPelem})}{=} \frac{1}{4} [\left(\xi_{00}^{3} + \xi_{02}^{3} + \xi_{20}^{3} + \xi_{22}^{3}\right) \mathbb{1} \varrho \mathbb{1} + \left(\xi_{01}^{3} + \xi_{03}^{3} + \xi_{21}^{3} + \xi_{23}^{3}\right) \mathbb{1} \varrho \sigma_{z} \\
		&+ \left(\xi_{10}^{3} + \xi_{12}^{3} + \xi_{30}^{3} + \xi_{32}^{3}\right) \sigma_{z} \varrho \mathbb{1} + \left(\xi_{11}^{3} + \xi_{13}^{3} + \xi_{31}^{3} + \xi_{33}^{3}\right) \sigma_{z} \varrho \sigma_{z}].
	\end{align*}
	Processors are structurally equivalent as they are all able to implement identical operations. However, processors $G_{1}$ and $G_{2}$ are strongly equivalent $G_{1} \sim_{pr} G_{2}$, while processors $G_{1}$ and $G_{3}$ are weakly equivalent $G_{1} \approx_{pr} G_{3}$. E.g. $G_{1}$ realizes $\mathbb{1}\varrho\mathbb{1}$ with probability $p_{1} = \frac{1}{2}$, while $G_{3}$ realizes the same channel with probability $p_{3} = \frac{1}{4}$. Thus, we cannot say that structural equivalence implies either strong or weak probabilistic equivalence. 
\end{example}
However, let us examine whether structural equivalence implies that two processors are either strongly or weakly equivalent and not only exclusively strongly or exclusively weakly equivalent. We shall consider two processors $G = \sum_{jk}^{P} A_{jk} \otimes \dyad{j}{k}$ and $\widetilde{G} = \sum_{mn}^{\widetilde{P}} B_{mn} \otimes \dyad{m}{n}$ that are structurally equivalent. Let us consider that $\mathfrak{G}_{G}^{pr} = \frac{1}{p} \mathfrak{O}_{G}^{pr}$ and realize that for structurally equivalent processors, it holds that $\mathfrak{O}_{G}^{pr} = K_{\xi, \widetilde{\xi}} \mathfrak{O}_{G}^{pr}$, which means that the following equation is true:
\begin{align*}
	\frac{1}{P} \sum_{kk^{\prime}}^{P} \xi_{kk^{\prime}} \mathscr{A}_{k} \varrho \mathscr{A}_{k^{\prime}}^{\dagger} = K_{\xi, \widetilde{\xi}} \frac{1}{\widetilde{P}} \sum_{nn^{\prime}}^{\widetilde{P}} \widetilde{\xi}_{nn^{\prime}} \mathscr{B}_{n} \varrho \mathscr{B}_{n^{\prime}}^{\dagger}.
\end{align*}
We can now take a look at the channels and measurements that $G$ implements:
\begin{align*}
	\mathfrak{G}_{G, \xi}^{pr} \overset{(\ref{elemOfImplPr})}{=} \frac{1}{pP} \sum_{kk^{\prime}}^{P} \xi_{kk^{\prime}} \mathscr{A}_{k}\varrho\mathscr{A}_{k^{\prime}}^{\dagger} = \frac{K_{\xi, \widetilde{\xi}}}{p\widetilde{P}} \sum_{nn^{\prime}}^{\widetilde{P}} \widetilde{\xi}_{nn^{\prime}} \mathscr{B}_{n} \varrho \mathscr{B}_{n^{\prime}}^{\dagger}.
\end{align*}
We can also check what processor $\widetilde{G}$ implements:
\begin{align*}
	\mathfrak{G}_{\widetilde{G}, \widetilde{\xi}}^{pr} \overset{(\ref{elemOfImplPr})}{=} \frac{1}{\widetilde{p}\widetilde{P}} \sum_{nn^{\prime}}^{\widetilde{P}} \widetilde{\xi}_{nn^{\prime}} \mathscr{B}_{n} \varrho \mathscr{B}_{n^{\prime}}^{\dagger}.
\end{align*}
We can see that the implemented measurements and channels for $G$ and $\widetilde{G}$ are that same, albeit the probabilities might differ:
\begin{align*}
	p &= \frac{K_{\xi, \widetilde{\xi}}}{\widetilde{P}} \sum_{nn^{\prime}}^{\widetilde{P}} \widetilde{\xi}_{nn^{\prime}} \Tr\left(\mathscr{B}_{n}\varrho\mathscr{B}_{n^{\prime}}^{\dagger}\right),\\
	\widetilde{p} &= \frac{1}{\widetilde{P}} \sum_{nn^{\prime}}^{\widetilde{P}} \widetilde{\xi}_{nn^{\prime}} \Tr\left(\mathscr{B}_{n}\varrho\mathscr{B}_{n^{\prime}}^{\dagger}\right).
\end{align*}
We also do not need to explicitly calculate sets of implemented channels $\mathfrak{C}_{G}^{pr}$ and $\mathfrak{C}_{\widetilde{G}}^{pr}$ as our condition is already even stronger than that. And thus, structural equivalence between processors implies that they are either strongly or weakly probabilistically equivalent.

Let us now give an example when weak probabilistic equivalence does not imply structural equivalence.
\begin{example}
	Consider the following processors:
	\begin{align*}
		G &= \frac{1}{\sqrt{2}} \left(\mathbb{1} \otimes \dyad{0}{0} + \sigma_{z} \otimes \dyad{1}{0} + \sigma_{x} \otimes \dyad{0}{1} - i \sigma_{y} \otimes \dyad{1}{1}\right) + \sigma_{z} \otimes \dyad{2}{2} \\
		\widetilde{G} &= \sigma_{z} \otimes \dyad{0}{0}.
	\end{align*}
	Processor $G$ is formed from processor that does not implement any channel and was already used in the example \ref{noChannelProcessor} and added to it is processor $\widetilde{G}$. This processor implements:
	\begin{align*}\label{randomProcessor}
		&\mathfrak{G}_{G,\xi}^{pr} \overset{(\ref{elemOfImplPr})}{=}  \numberthis \\&\frac{1}{3p} [\dyad{0}{0} \left(2 \xi_{00} \varrho_{00} + \sqrt{2} \xi_{02} \varrho_{00} + \sqrt{2} \xi_{20} \varrho_{00} + \xi_{22} \varrho_{00}\right) + \dyad{0}{1} \left(2 \xi_{01} \varrho_{00} - \sqrt{2} \xi_{02} \varrho_{01} + \sqrt{2} \xi_{21} \varrho_{00} - \xi_{22} \varrho_{01}\right) \\
		&\hspace*{2mm}+ \dyad{1}{0} \left(2 \xi_{00} \varrho_{00} + \sqrt{2} \xi_{12} \varrho_{00} - \sqrt{2} \xi_{20} \varrho_{10} - \xi_{22} \varrho_{10}\right) + \dyad{1}{1} \left(2 \xi_{11} \varrho_{00} - \sqrt{2} \xi_{12} \varrho_{01} - \sqrt{2} \xi_{21} \varrho_{10} + \xi_{22} \varrho_{11}\right)]. 
	\end{align*}
	Let us now figure out when the processor is able to implement channels.
	\begin{align*}
		p &\overset{(\ref{probQP})}{=} \frac{1}{3} \left(2 \xi_{00} \varrho_{00} + \sqrt{2} \xi_{02} \varrho_{00} + \sqrt{2} \xi_{20} \varrho_{00} + \xi_{22} \varrho_{00} + 2 \xi_{11} \varrho_{00} - \sqrt{2} \xi_{12} \varrho_{01} - \sqrt{2} \xi_{21} \varrho_{10} + \xi_{22} \varrho_{11}\right) \\
		&= \frac{1}{3} \left[\varrho_{00} \left(2\xi_{00} + \sqrt{2}\xi_{02} + \sqrt{2}\xi_{20} + 2\xi_{11}\right) + \xi_{22} \varrho_{11} - \sqrt{2}\xi_{12} \varrho_{01} - \sqrt{2}\xi_{21}\varrho_{10}\right].
	\end{align*}
	For processor to implement quantum channel, probability cannot depend on data state $\varrho$. Therefore, we want this to be equal to $C(\varrho_{00} + \varrho_{11})$, where $C \in \mathbb{R}$ is any real number and therefore $\xi_{00} = \xi_{11} = 0$, which in turn also means that $\xi_{02} = \xi_{20} = \xi_{12} = \xi_{21} = 0$. Thus, processor $G$ implements channel only for program state $\xi_{22} = 1$ with probability $p = \frac{1}{3}$. With the chosen program $\xi_{22} = 1$, $G$ implements $\frac{1}{3p} \sigma_{z} \varrho \sigma_{z}$. Processor $\widetilde{G}$ always implements $\frac{1}{\widetilde{p}} \sigma_{z} \varrho \sigma_{z}$ with probability $\widetilde{p} = 1$. Therefore, these processors are weakly probabilistically equivalent; however they are not structurally equivalent. Let us choose program state $\xi_{00}$ and let us take a look at what is applied on data state by processor $G$: $\mathfrak{G}_{G,\xi}^{pr} = \frac{1}{3p} \varrho_{00} \dyad{0}{0}$. There is no way in which processor $\widetilde{G}$ could implement the same operation multiplied by a positive real number $K_{\xi,\widetilde{\xi}}$. Therefore, weak equivalence does not imply structural equivalence.
\end{example}
Let us modify previous example to show that neither does strong equivalence guarantee structural equivalence.
\begin{example}
	We shall look at two following processors:
	\begin{align*}
		G &= \frac{1}{\sqrt{2}} \left(\mathbb{1} \otimes \dyad{0}{0} + \sigma_{z} \otimes \dyad{1}{0} + \sigma_{x} \otimes \dyad{0}{1} - i \sigma_{y} \otimes \dyad{1}{1}\right) + \sigma_{z} \otimes \dyad{2}{2} \\
		\widetilde{G} &= \frac{1}{\sqrt{2}} \left(\mathbb{1} \otimes \dyad{0}{0} - \sigma_{z} \otimes \dyad{1}{0} + \sigma_{x} \otimes \dyad{0}{1} + i \sigma_{y} \otimes \dyad{1}{1}\right) + \sigma_{z} \otimes \dyad{2}{2},
	\end{align*}
	which are formed in a similar fashion - as a composition of processor not being able to implement any channel and a processor being able to apply only Pauli matrix $\sigma_{z}$ on the input data state. Let us now look at what the second processor implements:
	\begin{align*}
		&\mathfrak{G}_{\widetilde{G},\widetilde{\xi}}^{pr} \overset{(\ref{elemOfImplPr})}{=} \\&\frac{1}{3\widetilde{p}} [\dyad{0}{0} \left(2 \widetilde{\xi}_{11} \varrho_{11} + \sqrt{2} \widetilde{\xi}_{12} \varrho_{10} + \sqrt{2} \widetilde{\xi}_{21} \varrho_{01} + \widetilde{\xi}_{22} \varrho_{00}\right) + \dyad{0}{1} \left(2 \widetilde{\xi}_{10} \varrho_{11} - \sqrt{2} \widetilde{\xi}_{12} \varrho_{11} + \sqrt{2} \widetilde{\xi}_{20} \varrho_{01} - \widetilde{\xi}_{22} \varrho_{01}\right) \\
		&\hspace*{0.3cm}+ \dyad{1}{0} \left(2 \widetilde{\xi}_{01} \varrho_{11} + \sqrt{2} \widetilde{\xi}_{02} \varrho_{10} - \sqrt{2} \widetilde{\xi}_{21} \varrho_{11} - \widetilde{\xi}_{22} \varrho_{10}\right) + \dyad{1}{1} \left(2 \widetilde{\xi}_{00} \varrho_{11} - \sqrt{2} \widetilde{\xi}_{02} \varrho_{11} - \sqrt{2} \widetilde{\xi}_{20} \varrho_{11} + \widetilde{\xi}_{22} \varrho_{11}\right)].
	\end{align*}
	This processor, exactly as processor $G$, is able to apply channel on data state only in case of program state being $\widetilde{\xi}_{22} = 1$. In such a case it implements $\frac{1}{3\widetilde{p}} \sigma_{z} \varrho \sigma_{z}$ with probability being $\widetilde{p} = \frac{1}{3}$. Thus, processors are strongly equivalent $G \sim_{pr} \widetilde{G}$. However, if we take program $\xi_{00} = 1$ for processor $G$, we can see from equation (\ref{randomProcessor}) that the implemented measurement is $\frac{1}{3p} \varrho_{00} \dyad{0}{0}$ which is impossible to obtain from processor $\widetilde{G}$ even if it is multiplied by $K_{\xi,\widetilde{\xi}}$. Therefore, strong probabilistic equivalence does not imply structural one.
\end{example}

In conclusion, the only type of equivalence that provides some additional information about other types of equivalences is structural one. If processors are structurally equivalent, we can say that they are also either strongly or weakly probabilistically equivalent.

\chapter{Quantum Networks}\label{QN}
Mathematical formalism describing elementary quantum circuits is already well established. Density matrices, quantum channels, quantum instruments and POVMs are at its heart. However, once one starts to combine these circuits together into complex quantum networks, their analysis becomes much more difficult and convoluted. There exist infinite possibilities to order various quantum gates in circuits, to combine them with various measurements and to combine quantum circuits into more complicated quantum networks. Furthermore, one has to take into account that every such network can be probabilistic. This means that optimizing every possible network for various tasks described by multitude of possible networks becomes rather complicated. Therefore, one would like to devise a more unifying formalism to simplify their description. Luckily, generalization of density matrices, quantum channels, quantum instruments and POVMs, that allow for easier manipulation of quantum networks, do exist \cite{TheoreticalFrameworkForQuantumNetworks, QuantumNetworks:GeneralTheoryAndApplications, QuantumCircuitArchitecture, OptimalQuantumNetworksAndOneShotEntropies}.

This generalization is based on Choi-Jamio\l{}kowski isomorphism linking linear maps with linear operators, or at the basic level, quantum channels with density matrices. One also has to take into account that one task can be executed by multiple circuits that form equivalence class. The entire equivalence class of networks, can, in the end, be described by only one operator called Choi operator. This simplifies optimization, because one does not need to specify all channels and measuring devices forming quantum network, but only one operator describing the entire network.

The new formalism is used in solving problems such as quantum channel discrimination \cite{TheoreticalFrameworkForQuantumNetworks, PerfectDiscriminationOfNoSignallingChannelsViaQuantumSuperpositionOfCausalStructures}, quantum tomography \cite{OptimalQuantumTomography}, cloning of unitary transformation \cite{OptimalCloningOfUnitaryTransformation}, or to study the causality \cite{DeterminismWithoutCausality, CausalInfluenceInOperationalProbabilisticTheories} or quantum learning of unitary transformation \cite{OptimalProbabilisticStorageAndRetrievalOfUnitaryChannels} on which we shall build further in this thesis.


\section{Choi-Jamio\l{}kowski Isomorphism}

Choi-Jamio\l{}kowski isomorphism is at the center of our generalized treatment of quantum networks. It introduces one-to-one correspondence between linear maps and linear operators.

But firstly, we shall examine isomorphism between states and operators. From now on, let ${\cal{L(H)}}$ denote set of linear operators on a Hilbert space ${\cal{H}}$ and ${\cal{L}}({\cal{H}}_{a}, {\cal{H}}_{b})$ denote linear operator from Hilbert space ${\cal{H}}_{a}$ to Hilbert space ${\cal{H}}_{b}$. Then, let us have an operator $A \in {\cal{L}}({\cal{H}}_{a}, {\cal{H}}_{b})$ and a quantum state $\dket{A}_{ab} \in {\cal{H}}_{a} \otimes {\cal{H}}_{b}$ (called double ket):
\begin{align*}
	A &= \sum_{n,m} \ket{n}_{b}\hspace*{-1mm}\bra{n}A\ket{m}_{a}\hspace*{-1mm}\bra{m}
\end{align*}
\begin{align}\label{opAndDKet}
	\dket{A}_{ab} &= \sum_{n,m} \bra{n}A\ket{m} \ket{n}_{b}\ket{m}_{a},
\end{align}
where $\{\ket{m}_{a}\}, \{\ket{n}_{b}\}$ denote orthonormal bases of ${\cal{H}}_{a}$ and ${\cal{H}}_{b}$ respectively. If dimensions of Hilbert spaces $\dim({\cal{H}}_{a}) = d_{a}$ and $\dim({\cal{H}}_{a^{\prime}}) = d_{a^{\prime}}$ are equal $d_{a} = d_{a^{\prime}}$, it means that ${\cal{H}}_{a} \cong {\cal{H}}_{a^{\prime}}$ are isomorphic and the relation:
\begin{align}\label{stateChannelDuality}
	\dket{A}_{ab} = (A_{ba} \otimes \mathbb{1}_{a^{\prime}})\dket{I}_{aa^{\prime}},
\end{align}
where $\mathbb{1}_{a^{\prime}}$ denotes identity on ${\cal{H}}_{a^{\prime}}$ and $\dket{I}_{aa^{\prime}} = \sum_{n}^{d_{a}}\ket{n}_{a}\ket{n}_{a^{\prime}}$ is nonnormalized maximally entangled state on Hilbert space ${\cal{H}}_{a} \otimes {\cal{H}}_{a^{\prime}}$, defines isomorphism between quantum states and quantum operators. By simple substitution, we are able to retrieve double ket from right-hand side of equation (\ref{opAndDKet}):
\begin{align*}
	(A \otimes \mathbb{1}_{a^{\prime}})\dket{I}_{aa^{\prime}} &= \sum_{nmm^{\prime}} (\ket{n}_{b}\hspace*{-1mm}\bra{n}A\ket{m^{\prime}}_{a}\hspace*{-1mm}\bra{m^{\prime}} \otimes \mathbb{1}_{a^{\prime}})\ket{m}_{a}\ket{m}_{a^{\prime}} \\&=\sum_{nmm^{\prime}} (\bra{n}A\ket{m^{\prime}}\ket{n}_{b}\prescript{}{a}{\bra{m^{\prime}}} \otimes \mathbb{1}_{a^{\prime}})\ket{m}_{a}\ket{m}_{a^{\prime}} = \sum_{nm} \bra{n}A\ket{m}\ket{n}_{b}\ket{m}_{a^{\prime}} \\&\hspace*{-4mm}\overset{{\cal{H}}_{a} \cong {\cal{H}}_{a^{\prime}}}{=} \sum_{nm} \bra{n}A\ket{m}\ket{n}_{b}\ket{m}_{a} = \dket{A}_{ab}.
\end{align*}

Choi-Jamio\l{}kowski isomorphism is a correspondence between linear maps (or superoperators) and linear operators. In previous paragraph, analogous role to the role of linear maps was taken by linear operators and to the role of linear operators by quantum states. Let ${\cal{L}}({\cal{L}}({\cal{H}}_{a}),{\cal{L}}({\cal{H}}_{b}))$ denote the set of linear maps from the set of linear operators ${\cal{L}}({\cal{H}}_{a})$ to the set of linear operators ${\cal{L}}({\cal{H}}_{b})$. Also, let $X^{T}$ denote the transposition of operator $X$. Then one-to-one correspondence between linear maps ${\cal{M}} \in {\cal{L}}({\cal{L}}({\cal{H}}_{a}),{\cal{L}}({\cal{H}}_{b}))$ and linear operators $M_{ba}$ on Hilbert space ${\cal{L}}({\cal{H}}_{a}) \otimes {\cal{L}}({\cal{H}}_{b})$ is given by the following definition:
\begin{definition} 
	Choi-Jamio\l{}kowski Isomorphism
	\\
	Let us have map ${\cal{C}}: {\cal{L}}({\cal{L}}({\cal{H}}_{a}),{\cal{L}}({\cal{H}}_{b})) \rightarrow {\cal{L}}({\cal{H}}_{a}) \otimes {\cal{L}}({\cal{H}}_{b})$. Then, Choi-Jamio\l{}kowski isomorphism is given by  
	\begin{align}\label{choi}
		M_{ba} = {\cal{C}}({\cal{M}}) = ({\cal{M}} \otimes {\cal{I}}_{a})(\dket{I}\dbra{I}),
	\end{align}
	where ${\cal{I}}_{a}$ is an identity map on ${\cal{L}}({\cal{H}}_{a})$ and $M_{ba}$ is called Choi operator of map ${\cal{M}}$.
\end{definition}
Inverse of the previous correspondence (\ref{choi}) is given by:
\begin{align}\label{inverseChoi}
	{\cal{M}}(X) = {\cal{C}}^{-1}(M_{ba})(X) = \Tr_{a}[(\mathbb{1}_{b} \otimes X^{T}_{a})M_{ba}],
\end{align}
where $X_{a}^{T} \in {\cal{L}}({\cal{H}}_{a})$.
\begin{proof}
	We shall prove Choi-Jamio\l{}kowski isomorphism by verifying the inverse relation (\ref{inverseChoi}) for any linear operator $X\in {\cal{L}}({\cal{H}}_{a})$:
	\begin{align*}
		{\cal{M}}(X) &\overset{(\ref{inverseChoi})}{=} \Tr_{a}\left[M_{ba}\left(\mathbb{1}_{b}\otimes X_{a}^{T}\right)\right] \overset{(\ref{choi})}{=} \Tr_{a} \left\{\left[\left({\cal{M}} \otimes {\cal{I}}_{a}\right) \sum_{i,j}\ket{i}_{a}\ket{i}_{b}\hspace*{1mm}\prescript{}{a}{\bra{j}}\prescript{}{b}{\bra{j}}\right]\left(\mathbb{1}_{b} \otimes X^{T}_{a}\right)\right\} \\&\hspace*{2mm}= \sum_{i,j} \left[{\cal{M}}(\ket{i}_{b}\hspace*{-1mm}\bra{j}) \Tr_{a}\left(\ket{i}_{a}\hspace*{-1mm}\bra{j}X^{T}_{a}\right)\right] = \sum_{i,j} {\cal{M}}(\ket{i}_{b}\hspace*{-1mm}\bra{j}) \bra{j}X^{T}\ket{i} = \sum_{i,j} {\cal{M}}(\ket{i}_{b}\hspace*{-1mm}\bra{j}) \bra{i}X\ket{j} \\&\hspace*{-4.5mm}\overset{\text{linearity of } {\cal{M}}}{=} {\cal{M}}\left(\sum_{i,j} \ket{i}_{b}\bra{i}X\ket{j}_{b}\bra{j}\right) = {\cal{M}}(X).
	\end{align*}
	This is true for any $X\in {\cal{L}}({\cal{H}}_{a})$. Therefore, we have shown that the Choi-Jamio\l{}kowski isomorphism is true.
\end{proof}

Quantum channels are completely positive, trace-preserving linear transformations and therefore, we shall examine how these qualities translate into Choi operators.
\begin{lemma}\label{TPlemma}
	Trace-preservation
	\\
	Linear map ${\cal{M}} \in {\cal{L}}({\cal{L}}({\cal{H}}_{a}),{\cal{L}}({\cal{H}}_{b}))$ is trace-preserving, if and only if for its corresponding Choi operator $M_{ba}$ the following equation $\Tr_{b}(M_{ba}) = \mathbb{1}_{a}$ is true.
\end{lemma}
\begin{proof}
	Map is trace-preserving if $\Tr\left[{\cal{M}}(X)\right] = \Tr(X)$. Let us now calculate trace of ${\cal{M}} \in {\cal{L}}({\cal{L}}({\cal{H}}_{a}),{\cal{L}}({\cal{H}}_{b}))$ on an arbitrary input $X \in {\cal{L}}({\cal{H}}_{a})$:
	\begin{align*}
		\Tr\left[{\cal{M}}(X)\right] \overset{(\ref{inverseChoi})}{=} \Tr\left[M_{ba}\left(\mathbb{1}_{b} \otimes X_{a}^{T}\right)\right] = \Tr_{a}\left[\Tr_{b}(M_{ba})X_{a}^{T}\right] \overset{!}{=} \Tr_{a}(X_{a}^{T}) \overset{(i)}{=} \Tr(X).
	\end{align*}
	where we require the expression on the left of the exclamation mark to be equal to the expression on its right side. In $(i)$ we are using the invariance of trace under transposition $\Tr(X^{T}_{a}) = \Tr(X_{a})$. For this equation to hold, the equation $\Tr_{b}(M_{ba}) = \mathbb{1}_{a}$ must be true. 
	
	On the other hand, if $\Tr_{b}(M_{ba}) = \mathbb{1}_{a}$ holds, then ${\cal{M}}$ is trace-preserving.
\end{proof}

\begin{lemma}\label{CPlemma}
	Complete Positivity
	\\
	Linear map ${\cal{M}} \in {\cal{L}}({\cal{L}}({\cal{H}}_{a}),{\cal{L}}({\cal{H}}_{b}))$ is completely positive, if and only if its corresponding Choi operator $M_{ba}$ is positive semi-definite.
\end{lemma}
\begin{proof}
	Firstly, let us suppose that linear map ${\cal{M}} \in {\cal{L}}({\cal{L}}({\cal{H}}_{a}),{\cal{L}}({\cal{H}}_{b}))$ is completely positive $({\cal{M}} \otimes {\cal{I}}_{c})(X) \geq 0$, where ${\cal{I}}_{c}$ is identity operator on ${\cal{L}}({\cal{H}}_{c})$. Now, let us put ${\cal{L}}({\cal{H}}_{c}) = {\cal{L}}({\cal{H}}_{a})$ and take $X = \dket{I}\dbra{I}$. Then, we obtain Choi operator of a given map as follows $({\cal{M}} \otimes {\cal{I}}_{a})(\dket{I}\dbra{I}) \overset{(\ref{choi})}{=} M_{ba} \geq 0$, which means, that $M_{ba}$ is positive semi-definite.
	
	Now, let us assume that $M_{ba} \geq 0$ is positive semi-definite and that ${\cal{M}} \in {\cal{L}}({\cal{L}}({\cal{H}}_{a}),{\cal{L}}({\cal{H}}_{b}))$:
	\begin{align*}
		&\left({\cal{M}} \otimes {\cal{I}}_{c}\right)(X_{ac}) \overset{(\ref{inverseChoi})}{=} \Tr_{a} \left[\left(M_{ba} \otimes \mathbb{1}_{c}\right)\left(\mathbb{1}_{b} \otimes X_{ac}^{T_{a}}\right)\right] \overset{(i)}{=} 
		\\&\sum_{m,n} \prescript{}{a}{\bra{m}}\left(M_{ba} \otimes \mathbb{1}_{c}\right) \ket{n}_{a}\prescript{}{a}{\bra{n}}(\mathbb{1}_{b} \otimes X_{ac}^{T_{a}})\ket{m}_{a} \overset{(ii)}{=} \sum_{m,n} \prescript{}{a}{\bra{m}}M_{ba}\ket{n}_{a}\prescript{}{a}{\bra{m}}X_{ac}\ket{n}_{a} \overset{(iii)}{=} 
		\\&\sum_{m,n} \prescript{}{a^{\prime}}{\bra{m}}M_{ba^{\prime}}\ket{n}_{a^{\prime}}\prescript{}{a}{\bra{m}}X_{ac} \ket{n}_{a} = \prescript{}{aa^{\prime}}{\dbra{I}}M_{ba^{\prime}} \otimes X_{ac} \dket{I}_{aa^{\prime}}.
	\end{align*}
	In $(i)$ we have expressed trace through summation and added identity, in $(ii)$ we have partially transposed operator $X$ over space ${\cal{L}}({\cal{H}}_{a})$ and in $(iii)$ we renamed first pair of indices $a$, which we are allowed to do, because it does not change the value of scalar that is given by "sandwiching" operator with bra and ket vectors, with isomorphic Hilbert spaces ${\cal{H}}_{a} \cong {\cal{H}}_{a^{\prime}}$. But what remains is still an operator on space ${\cal{L}}({\cal{H}}_{b} \otimes {\cal{H}}_{c})$. We have to show that this is positive semi-definite, which means to show $\prescript{}{bc}{\bra{\Psi}}\prescript{}{aa^{\prime}}{\dbra{I}}M_{ba^{\prime}} \otimes X_{ac} \dket{I}_{aa^{\prime}}\ket{\Psi}_{bc} \geq 0$ for any vector $\ket{\Psi}_{bc}$. But the operator $M_{ba^{\prime}}$ is positive semi-definite from the assumption and therefore, if and only if $X_{ac} \geq 0$, also their tensor product is positive semi-definite. That means that if $X_{ac}$ is positive then the whole expression is positive semi-definite and that in turn means that $\left({\cal{M}} \otimes {\cal{I}}_{c}\right) \geq 0$, which means that ${\cal{M}}$ is completely positive.	
\end{proof}
And because density operators are Hermitian, and density operators describe quantum states, we also add the following lemma. Let us also note that $M^{\dagger}$ denotes Hermitian conjugate of $M$.
\begin{lemma}\label{Hlemma}
	Hermiticity
	\\
	Linear map ${\cal{M}} \in {\cal{L}}({\cal{L}}({\cal{H}}_{a}),{\cal{L}}({\cal{H}}_{b}))$ is Hermitian, if and only if its corresponding Choi operator $M_{ba}$ is Hermitian.
\end{lemma}
\begin{proof}
	For a map to be Hermitian, it must hold the following ${\cal{M}}(X^{\dagger}) = \left[{\cal{M}}(X)\right]^{\dagger}$. Firstly, we shall calculate:
	\begin{align}\label{hermCJ}
		\left[{\cal{M}}(X)\right]^{\dagger} \overset{(\ref{inverseChoi})}{=} \left\{\Tr_{a}\left[M_{ba}\left(\mathbb{1}_{b}\otimes X_{a}^{T}\right)\right]\right\}^{\dagger} = \Tr_{a}\left[\left(\mathbb{1}_{b} \otimes X_{a}^{\ast}\right) M_{ba}^{\dagger}\right] \overset{\text{cyclicity}}{\underset{\text{of trace}}{=}} \Tr_{a} \left[M_{ba}^{\dagger}\left(\mathbb{1}_{b} \otimes (X_{a}^{\dagger})^{T}\right)\right],
	\end{align}
	where $X_{a}^{\ast}$ denotes complex conjugate of $X_{a}$. By comparing the expression after the last equal sign with ${\cal{M}}(X^{\dagger}) = \Tr_{a}\left[M_{ba} \left(\mathbb{1}_{b} \otimes (X^{\dagger}_{a})^{T}\right)\right]$, it can be seen, that in order for ${\cal{M}}$ to be Hermitian, $M_{ba} = M_{ba}^{\dagger}$ must be true.
	
	On the other hand, if $M_{ba}$ is Hermitian, then by substituting $M_{ba}$ for $M_{ba}^{\dagger}$ into the previous equation (\ref{hermCJ}), it can be seen that ${\cal{M}}$ is also Hermitian.
\end{proof}

\subsection{Link Product}

Multiple linear maps can be composed into one. Naturally, we shall examine how this property translates into composition of Choi operators. In the beginning, let us consider composition of two linear maps ${\cal{A}} \in {\cal{L}}({\cal{L}}({\cal{H}}_{a}), {\cal{L}}({\cal{H}}_{b}))$ and ${\cal{B}} \in {\cal{L}}({\cal{L}}({\cal{H}}_{b}), {\cal{L}}({\cal{H}}_{c}))$, using inverse Choi-Jamio\l{}kowski isomorphism (\ref{inverseChoi}) we obtain:
\begin{align*}
	({\cal{B}} \circ {\cal{A}})(X_{a}) = {\cal{B}}({\cal{A}}(X_{a})) &= \Tr_{b} \left\{ \left[\mathbb{1}_{c} \otimes \Tr_{a} \left(\left(\mathbb{1}_{b} \otimes X_{a}^{T_{a}}\right) A_{ab}\right)^{T_{b}} \right] B_{bc} \right\} \\&= \Tr_{a}\left\{\Tr_{b}\left[\left(\mathbb{1}_{c} \otimes \mathbb{1}_{b} \otimes X_{a}^{T_{a}}\right)(\mathbb{1}_{c} \otimes A_{ab}^{T_{b}})(\mathbb{1}_{a} \otimes B_{bc})\right]\right\} \\&= \Tr_{a}\left\{ \Tr_{b}\left[\left(\mathbb{1}_{c} \otimes A_{ab}^{T_{b}}\right)(\mathbb{1}_{a} \otimes B_{bc})\right] \left(\mathbb{1}_{c} \otimes X_{a}^{T_{a}}\right) \right\},
\end{align*}
where $A_{ab}^{T_{b}}$ denotes partial transposition of $A_{ab}$ over the space ${\cal{L}}({\cal{H}}_{b})$. Comparing this equation with the inverse of Choi-Jamio\l{}kowski (\ref{inverseChoi}) we can see that Choi operator of the composition of two maps is:
\begin{align*}
	B \star A \equiv {\cal{C}}({\cal{B}} \circ {\cal{A}}) = \Tr_{b}\left[\left(\mathbb{1}_{c} \otimes A_{ab}^{T_{b}}\right)(\mathbb{1}_{a} \otimes B_{bc})\right],
\end{align*}
where $B \star A$ is called link product, and it denotes Choi operator of composition of maps ${\cal{B}} \circ {\cal{A}}$. In this link product, trace and transposition were both taken over shared Hilbert space ${\cal{L}}({\cal{H}}_{b})$ of the respective maps ${\cal{A}}$ and ${\cal{B}}$.

Let us take more general example with first map being ${\cal{A}}: {\cal{L}}({\cal{H}}_{\alpha}) \rightarrow ({\cal{H}}_{\beta})$ and the second one being ${\cal{B}}: {\cal{L}}({\cal{H}}_{\gamma}) \rightarrow ({\cal{H}}_{\delta})$, where $\alpha, \beta, \gamma$ and $\delta$ denote set of indices. Then, composition of these maps is ${\cal{B}} \circ {\cal{A}}: {\cal{L}}({\cal{H}}_{(\alpha \otimes \gamma)\setminus(\beta \cap \gamma)}) \rightarrow ({\cal{H}}_{(\beta \otimes \delta)\setminus(\beta \cap \gamma)})$, where we have discarded the overlapping space $\beta \cap \gamma$ as can be seen from the figure \ref{genLink}.
\begin{figure}[H]
	\begin{center}
		\begin{quantikz}
			\rstick{\hspace{-0.1cm}$\alpha$\\~} \qw &\gate[wires=3][2cm]{\mathcal{A}}  & \qw \rstick{\hspace{-0.7cm}$\beta$\\~} & & &\\
			& &\qw & \\
			& &\qw &\qw \rstick{\hspace{-0.1cm}$\gamma$\\~}  &\gate[wires=3][2cm]{\mathcal{B}} &\qw \rstick{\hspace{-0.7cm}$\delta$\\~} \\
			& & &\qw & &\qw \\
			& & &\qw & &\qw
		\end{quantikz}
	\end{center}
	\caption[font = small]{Schematic depiction of composition of maps ${\cal{A}}: {\cal{L}}({\cal{H}}_{\alpha}) \rightarrow ({\cal{H}}_{\beta})$ and ${\cal{B}}: {\cal{L}}({\cal{H}}_{\gamma}) \rightarrow ({\cal{H}}_{\delta})$. Greek letters $\alpha, \beta, \gamma$ and $\delta$ denote set of indices. In this example, one of the outputs of ${\cal{A}}$ serves as one of the inputs into ${\cal{B}}$.}
	\label{genLink}
\end{figure}
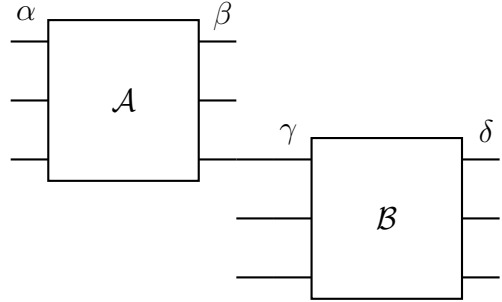
Let us now directly calculate the composition of ${\cal{A}}$ with ${\cal{B}}$:
\begin{align*}
	&({\cal{B}} \circ {\cal{A}})(X) = \left({\cal{B}}_{\delta\gamma} \otimes {\cal{I}}_{\beta \setminus \gamma}\right)\left({\cal{A}}_{\beta\alpha} \otimes {\cal{I}}_{\gamma \setminus \beta}\right)(X) \\&\overset{(\ref{inverseChoi})}{=} \left({\cal{B}}_{\delta\gamma} \otimes {\cal{I}}_{\beta \setminus \gamma}\right) \left\{\Tr_{\alpha}\left[\left(\mathbb{1}_{\gamma\setminus\beta} \otimes A_{\beta\alpha}\right)\left(\mathbb{1}_{\beta} \otimes X_{\alpha, \gamma\setminus\beta}^{T_{\alpha}}\right)\right]\right\} \\&\overset{(\ref{inverseChoi})}{=} \Tr_{\gamma}\left\{\left(\mathbb{1}_{\beta \setminus \gamma} \otimes {\cal{B}}_{\delta\gamma}\right)\left[\mathbb{1}_{\delta} \otimes \Tr_{\alpha}\left[\left(\mathbb{1}_{\gamma\setminus\beta} \otimes A_{\beta\alpha}\right)\left(\mathbb{1}_{\beta} \otimes X_{\alpha, \gamma\setminus\beta}^{T_{\alpha}}\right)\right]\right]^{T_{\gamma}}\right\} \\&\overset{T_{\gamma} \text{ goes inside the bracket}}{=} \Tr_{\gamma}\left\{\left(\mathbb{1}_{\beta \setminus \gamma} \otimes {\cal{B}}_{\delta\gamma}\right)\left[\mathbb{1}_{\delta} \otimes \Tr_{\alpha}\left[\left(\mathbb{1}_{\gamma\setminus\beta} \otimes A_{\beta\alpha}^{T_{\beta \cap \gamma}}\right)\left(\mathbb{1}_{\beta} \otimes X_{\alpha, \gamma\setminus\beta}^{T_{\alpha \cup \gamma\setminus\beta}}\right)\right]\right]\right\} \\&\overset{(i)}{=} \Tr_{\alpha} \underbrace{\Tr_{\gamma\setminus\beta} \Tr_{\beta\cap\gamma}}_{\Tr_{\gamma}} \left[\left(\mathbb{1}_{\alpha} \otimes \mathbb{1}_{\beta\setminus\gamma} \otimes {\cal{B}}_{\delta\gamma}\right) \left(\mathbb{1}_{\delta} \otimes \mathbb{1}_{\gamma\setminus\beta} \otimes A_{\beta\alpha}^{T_{\beta\cap\gamma}}\right) \left(\mathbb{1}_{\delta} \otimes \mathbb{1}_{\beta\cap\gamma} \otimes \mathbb{1}_{\beta\setminus\gamma} \otimes X_{\alpha, \gamma\setminus\beta}^{T_{\alpha \cup \gamma\setminus\beta}}\right)\right] \\&\overset{(ii)}{=} \Tr_{\alpha \cup \gamma\setminus\beta} \bigg\{\underbrace{\Tr_{\beta \cap \gamma} \left[\left(\mathbb{1}_{\alpha} \otimes \mathbb{1}_{\beta \setminus \gamma} \otimes B_{\delta\gamma}\right)\left(\mathbb{1}_{\delta} \otimes \mathbb{1}_{\gamma\setminus\beta} \otimes A_{\beta\alpha}^{T_{\beta\cap\gamma}}\right)\right]}_{\text{Choi operator of }{\cal{B}}\circ{\cal{A}}}\left(\mathbb{1}_{\delta} \otimes \mathbb{1}_{\beta\setminus\gamma} \otimes X_{\alpha, \gamma\setminus\beta}^{T_{\alpha \cup \gamma\setminus\beta}}\right)\bigg\}
\end{align*}
In $(i)$, $Tr_{\alpha}$ was taken out from the bracket and in $(ii)$, the identity $\mathbb{1}_{\beta\cap\gamma}$ disappears from the last set of parentheses because $\Tr_{\beta \cap \gamma}$ no longer concerns them. We can see, that in Choi operator of composed maps, there is trace and transposition taken over the overlapping space. Using this general example, we define link product:
\begin{definition}
	Link Product
	\\
	Let us have two Choi operators $A \in {\cal{L}}(\otimes_{a\in \alpha}{\cal{H}}_{a})$ and $B \in {\cal{L}}(\otimes_{b\in \beta}{\cal{H}}_{b})$, where $\alpha$ and $\beta$ denote finite set of indices. Then the link product is:
	\begin{align}\label{linkProduct}
		A \star B = {\cal{C}}({\cal{B}} \circ {\cal{A}}) = \Tr_{\alpha \cap \beta} \left[\left(\mathbb{1}_{\beta\setminus \alpha} \otimes A^{T_{\alpha \cap \beta}}\right)\left(B \otimes \mathbb{1}_{\alpha\setminus \beta}\right)\right].
	\end{align}
\end{definition}
The result of link product $A \star B$ is an operator in space ${\cal{L}}({\cal{H}}_{\alpha\setminus \beta} \otimes {\cal{H}}_{\beta\setminus \alpha})$. In the link product, transposition is applied over overlapping space, which is also traced out and operators $A$ and $B$ are expanded by such a Hilbert space so that they both span the same space in the end. If the intersection $\alpha \cap \beta$ is empty (i.e., there is no overlapping space), then link product reduces to tensor product $A \star B = A \otimes B$. Let us list basic properties of the link product:
\begin{lemma}\label{propLink}
	Properties of Link Product
	\\
	Let us have three operators $A \in {\cal{L}}(\otimes_{a\in \alpha}{\cal{H}}_{a})$, $B \in {\cal{L}}(\otimes_{b\in \beta}{\cal{H}}_{b})$ and $C \in {\cal{L}}(\otimes_{c\in \gamma}{\cal{H}}_{c})$. Then the following properties hold:
	\begin{itemize}
		\item symmetry: $A \star B = B \star A$
		\item linearity: $(a A + b B) \star C = a(A \star C) + b(B \star C)$ for any $a, b \in \mathbb{C}$
		\item hermiticity: if $A = A^{\dagger}$, $B = B^{\dagger}$ are Hermitian then also their link product $A \star B = (A \star B)^{\dagger}$ is Hermitian
		\item associativity: if $\alpha \cap \beta \cap \gamma = \emptyset$, then $A \star (B \star C) = (A \star B) \star C$
		\item positivity: if $A \geq 0$ and $B \geq 0$ are positive semi-definite, then also $A \star B \geq 0$ is positive semi-definite.
	\end{itemize}
\end{lemma}
\begin{proof}
	First four properties can be seen from the definition of the link product. We shall take closer look at the positivity. Let operators $A \geq 0$ and $B \geq 0$ be positive semi-definite. From lemma \ref{CPlemma}, it follows that ${\cal{A}}$ and ${\cal{B}}$ are completely positive and therefore also its composition ${\cal{B}}\circ{\cal{A}}$ is completely positive. Again, from lemma \ref{CPlemma}, Choi operator of said composition $A \star B \geq 0$ is positive semi-definite.
\end{proof}

\section{Diagrammatic Representation}

Graphical representation of quantum circuits is an extremely useful tool for understanding concrete circuits. Linear map ${\cal{A}}: {\cal{L}}({\cal{H}}_{0} \otimes {\cal{H}}_{1}) \rightarrow {\cal{L}}({\cal{H}}_{2} \otimes {\cal{H}}_{3})$ can be drawn in several equivalent ways, where we can permutate individual spaces. We shall depict three examples here:
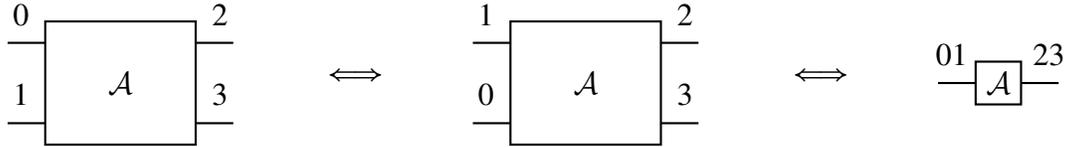
\begin{figure}[H]
	\begin{center}
		\begin{quantikz}
			\rstick{\hspace{-0.1cm}0\\~} \qw&\gate[wires=2][2cm]{\mathcal{A}}  & \qw \rstick{\hspace{-0.7cm}2\\~}\\
			\rstick{\hspace{-0.1cm}1\\~} & &\qw \rstick{\hspace{-0.7cm}3\\~}
		\end{quantikz}
		\qquad $\iff$ \qquad
		\begin{quantikz}
			\rstick{\hspace{-0.1cm}1\\~} \qw&\gate[wires=2][2cm]{\mathcal{A}}  & \qw \rstick{\hspace{-0.7cm}2\\~}\\
			\rstick{\hspace{-0.1cm}0\\~} & &\qw \rstick{\hspace{-0.7cm}3\\~}
		\end{quantikz}
		\qquad $\iff$ \qquad
		\begin{quantikz}
			\rstick{\hspace{-0.2cm}01\\~} \qw&\gate{\mathcal{A}}  & \qw \rstick{\hspace{-0.6cm}23\\~}
		\end{quantikz}
	\end{center}
	\caption[font=small]{Several possible pictorial representations of linear map ${\cal{A}}: {\cal{L}}({\cal{H}}_{0} \otimes {\cal{H}}_{1}) \rightarrow {\cal{L}}({\cal{H}}_{2} \otimes {\cal{H}}_{3})$. It is also possible to permutate only $0$ with $1$ or only $2$ with $3$.}
\end{figure}

Let us now consider composition ${\cal{B}} \circ {\cal{A}}$ of two linear maps ${\cal{A}}: {\cal{L}}({\cal{H}}_{0} \otimes {\cal{H}}_{2}) \rightarrow {\cal{L}}({\cal{H}}_{1} \otimes {\cal{H}}_{3})$ and ${\cal{B}}: {\cal{L}}({\cal{H}}_{3} \otimes {\cal{H}}_{4}) \rightarrow {\cal{L}}({\cal{H}}_{5} \otimes {\cal{H}}_{6})$, where map ${\cal{A}}$ is applied on the system as first one, followed by applying map ${\cal{B}}$ on the outcome. For map ${\cal{B}}$ to be meaningful, no its output can be an input of ${\cal{A}}$, because it would disrupt causality.
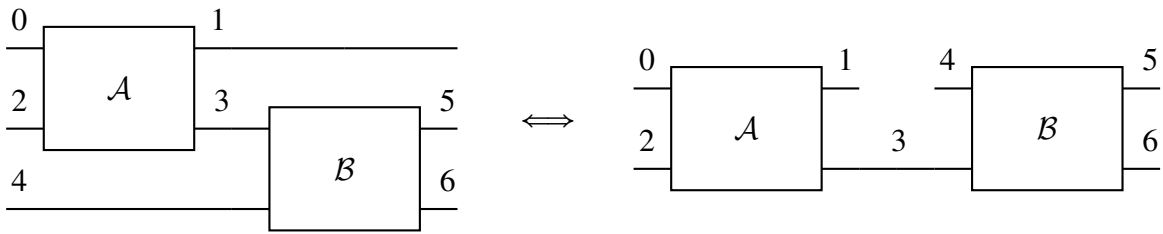
\begin{figure}[H]
	\begin{center}
\begin{quantikz}
	\rstick{\hspace{-0.1cm}0\\~}\qw &\gate[wires=2][2cm]{\mathcal{A}}  &\qw \rstick{\hspace{-0.7cm}1\\~} &\qw &\qw  \\
	\rstick{\hspace{-0.1cm}2\\~}\qw & &\qw \rstick{\hspace{-0.6cm}3\\~} &\gate[wires=2][2cm]{\mathcal{B}}&\qw\rstick{\hspace{-0.6cm}5\\~} \\
	\rstick{\hspace{-0.1cm}4\\~}\qw &\qw &\qw & &\qw\rstick{\hspace{-0.6cm}6\\~}
\end{quantikz}
		\quad $\iff$ \quad
\begin{quantikz}
	\rstick{\hspace{-0.1cm}0\\~} \qw&\gate[wires=2][2cm]{\mathcal{A}}  & \qw \rstick{\hspace{-0.7cm}1\\~} & &   \rstick{\hspace{-0.1cm}4\\~}&\gate[wires=2][2cm]{\mathcal{B}}&\qw\rstick{\hspace{-0.6cm}5\\~}\\
	\rstick{\hspace{-0.1cm}2\\~} &                     &\qw \rstick{\hspace{0.25cm}3\\~}& \qw & \qw &  &\qw\rstick{\hspace{-0.6cm}6\\~}
\end{quantikz}
	\end{center}
\caption[font=small]{Two equivalent ways of drawing composition of linear map ${\cal{A}}: {\cal{L}}({\cal{H}}_{0} \otimes {\cal{H}}_{2}) \rightarrow {\cal{L}}({\cal{H}}_{1} \otimes {\cal{H}}_{3})$ with linear map ${\cal{B}}: {\cal{L}}({\cal{H}}_{3} \otimes {\cal{H}}_{4}) \rightarrow {\cal{L}}({\cal{H}}_{5} \otimes {\cal{H}}_{6})$.}
\end{figure}

Let us now depict state preparation and measurement. State $\varrho$ can be viewed as Choi operator of preparation device, which can be viewed as channel ${\cal{C}}_{\varrho}: \mathbb{C} \rightarrow {\cal{L}}({\cal{H}}_{0})$ from one-dimensional Hilbert space (one-dimensional Hilbert space is isomorphic to space of complex numbers $\mathbb{C}$) into space ${\cal{L}}({\cal{H}}_{0})$. Channel ${\cal{C}}_{\varrho}$ is such that ${\cal{C}}_{\varrho}(c) = c\varrho$ for all $c \in \mathbb{C}$. Then state $\varrho$ can be expressed as Choi operator of said channel $C_{\varrho} \overset{(\ref{choi})}{=} ({\cal{C}}_{\varrho} \otimes {\cal{I}}_{\mathbb{C}}) (1 \otimes 1) = {\cal{C}}_{\varrho}(1) = \varrho$. Measurement is a map with one-dimensional output ${\cal{M}}: {\cal{L}}({\cal{H}}_{0}) \rightarrow \mathbb{C}$.
\begin{figure}[H]
	\begin{center}
	\includegraphics[scale=0.5]{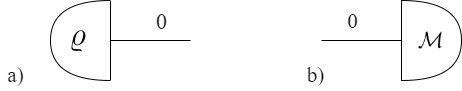}
	\end{center}
	\caption[font=small]{On the left $a)$, there is a depiction of state preparation $\varrho:\mathbb{C} \rightarrow {\cal{L}}({\cal{H}}_{0})$, where $\varrho$ is Choi operator of preparation device and on the right $b)$, there is a depiction of quantum measurement  ${\cal{M}}: {\cal{L}}({\cal{H}}_{0}) \rightarrow \mathbb{C}$.}
\end{figure}

\section{States, Channels, Instruments and POVMs}

Quantum circuits are described by quantum states, POVMs, quantum channels and quantum instruments, which were already introduced in chapter \ref{MF}. Here, we shall introduce corresponding Choi operators for these objects.

States are described by density matrices that are Hermitian, positive semi-definite operators with trace equal to one. Choi operators corresponding to quantum states must have one-dimensional input Hilbert space due to normalization condition. Let us have quantum state $\varrho \in {\cal{L}}({\cal{H}}_{in} \otimes {\cal{H}}_{out})$, where $\dim({\cal{H}}_{in}) = 1$. And let $C_{\varrho}$ be the corresponding Choi operator to state $\varrho$. Then using lemma \ref{TPlemma}, we obtain $\Tr_{out}C_{\varrho} = \mathbb{1}_{in}$. And because $\mathbb{1}_{in}$ is one-dimensional projector, its trace is $\Tr(\mathbb{1}_{in}) = 1$, which reflects that one quantum state can be seen as having one use of preparation device at our disposal.

Quantum channels, describing transformations in quantum circuits, are completely positive trace-preserving linear maps. Due to lemmas \ref{TPlemma} and \ref{CPlemma} their expression through Choi operators is straightforward. Let us have Choi operator of quantum channel $C_{{\cal{C}}} \in {\cal{L}}({\cal{H}}_{a} \otimes {\cal{H}}_{b})$, then normalization condition can be rewritten in the following way: $\Tr_{b}C_{{\cal{C}}} = \Tr_{b} \left[C_{{\cal{C}}}(\mathbb{1}_{a} \otimes \mathbb{1}_{b})\right] \overset{(\ref{linkProduct})}{=} C_{{\cal{C}}} \star \mathbb{1}_{b} \overset{\text{lemma } \ref{TPlemma}}{=} \mathbb{1}_{a}$.

Every quantum channel can be realized by an isometry on a larger Hilbert space. Isometry is an operator $V: {\cal{H}}_{a} \rightarrow {\cal{H}}_{b}$, where $V^{\dagger}V = \mathbb{1}_{a}$, that can form isometric channel ${\cal{V}}(\varrho) = V \varrho V^{\dagger}$.
\begin{theorem}\label{Stinespring}
	Stinespring Dilation
	\\
	Let ${\cal{C}} \in {\cal{L}}({\cal{L}}({\cal{H}}_{a}),{\cal{L}}({\cal{H}}_{b}))$ be a quantum channel and let $C$ be its Choi operator. Then there exists an ancillary Hilbert space ${\cal{H}}_{A}$ and isometry $V: {\cal{L}}({\cal{H}}_{a}) \rightarrow {\cal{L}}({\cal{H}}_{b} \otimes {\cal{H}}_{A})$ such that ${\cal{C}}(\varrho) = \Tr_{A}(V\varrho V^{\dagger})$.
\end{theorem}
In the proof, we shall make use of the operator $T_{a \rightarrow a^{\prime}} = \sum_{i} \ket{i}_{a^{\prime}} \prescript{}{a}{\bra{i}}$ that swaps the isomorphic Hilbert spaces ${\cal{H}}_{a} \cong {\cal{H}}_{a^{\prime}}$. Applying this operator to $O_{a} \in {\cal{L}}({\cal{H}}_{A})$ yields $O_{a^{\prime}} = T_{a \rightarrow a^{\prime}} O_{a} T_{a^{\prime} \rightarrow a}$.
\begin{figure}[H]
	\begin{center}
	\includegraphics[scale=0.4]{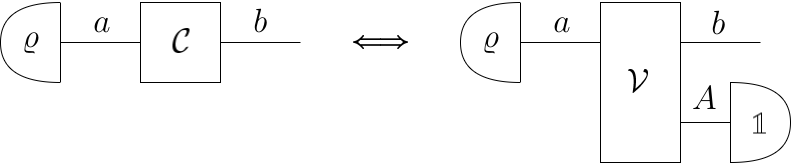}
	\end{center}
	\caption[font=small]{Graphical depiction of Stinespring dilation. On the left, there is a depiction of quantum channel ${\cal{C}}$ and on the right, there is a realization of the same channel through isometry $V$ and ancillary Hilbert space ${\cal{H}}_{A}$.}
\end{figure}
\begin{proof}
	Let us consider ancillary Hilbert space that spans the support of Choi operator of channel ${\cal{C}}$: ${\cal{H}}_{A} = \supp(C_{b^{\prime}a^{\prime}}^{\star})$, where $C^{\star}$ denotes complex conjugate of $C$. Therefore, ${\cal{H}}_{A} \subseteq {\cal{H}}_{a^{\prime}} \otimes {\cal{H}}_{b^{\prime}}$, with spaces ${\cal{H}}_{a} \cong {\cal{H}}_{a^{\prime}}$ and ${\cal{H}}_{b} \cong {\cal{H}}_{b^{\prime}}$ being isomorphic. Let us now define isometry:
	\begin{align}\label{isomStinespring}
		V = \left(\mathbb{1}_{b} \otimes C_{b^{\prime}a^{\prime}}^{\frac{1}{2}\star}\right)\left(\dket{I}_{bb^{\prime}} \otimes T_{a \rightarrow a^{\prime}}\right).
	\end{align}
	Let us check that this is really an isometry:
	\begin{align*}
		&V^{\dagger}V = \left(\prescript{}{{bb^{\prime}}}{\dbra{I}} \otimes T_{a^{\prime} \rightarrow a}\right) \left(\mathbb{1}_{b} \otimes C_{b^{\prime}a^{\prime}}^{\frac{1}{2}\star\dagger}\right) \left(\mathbb{1}_{b} \otimes C_{b^{\prime}a^{\prime}}^{\frac{1}{2}\star}\right)\dket{I}_{bb^{\prime}} \otimes T_{a \rightarrow a^{\prime}} = \prescript{}{{bb^{\prime}}}{\dbra{I}} \left(\mathbb{1}_{b} \otimes C_{b^{\prime}a}^{\star}\right) \dket{I}_{bb^{\prime}} \\&\overset{(\ref{opAndDKet})}{=} \sum_{m,n} \underbrace{\mel{n}{\mathbb{1}}{m}}_{\delta_{mn}}  \prescript{}{b}{\bra{n}} \prescript{}{b^{\prime}}{\bra{m}} (\mathbb{1}_{b} \otimes C_{b^{\prime}a}^{\ast}) \sum_{m^{\prime},n^{\prime}} \underbrace{\mel{n^{\prime}}{\mathbb{1}}{m^{\prime}}}_{\delta_{m^{\prime} n^{\prime}}} \ket{n^{\prime}}_{b} \ket{m^{\prime}}_{b^{\prime}} \\&= \sum_{m,m^{\prime}} \prescript{}{b}{\bra{m}} \prescript{}{b^{\prime}}{\bra{m}} (\mathbb{1}_{b} \otimes C_{b^{\prime}a}^{\ast}) \ket{m^{\prime}}_{b} \ket{m^{\prime}}_{b^{\prime}} = \sum_{m,m^{\prime}} \underbrace{\mel{m}{\mathbb{1}_{b}}{m^{\prime}}}_{\delta_{mm^{\prime}}} \prescript{}{b^{\prime}\hspace*{-1mm}}{\mel{m}{C_{b^{\prime}a}^{\ast}}{m}_{b^{\prime}}} \\&= \sum_{m} \prescript{}{b^{\prime}\hspace*{-1mm}}{\mel{m}{C_{b^{\prime}a}^{\ast}}{m}_{b^{\prime}}} = \Tr_{b^{\prime}}(C_{b^{\prime}a}^{\star}) \underset{\star \text{ and }\Tr}{\overset{\text{commutation of}}{=}}  \left[\Tr_{b^{\prime}}(C_{b^{\prime}a})\right]^{\star} \overset{\text{lemma } \ref{TPlemma}}{=} \mathbb{1}_{a}^{\star} = \mathbb{1}_{a}.
	\end{align*}
	Therefore, $V$ is really an isometry. Let us now show that $V$ and ancillary space can be used to realize a channel ${\cal{C}}$:
	\begin{align*}
		&\Tr_{A}(V\varrho V^{\dagger}) = \Tr_{A} \left[ \left(\mathbb{1}_{b} \otimes C_{b^{\prime}a^{\prime}}^{\frac{1}{2}\star}\right) \big(\dket{I}_{bb^{\prime}} \otimes T_{a \rightarrow a^{\prime}}\big) \varrho_{a} \big(\prescript{}{{bb^{\prime}}}{\dbra{I}} \otimes T_{a^{\prime} \rightarrow a}\big) \left(\mathbb{1}_{b} \otimes C_{b^{\prime}a^{\prime}}^{\frac{1}{2}\star\dagger}\right)\right] \\&\overset{{\cal{H}}_{A} \subseteq {\cal{H}}_{a^{\prime}} \otimes {\cal{H}}_{b^{\prime}}}{=} \Tr_{a^{\prime}b^{\prime}} \left[\left(\mathbb{1}_{b} \otimes C_{b^{\prime}a^{\prime}}^{\frac{1}{2}\star}\right) \left(\dket{I}_{bb^{\prime}}\dbra{I} \otimes \varrho_{a^{\prime}}\right) \left(\mathbb{1}_{b} \otimes C_{b^{\prime}a^{\prime}}^{\frac{1}{2}\star\dagger}\right)\right] \\&\underset{\text{and cyclicity of trace}}{\overset{\text{Hermiticity of } C_{b^{\prime}a^{\prime}}^{\frac{1}{2}\star}}{=}} \Tr_{a^{\prime}b^{\prime}} \left[\left(\mathbb{1}_{b} \otimes C_{b^{\prime}a^{\prime}}^{\star}\right) \left(\dket{I}_{bb^{\prime}}\dbra{I} \otimes \varrho_{a^{\prime}}\right)\right] \\&\overset{(i)}{=} \Tr_{a^{\prime}b^{\prime}} \left[\left(\mathbb{1}_{b} \otimes C_{b^{\prime}a^{\prime}}\right) \left(\dket{I}_{bb^{\prime}}\dbra{I}^{T_{b^{\prime}}} \otimes \varrho_{a^{\prime}}^{T_{a^{\prime}}}\right)\right] = \Tr_{a^{\prime}b^{\prime}} \left[\left(\mathbb{1}_{b} \otimes C_{b^{\prime}a^{\prime}}\right) \left(\sum_{ij}(\ket{i}_{b}\ket{i}_{b^{\prime}}\prescript{}{b}{\bra{j}}\prescript{}{b^{\prime}}{\bra{j}})^{T_{b^{\prime}}} \otimes \varrho_{a^{\prime}}^{T_{a^{\prime}}}\right)\right] \\&= \sum_{ij}\Tr_{a^{\prime}b^{\prime}} \left[\ket{i}_{b}\hspace*{-1mm}\bra{j} \otimes C_{b^{\prime}a^{\prime}}(\ket{j}_{b^{\prime}}\hspace*{-1mm}\bra{i} \otimes \varrho_{a^{\prime}}^{T_{a^{\prime}}})\right] \overset{(ii)}{=} \sum_{ij} \Tr_{a^{\prime}} \Big[\underbrace{\ket{i}_{b}\prescript{}{b^{\prime}}{\bra{i}} C_{b^{\prime}a^{\prime}} \ket{j}_{b^{\prime}}\prescript{}{b}{\bra{j}}}_{T_{b^{\prime} \rightarrow b} C_{b^{\prime}a^{\prime}} T_{b \rightarrow b^{\prime}}}  \otimes \varrho_{a^{\prime}}^{T_{a^{\prime}}}\Big] \\&= \Tr_{a^{\prime}} \left[C_{ba^{\prime}} \left(\mathbb{1}_{b} \otimes \varrho_{a^{\prime}}^{T_{a^{\prime}}}\right)\right] \overset{(\ref{inverseChoi})}{=} {\cal{C}}(\varrho).
	\end{align*}
	In $(i)$ we have used the invariance of trace with respect to transposition and the fact that $C_{b^{\prime}a^{\prime}}^{\dagger} = C_{b^{\prime}a^{\prime}}$ is Hermitian, because it is Choi operator of channel and channel is Hermitian (as can be seen from lemma \ref{Hlemma}). In $(ii)$, we have evaluated $\Tr_{b^{\prime}}$.
\end{proof}

Quantum operation ${\cal{O}}$ is a completely positive trace non-increasing linear map. Collection of quantum operations $\{{\cal{O}}_{i}\} \in {\cal{L}}({\cal{L}}({\cal{H}}_{a}), {\cal{L}}({\cal{H}}_{b}))$ form quantum instrument ${\cal{I}}$, if they sum up to quantum channel $\sum_{i}{\cal{O}}_{i} = {\cal{C}}$. Choi operators $C_{{\cal{O}}_{i}}$ correspond to every quantum operation and sum up to Choi operator of quantum channel $C_{{\cal{C}}} = \sum_{i}C_{{\cal{O}}_{i}}$ with normalization of channel being $C \star \mathbb{1}_{b} = \mathbb{1}_{a}$.

Choi operators describing quantum instrument with one-dimensional output space ${\cal{L}}({\cal{H}}_{b}) \cong \mathbb{C}$ form a POVM. POVM is composed of effects that sum up to identity $\sum_{i}E_{i} = \mathbb{1}_{b}$. Probability of measuring outcome $i$ given state $\varrho$ is given by: $p(i|\varrho) = \varrho \star E_{i} \overset{(\ref{linkProduct})}{=} \Tr(\varrho E_{i}^{T})$. The only difference between this expression and Born's rule is in transposition of an effect. But transposition can be absorbed in a definition of a POVM.

Any quantum instrument can be realized using an isometry and a POVM on a larger Hilbert space.
\begin{theorem}
	Realization of Quantum Instrument
	\\
	Let quantum operations $\{{\cal{O}}_{i}\}, {\cal{O}}_{i} \in {\cal{L}}({\cal{L}}({\cal{H}}_{a}),{\cal{L}}({\cal{H}}_{b}))$ form a quantum instrument. Then, there exist an ancillary Hilbert space ${\cal{H}}_{A}$, channel ${\cal{C}} \in {\cal{L}}({\cal{L}}({\cal{H}}_{a}),{\cal{L}}({\cal{H}}_{b} \otimes {\cal{H}}_{A}))$ and a POVM $\{P_{i}\}, P_{i} \in {\cal{L}}({\cal{H}}_{A})$ such that ${\cal{O}}_{i}(\varrho) = \Tr_{A}\left[{\cal{C}}(\varrho) \otimes P_{i,A}\right]$, where $P_{i,A}$ denotes element of POVM $P_{i} \in {\cal{L}}({\cal{H}}_{A})$.
\end{theorem}
\begin{figure}[H]
	\begin{center}
		\includegraphics[scale=0.4]{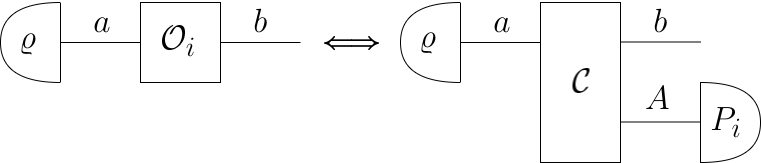}
	\end{center}
	\caption[font=small]{Graphical depiction of realization of quantum instrument. On the left, there is a depiction of quantum operation ${\cal{O}}_{i}$ and on the right, there is a realization of the same operation through quantum channel ${\cal{C}}$, element of POVM $P_{i}$ and an ancillary Hilbert space ${\cal{H}}_{A}$.}
\end{figure}
\begin{proof}
	Let ${\cal{C}} = \sum_{i}{\cal{O}}_{i}$ be a quantum channel. Let us denote its Choi operator by $C$ and Choi operators of individual quantum operations $O_{i}$. Based on Stinespring dilation \ref{Stinespring}, we define the same isometry $V = \left(\mathbb{1}_{b} \otimes C_{b^{\prime}a^{\prime}}^{\frac{1}{2}\star}\right)\dket{I}_{bb^{\prime}} \otimes T_{a \rightarrow a^{\prime}}$ as in (\ref{isomStinespring}) and the ancillary Hilbert space ${\cal{H}}_{A} = \supp(C)$. Further, let us define POVM $P_{i} = C_{b^{\prime}a^{\prime}}^{\frac{-1}{2}T} O_{i,b^{\prime}a^{\prime}}^{T} C_{b^{\prime}a^{\prime}}^{\frac{-1}{2}T}$ and let us show that it is, indeed, a POVM: $\sum_{i} P_{i} = C_{b^{\prime}a^{\prime}}^{\frac{-1}{2}T} \sum_{i} O_{i,b^{\prime}a^{\prime}}^{T} C_{b^{\prime}a^{\prime}}^{\frac{-1}{2}T} = C_{b^{\prime}a^{\prime}}^{\frac{-1}{2}T} C_{b^{\prime}a^{\prime}}^{T} C_{b^{\prime}a^{\prime}}^{\frac{-1}{2}T} = \mathbb{1}_{a^{\prime}b^{\prime}}$ and $P_{i}^{\dagger} = P_{i}$. Now, let us verify that this isometry and POVM can realize quantum instrument:
	\begin{align*}
		&\Tr_{A}\left[{\cal{C}}(\varrho)\left(\mathbb{1}_{b} \otimes P_{i,A}\right)\right] \overset{\text{theorem }\ref{Stinespring}}{=} \Tr_{A} \left[V\varrho V^{\dagger} \left(\mathbb{1}_{b} \otimes P_{i,A}\right)\right] \\&= \Tr_{a^{\prime}b^{\prime}} \left[\left(\mathbb{1}_{b} \otimes \left(C_{b^{\prime}a^{\prime}}^{\frac{1}{2}T}\right)^{\dagger}\right)\left(\varrho_{a^{\prime}} \otimes \dket{I}_{bb^{\prime}}\dbra{I}\right) \left(\mathbb{1}_{b} \otimes C_{b^{\prime}a^{\prime}}^{\frac{1}{2}T}\right) \left(\mathbb{1}_{b} \otimes C_{b^{\prime}a^{\prime}}^{\frac{-1}{2}T} O_{i,b^{\prime}a^{\prime}}^{T} C_{b^{\prime}a^{\prime}}^{\frac{-1}{2}T}\right)\right] \\&\overset{(i)}{=} \Tr_{a^{\prime}b^{\prime}}\left[\left(\varrho_{a^{\prime}} \otimes \dket{I}_{bb^{\prime}}\dbra{I}\right) \left(\mathbb{1}_{b} \otimes O_{i,b^{\prime}a^{\prime}}^{T}\right)\right] \overset{(ii)}{=}  \Tr_{a^{\prime}b^{\prime}} \left[\left(\varrho_{a^{\prime}} \otimes \sum_{jk} \ket{k}_{b}\ket{k}_{b^{\prime}} \prescript{}{b}{\bra{j}}\prescript{}{b^{\prime}}{\bra{j}}\right)^{T_{a^{\prime}b^{\prime}}} \hspace*{-0.1cm}\left(\mathbb{1}_{b} \otimes O_{i,b^{\prime}a^{\prime}}\right)\right] \\&= \Tr_{a^{\prime}b^{\prime}} \left[\left(\varrho_{a^{\prime}}^{T_{a^{\prime}}} \otimes \sum_{jk} \ket{k}_{b}\ket{j}_{b^{\prime}} \prescript{}{b}{\bra{j}}\prescript{}{b^{\prime}}{\bra{k}}\right)  \left(\mathbb{1}_{b} \otimes O_{i,b^{\prime}a^{\prime}}\right)\right] \overset{(iii)}{=} \Tr_{a^{\prime}} \left[\varrho_{a^{\prime}}^{T_{a^{\prime}}} \otimes \sum_{jk} \ket{k}_{b}\prescript{}{b^{\prime}}{\bra{k}} O_{i,b^{\prime}a^{\prime}} \ket{j}_{b^{\prime}}\prescript{}{b}{\bra{j}} \right] \\&= \Tr_{a^{\prime}} \left[O_{i,ba^{\prime}} \left(\mathbb{1}_{b} \otimes \varrho_{a^{\prime}}^{T_{a^{\prime}}}\right)\right] \overset{(\ref{inverseChoi})}{=} {\cal{O}}_{i}(\varrho).
	\end{align*}
	In $(i)$ we have used cyclicity of trace and Hermiticity of $C$, in $(ii)$ the invariance of trace under transposition and in $(iii)$ we have evaluated $\Tr_{b^{\prime}}$.
\end{proof}


\section{Graphical Construction}

We shall endeavor to explain our formalism with graphical tools. Quantum networks are combined quantum circuits, where the outputs of particular circuit serve as inputs into other quantum circuit. 

Simple graph theory can help us describe quantum networks as they can be viewed as directed acyclic graphs (DAGs). DAG is an ordered pair $G = (E,V)$, where $V$ is set of vertices and $E$ is set of ordered pairs of vertices called directed edges. Notation $e(v_{1},v_{2}) \in E$ means that there is an edge going from vertex $v_{1} \in V$ to vertex $v_{2} \in V$ \cite{GraphsTheoryAndAlgorithms}. Requirement for the graph to be without loops, or acyclic, just means that individual vertices have causal relations (for quantum networks that means, that individual operations are applied one at a time).

DAGs naturally possesses partial order $\preccurlyeq$ regarding vertices that can be changed in total order $\leq$. In general, this transition from partial order into total order is not unique. In an example given here through figure \ref{DAG}, there exists partial ordering of vertices $1 \preccurlyeq 2 \preccurlyeq 3 \preccurlyeq 4$ that can be changed in two different total orders $1 \leq 2 \leq 3 \leq 4$ or $1 \leq 3 \leq 2 \leq 4$.

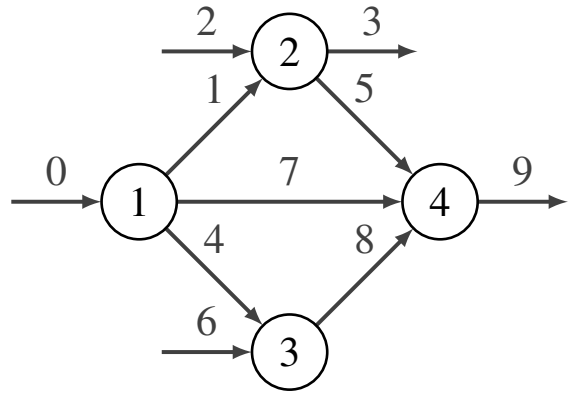
\begin{figure}[H]
	\begin{center}
	\begin{tikzpicture}[node distance={10mm}, main/.style = {draw, circle}]
		\Vertex[label = 1, fontscale = 2, x=2,y=0, opacity = 0, size = 1]{1}
		\Vertex[label = 2, fontscale = 2,x=4,y=2, opacity = 0, size = 1]{2}
		\Vertex[label = 3, fontscale = 2,x=4,y=-2, opacity = 0, size = 1]{3}
		\Vertex[label = 4, fontscale = 2,x=6,y=0, opacity = 0, size = 1]{4}
		
		\Vertex[Pseudo, x=0,y=0]{A}
		\Vertex[Pseudo, x=2,y=2]{B}
		\Vertex[Pseudo, x=2,y=-2]{C}
		\Vertex[Pseudo, x=8,y=0]{D}
		\Vertex[Pseudo, x=6,y=2]{E}
		
		\Edge[Direct,label=1,position={above=2mm},fontscale=2](1)(2)
		\Edge[Direct,label=4,position={above=2mm},fontscale=2](1)(3)
		\Edge[Direct,label=7,position={above=1mm},fontscale=2](1)(4)
		\Edge[Direct,label=5,position={above=2mm},fontscale=2](2)(4)
		\Edge[Direct,label=8,position={above=2mm},fontscale=2](3)(4)
		
		\Edge[Direct,label=0,position={above=1mm},fontscale=2](A)(1)
		\Edge[Direct,label=2,position={above=1mm},fontscale=2](B)(2)
		\Edge[Direct,label=3,position={above=1mm},fontscale=2](2)(E)
		\Edge[Direct,label=6,position={above=1mm},fontscale=2](C)(3)
		\Edge[Direct,label=9,position={above=1mm},fontscale=2](4)(D)
	\end{tikzpicture} 
	\end{center}
	\caption[font=small]{Example of directed acyclic graph (DAG) with partial order $1 \preccurlyeq 2 \preccurlyeq 3 \preccurlyeq 4$. This particular graph allows for two total orderings of vertices: $1 \leq 2 \leq 3 \leq 4$ or $1 \leq 3 \leq 2 \leq 4$.}
	\label{DAG}
\end{figure}

If we desire to interpret quantum networks as DAGs, then individual vertices take function of quantum channels (or quantum operations) and edges correspond to wires in quantum circuits. Let us redraw given graph with these substitutions on our mind:
\begin{figure}[H]
	\begin{center}
		\begin{quantikz}
			\rstick{\hspace{-0.1cm}0\\~}&\gate[5, nwires = {2,3,4,5}][2cm]{\mathcal{C}_{(1)}} &\rstick{\hspace{-0.1cm}2\\~} &\gate[2][2cm]{\mathcal{C}_{(2)}} &\rstick{\hspace{-0.6cm}3\\~}\qw &\gate[5, nwires = {1,4}][2cm]{\mathcal{C}_{(4)}} &\rstick{\hspace{-0.6cm}9\\~}\qw \\
			& &\rstick{\hspace{-0.1cm}1\\~}\qw & &\rstick{\hspace{-0.3cm}5\\~}\qw & & \\
			& &\ghost{X}\qw &\rstick{\hspace{-0.3cm}7\\~}\qw &\qw & & & \\
			& &\rstick{\hspace{-0.1cm}6\\~} &\gate[2][2cm]{\mathcal{C}_{(3)}} & & & &\\
			& &\rstick{\hspace{-0.3cm}4\\~}\qw & &\rstick{\hspace{-0.3cm}8\\~}\qw & & &
		\end{quantikz}
	\end{center}
	\label{DQNasDAG}
\end{figure}\vspace*{-1cm}
Here, quantum channel is denoted by ${\cal{C}}_{(i)}$. If we choose the total order to be $1 \leq 2 \leq 3 \leq 4$, we can redraw this by expanding gates ${\cal{C}}_{(2)}$ and ${\cal{C}}_{(3)}$. We create a new gate by expanding the old one with identities acting on Hilbert spaces ${\cal{H}}_{4}$ and ${\cal{H}}_{7}$: ${\cal{C}}_{(2)}^{\prime} = {\cal{C}}_{(2)} \otimes \mathbb{1}_{4} \otimes \mathbb{1}_{7}$. In the same fashion, we are also able to expand  ${\cal{C}}_{(3)}$ to create ${\cal{C}}_{(3)}^{\prime} = {\cal{C}}_{(3)} \otimes \mathbb{1}_{5} \otimes \mathbb{1}_{7}$:
\begin{figure}[H]
	\begin{center}
		\begin{quantikz}
			\rstick{\hspace{-0.1cm}0\\~}
			&\gate[4, nwires = {2,3,4}][2cm]{\mathcal{C}_{(1)}} &\rstick{\hspace{-0.1cm}2\\~} 
			&\gate[4][2cm]{\mathcal{C}_{(2)}^{\prime}} &\rstick{\hspace{-0.6cm}3\\~}\qw 
			& 
			&\rstick{\hspace{-0.1cm}6\\~} &\gate[4][2cm]{\mathcal{C}_{(3)}^{\prime}} 
			& 
			&\gate[4, nwires={1}][2cm]{\mathcal{C}_{(4)}} 
			&\rstick{\hspace{-0.6cm}9\\~}\qw 
			\\ 
			& 
			&\rstick{\hspace{-0.4cm}1\\~}\qw
			&\ghost{X}\qw 
			&\qw 
			&\rstick{\hspace{-0.3cm}5\\~}\qw 
			&\qw 
			&\qw 
			&\rstick{\hspace{-0.4cm}5\\~}\qw 
			& 
			& 
			\\ 
			& 
			&\rstick{\hspace{-0.4cm}4\\~}\qw  
			&\ghost{X}\qw  
			&\qw 
			&\rstick{\hspace{-0.3cm}4\\~}\qw 
			&\qw 
			&\qw 
			&\rstick{\hspace{-0.4cm}8\\~}\qw 
			& 
			& 
			\\ 
			& 
			&\rstick{\hspace{-0.4cm}7\\~}\qw  
			&\qw  
			&\qw 
			&\rstick{\hspace{-0.3cm}7\\~}\qw 
			&\qw 
			&\qw 
			&\rstick{\hspace{-0.4cm}7\\~}\qw 
			& 
			& 
		\end{quantikz}
	\end{center}
	\label{orderedDQN}
\end{figure}\vspace*{-1cm}
We can further draw this diagram more compactly by grouping together wires $1$, $4$ and $7$, then $5$, $4$ and $7$ and finally $5$, $8$ and $7$. This means that we are denoting all Hilbert spaces corresponding to these wires in a new way. Therefore, we rename Hilbert spaces accordingly ${\cal{H}}_{A_{1}} = {\cal{H}}_{1} \otimes {\cal{H}}_{4} \otimes {\cal{H}}_{5}$, ${\cal{H}}_{A_{2}} = {\cal{H}}_{4} \otimes {\cal{H}}_{5} \otimes {\cal{H}}_{7}$ and ${\cal{H}}_{A_{3}} = {\cal{H}}_{5} \otimes {\cal{H}}_{7} \otimes {\cal{H}}_{8}$:
\begin{figure}[H]
	\begin{center}
		\begin{quantikz}
			\rstick{\hspace{-0.1cm}0\\~}
			&\gate[2, nwires = {2}][2cm]{\mathcal{C}_{(1)}} &\rstick{\hspace{-0.1cm}2\\~} 
			&\gate[2][2cm]{\mathcal{C}_{(2)}^{\prime}} 
			&\rstick{\hspace{-0.6cm}3\\~}\qw 
			& 
			&\rstick{\hspace{-0.1cm}6\\~} 
			&\gate[2][2cm]{\mathcal{C}_{(3)}^{\prime}} 
			& 
			&\gate[2, nwires = {1}][2cm]{\mathcal{C}_{(4)}} &\rstick{\hspace{-0.6cm}9\\~}\qw 
			\\
			& 
			&\rstick{\hspace{-0.4cm}$A_{1}$\\~}\qw 
			&\qw 
			&\qw 
			&\rstick{\hspace{-0.3cm}$A_{2}$\\~}\qw 
			&\qw 
			&\qw 
			&\rstick{\hspace{-0.4cm}$A_{3}$\\~}\qw 
			& 
			&
		\end{quantikz}
	\end{center}
	\label{finalDQN}
\end{figure}\vspace*{-1cm}

In this fashion, we are able to redraw every possible quantum network with $N$ vertices in such a way that is equivalent to a concatenation of $N$ quantum operations (and which looks similarly to the last diagram). Therefore, Choi operator of quantum network is given by the link product of Choi operators of the individual maps $R^{N} = C_{1} \star C_{2} \star \cdots \star C_{N}$. Henceforth, the entire quantum network, which can be composed of myriad quantum circuits, is describable by only one operator. In this description, we also encompass that physical implementations of the whole quantum network can differ, albeit the final outcome stays the same. All these possible different implementations form class of equivalence of quantum networks.

\section{Deterministic Quantum Networks}

Deterministic quantum networks (DQN) are composed of quantum channels as they preserve traces of input quantum states, therefore the outcome of the DQN is not random. We can look at DQNs as DAGs:

\begin{remark} \label{DQNDAG}
	DQN as DAG
	\\
	DQNs are linear maps that correspond to DAGs, where:
	\begin{itemize}
		\item each edge is numbered by a unique integer
		\item edge numbered by an integer $j$ represents Hilbert space ${\cal{H}}_{j}$
		\item each vertex is numbered by a unique integer
		\item vertex numbered by an integer $i$ represents channel $C_{(i)}: {\cal{L}}({\cal{H}}_{in_{i}}) \rightarrow {\cal{L}}({\cal{H}}_{out_{i}})$
		\item edge between vertices $i$ and $i^{\prime}$ represents composition of channels $C_{(i)} \circ C_{(i^{\prime})}$.
	\end{itemize}
\end{remark}

Because these networks are formed by concatenation of quantum channels, their corresponding Choi operators are given by link product of the individual Choi operators of channels. Let us have a DQN ${\cal{R}}^{(N)}$ depicted in figure \ref{taco}, that maps input from ${\cal{L}}({\cal{H}}_{in})$ into ${\cal{L}}({\cal{H}}_{out})$, where $N$ is a finite number of vertices (i.e., channels), ${\cal{H}}_{in} = \otimes_{i=0}^{N-1} {\cal{H}}_{2i}$ denotes all input Hilbert spaces and ${\cal{H}}_{out} = \otimes_{i=0}^{N-1} {\cal{H}}_{2i+1}$ denotes all output Hilbert spaces. Then the Choi operator of such a DQN is:
\begin{align*}
	R^{N} = C_{1} \star C_{2} \star \cdots \star C_{N},
\end{align*}
where $C_{i}$'s are Choi operators of quantum channels ${\cal{C}}_{(i)}$ that form DQN ${\cal{R}}^{(N)}$.
\begin{figure}[H]
	\begin{center}
		\begin{quantikz}
			\rstick{\hspace{-0.1cm}0\\~}
			&\gate[2, nwires = {2}][2cm]{\mathcal{C}_{(1)}}
			&\qw \rstick{\hspace{-0.6cm}1\\~}
			& 
			&\rstick{\hspace{-0.1cm}2\\~}
			&\gate[2][2cm]{\mathcal{C}_{(2)}}
			&\rstick{\hspace{-0.4cm}3\\~}\qw
			&\ \ldots\ \qw			
			&\rstick{\hspace{-0.6cm}2N-2\\~}
			&\gate[2][2cm]{\mathcal{C}_{(N)}}
			&\rstick{\hspace{-0.6cm}2N-1\\~}\qw 
			\\ 
			&
			&\qw \rstick{\hspace{0.2cm}$A_{1}$\\~}
			&\qw
			&\qw
			&
			&\rstick{\hspace{-0.4cm}$A_{2}$\\~}\qw
			&\ \ldots\ \qw			
			&\rstick{\hspace{-0.6cm}$A_{N-1}$\\~}
			&
			&
		\end{quantikz}
	\end{center}
	\caption[font=small]{Depiction of DQN formed by concatenation of $N$ channels. Input spaces are ${\cal{H}}_{in} = \otimes_{i=0}^{N-1} {\cal{H}}_{2i}$, while output spaces are ${\cal{H}}_{out} = \otimes_{i=0}^{N-1} {\cal{H}}_{2i+1}$.}
	\label{taco}
\end{figure}
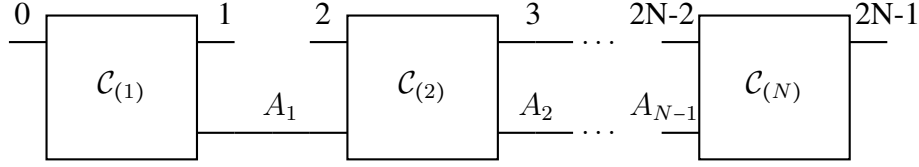

We shall state normalization condition for a DQN:
\begin{lemma}\label{normalizationConditionDQN}
	Normalization Condition for DQN
	\\
	Let us have Choi operator $R^{N}$ corresponding to the DQN ${\cal{R}}^{(N)} \in {\cal{L}}(\otimes_{i=0}^{2N-1}{\cal{H}}_{i})$. Then, $R^{N}$ is positive semi-definite and the following relation holds:
	\begin{align}\label{normalizationDQN}
		\Tr_{2k-1}(R^{k}) = \mathbb{1}_{2k-2} \otimes R^{k-1} \qquad \text{for } k = 1, \cdots, N,
	\end{align}
	where $R^{k-1}$ is a Choi operator of the DQN ${\cal{R}}^{(k-1)}$.
\end{lemma}
If $k = N$, then we obtain $\Tr_{2N-1}(R^{N}) = \mathbb{1}_{2N-2} \otimes R^{N-1}$ and if $k = 1$, then $\Tr_{1}(R^{1}) = \mathbb{1}_{0} \otimes R^{0} = \mathbb{1}_{0}$.
\begin{proof}
	DQN ${\cal{R}}^{(N)}$ with $N$ vertices can be expressed as a concatenation of $N$ channels ${\cal{R}}^{(N)} = {\cal{C}}_{(1)} \star {\cal{C}}_{(2)} \star \cdots \star {\cal{C}}_{(N)}$, where ${\cal{C}}_{(i)}: {\cal{L}}({\cal{H}}_{2i-2} \otimes {\cal{H}}_{A_{i-1}}) \rightarrow {\cal{L}}({\cal{H}}_{2i-1} \otimes {\cal{H}}_{A_{i}})$ and where we put first and last ancillary spaces to be one-dimensional ${\cal{H}}_{A_{0}} \cong {\cal{H}}_{A_{N}} \cong \mathbb{C}$. Let $C_{i} \in {\cal{L}}(\otimes_{k\in {\cal{J}}_{i}})$ denote Choi operator of channel ${\cal{C}}_{(i)}$, where ${\cal{J}}_{i} = \{2i-2, A_{i-1}, 2i-1, A_{i}\}$. Noting that ${\cal{J}}_{i} \cap {\cal{J}}_{j} \cap {\cal{J}}_{k} = \emptyset$, we can use associativity of link product from lemma \ref{propLink}, to express Choi operator of DQN as $R^{N} = C_{1} \star C_{2} \star \cdots \star C_{N}$. Let us now evaluate trace of DQN:
	\begin{align*}
		\Tr_{2N-1}(R^{N}) &= C_{1} \star C_{2} \star \cdots \star \Tr_{2N-1}(C_{N}) \overset{\text{lemma } \ref{TPlemma}}{=} C_{1} \star C_{2} \star \cdots \underbrace{C_{N-1} \star \mathbb{1}_{A_{N-1}}}_{\widetilde{C_{N-1}}} \otimes \mathbb{1}_{2N-2} \\&= C_{1} \star C_{2} \star \cdots \star \widetilde{C_{N-1}} \otimes \mathbb{1}_{2N-2} = R^{N-1} \otimes \mathbb{1}_{2N-2}.
	\end{align*}
	From here, we can proceed iteratively. Let us also explicitly write the normalization for DQN with one vertex: $\Tr_{1}(R^{1}) = \Tr_{1}(C_{1} \star \mathbb{1}_{A_{1}}) \overset{(\ref{linkProduct})}{=} \Tr_{1}\Tr_{A_{1}}(C_{1}) \overset{\text{lemma } \ref{TPlemma}}{=} \mathbb{1}_{0}$.
\end{proof}
In the following, we provide a "recipe" for realization of DQN using isometries, which is a consequence of Stinespring dilation \ref{Stinespring}:
\begin{theorem}\label{realitzationTheoremDQN}
	Realization Theorem for DQN
	\\
	Let us have positive semi-definite operator $R^{(N)} \in {\cal{L}}({\cal{H}}_{out} \otimes {\cal{H}}_{in})$, where ${\cal{H}}_{in} = \otimes_{i=0}^{N-1} {\cal{H}}_{2i}$ and ${\cal{H}}_{out} = \otimes_{i=0}^{N-1} {\cal{H}}_{2i+1}$ satisfying $\Tr_{2k-1}(R^{k}) = \mathbb{1}_{2k-2} \otimes R^{k-1}$ for $k = 1, \cdots, N$. Then, $R^{N}$ is Choi operator of DQN ${\cal{R}}^{(N)}$. Moreover, DQN ${\cal{R}}^{(N)}$ can be realized as a concatenation of isometries $V_{i}$ and tracing out the ancillary space:
	\begin{align}\label{realizationDQN}
		{\cal{R}}^{(N)}(\varrho) = \Tr_{A_{N}} \left(V_{N} V_{N-1} \cdots V_{1} \varrho V_{1}^{\dagger} \cdots V_{N-1}^{\dagger} V_{N}^{\dagger}\right).
	\end{align}
\end{theorem}
\begin{figure}[H]
	\begin{center}
		\begin{quantikz}
			\rstick{\hspace{-0.1cm}0\\~}
			&\gate[2, nwires = {2}][2cm]{\mathcal{V}_{(1)}}
			&\qw \rstick{\hspace{-0.6cm}1\\~}
			& 
			&\rstick{\hspace{-0.1cm}2\\~}
			&\gate[2][2cm]{\mathcal{V}_{(2)}}
			&\rstick{\hspace{-0.4cm}3\\~}\qw
			&\ \ldots\ \qw			
			&\rstick{\hspace{-0.6cm}2N-2\\~}
			&\gate[2][2cm]{\mathcal{V}_{(N)}}
			&\rstick{\hspace{-0.6cm}2N-1\\~}\qw 
			&
			\\ 
			{}
			&
			&\qw \rstick{\hspace{0.2cm}$A_{1}$\\~}
			&\qw
			&\qw
			&
			&\rstick{\hspace{-0.4cm}$A_{2}$\\~}\qw
			&\ \ldots\ \qw			
			&\rstick{\hspace{-0.6cm}$A_{N-1}$\\~}
			&
			&\rstick{\hspace{-0.6cm}$A_{N}$\\~}\qw 
			&\meterD{\mathbb{1}}
		\end{quantikz}
	\end{center}
	\caption[font=small]{Depiction of DQN with $N$ channels realized by an isometric channels ${\cal{V}}_{(i)}$, that are defined as ${\cal{V}}_{(i)}(\varrho) = V_{i}\varrho V_{i}^{\dagger}$ with $V_{i} \in {\cal{L}}({\cal{H}}_{2i-2} \otimes {\cal{H}}_{A_{i-1}}, {\cal{H}}_{2i-1} \otimes {\cal{H}}_{A_{i}})$ being an isometry.}
	\label{realDQN}
\end{figure}
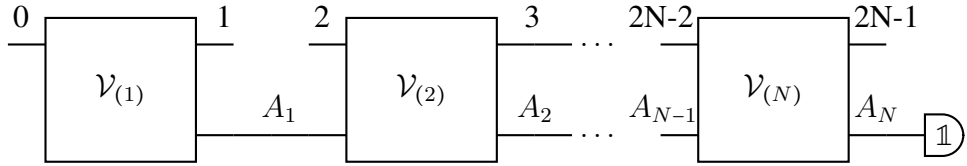
\begin{proof}
	This proof is similar to the proof of Stinespring dilation \ref{Stinespring}. Let us define Hilbert space ${\cal{H}}_{A_{k}} = \supp(R^{k\ast}_{A_{k}})$, where $\supp(R^{k\ast}_{A_{k}})$ denotes support of $R^{k\ast}_{A_{k}}$, and isometry:
	\begin{align*}
		V_{k} = \mathbb{1}_{2k-1} \otimes 	R^{k\frac{1}{2}\ast}_{A_{k}}R^{(k-1)(\frac{-1}{2})\ast}_{A_{(k-1)}} \dket{I}_{\substack{(2k-1)\\(2k-1)^{\prime}}} \otimes T_{(2k-2)\rightarrow(2k-2)^{\prime}}.
	\end{align*}  
	Let us remind that $T_{(2k-2)\rightarrow(2k-2)^{\prime}} = T_{\substack{(2k-2)\rightarrow\\(2k-2)^{\prime}}} = \sum_{k} \ket{k}_{(2k-2)^{\prime}}\hspace*{1mm}\prescript{}{(2k-2)}{\bra{k}}$ exchanges the isomorphic Hilbert spaces. Let us check that this really is an isometry:
	\begin{align*}
		&V_{k}^{\dagger}V_{k} \\&\overset{(i)}{=} \left[ T_{\substack{(2k-2)^{\prime}\rightarrow\\(2k-2)}} \otimes \prescript{}{\substack{(2k-1)\\(2k-1)^{\prime}}}{\dbra{I}} \mathbb{1}_{2k-1} \otimes R^{(k-1)\left(\frac{-1}{2}\right)\ast}_{A_{(k-1)}}R^{k\frac{1}{2}\ast}_{A_{k}} \right] \left[\mathbb{1}_{2k-1} \otimes R^{k\frac{1}{2}\ast}_{A_{k}}R^{(k-1)\left(\frac{-1}{2}\right)\ast}_{A_{(k-1)}} \dket{I}_{\substack{(2k-1)\\(2k-1)^{\prime}}} \otimes T_{\substack{(2k-2)\rightarrow\\(2k-2)^{\prime}}}\right] \\&\overset{(ii)}{=} T_{\substack{(2k-2)^{\prime}\rightarrow\\(2k-2)}} \hspace*{1mm} \Tr_{(2k-1)^{\prime}} \left[R^{(k-1)\left(\frac{-1}{2}\right)\ast}_{A_{(k-1)}}R^{k\ast}_{A_{k}}R^{(k-1)\left(\frac{-1}{2}\right)\ast}_{A_{(k-1)}}\right] T_{\substack{(2k-2)\rightarrow\\(2k-2)^{\prime}}} \\&\overset{(iii)}{=} T_{\substack{(2k-2)^{\prime}\rightarrow\\(2k-2)}} \left[\mathbb{1}_{(2k-2)^{\prime}} \otimes \mathbb{1}_{A_{(k-1)}}\right] T_{\substack{(2k-2)\rightarrow\\(2k-2)^{\prime}}} = \mathbb{1}_{2k-2} \otimes \mathbb{1}_{A_{(k-1)}}.
	\end{align*}
	Thus, we have proved that $V$ is indeed an isometry. In $(i)$ we have used the fact that $R^{(k)}$ is an Hermitian operator, in $(ii)$ we have used equation (\ref{stateChannelDuality}) and similar procedure as in proof of Stinespring dilation \ref{Stinespring}. And in $(iii)$ we have used the normalization condition \ref{normalizationConditionDQN}: $\Tr_{(2k-1)^{\prime}}(R^{k\ast}_{A_{k}}) = \mathbb{1}_{(2k-2)^{\prime}} \otimes R^{(k-1)\ast}_{A_{(k-1)}}$, where we have also used that ${\cal{H}}_{A_{k}} = \supp(R^{k\ast}_{A_{k}})$ and therefore ${\cal{H}}_{A_{k}} \subseteq \otimes_{i=0}^{2k-1}{\cal{H}}_{i^{\prime}}$.
	
	Let us now define an isometry $W_{N} = V_{N}\cdots V_{1}$ and evaluate the following expression:
	\begin{align*}
		V_{k}V_{k-1} &=
		\left[\mathbb{1}_{\substack{(2k-1)\\(2k-3)}} \otimes R^{k\frac{1}{2}\ast}_{A_{k}} R^{(k-1)\left(\frac{-1}{2}\right)\ast}_{A_{(k-1)}} \dket{I}_{\substack{(2k-1)\\(2k-1)^{\prime}}} \otimes T_{\substack{(2k-2)\rightarrow\\(2k-2)^{\prime}}} \right]\\ &\hspace*{0.37cm}\left[\mathbb{1}_{\substack{(2k-3)\\(2k-1)}} \otimes R^{(k-1)\frac{1}{2}\ast}_{A_{(k-1)}} R^{(k-2)(\frac{-1}{2})\ast}_{A_{(k-2)}} \dket{I}_{\substack{(2k-3)\\(2k-3)^{\prime}}} \otimes T_{\substack{(2k-4)\rightarrow\\(2k-4)^{\prime}}} \right] \\
		&=\underbrace{\mathbb{1}_{\substack{(2k-1)\\(2k-3)}}}_{\mathbb{1}_{out}} \otimes R^{k\frac{1}{2}}_{A_{k}} \otimes \underbrace{\mathbb{1}_{Supp(R^{(k-1)\ast})}}_{\subseteq \otimes_{i=0}^{2k-1}{\cal{H}}_{i^{\prime}}} \otimes R^{(k-2)\left(\frac{-1}{2}\right)\ast}_{A_{(k-2)}} \underbrace{\dket{I}_{\substack{(2k-1)\\(2k-1)^{\prime}}} \dket{I}_{\substack{(2k-3)\\(2k-3)^{\prime}}}}_{\dket{I}_{\substack{out\\out^{\prime}}}} \otimes \underbrace{T_{\substack{(2k-2)\rightarrow\\(2k-2)^{\prime}}} \otimes T_{\substack{(2k-4)\rightarrow\\(2k-4)^{\prime}}}}_{T_{in\rightarrow in^{\prime}}}.
	\end{align*}
	By evaluating the entire product of operators $V_{N}\cdots V_{1}$, we get the outcome (which is already foreshadowed in the last equation): $W_{N} = \mathbb{1}_{out} \otimes R^{N\frac{1}{2}\ast}_{A_{N}} \dket{I}_{\substack{out\\out^{\prime}}} \otimes T_{in \rightarrow in^{\prime}}$. Using Stinespring dilation \ref{Stinespring}, we can see that this is an isometry for channel ${\cal{R}}^{(N)}$ and that the entire DQN can be realized as ${\cal{R}}^{(N)}(\varrho) = \Tr_{A_{N}}\left(W_{N} \varrho W_{N}^{\dagger}\right)$.
\end{proof}

\section{Probabilistic Quantum Networks}

The role of quantum channels is in probabilistic quantum networks (PQNs) substituted by completely positive trace non-increasing linear maps (i.e., quantum operations). These networks produce stochastic output based on the outcome of measurement.

PQNs can also be viewed as DAGs. The only difference in comparison to interpreting DQNs as DAGs (as in \ref{DQNDAG}) is the substitution of quantum channels for quantum operations.

The following lemma generalizes relation between quantum operations and channels to quantum networks.
\begin{lemma}
	Sub-normalization of PQN
	\\
	If $R^{N}$ is a Choi operator of PQN ${\cal{R}}^{(N)}$, then there exists DQN ${\cal{S}}^{(N)}$ with corresponding Choi operator $S^{N}$ such that $R^{N} \leq S^{N}$.
\end{lemma}
\begin{proof}
	For $N = 1$ we have $R^{1}$, which is just a quantum operation. And for quantum operation there is always a channel $S^{1}$ such that $R^{1} \leq S^{1}$.
	
	We shall prove previous lemma through mathematical induction. Let us assume that for $N-1$, for all PQNs $R^{N-1}$ there exists DQN $S^{N-1}$ such that $R^{N-1} \leq S^{N-1}$. Choi operator of PQN is a concatenation of Choi operators of quantum operations $R^{N} = O_{1} \star \cdots \star O_{N}$. For every individual $O_{i}$ there exists a channel $C_{i}$ such that $O_{i} \leq C_{i}$. This fact together with assumption gives us $R^{N-1} \star O_{N} \leq S^{N-1} \star C_{N}$, but $R^{N} = R^{N-1} \star O_{N}$ and $S^{N} = S^{N-1} \star C_{N}$. Thus, previous lemma is proved.
\end{proof}

We shall provide realization theorem also for PQNs.
\begin{theorem}\label{realizationPQN}
	Realization Theorem for PQN
	\\
	Let us have positive semi-definite operator $R^{N} \in {\cal{L}}({\cal{H}}_{out} \otimes {\cal{H}}_{in})$, where ${\cal{H}}_{in} = \otimes_{i=0}^{N-1} {\cal{H}}_{2i}$ and ${\cal{H}}_{out} = \otimes_{i=0}^{N-1} {\cal{H}}_{2i+1}$ and suitable Choi operator of DQN $S^{N}$ such that $R^{N} \leq S^{N}$. Then, PQN ${\cal{R}}^{(N)}$ can be realized as a concatenation of $N$ isometries followed by an effect on an ancillary space:
	\begin{align}\label{RealizationPQN}
		{\cal{R}}^{(N)}(\varrho) = \Tr_{A_{N}} \left(V_{N} V_{N-1} \cdots V_{1} \varrho V_{1}^{\dagger} \cdots V_{N-1}^{\dagger} V_{N}^{\dagger} \mathbb{1}_{2N-1} \otimes E_{A_{N}}\right),
	\end{align}
	where $E_{A_{N}}$ denotes mentioned effect.
\end{theorem}
\begin{figure}[H]
	\begin{center}
		\begin{quantikz}
			\rstick{\hspace{-0.1cm}$0$\\~}
			&\gate[wires=2, nwires = {2}][2cm]{\mathcal{V}_{(1)}}
			&\qw \rstick{\hspace{-0.6cm}$1$\\~}
			& 
			&\rstick{\hspace{-0.1cm}$2$\\~}
			&\gate[wires=2][2cm]{\mathcal{V}_{(2)}}
			&\rstick{\hspace{-0.6cm}$3$\\~}\qw
			&\ \ldots\ 
			&\qw\rstick{\hspace{-0.9cm}$2N-2$\\~}
			&\gate[wires=2][2cm]{\mathcal{V}_{(N)}}
			&\rstick{\hspace{-0.6cm}$2N-1$\\~}\qw 
			&
			\\ 
			{}
			&
			&\qw \rstick{\hspace{0.2cm}$A_{1}$\\~}
			&\qw
			&\qw
			&
			&\rstick{\hspace{-0.6cm}$A_{2}$\\~}\qw
			&\ \ldots\ 
			&\qw\rstick{\hspace{-0.6cm}$A_{N-1}$\\~}
			&
			&\rstick{\hspace{-0.6cm}$A_{N}$\\~}\qw 
			&\meterD{E_{A_{N}}}
		\end{quantikz}
	\end{center}
	\caption[font=small]{Depiction of PQN with $N$ channels realized by isometric channels ${\cal{V}}_{(i)}$ and an effect $E_{A_{N}}$.}
	\label{realPQN}
\end{figure}
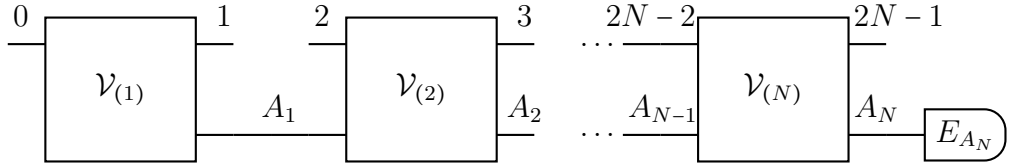
\begin{proof}
	Let us denote PQN with $N$ vertices as ${\cal{R}}^{(N)}$ and DQN as ${\cal{S}}^{(N)}$. Similarly, as in proof of realization theorem for DQN \ref{realitzationTheoremDQN}, we define ${\cal{H}}_{A_{k}} = Supp(S^{k\ast}_{A_{k}})$ and isometry:
	\begin{align}\label{isomPQN}
		V_{k} = \mathbb{1}_{2k-1} \otimes 	S^{k\frac{1}{2}\ast}_{A_{k}}S^{(k-1)(\frac{-1}{2})\ast}_{A_{(k-1)}} \dket{I}_{\substack{(2k-1)\\(2k-1)^{\prime}}} \otimes T_{(2k-2)\rightarrow(2k-2)^{\prime}}.
	\end{align}
	 For simpler notation, let us also define $W_{N} = V_{N}\cdots V_{1} = \mathbb{1}_{out} \otimes S^{N\frac{1}{2}\ast}_{A_{N}} \dket{I}_{\substack{out\\out^{\prime}}} \otimes T_{in \rightarrow in^{\prime}}$. In addition, we also need an effect on an ancillary space $E_{A_{N}} = S^{N\left(\frac{-1}{2}\right)\ast}_{A_{N}} R^{N\ast}_{A_{N}} S^{N\left(\frac{-1}{2}\right)\ast}_{A_{N}}$. And now, only the calculation of expression \ref{RealizationPQN} awaits us:
	\begin{align*}
		&\Tr_{A_{N}}\left[W_{N} \varrho W_{N}^{\dagger} E_{A_{N}}\right] \\&= \Tr_{A_{N}} \left[\left( \mathbb{1}_{out} \otimes S^{N\frac{1}{2}\ast}_{A_{N}} \dket{I}_{\substack{out\\out^{\prime}}} \otimes T_{\substack{in \rightarrow\\in^{\prime}}} \right) \varrho_{in} \left( T_{\substack{in^{\prime} \rightarrow\\in}} \otimes \prescript{}{\substack{out\\out^{\prime}}}{\dbra{I}} \mathbb{1}_{out} \otimes S^{N\frac{1}{2}\ast}_{A_{N}} \right) \left( S^{N\left(\frac{-1}{2}\right)\ast}_{A_{N}} R^{N\ast}_{A_{N}} S^{N\left(\frac{-1}{2}\right)\ast}_{A_{N}} \right)\right] \\&\overset{(i)}{=} \Tr_{A_{N}}\left[\left(\mathbb{1}_{out} \otimes R^{N\ast}_{A_{N}}\right)  \left(\dket{I}_{\substack{out\\out^{\prime}}}\dbra{I} \otimes \varrho_{in^{\prime}}\right)\right] \overset{(ii)}{=} \Tr_{\substack{in^{\prime}\\out^{\prime}}} \left[\left(\mathbb{1}_{out} \otimes R^{N}_{\substack{in^{\prime}\\out^{\prime}}}\right) \left(\dket{I}_{\substack{out\\out^{\prime}}}\dbra{I}^{T_{out^{\prime}}} \otimes \varrho_{in^{\prime}}^{T_{in^{\prime}}}\right)\right] \\&\overset{(iii)}{=} \sum_{i,j} \Tr_{in^{\prime}} \left[ \ket{i}_{out} \prescript{}{out^{\prime}}{\bra{i}} R^{N}_{\substack{in^{\prime}\\out^{\prime}}} \ket{j}_{out^{\prime}} \prescript{}{out}{\bra{j}} \otimes \varrho_{in^{\prime}}^{T_{in^{\prime}}} \right] = \Tr_{in^{\prime}} \left[ R^{N}_{\substack{in^{\prime}\\out}} \left(\mathbb{1}_{out} \otimes \varrho_{in^{\prime}}^{T_{in^{\prime}}}\right)\right] \overset{(\ref{inverseChoi})}{=} {\cal{R}}^{N}(\varrho),
	\end{align*}
		where in $(i)$ we have used cyclicity of $Tr_{A_{N}}$. In $(ii)$ we could change space, because ${\cal{H}}_{A_{N}} \subseteq \otimes_{i=0}^{2N-1}{\cal{H}}_{i^{\prime}} \cong {\cal{H}}_{in^{\prime}} \otimes {\cal{H}}_{out^{\prime}}$, we have also used the invariance of trace under transposition and Hermiticity of $R^{N}$. In $(iii)$ we have transposed the expression and traced out the space ${\cal{H}}_{out^{\prime}}$.
\end{proof}

\section{Generalized Quantum Instrument and Quantum Tester}

Generalized quantum instrument (GQI), as already the name indicates, fulfills the analogous function for quantum networks as does quantum instrument for channels. Quantum tester generalizes the notion of a POVM for quantum networks. Therefore, GQI has similar relation with quantum tester, where we are measuring quantum networks rather than quantum states, as quantum instrument has with POVM.
\begin{definition}
	Generalized Quantum Instrument
	\\
	Collection of probabilistic quantum networks $\{{\cal{R}}_{i}^{(N)}\}$, such that they sum up to deterministic quantum network $\sum_{i}{\cal{R}}_{i}^{(N)} = {\cal{R}}^{(N)}_{DQN}$, form generalized quantum instrument.
\end{definition}
Let us also state realization theorem for GQI:
\begin{theorem}\label{realizationGQI}
	Realization Theorem for Generalized Quantum Instrument
	\\
	Let $\{{\cal{R}}_{i}^{(N)}, {\cal{R}}_{i}^{(N)} \in {\cal{L}}({\cal{L}}({\cal{H}}_{in}),{\cal{L}}({\cal{H}}_{out}))\}, {\cal{R}}^{(N)} = \sum_{i}{\cal{R}}_{i}^{(N)}$ be a GQI. Then there exists a Hilbert space ${\cal{H}}_{A_{N}}$, deterministic quantum network ${\cal{S}}^{(N)} \in {\cal{L}}({\cal{L}}({\cal{H}}_{in}),{\cal{L}}({\cal{H}}_{out} \otimes {\cal{H}}_{A_{N}}))$ and a POVM $P_{i} \in {\cal{L}}({\cal{H}}_{A_{N}})$ such that for any $\varrho$:
	\begin{align*}
		{\cal{R}}_{i}^{(N)}(\varrho) = \Tr_{A_{N}}\left[S^{N}(\varrho) \left(\mathbb{1}_{out} \otimes P_{i}\right) \right].
	\end{align*}
\end{theorem}
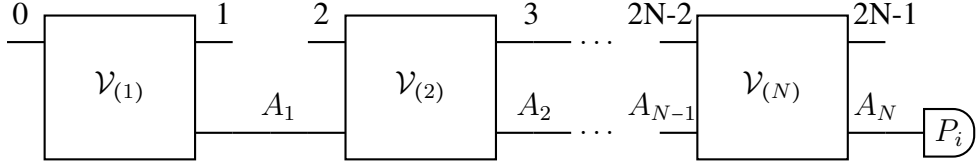
\begin{figure}[H]
	\begin{center}
		\begin{quantikz}
			\rstick{\hspace{-0.1cm}0\\~}
			&\gate[2, nwires = {2}][2cm]{\mathcal{V}_{(1)}}
			&\qw \rstick{\hspace{-0.6cm}1\\~}
			& 
			&\rstick{\hspace{-0.1cm}2\\~}
			&\gate[2][2cm]{\mathcal{V}_{(2)}}
			&\rstick{\hspace{-0.4cm}3\\~}\qw
			&\ \ldots\ \qw			
			&\rstick{\hspace{-0.6cm}2N-2\\~}
			&\gate[2][2cm]{\mathcal{V}_{(N)}}
			&\rstick{\hspace{-0.6cm}2N-1\\~}\qw 
			&
			\\ 
			{}
			&
			&\qw \rstick{\hspace{0.2cm}$A_{1}$\\~}
			&\qw
			&\qw
			&
			&\rstick{\hspace{-0.4cm}$A_{2}$\\~}\qw
			&\ \ldots\ \qw			
			&\rstick{\hspace{-0.6cm}$A_{N-1}$\\~}
			&
			&\rstick{\hspace{-0.6cm}$A_{N}$\\~}\qw 
			&\meterD{P_{i}}
		\end{quantikz}
	\end{center}
	\caption[font=small]{Depiction of a generalized quantum instrument with $N$ channels realized by isometric channels ${\cal{V}}_{(i)}$ and a POVM on an ancillary space $P_{i}$.}
	\label{realGQI}
\end{figure}
\begin{proof}
	Proof is analogous to the proof of realization theorem of PQN \ref{realizationPQN}. Let us define ${\cal{H}}_{A_{k}} = \supp(R^{k\ast}_{A_{k}})$, isometry $V_{k}$ exactly like in the proof of realization theorem for PQN in (\ref{isomPQN}), DQN ${\cal{S}}^{(N)} = {\cal{V}}_{(1)}\star \cdots \star{\cal{V}}_{(N)}$ and POVM $P_{i, A_{N}} = R^{N\left(\frac{-1}{2}\right)\ast}_{A_{N}} R^{N\ast}_{i, A_{N}} R^{N\left(\frac{-1}{2}\right)\ast}_{A_{N}}$, where $R^{N}_{i, A_{N}}$ denotes PQN corresponding to the $i$-th element of the POVM $P_{i,A_{N}}$ acting on Hilbert space ${\cal{H}}_{A_{N}}$. Then, we just have to verify, similarly to \ref{realizationPQN}, the following:
	\begin{align*}
		&\Tr_{A_{N}} \left[S^{N}(\varrho) P_{i}\right] = \Tr_{A_{N}} \left[V_{N} \cdots V_{1} \varrho V_{1}^{\dagger} \cdots V_{N}^{\dagger} P_{i, A_{N}}\right] \\&= \Tr_{A_{N}} \left[\left( \mathbb{1}_{out} \otimes R^{N\frac{1}{2}\ast}_{A_{N}} \right)\left( \varrho_{in^{\prime}} \otimes \dket{I}_{\substack{out\\out^{\prime}}}\dbra{I} \right)\left( \mathbb{1}_{out} \otimes R^{N\frac{1}{2}\ast}_{A_{N}} \right)\left( \mathbb{1}_{out} \otimes R^{N\left(\frac{-1}{2}\right)\ast}_{A_{N}} R^{N\ast}_{i, A_{N}} R^{N\left(\frac{-1}{2}\right)\ast}_{A_{N}} \right)\right] \\&\overset{(i)}{=} \Tr_{{\substack{in^{\prime}\\out^{\prime}}}} \left[\left( \varrho_{in^{\prime}} \otimes \dket{I}_{\substack{out\\out^{\prime}}}\dbra{I} \right) \left( \mathbb{1}_{out} \otimes R_{i, \substack{in^{\prime}\\out^{\prime}}}^{N\ast}\right)\right] \overset{(ii)}{=} \Tr_{in^{\prime}} \left[\left(\mathbb{1}_{out} \otimes \varrho_{in^{\prime}}^{T}\right)R_{i,\substack{in^{\prime}\\out}}^{N}\right] \overset{(\ref{inverseChoi})}{=} R_{i}^{(N)}(\varrho).
	\end{align*}
	In $(i)$ we have used that ${\cal{H}}_{A_{N}} \subseteq \otimes_{i=0}^{2N-1} H_{i^{\prime}} \cong {\cal{H}}_{in^{\prime}} \otimes           {\cal{H}}_{out^{\prime}}$ and cyclic property of trace. In $(ii)$ we have used the invariance of trace under transposition and Hermiticity of $R^{N}$ and we have also evaluated $\Tr_{out^{\prime}}$.
\end{proof}

As was already mentioned, quantum tester generalizes measurements for QNs. It is just a GQI with one-dimensional output. Quantum $N$-tester takes as an input quantum network with $N$ vertices and outputs probability.
\begin{definition}
	Quantum $N$-tester
	\\
	A quantum $N$-tester ${\cal{T}}^{(N)}$ is a generalized quantum instrument $\left\{{\cal{R}}_{i}^{(N)}\right\}$ such that $\dim({\cal{H}}_{0}) = \dim({\cal{H}}_{out}) = 1$, where ${\cal{H}}_{0}$ is the very first input and ${\cal{H}}_{out}$ is the output of tester.
\end{definition}
Based on normalization lemma of DQN \ref{normalizationConditionDQN}, quantum tester must fulfill the following conditions:
\begin{align*}
	&T^{N} = \mathbb{1}_{2N-2} \otimes T^{N-1},\\
	&\Tr_{2k-1}(T^{k}) = \mathbb{1}_{2k-2} \otimes T^{k-1} \quad\text{for } k = 2, \cdots (N-1), \\
	&\Tr_{1}(T^{1}) \overset{\text{one-dimensional first input}}{=} 1.
\end{align*}
We can see that $T^{1}$ is a quantum state because it is positive semi-definite operator with trace one. And let us also, unsurprisingly, state the realization theorem for quantum tester:
\begin{theorem}\label{realizationTester}
	Realization Theorem for Quantum $N$-tester
	\\
	Quantum $N$-tester ${\cal{T}}^{(N)}$ can be realized by a deterministic quantum network ${\cal{R}}^{(N)}$ with first input $\dim({\cal{H}}_{0}) = 1$ followed by a POVM on output Hilbert space ${\cal{H}}_{out}$.
\end{theorem}
\begin{figure}[H]
	\begin{center}
		\begin{quantikz}
			\rstick{\hspace{-0.1cm}0\\~}
			&\gate[2, nwires = {2}][2cm]{\mathcal{V}_{(1)}}
			&\qw \rstick{\hspace{-0.6cm}1\\~}
			& 
			&\rstick{\hspace{-0.1cm}2\\~}
			&\gate[2][2cm]{\mathcal{V}_{(2)}}
			&\rstick{\hspace{-0.4cm}3\\~}\qw
			&\ \ldots\ \qw			
			&\rstick{\hspace{-0.6cm}2N-2\\~}
			&\gate[2][2cm]{\mathcal{V}_{(N)}}
			&\rstick{\hspace{-0.6cm}$A_{N}$\\~}\qw 
			&\meterD{P_{i}}
			\\ 
			{}
			&
			&\qw \rstick{\hspace{0.2cm}$A_{1}$\\~}
			&\qw
			&\qw
			&
			&\rstick{\hspace{-0.4cm}$A_{2}$\\~}\qw
			&\ \ldots\ \qw			
			&\rstick{\hspace{-0.6cm}$A_{N-1}$\\~}
			&
			&
			&
		\end{quantikz}
	\end{center}
	\caption[font=small]{Depiction of a generalized quantum $N$-tester realized through isometric channels ${\cal{V}}_{(i)}$ with POVM $P_{i}$ on an ancillary space ${\cal{H}}_{A_{N}}$. Compared with figure \ref{realGQI}, ancillary space ${\cal{H}}_{A_{N}}$ is in a place where originally there was a Hilbert space ${\cal{H}}_{2N-1}$.}
	\label{realTester}
\end{figure}
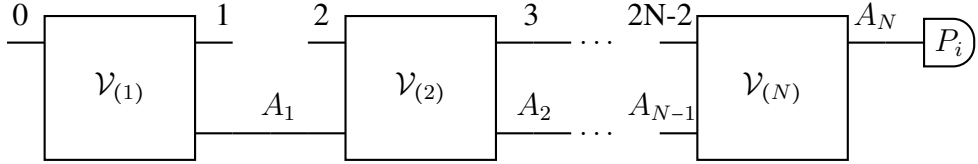
\begin{proof}
	Because quantum tester is only a generalized quantum instrument, the proof is the same as is proof of realization theorem of generalized quantum instrument \ref{realizationGQI}. We shall only relabel ancillary space as ${\cal{H}}_{A_{N}} = {\cal{H}}_{2N-1}$ as can be seen from the schemes of the respective realizations in figures \ref{realGQI} and \ref{realTester}. Since the first input space ${\cal{H}}_{0}$ is one-dimensional, the first isometry is only a preparation of a state $\dket{\Psi} = \left(\mathbb{1}_{1} \otimes R^{1\frac{1}{2}\ast}_{A_{1}}\right) \dket{I}_{11^{\prime}}$, which comes from state-channel duality from equation (\ref{stateChannelDuality}).
\end{proof}

\section{Composition of Quantum Networks}

In this section, we shall investigate the outcome of composition of quantum networks. They can be viewed as directed acyclic graphs, which means that also their composition has to remain DAG. If there are two directed edges connecting the same vertices, we denote them by the same integer as can be seen from example in figure \ref{compositionDAG}.
\begin{figure}
	\vspace*{-2cm}
\begin{center}
	\begin{tikzpicture}
		\Vertex[x=-5.5,opacity=0,label=$2$,fontscale=1.5]{2}
		\Vertex[x=-4,opacity=0,label=$5$,fontscale=1.5]{5}
		
		\Edge[Direct,label=2,position=above,fontscale=1.5](2)(5)
		
		\Vertex[x=-6.5,y=-1.5,style={color=white}]{E}
		\Vertex[x=-2,y=-1.5,style={color=white}]{F}
		\Vertex[x=-2,style={color=white}]{G}
		
		\Edge[Direct,label=1,position={above left=1mm},fontscale=1.5](E)(2)
		\Edge[Direct,label=8,position={below left=1mm},fontscale=1.5](5)(F)
		\Edge[Direct,label=9,position=above,fontscale=1.5](5)(G)
		
		\Text[x=-2,fontsize=\LARGE]{$\star$}
		
		\Vertex[x=0,opacity=0,label=$1$,fontscale=1.5]{1}
		\Vertex[x=1.5,opacity=0,label=$3$,fontscale=1.5]{3}
		\Vertex[x=1.5,y=1.5,opacity=0,label=$4$,fontscale=1.5]{4}
		\Vertex[x=3,opacity=0,label=$6$,fontscale=1.5]{6}
		
		\Edge[Direct,label=3,position=above,fontscale=1.5](1)(3)
		\Edge[Direct,label=5,position=left,fontscale=1.5](3)(4)
		\Edge[Direct,label=7,position=above,fontscale=1.5](3)(6)
		
		\Vertex[x=-1.5,style={color=white}]{A}
		\Vertex[x=0.5,y=1.5,style={color=white}]{B}
		\Vertex[x=1,y=-1.5,style={color=white}]{C}
		\Vertex[x=3,y=1.5,style={color=white}]{D}
		\Vertex[x=3,y=1.5,style={color=white}]{H}
		\Vertex[x=4.5,style={color=white}]{I}
		
		\Edge[Direct,label=0,position=above,fontscale=1.5](A)(1)
		\Edge[Direct,label=1,position=left,fontscale=1.5](1)(B)
		\Edge[Direct,label=4,position=left,fontscale=1.5](C)(3)
		\Edge[Direct,label=6,position=above,fontscale=1.5](4)(D)
		\Edge[Direct,label=8,position=right,fontscale=1.5](H)(6)
		\Edge[Direct,label=10,position=above,fontscale=1.5](6)(I)
		
		\Text[x=5,fontsize=\LARGE]{$=$}
		
		\Vertex[x=-3,y=-5,opacity=0,label=$1$,fontscale=1.5]{01}
		\Vertex[x=-2,y=-3.5,opacity=0,label=$2$,fontscale=1.5]{02}
		\Vertex[x=-2,y=-6.5,opacity=0,label=$3$,fontscale=1.5]{03}
		\Vertex[x=-1,y=-5,opacity=0,label=$4$,fontscale=1.5]{04}
		\Vertex[x=1,y=-4,opacity=0,label=$5$,fontscale=1.5]{05}
		\Vertex[x=1,y=-6.5,opacity=0,label=$6$,fontscale=1.5]{06}
		
		\Edge[Direct,label=1,position=left,fontscale=1.5](01)(02)
		\Edge[Direct,label=3,position=right,fontscale=1.5](01)(03)
		\Edge[Direct,label=5,position={above left=2mm},fontscale=1.5](03)(04)
		\Edge[Direct,label=2,position=above,fontscale=1.5](02)(05)
		\Edge[Direct,label=7,position=above,fontscale=1.5](03)(06)
		\Edge[Direct,label=8,position=right,fontscale=1.5](05)(06)
		
		\Vertex[x=-4.5,y=-5,style={color=white}]{H}
		\Vertex[x=-3,y=-7.5,style={color=white}]{I}
		\Vertex[x=0.5,y=-5,style={color=white}]{J}
		\Vertex[x=2.5,y=-3.5,style={color=white}]{K}
		\Vertex[x=2.5,y=-6.5,style={color=white}]{L}
		
		\Edge[Direct,label=0,position=above,fontscale=1.5](H)(01)
		\Edge[Direct,label=4,position={above left=1mm},fontscale=1.5](I)(03)
		\Edge[Direct,label=6,position=above,fontscale=1.5](04)(J)
		\Edge[Direct,label=9,position=above,fontscale=1.5](05)(K)
		\Edge[Direct,label=10,position=above,fontscale=1.5](06)(L)
	\end{tikzpicture}
\end{center}
\caption[font=small]{Composition of two DAGs has to, again, end up being a DAG in order to obtain relevant quantum network. On both of the first two DAGs, there are edges denoted with $1$, because in the final DAG, they denote the same edge.}
\label{compositionDAG}
\end{figure}
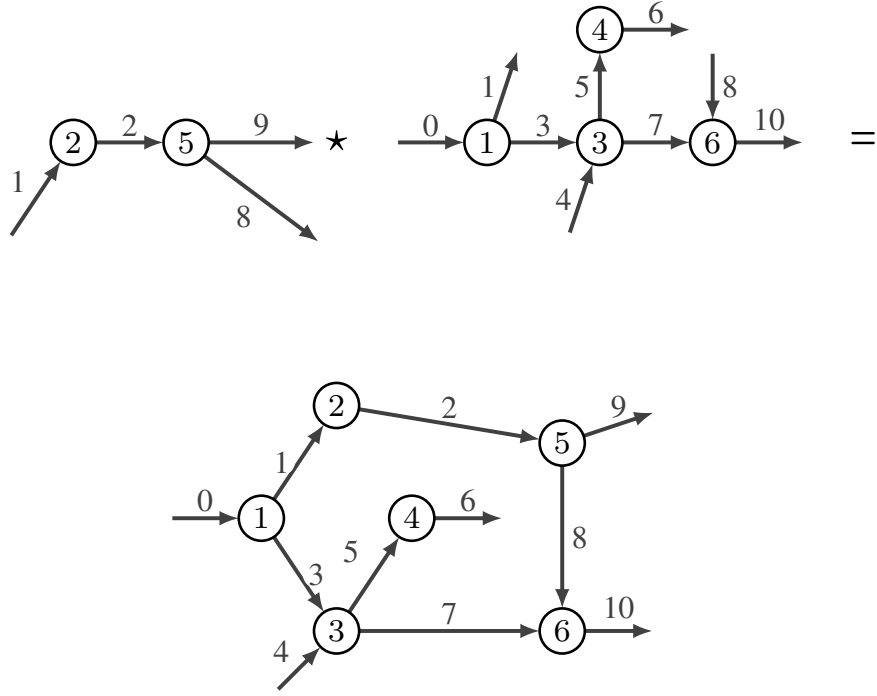
There exists a partial ordering in DAGs that can be changed to total order. But there is no such relation between the vertices of two different quantum networks ${\cal{R}}^{(M)}$ and ${\cal{S}}^{(N)}$. Fortunately, it is possible to define a total ordering of vertices from the union of sets of vertices ${\cal{R}}^{(M)} \cup {\cal{S}}^{(N)}$. Therefore, we can obtain new quantum network from composition of two previous ones as is depicted in the figure \ref{compositionQN}.
\begin{figure}[H]
	\begin{center}
		\begin{quantikz}
			\rstick{\hspace{-0.1cm}1\\~}
			&\gate[2, nwires = {2}][1cm]{\mathcal{C}_{(2)}}
			&\rstick{\hspace{-0.2cm}2\\~}\qw
			&\gate[2, nwires = {2}][1cm]{\mathcal{C}_{(5)}}
			&\rstick{\hspace{-1.5cm}9\\~}\qw
			&\rstick{\hspace{0cm}0\\~}
			&\gate[2, nwires = {2}][1cm]{\mathcal{C}_{(1)}} \rstick{\hspace{-0cm}1\\~}  
			&\rstick{\hspace{1cm}4\\~}\qw
			&
			&\gate[2][1cm]{\mathcal{C}_{(3)}}
			&\rstick{\hspace{-0.3cm}5\\~}\qw 
			&\gate[2][1cm]{\mathcal{C}_{(4)}} 
			&\rstick{\hspace{-1cm}6\\~}\qw
			&\gate[2, nwires = {1}][1cm]{\mathcal{C}_{(6)}} &\rstick{\hspace{-0.7cm}10\\~}\qw 
			\\
			&
			&
			&
			&\rstick{\hspace{-1.5cm}8\\~}\qw
			&\rstick{\hspace{-1cm}\huge$\star$\\~}\hspace{-1cm}
			&\rstick{\hspace{1cm}3\\~}
			&\qw
			&\rstick{\hspace{1.7cm}7\\~}\qw
			&\qw 
			&\qw 
			&\rstick{\hspace{0.2cm}7\\~}
			&\qw 
			&\qw
			&\rstick{\hspace{-0.2cm}\large$=$\\~}
			\\
			&
			&\rstick{\hspace{-0.1cm}0\\~}
			&\gate[2, nwires = {2}][1cm]{\mathcal{C}_{(1)}}
			&\gate[2, nwires = {1}][1cm]{\mathcal{C}_{(2)}}
			&\rstick{\hspace{-0.1cm}4\\~} 
			&\gate[2][1cm]{\mathcal{C}_{(3)}} 
			&\gate[2, nwires = {1}][1cm]{\mathcal{C}_{(4)}}\qw
			&\rstick{\hspace{-0.8cm}6\\~}\qw
			&\gate[2, nwires={1}][1cm]{\mathcal{C}_{(5)}} 
			&\rstick{\hspace{-0.8cm}9\\~}\qw
			&\gate[2, nwires = {1}][1cm]{\mathcal{C}_{(6)}}
			&\rstick{\hspace{-0.9cm}10\\~}\qw
			&
			&
			\\
			&
			&
			&
			& 
			&\qw
			&
			&
			&\qw
			& 
			&\qw 
			& 
			&
			&
			&
		\end{quantikz}
	\end{center}
	\caption[font=small]{Composition of two QNs corresponding to previous DAGs in the \ref{compositionDAG}}
	\label{compositionQN}
\end{figure}
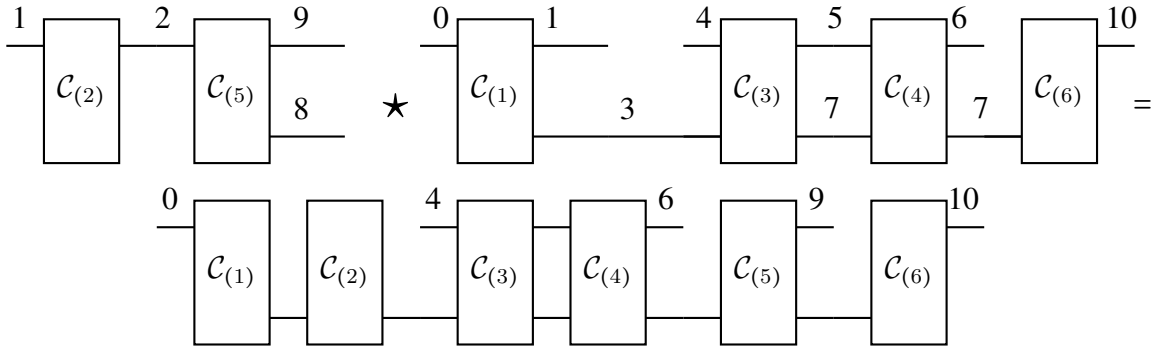
From the same figure \ref{compositionQN}, we can gather that composition of quantum networks can be made from concatenation of Choi operators forming the individual networks. Let us suppose that $R^{M} \in {\cal{L}}(\otimes_{i}{\cal{H}}_{i})$ is a Choi operator of quantum network ${\cal{R}}^{(M)}$ and $S^{N} \in {\cal{L}}(\otimes_{j}{\cal{H}}_{j})$ is a Choi operator of quantum network ${\cal{S}}^{(N)}$, then their composition is ${\cal{C}}({\cal{S}}^{(N)} \circ {\cal{R}}^{(M)}) = S^{N} \star R^{M}$ (which is a consequence of associativity of the link product \ref{propLink}).

Important case is a composition of quantum network ${\cal{R}}^{(N)}$ with quantum $(N+1)$-tester $\{{\cal{T}}_{i}^{(N+1)}\}$ as is depicted in figure \ref{compQNtester}. Then, probability of an outcome $i$ is given as follows $p(i|{\cal{R}}^{(N)}) = R^{N} \star T_{i}^{N+1} = \Tr\left[R^{N} T_{i}^{(N+1)T}\right]$, which is a generalization of a measurement to quantum networks. Tester $\{T_{i}^{(N+1)T}\}$ fulfills the role of a POVM and quantum network ${\cal{R}}^{(N)}$ the role of a state.
\begin{figure}[H]
	\begin{center}
	\includegraphics[scale=0.5]{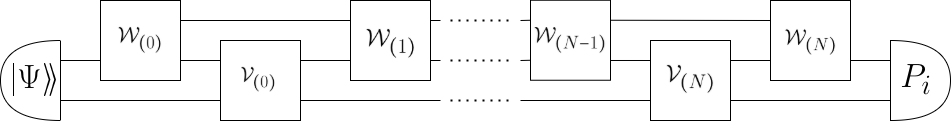}
	\end{center}
	\caption[font=small]{Composition of quantum network ${\cal{R}}^{(N)}$ that is being realized by isometries $W_{i}$ with quantum $(N+1)$-tester $\{T_{i}^{(N+1)}\}$ that is being realized by isometries $V_{i}$.}
	\label{compQNtester}
\end{figure}

\section{Relation with Quantum Processor}

Quantum programmable processors, which were the main interest of our investigation in chapter \ref{QP}, can be described using quantum networks. Certain DQNs can be interpreted as deterministic quantum processors, while GQIs can be viewed as probabilistic ones, as is depicted in figure \ref{QNasQP} (note that here, we are denoting program state with $\varrho$ and data state with $\xi$ as opposed to the notation used in chapter \ref{QP}).
\begin{figure}[H]
	\begin{center}
		\includegraphics[scale = 0.4]{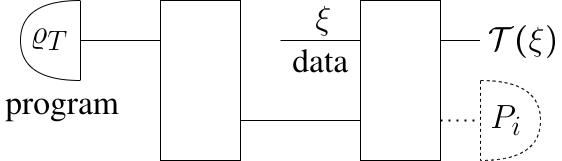}
	\end{center}
	\caption[font=small]{DQNs can be regarded as deterministic quantum processors, where quantum state $\varrho_{T}$ fulfills the function of program of quantum processor and the desired transformation ${\cal{T}}$ is applied on an input data state $\xi$. In case when we also add a POVM $\{P_{i}\}$ at the end of program register we obtain probabilistic quantum processor that can be viewed as a GQI.}
	\label{QNasQP}
\end{figure}
Quantum channels can also take the role of program in quantum processor as can be seen in figure \ref{QNasQPchannel}. Through Choi-Jamio\l{}kowski isomorphism one is able to encode information of the desired transformation in quantum state. 
\begin{figure}[H]
	\begin{center}
		\includegraphics[scale = 0.4]{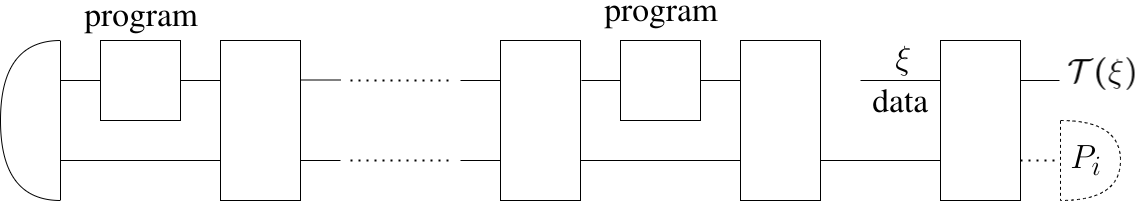}
	\end{center}
	\caption[font=small]{Quantum channels can function as programs of quantum processors if they are encoded in quantum state through Choi-Jamio\l{}kowski isomorphism. We desire to implement transformation ${\cal{T}}$ on input data state $\xi$. Again, by adding measurement, we obtain probabilistic quantum processor.}
	\label{QNasQPchannel}
\end{figure}

\section{Probabilistic Storage and Retrieving of Unitary Transformation}

In this section we shall take a closer look on the task of probabilistic storage and retrieval (PSAR) of unitary transformation (sometimes also called quantum learning) \cite{OptimalQuantumLearningOfAUnitaryTransformation, StoringQuantumDynamicsInQuantumStatesAStochasticProgrammableGate}. Let us imagine having access to unitary channel $N$ times but to lose this access in the future. The task itself consists of two parts. In the first one, one shall store the channel in a state and in the second part, one shall try to retrieve the stored channel and apply the retrieved channel on an unknown quantum state. Naturally, we shall make use of Choi-Jamio\l{}kowski isomorphism for storing the dynamics of the system in a quantum state. The questions now are how to optimally store given channel in the present and how to optimally access it later?   The device for storing and retrieving of transformation can be interpreted as a probabilistic quantum processor (due to no-programming theorem \ref{no-programming} \cite{ProgrammableQuantumGateArrays}, such a processor cannot be deterministic). The scheme of such a device is depicted in figure \ref{PSAR}. It is a generalized quantum instrument ${\cal{R}} = \{{\cal{R}}_{s}, {\cal{R}}_{f}\}$ with the two possible outcomes - either successful measurement which results in applying the desired transformation on the unknown state $\ket{\xi}$, or failure with implementation of some different transformation. One additional constraint is applied on PSAR device - the device is covariant, which means that probability of success $p_{success}$ is invariant under unitary transformation. The consequence of covariant property is that the probability of successfully retrieving the desired channel is equal for all considered channels.
\begin{figure}[H]
	\begin{center}
	\includegraphics[scale = 0.3]{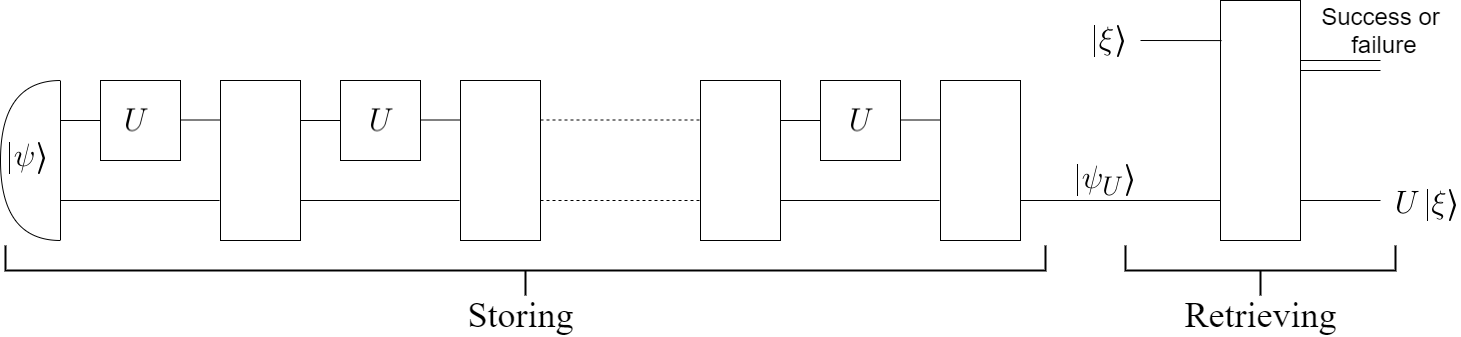}
	\end{center}
	\caption[font=small]{Schematic depiction of probabilistic storage and retrieving of unitary transformation. During the storing phase, unitary transformation $U$ is applied on input state $\ket{\psi}$ creating state $\ket{\psi_{U}}$ which is then used as input (program) state into the retrieving phase, where on the output we wish to find a state $U\ket{\xi}$.}
	\label{PSAR}
\end{figure}

Sedlák, Bisio and Ziman \cite{OptimalProbabilisticStorageAndRetrievalOfUnitaryChannels} showed the optimal probability of retrieving unitary channel being $p_{suc} = \frac{N}{N-1+d^{2}}$, where $N$ is the number of times one has access to a given channel and $d$ is dimension of Hilbert space to which the unitary channel belongs to.

In \cite{ProbabilisticStorageAndRetrievalOfQubitPhaseGates}, Sedlák and Ziman investigated a more restrained problem, where the unitary transformation was from $U(1)$ group and could be expressed in the computational basis in the following form:
\begin{align}\label{Ufi}
	U_{\phi} = \ket{0}\bra{0} + e^{i\phi}\ket{1}\bra{1}.
\end{align} 
The success probability of retrieving such a channel was $p_{success} = \frac{N}{N+1}$, where $N$ is the number of times one has access to a given transformation during storing phase.

\section{PSAR of Unitary Transformation with Noise}

In our work, we shall investigate robustness or resilience of PSAR device optimized for storing and retrieval of unitary transformation from $U(1)$ group, described in previous section, against noises \cite{RobustnessOfOptimalProbabilisticStorageAndRetrievalOfUnitaryChannelsToNoise}. Firstly, we shall examine the effect of depolarization on the performance of the device. Afterwards, we shall investigate the effects of phase damping. Depolarization channel is often used to model noise, errors and to investigate the possible noise mitigation strategies while the phase damping channel models decoherence \cite{QuantumComputationAndQuantumInformation, GenerativeMachineLearningWithTensorNetworksBenchmarksOnNear-TermQuantumComputers,  SimpleMitigationOfDepolarizingErrorsInQuantumSimulations, QuantumDecoherenceOfTwoQubits, QuantumFisherInformationOfTheGreenberger-Horne-ZeilingerStateInDecoherenceChannels, MinimalQuditCodeForAQubitInThePhase-dampingChannel}.

\subsection{Depolarization}

We shall implement channel that is a convex combination of a unitary channel and contraction in the total mixture on the device optimized for implementing unitary channel in order to investigate its resistance (or robustness) against this kind of noise. This means that the input state for storing phase and retrieving instrument remains the same as for the device optimized for implementing unitary channel in \cite{ProbabilisticStorageAndRetrievalOfQubitPhaseGates}.

Channel that we are going to implement has the following form:
\begin{align*}
	{\cal{E}}_{\phi} = q {\cal{U}}_{\phi} + (1-q) {\cal{C}}_{\mathbb{1}/2},
\end{align*}
where $0 \leq q \leq 1$, unitary channel is denoted by ${\cal{U}}_{\phi}(\varrho) = U_{\phi} \varrho U_{\phi}^{\dagger}$ with phase gate $U_{\phi}$ is from equation (\ref{Ufi}) and ${\cal{C}}_{\mathbb{1}/2}$ denotes contraction in the total mixture:
\begin{align}\label{contraction}
	{\cal{C}}_{\mathbb{1}/2}(\varrho) = \frac{\Tr(\varrho)}{\Tr(\frac{\mathbb{1}}{2})} \frac{\mathbb{1}}{2} = \frac{\mathbb{1}}{2}.
\end{align}
This channel can be interpreted as a depolarizing channel from equation (\ref{Dchannel}) for $p = 1$. Channel ${\cal{E}}_{\phi}$ transforms input qubit state by a unitary operation $U_{\phi}$ with probability $q$ and with probability $1-q$ it replaces the input state by a total mixture. Let us now calculate the Choi operator corresponding to the unitary channel:
\begin{align*}
	U_{\phi} &= ({\cal{U}}_{\phi} \otimes {\cal{I}}) \dket{I}\dbra{I} = ({\cal{U}}_{\phi} \otimes {\cal{I}}) \sum_{m,n}\ket{mm}\bra{nn} = \sum_{m,n} {\cal{U}}_{\phi}(\ket{m}\bra{n}) \otimes \ket{m}\bra{n} \\&= \sum_{m,n} U_{\phi}(\ket{m}\bra{n})U_{\phi}^{\dagger} \otimes \ket{m}\bra{n} = \sum_{m,n} (U_{\phi} \otimes \mathbb{1}) \ket{mm}\bra{nn} (U_{\phi}^{\dagger} \otimes \mathbb{1}) \overset{(\ref{stateChannelDuality})}{=} \dket{U_{\phi}}\dbra{U_{\phi}^{\dagger}}.
\end{align*}
And also let us calculate the Choi operator for contraction:
\begin{align*}
	C_{\mathbb{1}/2} &= ({\cal{C}}_{\mathbb{1}/2} \otimes {\cal{I}}) \dket{I}\dbra{I} = \sum_{m,n} {\cal{C}}_{\mathbb{1}/2}(\ket{m}\bra{n}) \otimes \ket{m}\bra{n} =  \sum_{m,n} \frac{\Tr(\ket{m}\bra{n})}{\Tr(\frac{\mathbb{1}}{2})} \frac{\mathbb{1}}{2} \otimes \ket{m}\bra{n} \\&= \frac{\mathbb{1}}{2} \otimes \sum_{n} \ket{n}\bra{n} = \frac{1}{2} \mathbb{1} \otimes \mathbb{1}.
\end{align*}
Therefore, the Choi operator of the whole channel ${\cal{E}}_{\phi}$ has the following form:
\begin{align*}
	E_{\phi} = q \dket{U_{\phi}}\dbra{U_{\phi}} + \frac{1-q}{2} \mathbb{1}\otimes\mathbb{1}.
\end{align*}
Input state in the device is \cite{OptimalQuantumLearningOfAUnitaryTransformation}:
\begin{align} \label{inpState}
	\ket{\psi} = \bigoplus_{j=0}^{N} \sqrt{p_{j}} \dket{I_{j}} \overset{(i)}{=} \sum_{j}^{N} \sqrt{p_{j}} \ket{j},
\end{align}
where in $(i)$ we are using a sort of dictionary:
\begin{align*}
	&\ket{0} \rightarrow \ket{0\cdots0} &&\ket{N+1} \rightarrow \ket{0\cdots010}\\
	&\ket{1} \rightarrow \ket{0\cdots1} &&\ket{N+2} \rightarrow \ket{0\cdots100}\\
	&\ket{2} \rightarrow \ket{0\cdots11} &&\ket{N+3} \rightarrow \ket{0\cdots101}\\
	&\ket{3} \rightarrow \ket{0\cdots111} &&\ket{N+4} \rightarrow \ket{0\cdots110}\\
	&\qquad\vdots &&\qquad\vdots\\
	&\ket{N} \rightarrow \ket{1\cdots1} &&\ket{2^{N}-1} \rightarrow \ket{1\cdots110}.
	\label{generalDictionary} \numberthis
\end{align*}
In the left column, $j$ from states $\ket{j}$ has the meaning of number of $1$'s in the state written from the right, while no other $1$'s can be in those particular states. The right column treats the remaining states in the increasing order. This notation corresponds to the decomposition of $U_{\phi}^{\otimes N}$ to irreducible representation:
\begin{align} \label{IrrepDecomposition}
	U_{\phi}^{\otimes N} = \sum_{j \in \hspace*{0.7mm} irreps} e^{ij\phi} \mathbb{1}_{j} \otimes \mathbb{1}_{m_{j}},
\end{align}
where $\mathbb{1}_{m_{j}}$ denotes identity on multiplicity spaces which correspond to the states from the right column in dictionary (\ref{generalDictionary}).

\subsubsection{Two-to-one}
Let us show model calculation for "two-to-one" case where we have channel ${\cal{E}}_{\phi}$ twice at our disposal and at the output, we require to have a single copy of resulting channel (one might be interested also in a more general task to produce more than one copy of the desired transformation). This concrete device is depicted in the figure \ref{PSARdepimage}. Choi operator for said channel has the following form:
\begin{align*}
	E_{\phi}^{\otimes2} = &q^{2} \dket{U_{\phi}}_{12}\dbra{U_{\phi}} \otimes \dket{U_{\phi}}_{34}\dbra{U_{\phi}} + q \frac{1-q}{2} \dket{U_{\phi}}_{12}\dbra{U_{\phi}} \otimes \mathbb{1}_{34} \\&+ q \frac{1-q}{2} \mathbb{1}_{12} \otimes \dket{U_{\phi}}_{34}\dbra{U_{\phi}} + \frac{(1-q)^{2}}{4} \mathbb{1}_{12} \otimes \mathbb{1}_{34}.
	\label{Efi}\numberthis
\end{align*}
\begin{figure}[H]
	\begin{center}
		\includegraphics[scale = 0.45]{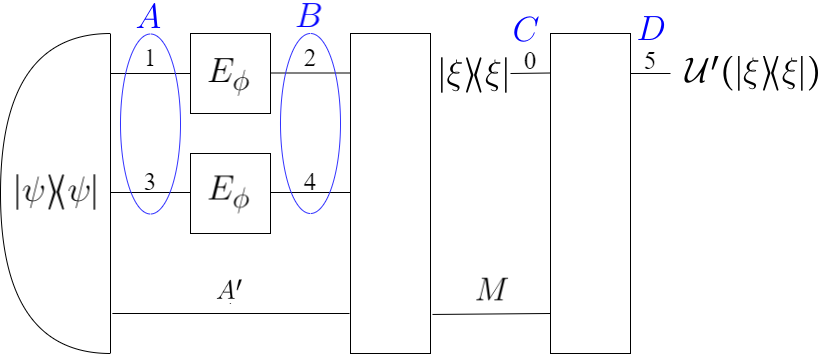}
	\end{center}
	\caption{Schematic image of PSAR implementing channel ${\cal{E}}_{\phi}$ twice with input state $\ket{\psi}$ explicitly written in equation (\ref{state}). At the output of register ${\cal{H}}_{5}$, we expect to retrieve unitary channel, possibly with some noise, in case of successful implementation. In our notation, Hilbert spaces ${\cal{H}}_{13} = {\cal{H}}_{A}$, ${\cal{H}}_{24} = {\cal{H}}_{B}$, ${\cal{H}}_{0} = {\cal{H}}_{C}$ and ${\cal{H}}_{5} = {\cal{H}}_{D}$ are identical.}
	\label{PSARdepimage}
\end{figure}
Before proceeding to the initial state, let us explicitly express decomposition of unitary operator applied on two qubits into irreducible representation:
\begin{align*}
	U^{\otimes2}_{\phi} =
	\begin{pmatrix}
		e^{i0\phi} & 0 & 0 & 0\\
		0 & e^{i1\phi} & 0 & 0\\
		0 & 0 & e^{i1\phi} & 0\\
		0 & 0 & 0 & e^{i2\phi}
	\end{pmatrix}
	= \bigoplus_{j=0}^{2} e^{ij\phi} \otimes \mathbb{1}_{m_{j}}.
	\label{Bono} \numberthis
\end{align*}
Matrices $\mathbb{1}_{m_{j}}$ denote identity on multiplicity spaces ${\cal{H}}_{m_{j}}$. We can see that in this case, only such identity with higher dimension than one is $\mathbb{1}_{m_{1}}$. PSAR device is optimized only for implementing unitary channel. The optimal storage state used in such a device is from the equation (\ref{inpState}):
\begin{align*}
	\ket{\Psi} &= \bigoplus_{j=0}^{2} \sqrt{p_{j}}\dket{I_{j}} = \sqrt{p_{0}}\ket{00} + \sqrt{p_{1}}\ket{11} + \sqrt{p_{2}}\ket{22} \\&= \sqrt{p_{0}}\ket{00,00} + \sqrt{p_{1}}\ket{01,01} + \sqrt{p_{2}}\ket{11,11}.
	\label{state}\numberthis
\end{align*}
In the last equation, we have used the fact that we discard the multiplicity spaces because in case of retrieving unitary channel they carry no additional information, and we are following the previous work where the PSAR was optimized for implementation of unitary channels \cite{ProbabilisticStorageAndRetrievalOfQubitPhaseGates}. From equation (\ref{Bono}) we can see that we either dismiss state $\ket{01}$ or state $\ket{10}$ - we have chosen to disregard the latter one. From previous work \cite{ProbabilisticStorageAndRetrievalOfQubitPhaseGates}, we already have the values for probabilities: $p_{j} = \frac{1}{3}$ for all $j$. The dictionary from the equation (\ref{generalDictionary}) in this case reduces to:
\begin{align*}
	&\ket{0} \rightarrow \ket{00} &&\ket{3} \rightarrow \ket{10}\\
	&\ket{1} \rightarrow \ket{01}\\
	&\ket{2} \rightarrow \ket{11},
	\label{dictionary} \numberthis
\end{align*}
where we have denoted the state $\ket{3}$ as the multiplicity state, that we discard. We have already used this concrete dictionary in the equation (\ref{state}) The significance of this notation will reveal itself later.

\subsubsection{Storing}

Firstly, let us investigate the storing phase of the device with the noisy channel ${\cal{E}}_{\phi}^{\otimes2}$ from equation (\ref{Efi}), applied on the state $\dyad{\psi}{\psi}$ from equation (\ref{state}):
\begin{align*}
	&\varrho_{E} \equiv {\cal{E}}_{\phi,AB}^{\otimes2} \star \ket{\Psi}_{AA^{\prime}}\bra{\Psi} \overset{(\ref{linkProduct})}{=} \Tr_{A} [(E_{\phi,AB}^{\otimes2} \otimes \mathbb{1}_{A^{\prime}}) (\ket{\Psi}_{AA^{\prime}}\bra{\Psi}^{T_{A}} \otimes \mathbb{1}_{B})] \underset{{\cal{H}}_{B} = {\cal{H}}_{24}}{\overset{{\cal{H}}_{A} = {\cal{H}}_{13}}{=}} \\&\Tr_{13}\Biggl\{\biggl[ \bigg( q^{2}\dket{U_{\phi}}_{1234}\dbra{U_{\phi}}^{\otimes2} + q \frac{1-q}{2} \dket{U_{\phi}}_{12}\dbra{U_{\phi}} \otimes \mathbb{1}_{34} + q \frac{1-q}{2} \mathbb{1}_{12} \otimes \dket{U_{\phi}}_{34}\dbra{U_{\phi}} + \frac{(1-q)^{2}}{4} \mathbb{1}_{1234}\bigg) \\& \otimes \mathbb{1}_{A^{\prime}} \biggl] \left[ \bigoplus_{j,k=0}^{2} \sqrt{p_{j}p_{k}} \dket{I_{j}}_{13A^{\prime}}\dbra{I_{k}}^{T_{13}} \otimes \mathbb{1}_{24} \right] \Biggl\},
	\label{storing} \numberthis
\end{align*}
where $\dket{U_{\phi}}_{1234}\dbra{U_{\phi}}^{\otimes2} = \dket{U_{\phi}}_{12}\dbra{U_{\phi}} \otimes \dket{U_{\phi}}_{34}\dbra{U_{\phi}}$ and we have denoted the resulting state with $\varrho_{E}$. In our notation $E_{\phi,AB}^{\otimes2}$ means that the operator is applied on space ${\cal{H}}_{AB}$. Let us cut the calculation into four parts, for each term forming Choi operator $E_{\phi}^{\otimes2}$ of the applied channel.

Firstly, we shall start evaluating the term next to $q^{2}$ from equation (\ref{storing}):
\begin{align*}
	&\varrho_{\phi,\phi} \equiv \Tr_{A} \left[ \left( \dket{U_{\phi}}_{AB}\dbra{U_{\phi}}^{\otimes2} \otimes \mathbb{1}_{A^{\prime}} \right) \left( \bigoplus_{j,k=0}^{2} \sqrt{p_{j}p_{k}} \dket{I_{j}}_{AA^{\prime}}\dbra{I_{k}}^{T_{A}} \otimes \mathbb{1}_{B} \right) \right] \overset{(\ref{stateChannelDuality})}{=} \\&\Tr_{A} \left\{ \left[ (U_{\phi,A}^{\otimes2} \otimes \mathbb{1}_{B})\dket{I}_{AB}\dbra{I}(U_{\phi,A}^{\dagger\otimes2} \otimes \mathbb{1}_{B}) \otimes \mathbb{1}_{A^{\prime}} \right] \left( \bigoplus_{j,k=0}^{2} \sqrt{p_{j}p_{k}} \ket{jj}_{AA^{\prime}}\bra{kk}^{T_{A}} \otimes \mathbb{1}_{B} \right) \right\} = \\&\Tr_{A} \biggl\{ \bigg[  \underbrace{\bigoplus_{m=0}^{2}\bigl( (e^{im\phi}\otimes\mathbb{1}_{m_{m}})_{A}}_{\text{decomposition of }U_{\phi,A}^{\otimes2}}\otimes\mathbb{1}_{B}\bigl)\ket{mm}_{AB} \underbrace{\bigoplus_{n=0}^{2} {}_{AB}\bra{nn} \big( (e^{-in\phi}\otimes\mathbb{1}_{m_{n}})_{A}}_{\text{decomposition of }U_{\phi,A}^{\dagger,\otimes2}}\otimes\mathbb{1}_{B}\big) \otimes \mathbb{1}_{A^{\prime}} \bigg] \\& \left( \bigoplus_{j,k=0}^{2} \sqrt{p_{j}p_{k}} \ket{k}_{A}\bra{j} \otimes \ket{j}_{A^{\prime}}\bra{k} \otimes \mathbb{1}_{B} \right) \biggl\} \underset{\text{multiplicity spaces}}{\overset{\text{discarding}}{=}} \\& \bigoplus_{j,k,m,n=0}^{2} \ket{m}_{B}\bra{n} \otimes \ket{j}_{A^{\prime}}\bra{k} e^{i(m-n)\phi} \sqrt{p_{j}p_{k}} \underbrace{\Tr_{A}\left(\ket{m}_{A}\bra{n}\ket{k}_{A}\bra{j}\right)}_{\delta_{jm}, \delta_{kn}} = \\& \bigoplus_{j,k=0}^{2} \sqrt{p_{j}} e^{ij\phi} \ket{jj}_{BA^{\prime}}\bra{kk} e^{-ik\phi} \sqrt{p_{k}} \overset{p_{j}=p_{k}=\frac{1}{3}}{=} \bigoplus_{j,k=0}^{2} \frac{1}{\sqrt{3}} e^{ij\phi} \ket{jj}_{BA^{\prime}}\bra{kk} e^{-ik\phi} \frac{1}{\sqrt{3}},
	\label{fifi} \numberthis
\end{align*}
where $\mathbb{1}_{m_{m}}$ denotes identity on multiplicity space ${\cal{H}}_{m}$, $U_{\phi,A}$ and $(.)_{A}$ denote operator applied on ${\cal{H}}_{A}$. Further let us take a look at the second term, next to the coefficient $q(1-q)$, from equation (\ref{storing}):
\begin{align*}
	&\varrho_{\phi,I} \equiv \Tr_{13} \left[ \left( \dket{U_{\phi}}_{12}\dbra{U_{\phi}} \otimes \frac{1}{2} \mathbb{1}_{34} \otimes \mathbb{1}_{A^{\prime}} \right) \left( \bigoplus_{j,k=0}^{2} \sqrt{p_{j}p_{k}} \dket{I_{j}}_{13A^{\prime}}\dbra{I_{k}}^{T_{13}} \otimes \mathbb{1}_{B} \otimes \mathbb{1}_{B} \right) \right] \overset{(\ref{stateChannelDuality})}{=} \\& \frac{1}{2} \Tr_{13} \Bigg\{ \left[ (U_{\phi,1} \otimes \mathbb{1}_{2}) \dket{I}_{12}\dbra{I} (U_{\phi,1}^{\dagger} \otimes \mathbb{1}_{2}) \otimes \mathbb{1}_{34} \otimes \mathbb{1}_{A^{\prime}} \right] \left( \bigoplus_{j,k=0}^{2} \sqrt{p_{j}p_{k}}\ket{jj}_{13A^{\prime}}\bra{kk}^{T_{13}} \right) \Bigg\} = \\& \frac{1}{2} \Tr_{13} \Biggl\{ \biggl[ \bigoplus_{m=0}^{1} \left[(e^{im\phi} \otimes \mathbb{1}_{m_{m}})_{1} \otimes \mathbb{1}_{2} \right] \ket{mm}_{12} \bigoplus_{n=0}^{1} {}_{12}\bra{nn} \left[ (e^{-in\phi} \otimes \mathbb{1}_{m_{n}})_{1} \otimes \mathbb{1}_{2} \right] \\& \otimes \sum_{d=0}^{1} \ket{d}_{3}\bra{d} \otimes \sum_{c=0}^{1} \ket{c}_{4}\bra{c} \otimes \mathbb{1}_{A^{\prime}} \biggr] \left( \bigoplus_{j,k=0}^{2} \sqrt{p_{j}p_{k}} \ket{k}_{13}\bra{j} \otimes \ket{j}_{A^{\prime}}\bra{k} \otimes \mathbb{1}_{B} \right) \Biggr\} \underset{k = bb^{\prime}}{\overset{j = aa^{\prime}}{=}} \\& \frac{1}{2} \bigoplus_{\substack{a,a^{\prime},b,b^{\prime}, \\ m,n,c,d=0}}^{1} \Tr_{13} \Biggl\{ \biggl[ \left[ (e^{im\phi} \otimes \mathbb{1}_{m_{m}})_{1} \otimes \mathbb{1}_{2} \right] \ket{mm}_{12} \hspace*{-0.5mm}\bra{nn} \left[ (e^{-in\phi} \otimes \mathbb{1}_{m_{n}})_{1} \otimes \mathbb{1}_{2} \right] \\& \otimes \ket{d}_{3}\bra{d} \otimes \ket{c}_{4}\bra{c} \otimes \mathbb{1}_{A^{\prime}} \biggr] \big( \sqrt{p_{aa^{\prime}}p_{bb^{\prime}}} \underbrace{\ket{bb^{\prime}}_{13}\bra{aa^{\prime}} \otimes \ket{aa^{\prime}}_{A^{\prime}}\bra{bb^{\prime}}}_{\text{multiplicity state }\ket{10} = \ket{3}\text{ is not allowed}} \otimes \mathbb{1}_{B} \big) \Biggr\} = \\& \frac{1}{2} \bigoplus_{\substack{a,a^{\prime},b,b^{\prime}, \\ m,n,c,d=0}}^{1} \sqrt{p_{aa^{\prime}}p_{bb^{\prime}}} e^{i(m-n)\phi} \ket{mc}_{24}\bra{nc} \otimes \ket{aa^{\prime}}_{A^{\prime}}\bra{bb^{\prime}}  \underbrace{\Tr_{13} (\ket{md}_{13}\bra{nd} \ket{bb^{\prime}}_{13}\bra{aa^{\prime}}}_{\delta_{nb},\delta_{db^{\prime}}}) = \\& \frac{1}{2} \bigoplus_{\substack{a,a^{\prime},b,b^{\prime}, \\ m,c=0}}^{1} \sqrt{p_{aa^{\prime}}p_{bb^{\prime}}} e^{i(m-b)\phi} \ket{mc}_{24}\bra{bc} \otimes \ket{aa^{\prime}}_{A^{\prime}}\bra{bb^{\prime}} \underbrace{{}_{13}\bra{aa^{\prime}}\ket{mb^{\prime}}_{13}}_{\delta_{am}, \delta_{a^{\prime}b^{\prime}}} = \\& \frac{1}{2} \bigoplus_{a,a^{\prime},b,c=0}^{1} \sqrt{p_{aa^{\prime}}p_{bb^{\prime}}} e^{i(a-b)\phi} \ket{ac}_{B}\bra{bc} \otimes \ket{aa^{\prime}}_{A^{\prime}}\bra{ba^{\prime}} \underset{\text{ for all }a,a^{\prime}}{\overset{p_{aa^{\prime}} = \frac{1}{3}}{=}} \\& \frac{1}{6} \bigoplus_{c=0}^{1} (\ket{0c}_{B}\bra{0c} \otimes \ket{00}_{A^{\prime}}\bra{00} +
	e^{-i\phi} \ket{0c}_{B}\bra{1c} \otimes \xcancel{\ket{00}_{A^{\prime}}\bra{10}} + \ket{0c}_{B}\bra{0c} \otimes \ket{01}_{A^{\prime}}\bra{01} \\&+ e^{-\phi}\ket{0c}_{B}\bra{1c} \otimes \ket{01}_{A^{\prime}}\bra{11} + e^{i\phi} \ket{1c}_{B}\bra{0c} \otimes \xcancel{\ket{10}_{A^{\prime}}\bra{00}} + \ket{1c}_{B}\bra{1c} \otimes \xcancel{\ket{10}_{A^{\prime}}\bra{10}} \\&+ e^{i\phi}\ket{1c}_{B}\bra{0c} \otimes \ket{11}_{A^{\prime}}\bra{01} + \ket{1c}_{B}\bra{1c} \otimes \ket{11}_{A^{\prime}}\bra{11}) \underset{\text{discarding multiplicity state }\ket{10} = \ket{3}}{\overset{\text{dictionary } (\ref{dictionary})}{=}} \\& \frac{1}{6} ( \ket{00}_{BA^{\prime}}\bra{00} + \ket{10}_{BA^{\prime}}\bra{10} + \ket{01}_{BA^{\prime}}\bra{01} + \ket{11}_{BA^{\prime}}\bra{11} + e^{-i\phi}\ket{01}_{BA^{\prime}}\bra{32} \\&+ e^{-\phi}\ket{11}_{BA^{\prime}}\bra{22} + e^{i\phi}\ket{32}_{BA^{\prime}}\bra{01} + e^{i\phi}\ket{22}_{BA^{\prime}}\bra{11} + \ket{32}_{BA^{\prime}}\bra{32} + \ket{22}_{BA^{\prime}}\bra{22}).
	\label{fii} \numberthis
\end{align*}
In the analogous manner, we can evaluate the third term in equation (\ref{storing}):
\begin{align*}
	&\varrho_{I,\phi} \equiv \Tr_{13} \left[ \left(  \frac{1}{2} \mathbb{1}_{12} \otimes \dket{U_{\phi}}_{34}\dbra{U_{\phi}} \otimes \mathbb{1}_{A^{\prime}} \right) \left( \bigoplus_{j,k=0}^{2} \sqrt{p_{j}p_{k}} \dket{I_{j}}_{13A^{\prime}}\dbra{I_{k}}^{T_{13}} \otimes \mathbb{1}_{B} \right) \right] = \cdots = \\& \frac{1}{6} (\ket{00}_{BA^{\prime}}\bra{00} + \ket{30}_{BA^{\prime}}\bra{30} + e^{i\phi}\ket{11}_{BA^{\prime}}\bra{00} + e^{i\phi}\ket{21}_{BA^{\prime}}\bra{30} + e^{-i\phi}\ket{00}_{BA^{\prime}}\bra{11} \\& + e^{-i\phi}\ket{30}_{BA^{\prime}}\bra{21} + \ket{11}_{BA^{\prime}}\bra{11} + \ket{21}_{BA^{\prime}}\bra{21} + \ket{12}_{BA^{\prime}}\bra{12} + \ket{22}_{BA^{\prime}}\bra{22}).
	\label{ifi} \numberthis
\end{align*}
In the end, let us calculate the last term from equation (\ref{storing}):
\begin{align*}
	&\varrho_{I,I} \equiv \Tr_{13} \left[ \left(\frac{1}{2} \mathbb{1}_{12} \otimes \frac{1}{2} \mathbb{1}_{34} \otimes \mathbb{1}_{A^{\prime}}\right) \left( \bigoplus_{jk=0}^{2} \sqrt{p_{j}p_{k}} \ket{k}_{A}\bra{j} \otimes \ket{j}_{A^{\prime}}\bra{k} \otimes \mathbb{1}_{B} \right) \right] \underset{{\cal{H}}_{A} = {\cal{H}}_{13}}{\overset{{\cal{H}}_{B} = {\cal{H}}_{24}}{=}} \\& \frac{1}{4} \sqrt{p_{j}p_{k}} \left( \mathbb{1}_{B} \otimes \bigoplus_{j=0}^{2} \ket{j}_{A^{\prime}}\bra{k} \right) \Tr_{A}\left(\ket{k}_{A}\bra{j}\right) = \frac{1}{4} \left( \mathbb{1}_{B} \otimes \bigoplus_{j=0}^{2} p_{j} \ket{j}_{A^{\prime}}\bra{j} \right) \overset{(i)}{=} \\& \frac{1}{4} \frac{1}{3} \left(\mathbb{1}_{B} \otimes \mathbb{1}_{A^{\prime}}^{MS}\right),
	\label{ii} \numberthis
\end{align*}
where in $(i)$, we know that $p_{j} = p_{k} = \frac{1}{3}$ for all $j,k$ and where $\mathbb{1}^{MS} = \sum_{j=0}^{2}\ket{j}_{A^{\prime}}\bra{j}$ denotes identity on space without multiplicities, i.e., on space that is spanned by states $\{\ket{00}, \ket{01}, \ket{11}\}$, while $\mathbb{1}$ is an identity on space spanned by states $\{\ket{00}, \ket{01}, \ket{10}, \ket{11}\}$.

\subsubsection{Retrieving}

To retrieve the transformation that is implemented in the end, we shall calculate the following expression:
\begin{align*}
	&R_{s} \star \varrho_{E} = \Tr_{M} \left[R_{MCD}^{s} (\varrho_{E}^{T_{M}} \otimes \mathbb{1}_{CD})\right] =\\& \Tr_{M} \left\{R_{MCD}^{s} \left[(q^{2}\varrho_{\phi,\phi}^{T_{M}} + q(1-q)\varrho_{\phi,I}^{T_{M}} + q(1-q)\varrho_{I,\phi}^{T_{M}} + (1-q)^{2}\varrho_{I,I}^{T_{M}}) \otimes \mathbb{1}_{CD}\right]\right\},
	\label{retrieving} \numberthis
\end{align*}
where $R_{s} = R_{MCD}^{s}$ denotes Choi operator of quantum instrument corresponding to successful retrieving in case of optimal device for probabilistic storage and retrieval of phase gate. The form of this instrument is as follows \cite{ProbabilisticStorageAndRetrievalOfQubitPhaseGates}:
\begin{align*}
	R_{s} = \bigoplus_{J=-1}^{2} \mathbb{1}_{J} \otimes \mathbb{1}_{J} \otimes s^{(J)},
	\label{retr instr} \numberthis
\end{align*}
where $s^{(J)}_{jj^{\prime}} = \sum_{j,j^{\prime}} \dket{I_{m_{J}^{j}}}_{CD}\dbra{I_{m_{J}^{j^{\prime}}}}$ for $J = \{0, 1\}$ with $j, j^{\prime} \in \{J, J+1\}$ and $s^{(-1)}_{jj^{\prime}} = s^{(2)}_{jj^{\prime}} = 0$ and $s^{(0)}_{jj^{\prime}} = s^{(1)}_{jj^{\prime}} = 1$ for all $j,j^{\prime}$. Identity $I_{m_{J}^{j}}$ relates to the following decomposition:
\begin{align*}
	e^{ij\phi} I_{j} \otimes U_{\phi}^{\ast} = \bigoplus_{J \in \{j-1,j\}} e^{iJ\phi} \otimes I_{m_{J}^{j}},
\end{align*} 
where it denotes multiplicity spaces, and the index $j$ labels irreps in the decomposition of $U_{\phi}^{\otimes N}$ from equation (\ref{IrrepDecomposition}). The identity $\mathbb{1}_{J} \otimes \mathbb{1}_{J}$ acts on Hilbert space ${\cal{H}}_{M}$ that is effectively also the Hilbert space in which the state, denoted by $\varrho_{E}$, resides after the storing phase. Let us also express the form of $R_{s}$ for a concrete index $J$:
\begin{align*}
	&R_{s}^{(J)} = \ket{J,J}_{M}\bra{J,J} \otimes \ket{00}_{CD}\bra{00} + \ket{J,J}_{M}\bra{J+1,J+1} \otimes \ket{00}_{CD}\bra{11} \\& + \ket{J+1,J+1}_{M}\bra{J,J} \otimes \ket{11}_{CD}\bra{00} + \ket{J+1,J+1}_{M}\bra{J+1,J+1} \otimes \ket{11}_{CD}\bra{11}.
	\label{Jretr} \numberthis
\end{align*}

Again, let us evaluate the expression in equation (\ref{retrieving}) by parts. Firstly, we shall start with the term next to $q^{2}$:
\begin{align*} \label{retrUfi}
	&\Tr_{M} \left[R_{MCD}^{s} (\varrho_{\phi,\phi}^{T_{M}} \otimes \mathbb{1}_{CD})\right] \overset{(\ref{fifi})}{=} \Tr_{M} \Bigg[\Big(\bigoplus_{J=-1}^{2} R_{s}^{(J)} \Big) \bigg(\bigoplus_{j,k=0}^{2} \frac{1}{\sqrt{3}} e^{ij\phi} \underbrace{\ket{jj}_{M}\bra{kk}^{T_{M}}}_{\dket{I_{k}}_{M}\dbra{I_{j}}} e^{-ik\phi} \frac{1}{\sqrt{3}} \otimes \mathbb{1}_{CD}\bigg)\Bigg] = \\& \frac{1}{3} \bigoplus_{J=-1}^{2}\bigoplus_{j,k=0}^{2} e^{i(j-k)\phi} \prescript{}{M}{\dbra{I_{j}}} R_{s}^{(J)} \dket{I_{k}}_{M} \overset{(\ref{Jretr})}{=} \\& \frac{1}{3}\bigoplus_{J=-1}^{2}\bigoplus_{j,k=0}^{2} e^{i(j-k)\phi} \prescript{}{M}{\bra{jj}} \bigl( s_{J,J}^{(J)} \underbrace{\ket{J,J}_{M}\bra{J,J}}_{j=k=J} \otimes \ket{00}_{CD}\bra{00} + s_{J,J+1}^{(J)} \underbrace{\ket{J,J}_{M}\bra{J+1,J+1}}_{j=J,k=J+1} \otimes \ket{00}_{CD}\bra{11} + \\&s_{J+1,J}^{(J)} \underbrace{\ket{J+1,J+1}_{M}\bra{J,J}}_{j=J+1,k=1} \otimes \ket{11}_{CD}\bra{00} + s_{J+1,J+1}^{(J)} \underbrace{\ket{J+1,J+1}_{M}\bra{J+1,J+1}}_{j=k=J+1} \otimes \ket{11}_{CD}\bra{11} \bigr) \ket{kk}_{M} \\&\overset{s^{(-1)}_{jj^{\prime}} = s^{(2)}_{jj^{\prime}} = 0}{\underset{s^{(0)}_{jj^{\prime}} = s^{(1)}_{jj^{\prime}} = 1}{=}} \frac{1}{3} \bigoplus_{J=0}^{1} \left( \ket{00}_{CD}\bra{00} + e^{-i\phi}\ket{00}_{CD}\bra{11} + e^{i\phi}\ket{11}_{CD}\bra{00} + \ket{11}_{CD}\bra{11} \right) \overset{(i)}{=} \frac{2}{3} \dket{U_{\phi}}_{CD}\dbra{U_{\phi}},  \numberthis
\end{align*}
where $\dket{U_{\phi}}_{CD}\dbra{U_{\phi}} = \ket{00}_{CD}\bra{00} + e^{-i\phi}\ket{00}_{CD}\bra{11} + e^{i\phi}\ket{11}_{CD}\bra{00} + \ket{11}_{CD}\bra{11}$, therefore, in $(i)$ we have only evaluated the direct summation. Let us now move on to the second term from equation (\ref{retrieving}):
\begin{align*}
	\Tr_{M} \left[R_{MCD}^{s} (\varrho_{\phi,I}^{T_{M}} \otimes \mathbb{1}_{CD})\right].
	\label{retrFii} \numberthis
\end{align*}
Let us remind the form of $\varrho_{\phi,I}^{T_{M}}$ from equation (\ref{fii}) where we are also denoting individual terms with numbers for future use:
\begin{align*}\label{Tfii} 
	&\varrho_{\phi,I}^{T_{M}} = \frac{1}{6} ( \underbrace{\ket{00}_{M}\bra{00}}_{1} + \underbrace{\ket{10}_{M}\bra{10}}_{2} + \underbrace{\ket{01}_{M}\bra{01}}_{3} + \underbrace{\ket{11}_{M}\bra{11}}_{4} + \underbrace{e^{-i\phi}\ket{32}_{M}\bra{01}}_{5} \\&+ 
	\underbrace{e^{-i\phi}\ket{22}_{M}\bra{11}}_{6} + \underbrace{e^{i\phi}\ket{01}_{M}\bra{32}}_{7} +  \underbrace{e^{i\phi}\ket{11}_{M}\bra{22}}_{8} +
	\underbrace{\ket{32}_{M}\bra{32}}_{9} +
    \underbrace{\ket{22}_{M}\bra{22}}_{10}).
	\numberthis
\end{align*}
Let us now express the retrieving operator for $J = 0$ and $J = 1$ (while remembering $s^{(-1)} = s^{(2)} = 0$):
\begin{align*}\label{Rs01}
	R_{s}^{(0)} &= \underbrace{\ket{00}_{M}\bra{00}}_{a} \otimes \ket{00}_{CD}\bra{00} + \underbrace{\ket{00}_{M}\bra{11}}_{b} \otimes \ket{00}_{CD}\bra{11} \\&+ \underbrace{\ket{11}_{M}\bra{00}}_{c} \otimes \ket{11}_{CD}\bra{00} + \underbrace{\ket{11}_{M}\bra{11}}_{d} \otimes \ket{11}_{CD}\bra{11}, \\
	R_{s}^{(1)} &= \underbrace{\ket{11}_{M}\bra{11}}_{e} \otimes \ket{00}_{CD}\bra{00} + \underbrace{\ket{11}_{M}\bra{22}}_{f} \otimes \ket{00}_{CD}\bra{11} \\&+ \underbrace{\ket{22}_{M}\bra{11}}_{g} \otimes \ket{11}_{CD}\bra{00} + \underbrace{\ket{22}_{M}\bra{22}}_{h} \otimes \ket{11}_{CD}\bra{11}.
	\numberthis
\end{align*}
Putting equations (\ref{Tfii}) and (\ref{Rs01}) into equation (\ref{retrFii}) we obtain:
\begin{align*}
	\Tr_{M} \left[R_{s}^{(0)} (\varrho_{\phi,I}^{T_{M}} \otimes \mathbb{1}_{CD})\right] &= \frac{1}{6}\bigl(\hspace*{-0.5cm}\underbrace{\ket{00}_{CD}\bra{00}}_{1 \text{ from (\ref{Tfii}), } a \text{ from (\ref{Rs01})}} \hspace*{-3mm}+\hspace*{3mm} \underbrace{\ket{11}_{CD}\bra{11}}_{4d}\bigr),\\
	\Tr_{M} \left[R_{s}^{(1)} (\varrho_{\phi,I}^{T_{M}} \otimes \mathbb{1}_{CD})\right] &= \frac{1}{6}(\underbrace{\ket{00}_{CD}\bra{00}}_{4e} + \underbrace{e^{-i\phi}\ket{00}_{CD}\bra{11}}_{6f} + \underbrace{e^{i\phi}\ket{11}_{CD}\bra{00}}_{8g} + \underbrace{\ket{11}_{CD}\bra{11}}_{10h}) \\&= \frac{1}{6} \dket{U_{\phi}}_{CD}\dbra{U_{\phi}}.
\end{align*}
Putting these two equations together, we get the following result:
\begin{align*}
	\Tr_{M} \left[R_{MCD}^{s} (\varrho_{\phi,I}^{T_{M}} \otimes \mathbb{1}_{CD})\right] = \frac{1}{6} \left[\dket{U_{\phi}}_{CD}\dbra{U_{\phi}} + \left(\ket{00}_{CD}\bra{00} + \ket{11}_{CD}\bra{11}\right) \right]
\end{align*}
In the same vein, we can calculate also the third term in (\ref{retrieving}):
\begin{align*}
	\Tr_{M} \left[R_{MCD}^{s} (\varrho_{I,\phi}^{T_{M}} \otimes \mathbb{1}_{CD})\right] = \frac{1}{6} \left[\dket{U_{\phi}}_{CD}\dbra{U_{\phi}} + \left(\ket{00}_{CD}\bra{00} + \ket{11}_{CD}\bra{11}\right) \right]
\end{align*}
Let us finish with the last term:
\begin{align*}
	&\Tr_{M} \left[R_{MCD}^{s} (\varrho_{I,I}^{T_{M}} \otimes \mathbb{1}_{CD})\right] \underset{(\ref{ii})}{\overset{s^{(-1)} = s^{(2)} = 0}{=}} \frac{1}{4 \cdot 3} \Tr_{M} \left[\left(\bigoplus_{J=0}^{1}R_{s}^{(J)}\right) \left(\sum_{j=0}^{3}\sum_{k=0}^{2} \ket{jk}_{M}\bra{jk} \otimes \mathbb{1}_{CD}\right)\right] = \\& \frac{1}{4 \cdot 3} \left\{ \Tr_{M} \left[R_{s}^{(0)} \left(\sum_{j=0}^{3}\sum_{k=0}^{2} \ket{jk}_{M}\bra{jk} \otimes \mathbb{1}_{CD}\right)\right] + \Tr_{M} \left[R_{s}^{(1)} \left(\sum_{j=0}^{3}\sum_{k=0}^{2} \ket{jk}_{M}\bra{jk} \otimes \mathbb{1}_{CD}\right)\right] \right\} \overset{(\ref{Rs01})}{=} \\& \frac{1}{4 \cdot 3} \biggl\{ \bigl( \underbrace{\ket{00}_{CD}\bra{00}}_{a} + \underbrace{\ket{11}_{CD}\bra{11}}_{d} \bigr) + \bigl( \underbrace{\ket{00}_{CD}\bra{00}}_{e} + \underbrace{\ket{11}_{CD}\bra{11}}_{h} \bigr) \biggr\} = \\&
	\frac{2}{4\cdot3} \left(\ket{00}_{CD}\bra{00} + \ket{11}_{CD}\bra{11}\right).
\end{align*}
Putting all the above results in equation (\ref{retrieving}), we obtain the channel applied on state $\dyad{\xi}{\xi}$ in case of successful implementation:
\begin{align*} 
	R_{s} \star \varrho_{E} = \frac{2}{3} \left[\left(q^{2} + q\frac{1-q}{2}\right) \dket{U_{\phi}}_{CD}\dbra{U_{\phi}} + \left(q\frac{1-q}{2} + \frac{(1-q)^{2}}{4}\right) \left(\ket{00}_{CD}\bra{00} + \ket{11}_{CD}\bra{11}\right)\right] 
\end{align*}

\subsubsection{Generalization and Probabilities}

Let us remind that the number of implementations of channel ${\cal{E}}_{\phi}$ is denoted with the letter $N$. We have also explicitly calculated the case for $N = 3$ with the following result:
\begin{align*} \label{rhoE}
	R_{s} \star {\varrho_{E}} &= \frac{3}{4} \bigg\{ \left[ q^{3} + q^{2}(1-q) + q\frac{(1-q)^{2}}{4} \right] \dket{U_{\phi}}_{CD}\dbra{U_{\phi}} \\&+ \left[ q^{2}\frac{1-q}{2} + q\frac{(1-q)^{2}}{2} + \frac{(1-q)^{3}}{8} \right] (\ket{00}_{CD}\bra{00} + \ket{11}_{CD}\bra{11}) \bigg\}. \numberthis
\end{align*}
Now, we shall generalize our result for arbitrary number $N$ of times we have access to channel ${\cal{E}}_{\phi}$:
\begin{align*}
	&R_{s} \star \varrho_{E} = \frac{1}{N+1} \\&\biggl\{ \left[ 1\frac{N}{2^{0}}q^{N}(1-q)^{0} + N \frac{N-1}{2^{1}}q^{N-1}(1-q)^{1} + \dots + N\frac{1}{2^{N-1}}q^{1}(1-q)^{N-1} + 1\frac{0}{2^{N}}q^{0}(1-q)^{N} \right] \dket{U_{\phi}}_{CD}\dbra{U_{\phi}} + \\& \left[ 1\frac{0}{2^{0}}q^{N}(1-q)^{0} + N \frac{1}{2^{1}}q^{N-1}(1-q)^{1} + \dots + N\frac{N-1}{2^{N-1}}q^{1}(1-q)^{N-1} + 1\frac{N}{2^{N}}q^{0}(1-q)^{N} \right] \\&\left(\ket{00}_{CD}\bra{00} + \ket{11}_{CD}\bra{11}\right) \biggr\}
\end{align*}
Factor $\frac{1}{N+1}$ comes from normalization of input state to PSAR device from equation (\ref{state}): $\ket{\psi} = \bigoplus_{j=0}^{N} \frac{1}{\sqrt{N-1}} \dket{I_{j}}$. First number of each term in square brackets follow Pascal's triangle as the implemented channel ${\cal{E}}_{\phi}^{\otimes N} = \left[q {\cal{U}}_{\phi} + (1-q) {\cal{C}}_{\mathbb{1}/2}\right]^{\otimes N}$ follow binomial distribution. Each term is divided by $2^{I}$, where $I$ denotes the number of times we apply depolarizing channel in the particular term as every such channel sends a quantum state in the total mixture. Number in the nominator of every fraction comes from the number of times we apply the particular channels - in the first case the unitary one and in the second case the depolarizing one. Finally, factors $q^{(N-I)}(1-q)^{I}$ come from number of used unitary channels $N-I$ and number of used depolarizing channels $I$ in the particular terms. Let us further simplify the previous expression:
\begin{align*}\label{resWhiteNoise}
	&R_{s} \star \varrho_{E} = \\& \frac{1}{N+1} \biggl\{ \sum_{I=0}^{N}\binom{N}{I} q^{N-I}(1-q)^{I} \frac{N-I}{2^{I}}\dket{U_{\phi}}_{CD}\dbra{U_{\phi}} \\&+ \sum_{I=0}^{N}\binom{N}{I} q^{N-I}(1-q)^{I} \frac{I}{2^{I}} \left(\ket{00}_{CD}\bra{00} + \ket{11}_{CD}\bra{11}\right) \biggr\} = \\& \frac{1}{N+1} \biggl\{ \sum_{I=0}^{N}\binom{N}{I} q^{N-I}(1-q)^{I} 2^{-I} \left[ (N-I)\dket{U_{\phi}}_{CD}\dbra{U_{\phi}} + I\left(\ket{00}_{CD}\bra{00} + \ket{11}_{CD}\bra{11}\right) \right] \biggr\} \overset{(i)}{=} \\& \frac{1}{(N+1)D} \biggl\{ \sum_{k=0}^{N}\binom{N}{I} q^{N-I}(1-q)^{I} 2^{-I} D \left[ (N-I)\dket{U_{\phi}}_{CD}\dbra{U_{\phi}} + I\left(\ket{00}_{CD}\bra{00} + \ket{11}_{CD}\bra{11}\right) \right] \biggr\} = \\& \frac{N(1+q)^{N}}{2^{N}(N+1)} \biggl\{ \sum_{k=0}^{N}\binom{N}{I} q^{N-I}(1-q)^{I} 2^{-I} \frac{2^{N}}{N(1+q)^{N}} \\&\left[ (N-I)\dket{U_{\phi}}_{CD}\dbra{U_{\phi}} + I\left(\ket{00}_{CD}\bra{00} + \ket{11}_{CD}\bra{11}\right) \right] \biggr\} \overset{(ii)}{=} \\& \frac{N(1+q)^{N}}{2^{N}(N+1)} \left[ \frac{2q}{1+q}\dket{U_{\phi}}_{CD}\dbra{U_{\phi}} + \frac{1-q}{1+q}\left(\ket{00}_{CD}\bra{00} + \ket{11}_{CD}\bra{11}\right) \right].
	\numberthis
\end{align*}
In $(i)$, we have multiplied the expression by $\frac{D}{D}$, where $D = \frac{2^{N}}{N(1+q)^{N}}$. Evaluation of summation in $(ii)$ was done in Wolfram Mathematica. Success probability of implementing channel $\frac{2q}{1+q}\dket{U_{\phi}}\dbra{U_{\phi}} + \frac{1-q}{1+q}\left(\ket{00}\bra{00} + \ket{11}\bra{11}\right)$ is $p_{suc} = \frac{N(1+q)^{N}}{2^{N}(N+1)}$. The probability of successful retrieval $p_{suc}$ of channel $\frac{2q}{1+q}\dket{U_{\phi}}\dbra{U_{\phi}} + \frac{1-q}{1+q}\left(\ket{00}\bra{00} + \ket{11}\bra{11}\right)$ is diminishing for majority of interval of values of $N$, however for high values of $N$, the success probability is increasing. This is happening because there are two competing factors, $\frac{N}{N+1}$ and $\left(\frac{1+q}{2}\right)^{N}$. While the term $\left(\frac{1+q}{2}\right)^{N}$ is dominant, the probability is decreasing, however when $\frac{N}{N+1}$ prevails, the probability starts to increase. Optimal value of $N$ can be obtained by finding stationary point:
\begin{align*}
	\derivative{p_{suc}}{N} = - \frac{2^{-N} N (1+q)^{N}}{(1+N)^{2}} + \frac{2^{-N} (1+q)^{N}}{1+N} - \frac{2^{-N} N \log(2)}{1+N} + \frac{2^{-N} N \log(1+q)}{1+N} \overset{!}{=} 0.
\end{align*}
By solving previous equation, we obtain two solutions for $-1 \ge q < 1$:
\begin{align*}
	1) \quad N &= \frac{1}{2} \left(-1 - \sqrt{\frac{4 + \log(\frac{2}{1+q})}{\log(\frac{2}{1+q})}}\right) \qquad
	2) \quad N &= \frac{1}{2} \left(-1 + \sqrt{\frac{4 + \log(\frac{2}{1+q})}{\log(\frac{2}{1+q})}}\right).
\end{align*}
However, the first solution results in negative $N$, thus the optimal $N$ is the second one. If we take a look only at the retrieval of a unitary transformation as that is what we desire to do, we find a similar behavior. Factor next to $\dket{U_{\phi}}_{CD}\dbra{U_{\phi}}$ is $\frac{N(1+q)^{N}}{2^{N}(N+1)}\frac{2q}{1+q} = \frac{Nq}{N+1} \frac{(1+q)^{N-1}}{2^{N-1}}$ and for most values of $q$ with growing number $N$ is diminishing, however there always exists some interval of high $q$ for which it is increasing with increasing $N$. It is also noteworthy that the noisy part of the retrieved channel is no longer depolarizing noise, but rather phase damping $P_{CD} = \left(\ket{00}_{CD}\bra{00} + \ket{11}_{CD}\bra{11}\right)$.

\subsection{Phase Damping}

Again, we use the same device as in the previous section but in this case, we shall implement the following channel:
\begin{align*}
	{\cal{F}}_{\phi} = q {\cal{U}}_{\phi} + (1-q) {\cal{P}},
\end{align*}
where $0 \leq q \leq 1$, channel ${\cal{P}} = \frac{1}{2} ({\cal{I}} + \varsigma_{z})$ and where $\varsigma_{z} (\cdot) = \sigma_{z} (\cdot) \sigma_{z}$ with $\sigma_{z} = \left(\begin{smallmatrix}
	1 & 0\\
	0 & -1
\end{smallmatrix}\right)$ being Pauli matrix. Let us have quantum state $\varrho = \dyad{\psi}{\psi}$, where $\ket{\psi} = a\ket{0} + b\ket{1}$ and apply channel ${\cal{P}}$ on this state:
\begin{align*}
	{\cal{P}}(\varrho) &= \frac{1}{2} \left[ {\cal{I}} (\varrho) + \varsigma_{z} (\varrho) \right]
	\\&= \frac{1}{2} \big[\left(\abs{a}^{2} \dyad{0}{0} + ab^{\star} \dyad{0}{1} + a^{\star}b \dyad{1}{0} + \abs{b}^{2} \dyad{1}{1}\right) 
	\\&\quad+ \sigma_{z} \left(\abs{a}^{2} \dyad{0}{0} + ab^{\star} \dyad{0}{1} + a^{\star}b \dyad{1}{0} + \abs{b}^{2} \dyad{1}{1}\right) \sigma_{z}\big]
	\\&= \frac{1}{2} \left[\left(\abs{a}^{2} \dyad{0}{0} + ab^{\star} \dyad{0}{1} + a^{\star}b \dyad{1}{0} + \abs{b}^{2} \dyad{1}{1}\right) + \left(\abs{a}^{2} \dyad{0}{0} - ab^{\star} \dyad{0}{1} - a^{\star}b \dyad{1}{0} + \abs{b}^{2} \dyad{1}{1}\right)\right]
	\\&= \abs{a}^{2} \dyad{0}{0} + \abs{b}^{2} \dyad{1}{1}.
\end{align*}
Or by writing previous equation in a matrix form, we can see that this channel is basically phase damping channel for $\lambda = 1$ in equation (\ref{PDchannel}).
\begin{align*}
	{\cal{P}}(\varrho) = \left(\begin{matrix}
		\abs{a}^{2} & 0\\
		0 & \abs{b}^{2}
	\end{matrix}\right).
\end{align*}
Choi operator corresponding to channel ${\cal{P}}$ is:
\begin{align*} \label{P}
	P &= \left[\frac{1}{2} ({\cal{I}} + \varsigma_{z}) \otimes {\cal{I}}\right] \dket{I}\dbra{I} = \frac{1}{2} \left[({\cal{I}} \otimes {\cal{I}}) + (\varsigma_{z} \otimes {\cal{I}})\right] \sum_{m,n=0}^{1} \ket{m}\bra{n} \otimes \ket{m}\bra{n} = \\& \frac{1}{2} \sum_{m,n=0}^{1} \left[(\ket{mm}\bra{nn} + \sigma_z(\ket{m}\bra{n})\sigma_{z} \otimes \ket{m}\bra{n}) \right] = \\& \frac{1}{2} \left[ \left( \ket{00}\bra{00} + \ket{00}\bra{11} + \ket{11}\bra{00} + \ket{11}\bra{11} \right) + \left( \ket{00}\bra{00} - \ket{00}\bra{11} - \ket{11}\bra{00} + \ket{11}\bra{11} \right) \right] = \\& \ket{00}\bra{00} + \ket{11}\bra{11} = \bigoplus_{j=0}^{1} \ket{jj}\bra{jj}.\numberthis
\end{align*}
That means that Choi operator of the entire channel is:
\begin{align*}
	F_{\phi} = q \dket{U_{\phi}}\dbra{U_{\phi}} + (1-q) P.
\end{align*}

\subsubsection{Two-to-one}

Again, let us show the calculation when we have access to channel ${\cal{F}}_{\phi}$ twice. The device for this particular case is depicted in the figure \ref{PSARpdimage}. Choi operator for this channel has the following form:
\begin{align*}
	F_{\phi}^{\otimes2} &= q^{2} \dket{U_{\phi}}_{12}\dbra{U_{\phi}} \otimes \dket{U_{\phi}}_{34}\dbra{U_{\phi}} + q (1-q) \dket{U_{\phi}}_{12}\dbra{U_{\phi}} \otimes P_{34} \\&+ q (1-q)P_{12} \otimes \dket{U_{\phi}}_{34}\dbra{U_{\phi}} + (1-q)^{2} P_{12} \otimes P_{34}.
\end{align*}
\begin{figure}[H]
	\begin{center}
		\includegraphics[scale = 0.45]{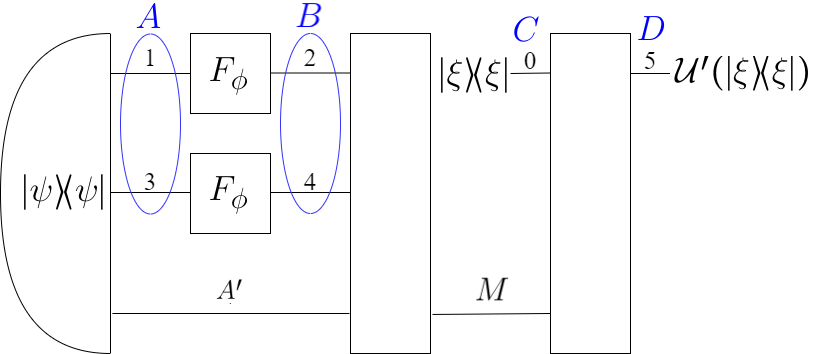}
	\end{center}
	\caption{Schematic image of PSAR implementing channel ${\cal{F}}_{\phi}$ twice with input state $\ket{\psi}$ explicitly written in equation (\ref{state}). At the output of register ${\cal{H}}_{5}$, we expect to retrieve unitary channel, possibly with some noise, in case of successful implementation. In our notation, Hilbert spaces ${\cal{H}}_{13} = {\cal{H}}_{A}$, ${\cal{H}}_{24} = {\cal{H}}_{B}$, ${\cal{H}}_{0} = {\cal{H}}_{C}$ and ${\cal{H}}_{5} = {\cal{H}}_{D}$ are identical.}
	\label{PSARpdimage}
\end{figure}
Initial state entering the device is the same as was in the previous case of white noise written in equation (\ref{state}). "Dictionary" remains the same as in equation (\ref{dictionary}).

\subsubsection{Storing}

We shall calculate what shall be stored in state $\dyad{\psi}{\psi}$ if channel ${\cal{F}}_{\phi}^{\otimes2}$ is applied twice on such an input state:
\begin{align*}
	&\varrho_{F} \equiv F_{\phi,AB}^{\otimes2} \star \ket{\Psi}_{AA^{\prime}}\bra{\Psi} = \Tr_{A} [(F_{\phi,AB}^{\otimes2} \otimes \mathbb{1}_{A^{\prime}}) (\ket{\Psi}_{AA^{\prime}}\bra{\Psi}^{T_{A}} \otimes \mathbb{1}_{B})] \\&\underset{{\cal{H}}_{B} = {\cal{H}}_{24}}{\overset{{\cal{H}}_{A} = {\cal{H}}_{13}}{=}} \Tr_{13}\Biggl\{\biggl[ (q^{2}\dket{U_{\phi}}_{12}\dbra{U_{\phi}} \otimes \dket{U_{\phi}}_{34}\dbra{U_{\phi}} + q (1-q) \dket{U_{\phi}}_{12}\dbra{U_{\phi}} \otimes P_{34} \\&+ q (1-q)P_{12} \otimes \dket{U_{\phi}}_{34}\dbra{U_{\phi}} + (1-q)^{2} P_{12} \otimes P_{34}) \biggl] \left[ \bigoplus_{j,k=0}^{2} \sqrt{p_{j}p_{k}} \dket{I_{j}}_{13A^{\prime}}\dbra{I_{k}}^{T_{A}} \otimes \mathbb{1}_{24} \right] \Biggl\}.
	\label{storingF} \numberthis
\end{align*}
We shall use the similar approach as in the previous section. We start by evaluating the first term which is the same as was in the case of depolarizing noisy channel calculated in equation (\ref{fifi}):
\begin{align*}
	&\varrho_{\phi,\phi} \equiv \Tr_{A} \left[ \left( \dket{U_{\phi}}_{AB}\dbra{U_{\phi}}^{\otimes2} \otimes \mathbb{1}_{A^{\prime}} \right) \left( \bigoplus_{j,k=0}^{2} \sqrt{p_{j}p_{k}} \dket{I_{j}}_{AA^{\prime}}\dbra{I_{k}}^{T_{A}} \otimes \mathbb{1}_{B} \right) \right] = \\& \bigoplus_{j,k=0}^{2} \frac{1}{\sqrt{3}} e^{ij\phi} \ket{jj}_{BA^{\prime}}\bra{kk} e^{-ik\phi} \frac{1}{\sqrt{3}}.
	\label{fifiF} \numberthis
\end{align*}
The next term, next to the factor $q(1-q)$, is also similar to the one in the previous section:
\begin{align*}
	&\varrho_{\phi,P} \equiv \Tr_{13} \left[ \left( \dket{U_{\phi}}_{12}\dbra{U_{\phi}} \otimes P_{34} \otimes \mathbb{1}_{A^{\prime}} \right) \left( \bigoplus_{j,k=0}^{2} \sqrt{p_{j}p_{k}} \dket{I_{j}}_{13A^{\prime}}\dbra{I_{k}}^{T_{13}} \otimes \mathbb{1}_{B} \right) \right] \underset{(\ref{P})}{\overset{(\ref{stateChannelDuality})}{=}} \\& \Tr_{13} \biggl\{ \left[ \bigoplus_{m,n=0}^{1} \left[(e^{im\phi} \otimes \mathbb{1}_{m_{m}})_{1} \otimes \mathbb{1}_{2} \right] \ket{mm}_{12}\bra{nn} \left[ (e^{-in\phi} \otimes \mathbb{1}_{m_{n}})_{1} \otimes \mathbb{1}_{2} \right] \otimes \bigoplus_{i=0}^{1}\ket{ii}_{34}\bra{ii} \otimes \mathbb{1}_{A^{\prime}} \right] \\& \left( \bigoplus_{j,k=0}^{2} \sqrt{p_{j}p_{k}} \ket{k}_{13}\bra{j} \otimes \ket{j}_{A^{\prime}}\bra{k} \otimes \mathbb{1}_{B} \right) \biggr\} \underset{k=bb^{\prime}}{\overset{j=aa^{\prime}}{=}} \\&  \bigoplus_{\substack{a,a^{\prime},b,b^{\prime}, \\ m,n,i=0}}^{1} \underbrace{\Tr_{13}(\ket{mi}_{13}\bra{ni} \ket{bb^{\prime}}_{13}\bra{aa^{\prime}})}_{i\rightarrow a^{\prime}, b^{\prime} \rightarrow a^{\prime}, m \rightarrow a, n \rightarrow b} e^{i(m-n)\phi} \sqrt{p_{aa^{\prime}}p_{bb^{\prime}}} \ket{mi}_{24}\bra{ni} \otimes \ket{aa^{\prime}}_{A^{\prime}}\bra{bb^{\prime}} = \\& \bigoplus_{a,a^{\prime},b = 0}^{1} e^{i(a-b)\phi} \sqrt{p_{aa^{\prime}}p_{ba^{\prime}}} \ket{aa^{\prime}}_{24}\bra{ba^{\prime}} \otimes \ket{aa^{\prime}}_{A^{\prime}}\bra{ba^{\prime}} \underset{{\cal{H}}_{24} = {\cal{H}}_{B}}{\overset{p_{aa^{\prime}} = \frac{1}{3}}{=}} \\& \frac{1}{3} \bigl( \ket{00}_{B}\bra{00} \otimes \ket{00}_{A^{\prime}}\bra{00} + e^{-i\phi} \ket{00}_{B}\bra{10} \otimes \xcancel{\ket{00}_{A^{\prime}}\bra{10}} + \ket{01}_{B}\bra{01} \otimes \ket{01}_{A^{\prime}}\bra{01} + \\& e^{-i\phi} \ket{01}_{B}\bra{11} \otimes \ket{01}_{A^{\prime}}\bra{11} + e^{i\phi} \ket{10}_{B}\bra{00} \otimes \xcancel{\ket{10}_{A^{\prime}}\bra{00}} + \ket{10}_{B}\bra{10} \otimes \xcancel{\ket{10}_{A^{\prime}}\bra{10}} + \\& e^{i\phi} \ket{11}_{B}\bra{01} \otimes \ket{11}_{A^{\prime}}\bra{01} + \ket{11}_{B}\bra{11} \otimes \ket{11}_{A^{\prime}}\bra{11} \bigr) \underset{\text{disregarding multipliticity states}}{\overset{\text{dictionary } (\ref{dictionary})}{=}} \\& \frac{1}{3} \left( \ket{00}_{BA^{\prime}}\bra{00} + \ket{11}_{BA^{\prime}}\bra{11} + e^{-i\phi}\ket{11}_{BA^{\prime}}\bra{22} + e^{i\phi}\ket{22}_{BA^{\prime}}\bra{11} + \ket{22}_{BA^{\prime}}\bra{22} \right).
	\label{fip} \numberthis
\end{align*}
In the similar fashion, we can calculate also the third term:
\begin{align*}
	&\varrho_{P, \phi} \equiv \Tr_{13} \left[ \left( P_{12} \otimes \dket{U_{\phi}}_{34}\dbra{U_{\phi}} \otimes   \mathbb{1}_{A^{\prime}} \right) \left( \ket{\Psi}_{13A^{\prime}}\bra{\Psi}^{T_{13}} \otimes \mathbb{1}_{B} \right) \right] = \\& \frac{1}{3} \left( \ket{00}_{BA^{\prime}}\bra{00} + e^{-i\phi}\ket{00}_{BA^{\prime}}\bra{11} + e^{i\phi}\ket{11}_{BA^{\prime}}\bra{00} + \ket{11}_{BA^{\prime}}\bra{11} + \ket{22}_{BA^{\prime}}\bra{22} \right).
	\label{pfi} \numberthis
\end{align*}
And the last term is:
\begin{align*}
	\Tr_{A} \left[ \left(P_{12} \otimes P_{34} \otimes \mathbb{1}_{A^{\prime}}\right) \left( \ket{\Psi}_{AA^{\prime}}\bra{\Psi}^{T_{A}} \otimes \mathbb{1}_{B} \right) \right].
	\label{tmp} \numberthis
\end{align*}
Let us first investigate the expression $P_{12} \otimes P_{34}$:
\begin{align*}
	&P_{12} \otimes P_{34} = (\ket{00}_{12}\bra{00} + \ket{11}_{12}\bra{11}) \otimes (\ket{00}_{34}\bra{00} + \ket{11}_{34}\bra{11}) = \\& \ket{00,00}_{1234}\bra{00,00} + \ket{11,00}_{1234}\bra{11,00} + \ket{00,11}_{1234}\bra{00,11} + \ket{11,11}_{1234}\bra{11,11}  \underset{{\cal{H}}_{B} = {\cal{H}}_{24}}{\overset{{\cal{H}}_{A} = {\cal{H}}_{13}}{=}} \\& \ket{00}_{A}\bra{00} \otimes \ket{00}_{B}\bra{00} + \ket{10}_{A}\bra{10} \otimes \ket{10}_{B}\bra{10} + \ket{01}_{A}\bra{01} \otimes \ket{01}_{B}\bra{01} + \ket{11}_{A}\bra{11} \otimes \ket{11}_{B}\bra{11} \overset{(i)}{=} \\& \ket{0}_{A}\bra{0} \otimes \ket{0}_{B}\bra{0} + \ket{3}_{A}\bra{3} \otimes \ket{3}_{B}\bra{3} + \ket{1}_{A}\bra{1} \otimes \ket{1}_{B}\bra{1} + \ket{2}_{A}\bra{2} \otimes \ket{2}_{B}\bra{2} = \sum_{i=0}^{3} \ket{ii}_{AB}\bra{ii}.
\end{align*}
In $(i)$, we have used the dictionary from (\ref{dictionary}). Let us use this result and combine it with equation (\ref{tmp}):
\begin{align*}
	&\Tr_{A} \left[ \left(\sum_{i=0}^{3} \ket{ii}_{AB}\bra{ii} \otimes \mathbb{1}_{A^{\prime}}\right) \left( \bigoplus_{j,k=0}^{2} \sqrt{p_{j}p_{k}} \ket{k}_{A}\bra{j} \otimes \ket{j}_{A^{\prime}}\bra{k} \otimes \mathbb{1}_{B} \right) \right] = \\& \bigoplus_{i=0}^{3}\bigoplus_{j,k=0}^{2} \sqrt{p_{j}p_{k}} \ket{i}_{B}\bra{i} \otimes \ket{j}_{A^{\prime}}\bra{k} \Tr_{A}(\ket{i}_{A}\bra{i}\ket{k}_{A}\bra{j}) \underset{\text{for all }j,k}{\overset{p_{j} = p_{k} = \frac{1}{3}}{=}} \frac{1}{3} \bigoplus_{i=0}^{2} \ket{ii}_{BA^{\prime}}\bra{ii} \equiv \varrho_{P,P}
	\label{PP}\numberthis
\end{align*}

\subsubsection{Retrieving}

Again, we are going to use the same retrieving instrument as before, expressed in equation (\ref{retr instr}). Let us calculate:
\begin{align*}
	R_{s} \star \varrho_{F} = \Tr_{M} \left[R_{MCD}^{s} \left(\varrho_{F}^{T_{M}} \otimes \mathbb{1}_{CD }\right)\right].
\end{align*}
The first term is the same as was in the previous case for the depolarizing channel:
\begin{align*}
	\Tr_{M} \left[R_{MCD}^{s} (\varrho_{\phi,\phi}^{T_{M}} \otimes \mathbb{1}_{CD})\right] \overset{(\ref{retrUfi})}{=} \frac{2}{3} \dket{U_{\phi}}_{CD}\dbra{U_{\phi}}.
\end{align*}
The second term is a little bit more complicated. Let us begin by expressing the transposed state from equation (\ref{fip}):
\begin{align*}
	\varrho_{\phi,P}^{T_{M}} = \frac{1}{3} \bigl( \underbrace{\ket{00}_{M}\bra{00}}_{1} + \underbrace{\ket{11}_{M}\bra{11}}_{2} + \underbrace{e^{-i\phi}\ket{22}_{M}\bra{11}}_{3} + \underbrace{e^{i\phi}\ket{11}_{M}\bra{22}}_{4} + \underbrace{\ket{22}_{M}\bra{22}}_{5} \bigr)
	\label{transp fip} \numberthis
\end{align*}
And now, the core of the calculation:
\begin{align*}
	&\Tr_{M} \left[R_{MCD}^{s} (\varrho_{\phi,P}^{T_{M}} \otimes \mathbb{1}_{CD})\right] \overset{s^{(-1)} = s^{(2)} = 0}{=} \Tr_{M} \left[R_{s}^{(0)} (\varrho_{\phi,P}^{T_{M}} \otimes \mathbb{1}_{CD})\right] + \Tr_{M} \left[R_{s}^{(1)} (\varrho_{\phi,P}^{T_{M}} \otimes \mathbb{1}_{CD})\right] = \\& \frac{1}{3} \Bigl[ \bigl( \hspace*{-5mm}\underbrace{\ket{00}_{CD}\bra{00}}_{1 \text{ from (\ref{transp fip}), } a \text{ from (\ref{Rs01})}} \hspace*{-3mm}+\hspace*{3mm} \underbrace{\ket{11}_{CD}\bra{11}}_{4b} \bigr) + \bigl( \underbrace{\ket{00}_{CD}\bra{00}}_{1e} +  e^{-i\phi}\underbrace{\ket{00}_{CD}\bra{11}}_{3f} +
	e^{i\phi}\underbrace{\ket{11}_{CD}\bra{00}}_{2g} + \underbrace{\ket{11}_{CD}\bra{11}}_{4h} \bigr) \Bigr] = \\&
	\frac{1}{3} \left(\dket{U_{\phi}}_{CD}\dbra{U_{\phi}} + P_{CD}\right)
\end{align*}
We get the same result also for the third term:
\begin{align*}
	\Tr_{M} \left[R_{MCD}^{s} (\varrho_{P, \phi}^{T_{M}} \otimes \mathbb{1}_{CD})\right] = 	\frac{1}{3} \left(\dket{U_{\phi}}_{CD}\dbra{U_{\phi}} + P_{CD}\right).
\end{align*}
And finally, the last term:
\begin{align*}
	&\Tr_{M} \left[R_{MCD}^{s} (\varrho_{P, P}^{T_{M}} \otimes \mathbb{1}_{CD})\right] = \Tr_{M} \left[R_{MCD}^{s} \left(\frac{1}{3}\bigoplus_{i=0}^{2} \ket{ii}_{M}\bra{ii} \otimes \mathbb{1}_{CD}\right)\right] = \\&\frac{1}{3} \left( \bigoplus_{i=0}^{2} \bra{ii}R_{s}^{(0)}\ket{ii} + \bigoplus_{i=0}^{2} \bra{ii}R_{s}^{(1)}\ket{ii} \right) = \frac{1}{3} \bigl( \underbrace{\ket{00}_{CD}\bra{00}}_{a \text{ from (\ref{Rs01})}} + \underbrace{\ket{11}_{CD}\bra{11}}_{d} + \underbrace{\ket{00}_{CD}\bra{00}}_{e} + \underbrace{\ket{11}_{CD}\bra{11}}_{h} \bigr) = \\& \frac{2}{3} P_{CD}
\end{align*}
Putting it all together, we obtain the retrieved channel:
\begin{align*} \label{2rhoF}
	R_{s} \star \varrho_{F} &= \frac{2}{3} q^{2} \dket{U_{\phi}}_{CD}\dbra{U_{\phi}} + \frac{1}{3} q(1-q) \left(\dket{U_{\phi}}_{CD}\dbra{U_{\phi}} + P_{CD}\right) \\&+ \frac{1}{3} q(1-q) \left(\dket{U_{\phi}}_{CD}\dbra{U_{\phi}} + P_{CD}\right) + \frac{2}{3} (1-q)^{2}P_{CD} \\&= \frac{2}{3} \left[ (q^{2} + q(1-q))\dket{U_{\phi}}_{CD}\dbra{U_{\phi}} + (q(1-q) + (1-q)^{2})P_{CD} \right] \\&= \frac{2}{3} \left[q\dket{U_{\phi}}_{CD}\dbra{U_{\phi}} + (1-q) P_{CD} \right]. \numberthis
\end{align*}
From here, we can directly see that we are implementing the channel ${\cal{F}}_{\phi} = q\dket{U_{\phi}}_{CD}\dbra{U_{\phi}} + (1-q) P_{CD}$ with probability $\frac{2}{3}$.

\subsubsection{Generalization and Probabilities}

Generalization of result for $N$ number of implementations of channel ${\cal{F}}_{\phi}$ is rather straightforward as the only variable is probability of success which depends only on the number of times we have access to noisy channel $N$:
\begin{align}\label{redPhaseDamp}
	R_{s} \star \varrho_{F} = \frac{N}{N+1} \left[ q \dket{U_{\phi}}_{CD}\dbra{U_{\phi}} + (1-q) P_{CD} \right].
\end{align}
We have retrieved the same channel with which we have started, ${\cal{F}}_{\phi}$, with the success probability being $p_{suc} = \frac{N}{N+1}$. This is the same success probability as in the original work of Sedlák and Ziman, where they were considering having access to phase gates without any noise, or equivalently, for the case, when $q = 1$. The probability of success $p_{suc}$ is increasing for the rising $N$ as well as the fraction $\frac{Nq}{N+1}$ of the entire retrieved channel, which belongs to unitary part of the channel, is increasing.

\subsection{Comparison}

In this section we shall compare the robustness of PSAR to depolarizing and phase damping channels.
\begin{figure}[h!]
	\vspace*{-0cm}
	\begin{subfigure}{.65\textwidth}
		\centering
		\hspace*{-2.5cm}
		\vspace*{2cm}
		\includegraphics[width=.81\linewidth]{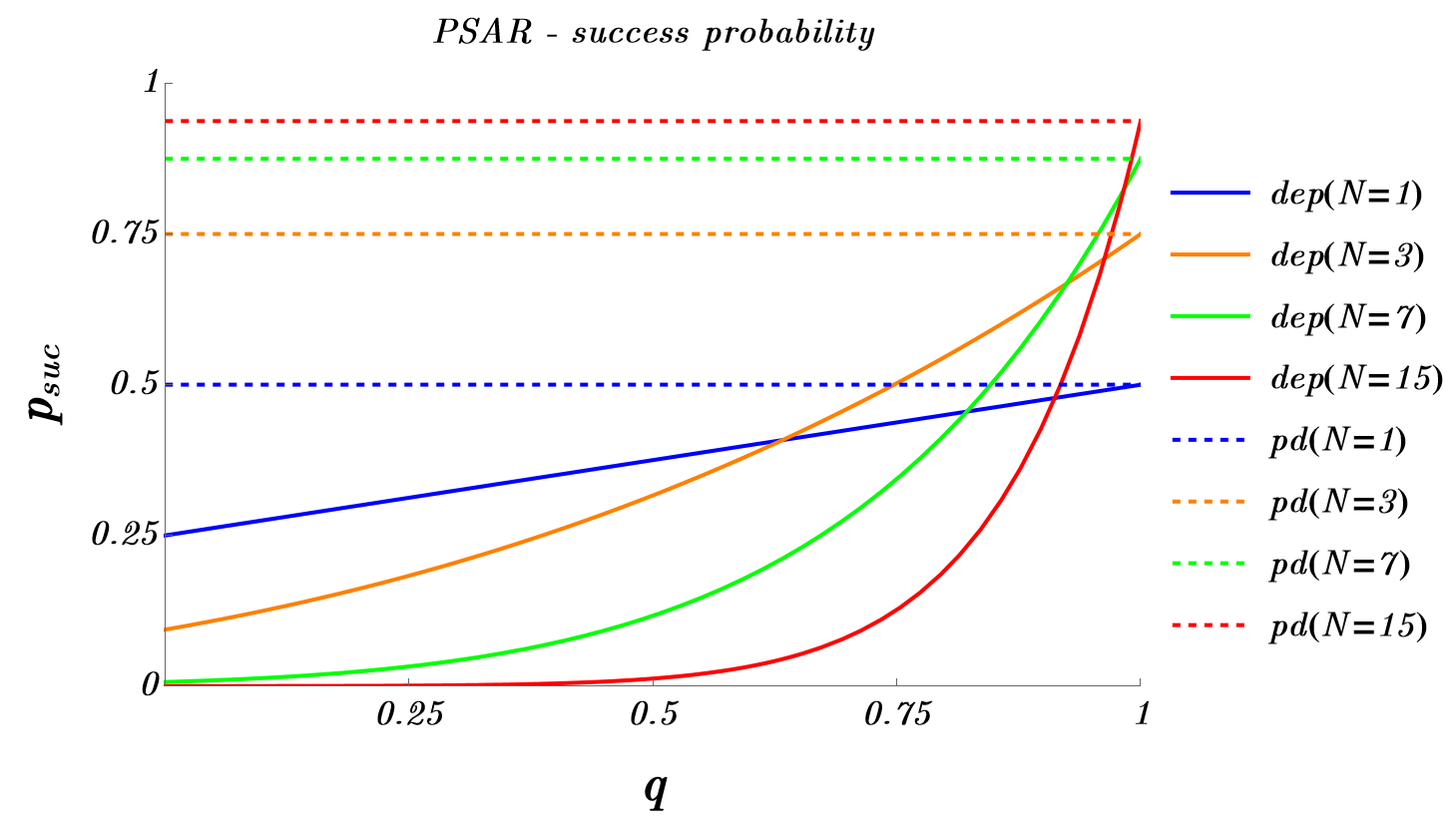}
		\vspace*{-1.5cm}
		\caption{Comparison of success probability $p_{suc}$ in case of \\depolarizing channel (solid lines) and phase damping \\(dashed lines) for $N=1$, $N=3$, $N=7$, and $N=15$.}
		\label{compPSARbetweenPDandDep:a}
	\end{subfigure}
	\hspace*{-2cm}
	\begin{subfigure}{.65\textwidth}
		\vspace*{0.3cm}
		\centering
		\hspace*{-3cm}
		\includegraphics[width=.81\linewidth]{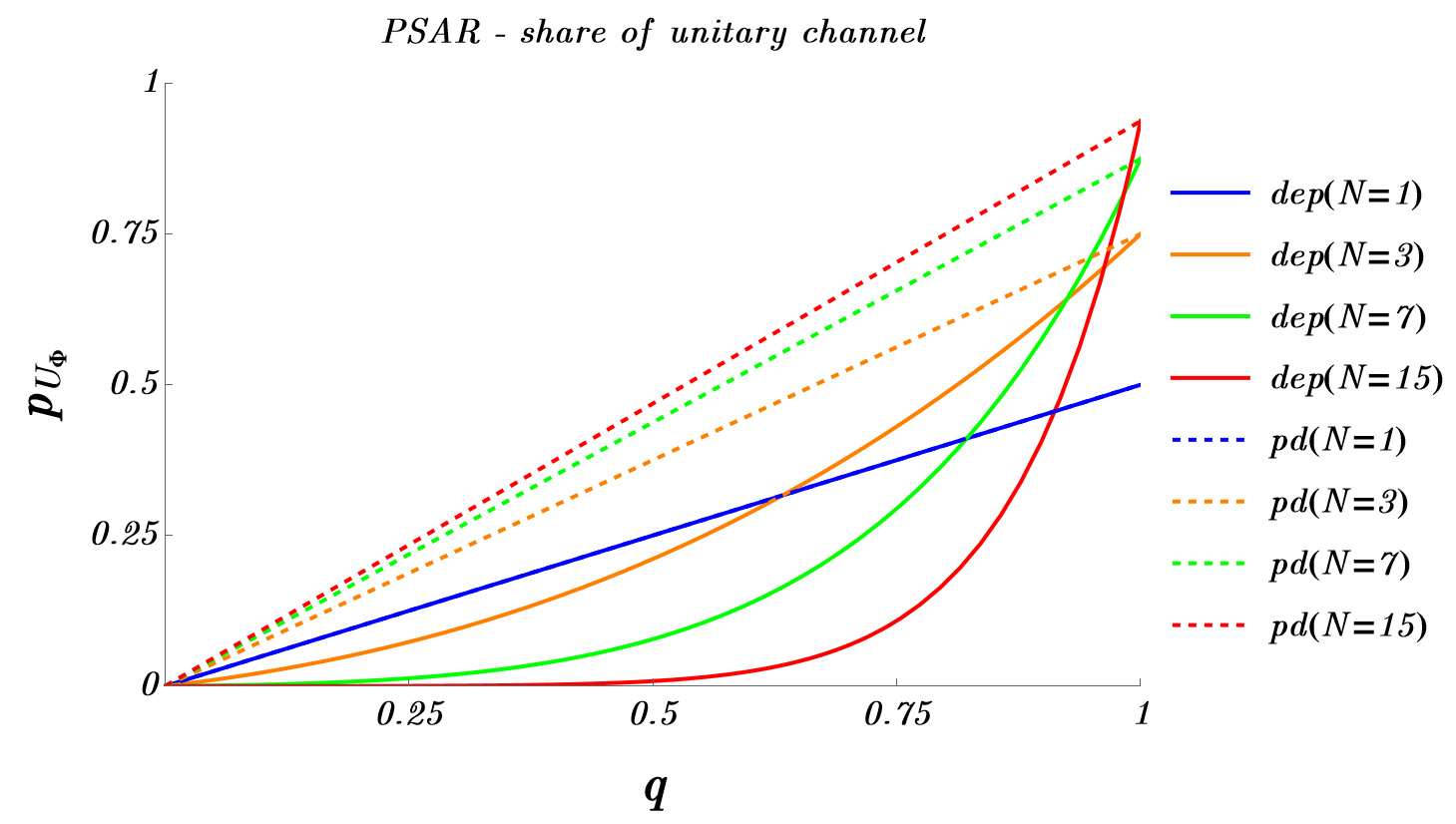}
		\caption{Ratio of unitary channel $p_{U_{\phi}}$ compared to the \\entire implemented channel  in case of depolarizing \\channel (solid lines) and phase damping (dashed \\lines) for $N=1$, $N=3$, $N=7$, and $N=15$.}
		\label{compPSARbetweenPDandDep:b}
	\end{subfigure}
	\hspace*{3cm}
	\begin{subfigure}{.6\textwidth}
		\vspace*{0.3cm}
		\centering
		\includegraphics[width=.8\linewidth]{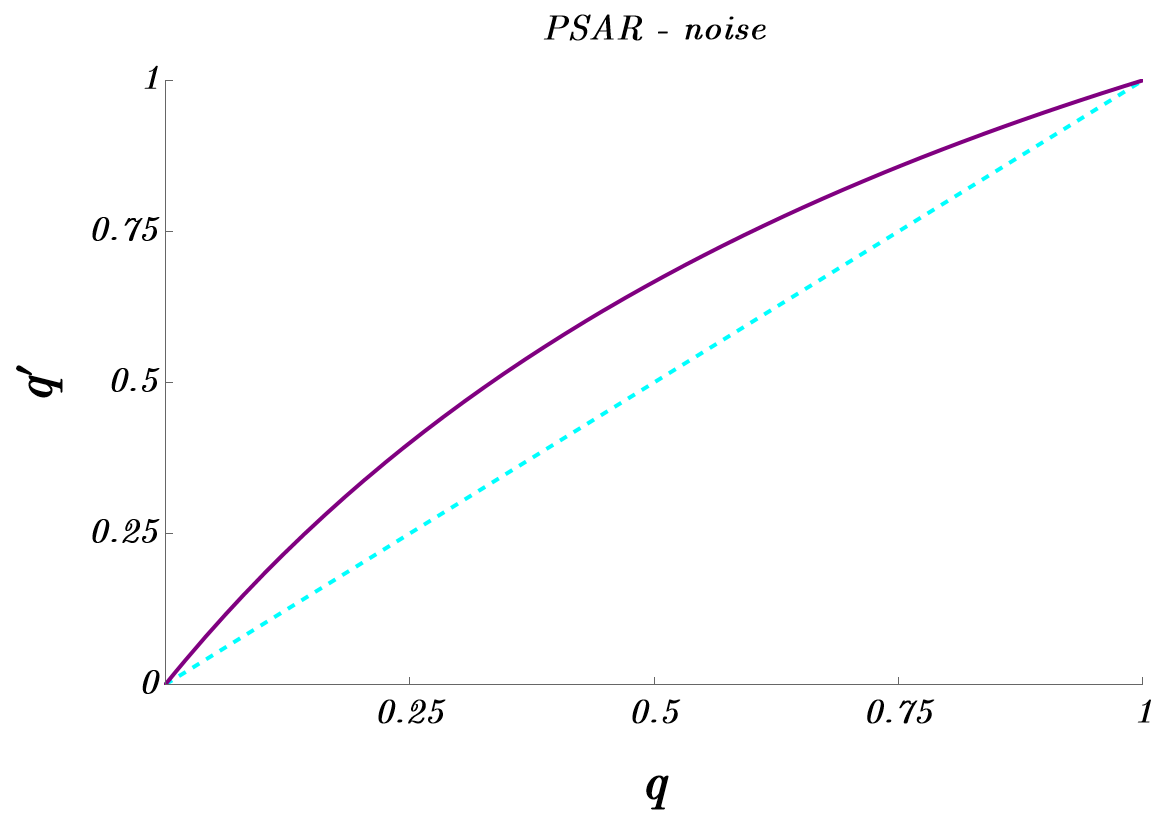}
		\caption{Comparison of noise degrees before storage phase $q$ with the noise present after the retrieval $q^{\prime}$ for depolarizing channel (solid line) and phase damping (dashed line).}
		\label{compPSARbetweenPDandDep:c}
	\end{subfigure}
	\vspace*{0.4cm}
	\caption{In the legends of respective figures, "\emph{pd}" denotes phase damping, while "\emph{dep}" denotes depolarization.}
	\label{compPSARbetweenPDandDep}
\end{figure}
In the figure \ref{compPSARbetweenPDandDep:a} there is a comparison of success probability for implementing depolarizing channel (solid lines) and phase damping (dashed lines) for $N=1$, $N=3$, $N=7$, and $N=15$. We can see that the success probability is always better in case of phase damping except for $q=1$.  We can also see that success probability for depolarization depends also on mixing parameter $q$, which is not the case for phase damping. Therefore, the resilience of PSAR optimized for implementing phase gates is higher against noise caused by phase damping than against the one caused by depolarization. Succes probability for phase damping goes to $1$, while for depolarization the success probability is diminishing for smaller values of $q$ with growing $N$, but for every value $N$ there exists an interval of $q$ for which the success probability is increasing.

In the figure \ref{compPSARbetweenPDandDep:b} we have depicted a comparison between depolarizing and phase damping channels for how big of a share of the entire implemented channel constitutes a unitary channel $p_{U_{\phi}}$ for the same values of $N$ as in the previous figure. Value for depolarization is $p_{U_{\phi}}(dep) = \frac{Nq}{N+1} (\frac{1+q}{2})^{N-1}$ and for phase damping $p_{U_{\phi}}(pd) = \frac{Nq}{N+1}$ as can be seen from equations (\ref{resWhiteNoise}) and (\ref{redPhaseDamp}) respectively. We can see that, except for $N=1$ and $q=1$, the PSAR is closer to implementation of pure unitary channel in case of using channel mixed with phase damping than with depolarization. For depolarizing channel, for every value of $N$ there exists an interval of high value $q$ where the implementation of unitary channel is more successful with growing $N$. Albeit, this region is also diminishing with growing $N$.

Figure \ref{compPSARbetweenPDandDep:c} shows a comparison between degrees of noise constituting both of original channels applied to input state versus degrees of noise still present in the channel after retrieval. This relation is independent of number of uses of channels $N$. In case of dephasing, the noise remains the same after the retrieval. However, for depolarizing noise, $q^{\prime} = 2q/(1+q) \geq q$. Therefore, we shall conclude that PSAR decreases depolarizing noise, and this decrease is not dependent on the number of uses of the noisy channel. It is because the evolution of depolarizing channel can be written as \cite{PreskillLectureNotes}:
\begin{align*}
	\varrho \mapsto \varrho^{\prime} = q\varrho + \frac{1-q}{3} (\sigma_{x} \varrho \sigma_{x} + \sigma_{y} \varrho \sigma_{y} + \sigma_{z} \varrho \sigma_{z}).
\end{align*}
And the retrieving instrument $R_{s}$ disregards the contributions from $\sigma_{x}$ and $\sigma_{y}$ and effectively changes the depolarizing noise to phase damping. This is because Choi matrices of $\varsigma_{x} = \sigma_{x} \cdot \sigma_{x}$ and $\varsigma_{y} = \sigma_{y} \cdot \sigma_{y}$ are:
\begin{align*}
	(\varsigma_{x} \otimes \mathcal{I}) \sum_{ij} \dyad{jj}{ii} &= \dyad{10}{10} + \dyad{10}{01} + \dyad{01}{10} + \dyad{01}{01} \\
	(\varsigma_{y} \otimes \mathcal{I}) \sum_{ij} \dyad{jj}{ii} &= \dyad{10}{10} - \dyad{10}{01} - \dyad{01}{10} + \dyad{01}{01},
\end{align*}
and because there are no terms with $\ket{10}$ or $\ket{01}$ within the retrieving instrument $R_{s}$, as can be seen from equations (\ref{retr instr}) and (\ref{Jretr}), the contributions from $\sigma_{x}$ and $\sigma_{y}$ errors are disregarded.


\subsection{Implementations}

In this section we shall investigate the proposed implementations of PSAR optimized for phase gates in the work \cite{ProbabilisticStorageAndRetrievalOfQubitPhaseGates}.

\subsubsection{Vidal-Masanes-Cirac}

First implementation of phase-gate learning is through the Vidal-Masanes-Cirac protocol \cite{StoringQuantumDynamicsInQuantumStatesAStochasticProgrammableGate}. This device is depicted for general case $N = 2^{k}-1$ of implementing unitary channel in the figure \ref{VMCimage}. The idea is to recycle the unsuccessful result and thus improve the total probability of success. In case of successful measurement, the protocol terminates while in case of unsuccessful one, the corresponding state is reused as the input state for the next register, which possesses more gates than the previous register serving as correction mechanism to the undesired transformation which correspond to the failed measurement. Measurement corresponding to success is $\dyad{0}{0}$ due to successful implementation in case of phase-gate implementation. Thus, $\dyad{1}{1}$ corresponds to unsuccessful measurement. We shall examine how this implementation fares in case of applying noisy channels instead of the unitary ones.
\begin{figure}[H]
	\begin{center}
			\begin{quantikz}
			\ket{\xi}
			\rstick{\hspace{-0.1cm}\\~}
			&\qw
			&\ctrl{1}
			&\qw
			&\ctrl{2}
			&\qw
			&\ctrl{4}
			&\qw
			\\ 
			\ket{+}\rstick{\hspace{-0.1cm}\\~}
			&\gate{U_{\varphi}}
			&\targ{}
			&\meter{}
			&\cw
			&\cw
			&\cw
			&\cw
			\\
			\ket{+}\rstick{\hspace{-0.1cm}\\~}
			&\gate{U_{\varphi}}
			&\gate{U_{\varphi}}
			&\qw
			&\targ{}
			&\meter{}
			&\cw
			&\cw
			\\
			\vdots
			\\
			\ket{+}\rstick{\hspace{-0.1cm}\\~}
			&\gate{U_{\varphi}}
			&\qw\cdots
			&\gate{U_{\varphi}}
			&\qw
			&\qw
			&\targ{}
			&\meter{}
		\end{quantikz}
		\vspace*{0.4cm}
		\caption{Vidal-Masanes-Cirac realization scheme for arbitrary $N=2^{k}-1$ times of applying unitary channel with $k$ being the number of registers. The input states are $\ket{\xi} = a\ket{0} + b\ket{1}$ and $\ket{+} = \frac{1}{\sqrt{2}} (\ket{0} + \ket{1})$. After the successful measurement, the procedure ends, while after the failed measurement, the state is recycled and used as input in the next register with gates correcting for the error in the original case of implementing only unitary channel.}
		\label{VMCimage}
	\end{center}
\end{figure}
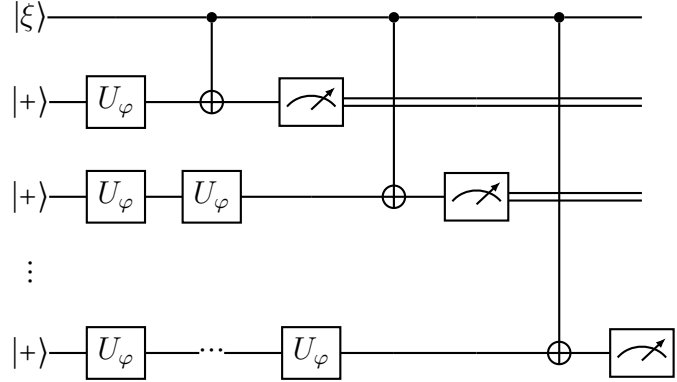

\paragraph{Depolarization}

We shall calculate the implementation for the $N = 3$ case to be able to directly compare the implementation with the equation (\ref{rhoE}). This device is pictured in the figure \ref{VMC}.
\begin{figure}[H]
	\begin{center}
		\begin{quantikz}
			\ket{\xi}
			\rstick{\hspace{-0.1cm}\\~}
			&\qw
			&\ctrl{1}
			&\qw
			&\ctrl{2}
			&\qw
			\\ 
			\ket{+}\rstick{\hspace{-0.1cm}\\~}
			&\gate{E_{\varphi}}
			&\targ{}
			&\meter{}
			&\cw
			&\cw
			\\
			\ket{+}\rstick{\hspace{-0.1cm}\\~}
			&\gate{E_{\varphi}}
			&\gate{E_{\varphi}}
			&\qw
			&\targ{}
			&\meter{}
		\end{quantikz}
	\end{center}
	\caption[font=small]{Vidal-Masanes-Cirac realization scheme for $N=3$ case, when we are applying channel ${\cal{E}}_{\phi}$ three times. Input states are $\ket{\xi} = a\ket{0} + b\ket{1}$ and $\ket{+} = \frac{1}{\sqrt{2}} (\ket{0} + \ket{1})$.}
	\label{VMC}
\end{figure}

Firstly, we have to calculate what happens when we are implementing channel ${\cal{E}}_{\phi}$ on the state $\ket{+} = \frac{1}{\sqrt{2}} (\ket{0} + \ket{1})$ only once:
\begin{align*} \label{E++}
	&{\cal{E}}_{\phi} (\dyad{+}{+}) = q {\cal{U}}_{\phi} (\dyad{+}{+}) + (1-q) {\cal{C}}_{\mathbb{1}/2} (\dyad{+}{+}) \\&= q (\dyad{0}{0} + e^{i\phi} \dyad{1}{1}) \frac{1}{2} (\dyad{0}{0} + \dyad{0}{1} + \dyad{1}{0} + \dyad{1}{1}) (\dyad{0}{0} + e^{-i\phi} \dyad{1}{1}) \\&+ (1-q) \frac{1}{2} (\dyad{0}{0} + \dyad{1}{1}) \\&= \frac{1}{2} \left[q (\dyad{0}{0} + e^{-i\phi}\dyad{0}{1} + e^{i\phi}\dyad{1}{0} + \dyad{1}{1}) + (1-q) (\dyad{0}{0} + \dyad{1}{1})\right] \numberthis
\end{align*}
Now, let us calculate the tensor product of $\dyad{\xi}{\xi} = a^{2} \dyad{0}{0} + ab^{\ast} \dyad{0}{1} + a^{\ast}b \dyad{1}{0} + b^{2} \dyad{1}{1}$ with ${\cal{E}}_{\phi} (\dyad{+}{+})$:
\begin{align*} \label{xiE++}
	&\dyad{\xi}{\xi} \otimes {\cal{E}}_{\phi} (\dyad{+}{+}) \\&= \frac{1}{2} \big\{q \big[a^{2} (\dyad{00}{00} + e^{-i\phi}\dyad{00}{01} + e^{i\phi}\dyad{01}{00} + \dyad{01}{01}) \\&+ ab^{\ast} (\dyad{00}{10} + e^{-i\phi} \dyad{00}{11} + e^{i\phi} \dyad{01}{10} + \dyad{01}{11}) \\&+ a^{\ast}b (\dyad{10}{00} + e^{-i\phi} \dyad{10}{01} + e^{i\phi} \dyad{11}{00} + \dyad{11}{01}) \\&+ b^{2} (\dyad{10}{10} + e^{-i\phi} \dyad{10}{11} + e^{i\phi} \dyad{11}{10} + \dyad{11}{11})\big] \\&+ (1-q) \big[a^{2} (\dyad{00}{00} + \dyad{01}{01}) + ab^{\ast} (\dyad{00}{10} + \dyad{01}{11}) \\&+ a^{\ast}b (\dyad{10}{00} + \dyad{11}{01}) + b^{2} (\dyad{10}{10} + \dyad{11}{11})\big]\big\} \numberthis
\end{align*}
And we have to also apply controlled $NOT$ gate on the previous result with the first qubit being the control one and the second qubit being the target qubit:
\begin{align*} \label{CNOT1}
	&CNOT \left[\dyad{\xi}{\xi} \otimes {\cal{E}}_{\phi} (\dyad{+}{+})\right] \\&= \frac{1}{2} \big\{q \big[a^{2} (\dyad{00}{00} + e^{-i\phi}\dyad{00}{01} + e^{i\phi}\dyad{01}{00} + \dyad{01}{01}) \\&+ ab^{\ast} (\dyad{00}{11} + e^{-i\phi} \dyad{00}{10} + e^{i\phi} \dyad{01}{11} + \dyad{01}{10}) \\&+ a^{\ast}b (\dyad{11}{00} + e^{-i\phi} \dyad{11}{01} + e^{i\phi} \dyad{10}{00} + \dyad{10}{01}) \\&+ b^{2} (\dyad{11}{11} + e^{-i\phi} \dyad{11}{10} + e^{i\phi} \dyad{10}{11} + \dyad{10}{10})\big] \\&+ (1-q) \big[a^{2} (\dyad{00}{00} + \dyad{01}{01}) + ab^{\ast} (\dyad{00}{11} + \dyad{01}{10}) \\&+ a^{\ast}b (\dyad{11}{00} + \dyad{10}{01}) + b^{2} (\dyad{11}{11} + \dyad{10}{10})\big] \big\} \\&\overset{(i)}{=} \frac{1}{2} \Big\{q \big[(a^{2} \dyad{0}{0} + ab^{\ast} e^{-i\phi} \dyad{0}{1} + a^{\ast}b e^{i\phi} \dyad{1}{0} + b^{2} \dyad{1}{1}) \otimes \dyad{0}{0} \\&+ (a^{2} \dyad{0}{0} + ab^{\ast} e^{i\phi} \dyad{0}{1} + a^{\ast}b e^{-i\phi} \dyad{1}{0} + b^{2} \dyad{1}{1}) \otimes \dyad{1}{1}  \big] \\&+ (1-q) \left[(a^{2} \dyad{0}{0} + b^{2} \dyad{1}{1}) \otimes \dyad{0}{0} + (a^{2} \dyad{0}{0} + b^{2} \dyad{1}{1}) \otimes \dyad{1}{1}\right]\Big\} \\&= \frac{1}{2} \Big\{\left[q U_{\phi} \dyad{\xi}{\xi} U_{\phi}^{\dagger} + (1-q) (a^{2} \dyad{0}{0} + b^{2} \dyad{1}{1})\right] \otimes \dyad{0}{0} \\&+ \left[q U_{-\phi} \dyad{\xi}{\xi} U_{-\phi}^{\dagger} + (1-q) (a^{2} \dyad{0}{0} + b^{2} \dyad{1}{1})\right] \otimes \dyad{1}{1}\Big\}, \numberthis
\end{align*}
where in $(i)$ we are discarding non-diagonal states as we only measure states $\ket{0}$ and $\ket{1}$. Therefore, the probability of successful implementation with Vidal-Masanes-Cirac scheme in case of $N = 1$ is $\frac{1}{2}$ and is in agreement with our calculation for the PSAR device. We proceed further and recycle the unsuccessful result and calculate what happens then. Let us calculate what happens if we apply channel ${\cal{E}}_{\phi}$ twice on $\ket{+}$:
\begin{align*}\label{tensoredStates}
	&{\cal{E}}_{\phi} \left[{\cal{E}}_{\phi} (\dyad{+}{+})\right] = q {\cal{U}}_{\phi} \left[q {\cal{U}}_{\phi}(\dyad{+}{+}) + (1-q) {\cal{C}}_{\mathbb{1}/2}(\dyad{+}{+})\right] + (1-q) {\cal{C}}_{\mathbb{1}/2} \left[{\cal{E}}_{\phi}(\dyad{+}{+})\right] \\&= q^{2} {\cal{U}}_{\phi}[{\cal{U}}_{\phi}(\dyad{+}{+})] + q(1-q) {\cal{U}}_{\phi}[{\cal{C}}_{\mathbb{1}/2}(\dyad{+}{+})] + (1-q) {\cal{C}}_{\mathbb{1}/2} \left[{\cal{E}}_{\phi}(\dyad{+}{+})\right] \\&= q^{2} (\dyad{0}{0} + e^{i\phi} \dyad{1}{1}) \frac{1}{2} (\dyad{0}{0} + e^{-\phi} \dyad{0}{1} + e^{i\phi} \dyad{1}{0} + \dyad{1}{1}) (\dyad{0}{0} + e^{-i\phi} \dyad{1}{1}) \\&+ q(1-q) (\dyad{0}{0} + e^{i\phi} \dyad{1}{1}) \frac{1}{2} (\dyad{0}{0} + \dyad{1}{1}) (\dyad{0}{0} + e^{-i\phi} \dyad{1}{1}) + (1-q) \frac{1}{2} (\dyad{0}{0} + \dyad{1}{1}) \\&= q^{2} \frac{1}{2} (\dyad{0}{0} + e^{-2i\phi} \dyad{0}{1} + e^{2i\phi} \dyad{1}{0} + \dyad{1}{1}) + \left[q(1-q) + (1-q)\right] \frac{1}{2} (\dyad{0}{0} + \dyad{1}{1}) \\&= q^{2} U_{2\phi} \dyad{\xi}{\xi} U_{2\phi}^{\dagger} + \frac{1-q^{2}}{2} \mathbb{1}. \numberthis
\end{align*}
Let us take the state corresponding to the failed measurement from the equation (\ref{CNOT1}):
\begin{align*}  \label{StateRecycled}
	\varrho_{recycled} &\equiv \frac{1}{2} \left[q U_{-\phi} \dyad{\xi}{\xi} U_{-\phi}^{\dagger} + (1-q) (a^{2} \dyad{0}{0} + b^{2} \dyad{1}{1})\right] \otimes \left[q^{2} U_{2\phi} \dyad{+}{+} U_{2\phi}^{\dagger} + \frac{1-q^{2}}{2} \mathbb{1}\right] \\& = \left[q \frac{1}{2} (a^{2} \dyad{0}{0} + e^{i\phi} ab^{\ast} \dyad{0}{1} + e^{-i\phi} a^{\ast}b \dyad{1}{0} + b^{2} \dyad{1}{1}) + \frac{1-q}{2} (a^{2}\dyad{0}{0} + b^{2} \dyad{1}{1})\right] \\&\otimes \left[\frac{q^{2}}{2} (\dyad{0}{0} + e^{-2i\phi} \dyad{0}{1} + e^{i2\phi} \dyad{1}{0} + \dyad{1}{1}) + \frac{1-q^{2}}{2} (\dyad{0}{0} + \dyad{1}{1})\right] \\&= \frac{q^{3}}{4} \big[a^{2} (\dyad{00}{00} + e^{-i2\phi} \dyad{00}{01} + e^{i2\phi} \dyad{01}{00} + \dyad{01}{01}) \\&+ ab^{\ast} (e^{i\phi} \dyad{00}{10} + e^{-i\phi} \dyad{00}{11} + e^{i3\phi}\dyad{01}{10} + e^{i\phi} \dyad{01}{11}) \\&+ a^{\ast}b (e^{-i\phi}\dyad{10}{00} + e^{-i3\phi} \dyad{10}{01} + e^{i\phi} \dyad{11}{00} + e^{-i\phi} \dyad{11}{01}) \\&+ b^{2} (\dyad{10}{10} + e^{-i2\phi} \dyad{10}{11} + e^{i2\phi} \dyad{11}{10} + \dyad{11}{11})\big] \\&+ \frac{1}{4} q(1-q^{2}) \big[a^{2} (\dyad{00}{00} + \dyad{01}{01}) + ab^{\ast} (e^{i\phi} \dyad{00}{10} + e^{i\phi} \dyad{01}{11}) \\&+ a^{\ast}b (e^{-i\phi} \dyad{10}{00} + e^{-i\phi} \dyad{11}{01}) + b ^{2} (\dyad{10}{10} + \dyad{11}{11})\big] \\&+ \frac{1}{4}q^{2}(1-q) \big[a^{2} (\dyad{00}{00} + e^{-i2\phi} \dyad{00}{01} + e^{i2\phi} \dyad{01}{00} + \dyad{01}{01}) \\&+ b^{2} (\dyad{10}{10} + e^{-2i\phi} \dyad{10}{11} + e^{i2\phi} \dyad{11}{10} + \dyad{11}{11})\big] \\&+ \frac{1}{4} (1-q)(1-q^{2}) \left[a^{2}(\dyad{00}{00} + \dyad{01}{01}) + b^{2} (\dyad{10}{10} + \dyad{11}{11})\right], \numberthis
\end{align*}
where we have denoted the resulting state as $\varrho_{recycled}$. Now we only have to apply $CNOT$ gate on the previous state:
\begin{align*} \label{CNOT3}
	&CNOT\left\{\varrho_{recycled} \otimes {\cal{E}}_{\phi}\left[{\cal{E}}_{\phi}(\dyad{+}{+})\right]\right\} \\&= \frac{q^{3}}{4} \big[a^{2} (\dyad{00}{00} + e^{-i2\phi} \dyad{00}{01} + e^{i2\phi} \dyad{01}{00} + \dyad{01}{01}) \\&+ ab^{\ast} (e^{i\phi} \dyad{00}{11} + e^{-i\phi} \dyad{00}{10} + e^{i3\phi}\dyad{01}{11} + e^{i\phi} \dyad{01}{10}) \\&+ a^{\ast}b (e^{-i\phi}\dyad{11}{00} + e^{-i3\phi} \dyad{11}{01} + e^{i\phi} \dyad{10}{00} + e^{-i\phi} \dyad{10}{01}) \\&+ b^{2} (\dyad{11}{11} + e^{-i2\phi} \dyad{11}{10} + e^{i2\phi} \dyad{10}{11} + \dyad{10}{10})\big] \\&+ \frac{1}{4} q(1-q^{2}) \big[a^{2} (\dyad{00}{00} + \dyad{01}{01}) + ab^{\ast} (e^{i\phi} \dyad{00}{11} + e^{i\phi} \dyad{01}{10}) \\&+ a^{\ast}b (e^{-i\phi} \dyad{11}{00} + e^{-i\phi} \dyad{10}{01}) + b ^{2} (\dyad{11}{11} + \dyad{10}{10})\big] \\&+ \frac{1}{4}q^{2}(1-q) \big[a^{2} (\dyad{00}{00} + e^{-i2\phi} \dyad{00}{01} + e^{i2\phi} \dyad{01}{00} + \dyad{01}{01}) \\&+ b^{2} (\dyad{11}{11} + e^{-2i\phi} \dyad{11}{10} + e^{i2\phi} \dyad{10}{11} + \dyad{10}{10})\big] \\&+ \frac{1}{4} (1-q)(1-q^{2}) \left[a^{2}(\dyad{00}{00} + \dyad{01}{01}) + b^{2} (\dyad{11}{11} + \dyad{10}{10})\right] \\&\overset{(i)}{=} \frac{q^{3}}{4} \left[a^{2} \dyad{0}{0} + ab^{\ast} e^{-i\phi} \dyad{1}{0} + a^{\ast}b e^{i\phi} \dyad{1}{0} + b^{2} \dyad{1}{1}\right] \otimes \dyad{0}{0} \\&+ \frac{q^{3}}{4} \left[a^{2} \dyad{0}{0} + ab^{\ast} e^{i3\phi} \dyad{0}{1} + a^{\ast}b e^{-i3\phi} \dyad{1}{0} + b^{2} \dyad{1}{1} \right] \otimes \dyad{1}{1} \\&+ \frac{q}{4}(1-q^{2}) (a^{2} \dyad{0}{0} + b^{2} \dyad{1}{1}) \otimes \dyad{0}{0} + \frac{q}{4}(1-q^{2}) (a^{2} \dyad{0}{0} + b^{2} \dyad{1}{1}) \otimes \dyad{1}{1} \\&+ \frac{q^{2}}{4}(1-q) (a^{2} \dyad{0}{0} + b^{2} \dyad{1}{1}) \otimes \dyad{0}{0} + \frac{q^{2}}{4}(1-q) (a^{2} \dyad{0}{0} + b^{2} \dyad{1}{1}) \otimes \dyad{1}{1} \\&+ \frac{1}{4} (1-q)(1-q^{2}) (a^{2} \dyad{0}{0} + b^{2} \dyad{1}{1}) \otimes \dyad{0}{0} + \frac{1}{4} (1-q)(1-q^{2}) (a^{2} \dyad{0}{0} + b^{2} \dyad{1}{1}) \otimes \dyad{1}{1} \\&= \frac{q^{3}}{4} U_{\phi} \dyad{\xi}{\xi} U_{\phi}^{\dagger} \otimes \dyad{0}{0} + \left[\frac{1}{4} q(1-q^{2}) + \frac{1}{4} q^{2}(1-q) + \frac{1}{4} (1-q)(1-q^{2})\right] (a^{2} \dyad{0}{0} + b^{2} \dyad{1}{1}) \otimes \dyad{0}{0} \\&+ \frac{q^{3}}{4} U_{-3\phi} \dyad{\xi}{\xi} U_{-3\phi}^{\dagger} \otimes \dyad{1}{1} + \left[\frac{1}{4} q(1-q^{2}) + \frac{1}{4} q^{2}(1-q) + \frac{1}{4} (1-q)(1-q^{2})\right] (a^{2} \dyad{0}{0} + b^{2} \dyad{1}{1}) \otimes \dyad{1}{1} \\&= \frac{1}{4}  \Big\{ \left[q^{3} U_{\phi} \dyad{\xi}{\xi} U_{\phi}^{\dagger} + (1-q^{3}) (a^{2} \dyad{0}{0} + b^{2} \dyad{1}{1})\right] \otimes \dyad{0}{0} \\&+ \left[q^{3} U_{-3\phi} \dyad{\xi}{\xi} U_{-3\phi}^{\dagger} + (1-q^{3}) (a^{2} \dyad{0}{0} + b^{2} \dyad{1}{1})\right] \otimes \dyad{1}{1} \Big\}. \numberthis
\end{align*}
In $(i)$ we are discarding the states that are not on the diagonal as they do not add anything due to the measurement we employ. Here we obtain probability of success being $\frac{1}{4}$. To get the final probability of successful implementation, we must sum up the probabilities that correspond to successful measurement in case of successful measurement with one register from equation (\ref{CNOT1}) with the probability of successful measurement with two registers from equation (\ref{CNOT3}): $\frac{1}{2} + \frac{1}{4} = \frac{3}{4}$. Therefore, the probability of successful measurement differs from equation (\ref{rhoE}) for the PSAR device, where it depends also on mixing parameter $q$. Another difference here is, that if we succeed with a measurement on the first try, with only one register used, we implement different channel $qU_{\phi}\dyad{\xi}{\xi}U_{\phi}^{\dagger} + (1-q) (a^{2} \dyad{0}{0} + b^{2} \dyad{1}{1})$ compared to the case if we succeed with the measurement when we are using two registers $q^{3}U_{\phi}\dyad{\xi}{\xi}U_{\phi}^{\dagger} + (1-q^{3}) (a^{2} \dyad{0}{0} + b^{2} \dyad{1}{1})$.

\paragraph{Phase Damping}
Let us repeat the same procedure also for the phase damping where we are implementing noisy channel ${\cal{F}}_{\phi}$ three times as depicted in the figure \ref{VMCpd}.
\begin{figure}[H]
	\begin{center}
		\begin{quantikz}
			\ket{\xi}
			\rstick{\hspace{-0.1cm}\\~}
			&\qw
			&\ctrl{1}
			&\qw
			&\ctrl{2}
			&\qw
			\\ 
			\ket{+}\rstick{\hspace{-0.1cm}\\~}
			&\gate{F_{\phi}}
			&\targ{}
			&\meter{}
			&\cw
			&\cw
			\\
			\ket{+}\rstick{\hspace{-0.1cm}\\~}
			&\gate{F_{\phi}}
			&\gate{F_{\phi}}
			&\qw
			&\targ{}
			&\meter{}
		\end{quantikz}
	\end{center}
	\caption[font=small]{Vidal-Masanes-Cirac realization scheme for $N=3$ case, when we are applying channel ${\cal{F}}_{\phi}$ three times. The input states are $\ket{\xi} = a\ket{0} + b\ket{1}$ and $\ket{+} = \frac{1}{\sqrt{2}} (\ket{0} + \ket{1})$.}
	\label{VMCpd}
\end{figure}
Firstly, we have to calculate what happens if we implement channel ${\cal{F}}_{\phi}$ once:
\begin{align*}
	{\cal{F}}_{\phi}(\dyad{+}{+}) = q {\cal{U}}_{\phi} (\dyad{+}{+}) + (1-q) {\cal{P}}(\dyad{+}{+}) \overset{(\ref{E++})}{=} {\cal{E}}_{\phi}(\dyad{+}{+}).
\end{align*}
And because the previous equation is true, also if we apply $CNOT$ on $	\dyad{\xi}{\xi} \otimes {\cal{F}}_{\phi}(\dyad{+}{+})$ we obtain the same result as in the case of depolarization:
\begin{align*}
	CNOT \left[\dyad{\xi}{\xi} \otimes {\cal{F}}_{\phi} (\dyad{+}{+})\right] \overset{(\ref{CNOT1})}{=} CNOT \left[\dyad{\xi}{\xi} \otimes {\cal{E}}_{\phi} (\dyad{+}{+})\right].
\end{align*} 
Now, let us calculate what happens if we apply channel ${\cal{F}}_{\phi}$ twice on the state $\ket{+}$:
\begin{align*}
	&{\cal{F}}_{\phi} \left[{\cal{F}}_{\phi}(\dyad{+}{+})\right] = q {\cal{U}}_{\phi} \left[q {\cal{U}}_{\phi} (\dyad{+}{+}) + (1-q) {\cal{P}}(\dyad{+}{+})\right] + (1-q) {\cal{P}}\left[{\cal{F}}_{\phi}(\dyad{+}{+})\right] \\&= q^{2} {\cal{U}}_{\phi} \left[{\cal{U}}_{\phi}(\dyad{+}{+})\right] + \left[q(1-q) + (1-q)\right] {\cal{P}}(\dyad{+}{+}) = q^{2} U_{2\phi} \dyad{\xi}{\xi} U_{2\phi}^{\dagger} + \frac{1-q^{2}}{2} \mathbb{1} \\&\overset{(\ref{tensoredStates})}{=} {\cal{E}}_{\phi} \left[{\cal{E}}_{\phi}(\dyad{+}{+})\right], 
\end{align*}
where ${\cal{P}}(\dyad{+}{+}) = \frac{1}{2} ({\cal{I}} + \varsigma_{z}) (\dyad{+}{+}) = \frac{1}{2} (\dyad{0}{0} + \dyad{1}{1})$ and therefore also ${\cal{U}}_{\phi}\left[{\cal{P}}(\dyad{+}{+})\right] = \\{\cal{P}}\left[{\cal{F}}(\dyad{+}{+})\right] = \frac{1}{2} \mathbb{1}$. Because the resulting state is the same as was in the depolarizing case, also the further steps will not divert from depolarization. Therefore, we are implementing the same final operation as in the depolarization case, where the state $\varrho_{recycled}$ is from the equation (\ref{StateRecycled}):
\begin{align*}
	&CNOT \left[\varrho_{recycled} \otimes {\cal{F}}_{\phi}\left({\cal{F}}_{\phi}(\dyad{+}{+})\right)\right] \overset{(\ref{CNOT3})}{=} CNOT \left[\varrho_{recycled} \otimes {\cal{E}}_{\phi}\left({\cal{E}}_{\phi}(\dyad{+}{+})\right)\right].
\end{align*}
Therefore, also after putting up the probabilities together the successful measurements in case of implementing the channel only once and in case of repairing it once for the failed measurement we are obtaining the same success probability of $\frac{3}{4}$, which is in agreement with the probability from equation (\ref{redPhaseDamp}) for the PSAR device.

\paragraph{Generalization}
We shall generalize the implementation using Vidal-Masanes-Cirac scheme, which is rather straightforward process because the only variables are the normalization of the entire state, powers of $q^{N}$ and $(1-q)^{N}$ and phase factor in case of unitary transformation corresponding to failed measurement $U_{-N\phi}\dyad{\xi}{\xi}U_{-N\phi}^{\dagger}$. In case of $N$ gates used we obtain:
\begin{align*} \label{genVMC}
	&\frac{1}{2^{k}} \big\{ \left[q^{N} U_{\phi} \dyad{\xi}{\xi} U_{\phi}^{\dagger} + (1-q^{N}) (a^{2} \dyad{0}{0} + b^{2} \dyad{1}{1})\right] \otimes \dyad{0}{0} \\&+ \left[q^{N} U_{-N\phi} \dyad{\xi}{\xi} U_{-N\phi}^{\dagger} + (1-q^{N}) (a^{2} \dyad{0}{0} + b^{2} \dyad{1}{1})\right] \otimes \dyad{1}{1} \big\}, \numberthis
\end{align*}
with $k$ being number of registers in the scheme and with the relation to number of total gates used being $N = 2^{k} - 1$. With the success probability being $\frac{1}{2} + \frac{1}{4} + \cdot + \frac{1}{2^{k}} = \frac{N}{N+1}$, i.e., the same as for the original case of implementing only phase gate. Which is also the same as PSAR device implementing noisy channel with phase damping as in equation (\ref{redPhaseDamp}) but different compared to success probability of implementing noisy channel with depolorazition as in equation (\ref{resWhiteNoise}). Also, the implemented channel is noisier and noisier with each correction, because with each new implementation the term next to unitarily transformed state in equation (\ref{genVMC}) is diminishing because $q \leq 1$.

\subsubsection{Virtual Qudit}
In this section, we shall describe one more implementation that is heavily inspired by \cite{ProbabilisticStorageAndRetrievalOfQubitPhaseGates} but is slightly modified from the original. The scheme is depicted in the figure \ref{ancImpl}. As the input state we take:
\begin{align}\label{ancImpSt}
	\ket{\psi} = \frac{1}{\sqrt{N+1}} \sum_{j=0}^{N} \ket{j} \in {\cal{H}}_{A}.
\end{align}
Crucial step in this particular implementation is defining virtual qudit with the dimension being $2^{N}$. Which is different compared to the previous work \cite{ProbabilisticStorageAndRetrievalOfQubitPhaseGates}, where the dimension was $N$ but here, the depolarizing channel also touches the multiplicity spaces. Again, slight difference compared to the previous work is that the channel that maps states from ${\cal{H}}_{A}$ in the space of our virtual qudit behaves as identity everywhere. Also, the shift-down operator that takes as input a control qubit and a target qudit must be extended in the following way:
\begin{align}
	\label{shiftDown}
	{\cal{C}}_{\ominus} (\ket{c} \otimes \ket{t}) = \begin{cases}
		&\ket{c} \otimes \ket{t \ominus c} \quad if \hspace*{0.1cm} \ket{t} \in {\cal{H}}_{A},\\
		&\ket{c} \otimes \ket{t} \qquad\quad\hspace*{-2mm} otherwise,
	\end{cases} \numberthis
\end{align}
\begin{figure}[H]
	\begin{center}
		\includegraphics[width=8cm]{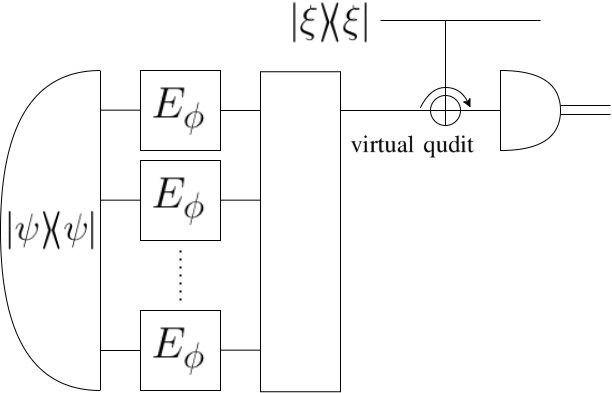}
		\vspace{1cm}
		\caption{We have chosen the implementation of noisy channel ${\cal{E}}_{\phi}^{\otimes N}$ for this image, as the difference for implementing channel ${\cal{F}}_{\phi}^{\otimes N}$ would be only in using different gates in the storing phase. Firstly, the channel is stored in the state $\ket{\psi}$ expressed in the equation (\ref{ancImpSt}). Then it is transformed to the virtual qudit by identity channel. Following this, the shift-down operator ${\cal{C}}_{\ominus}$ acts on the control qubit $\ket{\xi} = a\ket{0} + b\ket{1}$ and as target it takes the virtual qudit. In the end, the measurement in the basis of a virtual qudit is performed, with states $\ket{0}, \cdots, \ket{N-1}$ forming successful measurement.}
		\label{ancImpl}
	\end{center}
\end{figure}
where the operator behaves as identity on the multiplicity states. After the application of the shift-down operator we perform a measurement in the basis of virtual qudit where the first $N$ states $\{\ket{j}\}_{j=0}^{N-1}$ correspond to successful measurement while the remaining $2^{N} - N$ states correspond to failed measurement.

\paragraph{Depolarization}
Let us calculate what happens using virtual qudit implementation of our quantum network and therefore also uncover the result in case of the failed measurement. Input state into the scheme is:
\begin{align*} \label{impInSt}
	\ket{\psi} = \frac{1}{\sqrt{3}} (\ket{00} + \ket{01} + \ket{11}) \overset{(\ref{dictionary})}{=} \frac{1}{\sqrt{3}} (\ket{0} + \ket{1} + \ket{2}). \numberthis
\end{align*}
In our present case, we are applying the channel ${\cal{E}}_{\phi}^{\otimes2}$ on the input state:
\begin{align*}
	{\cal{E}}_{\phi}^{\otimes2}(\dyad{\psi}{\psi}) = q^{2} {\cal{U}}_{\phi}^{\otimes2} (\dyad{\psi}{\psi}) + q(1-q) \left[({\cal{U}}_{\phi} \otimes {\cal{C}}_{\mathbb{1}/2}) + ({\cal{C}}_{\mathbb{1}/2} \otimes {\cal{U}}_{\phi})\right] (\dyad{\psi}{\psi}) + (1-q)^{2} {\cal{C}}_{\mathbb{1}/2}^{\otimes2} (\dyad{\psi}{\psi}).
\end{align*}
Let us calculate the individual terms of the previous equation:
\begin{align*} \label{StInpStImpl}
	{\cal{U}}_{\phi}^{\otimes2} (\dyad{\psi}{\psi}) &= \frac{1}{3} (\dyad{00}{00} + e^{-i\phi} \dyad{00}{01} + e^{-2i\phi} \dyad{00}{11} + e^{i\phi} \dyad{01}{00} + \dyad{01}{01} \\&+ e^{-i\phi} \dyad{01}{11} + e^{2i\phi} \dyad{11}{00} + e^{i\phi} \dyad{11}{01} + \dyad{11}{11}) \\&\hspace*{-0.3cm}\overset{(\ref{dictionary})}{=} \frac{1}{3} (\dyad{0}{0} + e^{-i\phi} \dyad{0}{1} + e^{-2i\phi} \dyad{0}{2} + e^{i\phi} \dyad{1}{0} + \dyad{1}{1} + e^{-i\phi} \dyad{1}{2} \\&+ e^{2i\phi} \dyad{2}{0} + e^{i\phi} \dyad{2}{1} + \dyad{2}{2}),\\
	{\cal{C}}_{\mathbb{1}/2}^{\otimes2} (\dyad{\psi}{\psi}) &= \frac{1}{4} (\dyad{00}{00} + \dyad{01}{01} + \dyad{10}{10} + \dyad{11}{11}) \\&\hspace*{-0.3cm}\overset{(\ref{dictionary})}{=} \frac{1}{4} (\dyad{0}{0} + \dyad{1}{1} + \dyad{3}{3} + \dyad{2}{2}),\\
	({\cal{U}}_{\phi} \otimes {\cal{C}}_{\mathbb{1}/2}) (\dyad{\psi}{\psi}) &= ({\cal{U}}_{\phi} \otimes {\cal{C}}_{\mathbb{1}/2}) \bigg[\frac{1}{3} (\dyad{00}{00} + \dyad{00}{01} + \dyad{00}{11} \\&+ \dyad{01}{00} + \dyad{01}{01} + \dyad{01}{11} + \dyad{11}{00} + \dyad{11}{01} + \dyad{11}{11})\bigg] \\&= \frac{1}{3} \big[{\cal{U}}_{\phi}(\dyad{0}{0}) \otimes {\cal{C}}_{\mathbb{1}/2}(\dyad{0}{0}) + {\cal{U}}_{\phi}(\dyad{0}{0}) \otimes {\cal{C}}_{\mathbb{1}/2}(\dyad{0}{1}) \\&+ {\cal{U}}_{\phi}(\dyad{0}{1}) \otimes {\cal{C}}_{\mathbb{1}/2}(\dyad{0}{1}) + {\cal{U}}_{\phi}(\dyad{0}{0}) \otimes {\cal{C}}_{\mathbb{1}/2}(\dyad{1}{0}) \\&+ {\cal{U}}_{\phi}(\dyad{0}{0}) \otimes {\cal{C}}_{\mathbb{1}/2}(\dyad{1}{1}) + {\cal{U}}_{\phi}(\dyad{0}{1}) \otimes {\cal{C}}_{\mathbb{1}/2}(\dyad{1}{1}) \\&+ {\cal{U}}_{\phi}(\dyad{1}{0}) \otimes {\cal{C}}_{\mathbb{1}/2}(\dyad{1}{0}) + {\cal{U}}_{\phi}(\dyad{1}{0}) \otimes {\cal{C}}_{\mathbb{1}/2}(\dyad{1}{1}) \\&+ {\cal{U}}_{\phi}(\dyad{1}{1}) \otimes {\cal{C}}_{\mathbb{1}/2}(\dyad{1}{1})\big] \\&\hspace*{-0.1cm}\overset{(i)}{=} \frac{1}{3} \bigg(\dyad{0}{0} \otimes \frac{\mathbb{I}}{2} + \dyad{0}{0} \otimes \frac{\mathbb{I}}{2} + e^{-i\phi}\dyad{0}{1} \otimes \frac{\mathbb{I}}{2} + e^{i\phi}\dyad{1}{0} \otimes \frac{\mathbb{I}}{2} + \dyad{1}{1} \otimes \frac{\mathbb{I}}{2}\bigg) \\&= \frac{1}{6} [2(\dyad{00}{00} + \dyad{01}{01}) + e^{-i\phi} (\dyad{00}{10} + \dyad{01}{11}) + e^{i\phi} (\dyad{10}{00} + \dyad{11}{01}) \\&+ \dyad{10}{10} + \dyad{11}{11}] \\&\hspace*{-0.3cm}\overset{(\ref{dictionary})}{=} \frac{1}{6} [2(\dyad{0}{0} + \dyad{1}{1}) + e^{-i\phi} (\dyad{0}{3} + \dyad{1}{2}) + e^{i\phi} (\dyad{3}{0} + \dyad{2}{1}) + \dyad{3}{3} + \dyad{2}{2}],\\
	({\cal{C}}_{\mathbb{1}/2} \otimes {\cal{U}}_{\phi}) (\dyad{\psi}{\psi}) &= \frac{1}{6} \big[\dyad{0}{0} + \dyad{3}{3} + 2(\dyad{1}{1} + \dyad{2}{2}) + e^{-i\phi} (\dyad{0}{1} + \dyad{3}{2}) + e^{i\phi}(\dyad{1}{0} + \dyad{2}{3})\big]. \numberthis
\end{align*}
In $(i)$ we have used that ${\cal{C}}_{\mathbb{1}/2} (\dyad{0}{1}) = {\cal{C}}_{\mathbb{1}/2} (\dyad{1}{0}) = 0$ as can be seen from the equation (\ref{contraction}). Let us assume that the retrieved transformation should act on the state $\ket{\xi} = a\ket{0} + b\ket{1}$:
\begin{align*} \label{xiuu}
	&\dyad{\xi}{\xi} \otimes {\cal{U}}_{\phi}^{\otimes2} (\dyad{\psi}{\psi}) = \\&\frac{1}{3} \big[a^{2} (\dyad{00}{00} + e^{-i\phi} \dyad{00}{01} + e^{-2i\phi} \dyad{00}{02} + e^{i\phi} \dyad{01}{00} \\
	&+ \dyad{01}{01} + e^{-i\phi} \dyad{01}{02} + e^{2i\phi} \dyad{02}{00} + e^{i\phi} \dyad{02}{01} + \dyad{02}{02}) \\
	&+ ab^{\ast} (\dyad{00}{10} + e^{-i\phi} \dyad{00}{11} + e^{-2i\phi} \dyad{00}{12} + e^{i\phi} \dyad{01}{10} \\
	&+ \dyad{01}{11} + e^{-i\phi} \dyad{01}{12} + e^{2i\phi} \dyad{02}{10} + e^{i\phi} \dyad{02}{11} + \dyad{02}{12}) \\
	&+ a^{\ast}b (\dyad{10}{00} + e^{-i\phi} \dyad{10}{01} + e^{-2i\phi} \dyad{10}{02} + e^{i\phi} \dyad{11}{00} \\
	&+ \dyad{11}{01} + e^{-i\phi} \dyad{11}{02} + e^{2i\phi} \dyad{12}{00} + e^{i\phi} \dyad{12}{01} + \dyad{12}{02}) \\
	&+b^{2} (\dyad{10}{10} + e^{-i\phi} \dyad{10}{11} + e^{-2i\phi} \dyad{10}{12} + e^{i\phi} \dyad{11}{10} \\
	&+ \dyad{11}{11} + e^{-i\phi} \dyad{11}{12} + e^{2i\phi} \dyad{12}{10} + e^{i\phi} \dyad{12}{11} + \dyad{12}{12})\big]. \numberthis
\end{align*}
Acting on the last term:
\begin{align*} \label{xicc}
	&\dyad{\xi}{\xi} \otimes {\cal{C}}_{\mathbb{1}/2}^{\otimes2} (\dyad{\psi}{\psi}) = \\&\frac{1}{4} \big[a^{2} (\dyad{00}{00} + \dyad{01}{01} + \dyad{02}{02} + \dyad{03}{03}) + ab^{\ast} (\dyad{00}{10} + \dyad{01}{11} + \dyad{02}{12} + \dyad{03}{13}) \\& a^{\ast}b (\dyad{10}{00} + \dyad{11}{01} + \dyad{12}{02} + \dyad{13}{03}) + b^{2} (\dyad{10}{10} + \dyad{11}{11} + \dyad{12}{12} + \dyad{13}{13})\big]. \numberthis
\end{align*}
And the mixed term gives us the following result:
\begin{align*} \label{xiuc}
	&\dyad{\xi}{\xi} \otimes ({\cal{U}}_{\phi} \otimes {\cal{C}}_{\mathbb{1}/2}) (\dyad{\psi}{\psi}) = \\&\frac{1}{6} \big[a^{2} (2\dyad{00}{00} + 2 \dyad{01}{01} + e^{i\phi} \dyad{03}{00} + e^{i\phi} \dyad{02}{01} \\&+ \dyad{03}{03} + \dyad{02}{02} + e^{-i\phi} \dyad{00}{03} + e^{-i\phi} \dyad{01}{02}) \\&+ ab^{\ast} (2\dyad{00}{10} + 2 \dyad{01}{11} + e^{i\phi} \dyad{03}{10} + e^{i\phi} \dyad{02}{11} \\&+ \dyad{03}{13} + \dyad{02}{12} + e^{-i\phi} \dyad{00}{13} + e^{-i\phi} \dyad{01}{12}) \\&+ a^{\ast}b (2\dyad{10}{00} + 2\dyad{11}{01} + e^{i\phi} \dyad{13}{00} + e^{i\phi} \dyad{12}{01} \\&+ \dyad{13}{03} + \dyad{12}{02} + e^{-i\phi} \dyad{10}{03} + e^{-i\phi} \dyad{11}{02}) \\&+ b^{2} (2\dyad{10}{10} + 2\dyad{11}{11} + e^{i\phi} \dyad{13}{10} + e^{i\phi} \dyad{12}{11} \\&+ \dyad{13}{13} + \dyad{12}{12} + e^{-i\phi} \dyad{10}{13} + e^{-i\phi} \dyad{11}{12}) \big]. \numberthis
\end{align*}
Analogously, we can calculate also $\dyad{\xi}{\xi} \otimes ({\cal{C}}_{\mathbb{1}/2} \otimes {\cal{U}}_{\psi}) (\dyad{\psi}{\psi})$. Let us now apply shift-down operator on the equation (\ref{xiuu}):
\begin{align*}\label{VQC1}
	&{\cal{C}}_{\ominus} \left[\dyad{\xi}{\xi} \otimes {\cal{U}}_{\phi}^{\otimes2}(\dyad{\psi}{\psi})\right] = \\&\frac{1}{3} \big[a^{2} (\dyad{00}{00} + e^{-i\phi} \dyad{00}{01} + e^{-2i\phi} \dyad{00}{02} + e^{i\phi} \dyad{01}{00} \\&+ \dyad{01}{01} + e^{-i\phi} \dyad{01}{02} + e^{2i\phi} \dyad{02}{00} + e^{i\phi} \dyad{02}{01} + \dyad{02}{02}) \\&+ ab^{\ast} (\dyad{00}{12} + e^{-i\phi} \dyad{00}{10} + e^{-2i\phi} \dyad{00}{11} + e^{i\phi} \dyad{01}{12} \\&+ \dyad{01}{10} + e^{-i\phi} \dyad{01}{11} + e^{2i\phi} \dyad{02}{12} + e^{i\phi} \dyad{02}{10} + \dyad{02}{11}) \\&+ a^{\ast}b (\dyad{12}{00} + e^{-i\phi} \dyad{12}{01} + e^{-2i\phi} \dyad{12}{02} + e^{i\phi} \dyad{10}{00} \\&+ \dyad{10}{01} + e^{-i\phi} \dyad{10}{02} + e^{2i\phi} \dyad{11}{00} + e^{i\phi} \dyad{11}{01} + \dyad{11}{02}) \\&+ b^{2} (\dyad{12}{12} + e^{-i\phi} \dyad{12}{10} + e^{-2i\phi} \dyad{12}{11} + e^{i\phi} \dyad{10}{12} \\&+ \dyad{10}{10} + e^{-i\phi} \dyad{10}{11} + e^{2i\phi} \dyad{11}{12} + e^{i\phi} \dyad{11}{10} + \dyad{11}{11})\big] \overset{(i)}{=} \\& \frac{1}{3} \big[(a^{2} \dyad{0}{0} + ab^{\ast} e^{-i\phi} \dyad{1}{0} + a^{\ast}b e^{-i\phi} \dyad{1}{0} + b^{2} \dyad{1}{1}) \otimes (\dyad{0}{0} + \dyad{1}{1}) \\&+ (a^{2} \dyad{0}{0} + ab^{\ast} e^{2i\phi} \dyad{0}{1} + a^{\ast}b e^{-2i\phi} \dyad{1}{0} + b^{2} \dyad{1}{1}) \otimes \dyad{2}{2}\big]= \\& \frac{1}{3} \left[U_{\phi}\dyad{\xi}{\xi}U_{\phi}^{\dagger} \otimes (\dyad{0}{0} + \dyad{1}{1}) + U_{-2\phi}\dyad{\xi}{\xi}U_{-2\phi}^{\dagger
	} \otimes \dyad{2}{2}\right]. \numberthis
\end{align*}
In $(i)$ we only explicitly write out diagonal states in the second space because we are doing measurement corresponding to the basis $\{\ket{j}\}_{j=0}^{2^{N}-1}$, where now $N=2$. We proceed by applying the shift-down operator on the equation (\ref{xicc}):
\begin{align*}\label{VQC2}
	& {\cal{C}}_{\ominus} \left[\dyad{\xi}{\xi} \otimes {\cal{C}}_{\mathbb{1}/2}^{\otimes2} (\dyad{\psi}{\psi})\right] = \\&\frac{1}{4} \big[a^{2} (\dyad{00}{00} + \dyad{01}{01} + \dyad{02}{02} + \dyad{03}{03}) + ab^{\ast} (\dyad{00}{12} + \dyad{01}{10} + \dyad{02}{11} + \dyad{03}{13}) \\& a^{\ast}b (\dyad{12}{00} + \dyad{10}{01} + \dyad{11}{02} + \dyad{13}{03}) + b^{2} (\dyad{12}{12} + \dyad{10}{10} + \dyad{11}{11} + \dyad{13}{13})\big] \\&= \frac{1}{4} \big[(a^{2}\dyad{0}{0} + b^{2}\dyad{1}{1}) \otimes (\dyad{0}{0} + \dyad{1}{1} + \dyad{2}{2}) \\&+ (a^{2}\dyad{0}{0} + ab^{\ast} \dyad{0}{1} + a^{\ast}b \dyad{1}{0} + b^{2}\dyad{1}{1}) \otimes \dyad{3}{3}\big] \\&\overset{(i)}{=} \frac{1}{4} \big[(a^{2}\dyad{0}{0} + b^{2}\dyad{1}{1}) \otimes (\dyad{0}{0} + \dyad{1}{1} + \dyad{2}{2}) + \dyad{\xi}{\xi} \otimes \dyad{3}{3}\big]. \numberthis
\end{align*}
In $(i)$ we again only write diagonal states. And the last application of ${\cal{C}}_{\ominus}$ on the equation (\ref{xiuc}) yields:
\begin{align*}\label{VQC3}
	&{\cal{C}}_{\ominus}\left[\dyad{\xi}{\xi} \otimes ({\cal{U}}_{\phi} \otimes {\cal{C}}_{\mathbb{1}/2}) (\dyad{\psi}{\psi})\right] = \\&\frac{1}{6} \big[a^{2} (2\dyad{00}{00} + 2 \dyad{01}{01} + e^{i\phi} \dyad{03}{00} + e^{i\phi} \dyad{02}{01} \\&+ \dyad{03}{03} + \dyad{02}{02} + e^{-i\phi} \dyad{00}{03} + e^{-i\phi} \dyad{01}{02}) \\&+ ab^{\ast} (2\dyad{00}{12} + 2 \dyad{01}{10} + e^{i\phi} \dyad{03}{12} + e^{i\phi} \dyad{02}{10} \\&+ \dyad{03}{13} + \dyad{02}{11} + e^{-i\phi} \dyad{00}{13} + e^{-i\phi} \dyad{01}{11}) \\&+ a^{\ast}b (2\dyad{12}{00} + 2\dyad{10}{01} + e^{i\phi} \dyad{13}{00} + e^{i\phi} \dyad{11}{01} \\&+ \dyad{13}{03} + \dyad{11}{02} + e^{-i\phi} \dyad{12}{03} + e^{-i\phi} \dyad{10}{02}) \\&+ b^{2} (2\dyad{12}{12} + 2\dyad{10}{10} + e^{i\phi} \dyad{13}{12} + e^{i\phi} \dyad{11}{10} \\&+ \dyad{13}{13} + \dyad{11}{11} + e^{-i\phi} \dyad{12}{13} + e^{-i\phi} \dyad{10}{11}) \big] \\&\overset{(i)}{=} \frac{1}{6} \big[(a^{2} \dyad{0}{0} + b^{2} \dyad{1}{1}) \otimes 2\dyad{0}{0} \\&+ (2a^{2}\dyad{0}{0} + ab^{\ast} e^{-i\phi} \dyad{0}{1} + a^{\ast}b e^{i\phi} \dyad{1}{0} + b^{2} \dyad{1}{1}) \otimes \dyad{1}{1} \\&+ (a^{2}\dyad{0}{0} + 2b^{2}\dyad{1}{1}) \otimes \dyad{2}{2} + (a^{2}\dyad{0}{0} + ab^{\ast}\dyad{0}{1} + a^{\ast}b\dyad{1}{0} + b^{2}\dyad{1}{1}) \otimes \dyad{3}{3}\big] \\&= \frac{1}{6} \big[U_{\phi}\dyad{\xi}{\xi}U_{\phi}^{\dagger} \otimes \dyad{1}{1} + (a^{2}\dyad{0}{0} + b^{2}\dyad{1}{1}) \otimes (2\dyad{0}{0} + \dyad{2}{2}) \\&+ a^{2}\dyad{0}{0} \otimes \dyad{1}{1} + b^{2}\dyad{1}{1} \otimes \dyad{2}{2} + \dyad{\xi}{\xi} \otimes \dyad{3}{3}\big]. \numberthis
\end{align*}
In $(i)$ we only write out states corresponding to measuring diagonal states in the second space. And analogous calculation can be done also for the last remaining state:
\begin{align*}\label{VQC4}
	&{\cal{C}}_{\ominus}\left[\dyad{\xi}{\xi} \otimes ({\cal{C}}_{\mathbb{1}/2} \otimes {\cal{U}}_{\phi}) (\dyad{\psi}{\psi})\right] = \\&\frac{1}{6} \big[U_{\phi}\dyad{\xi}{\xi}U_{\phi}^{\dagger} \otimes \dyad{0}{0} + (a^{2}\dyad{0}{0} + b^{2}\dyad{1}{1}) \otimes (2\dyad{1}{1} + \dyad{2}{2}) + a^{2}\dyad{0}{0} \otimes \dyad{2}{2} \\&+ b^{2}\dyad{1}{1} \otimes \dyad{0}{0} + \dyad{\xi}{\xi} \otimes \dyad{3}{3}\big].\numberthis
\end{align*}
By putting together, the previous results (\ref{VQC1}), (\ref{VQC2}), (\ref{VQC3}) and (\ref{VQC4}), we obtain the state after the implemented transformation:
\begin{align*}
	&{\cal{C}}_{\ominus}\big[\dyad{\xi}{\xi} \otimes {\cal{E}}_{\phi}^{\otimes2} (\dyad{\psi}{\psi})\big] = \\&\frac{1}{3} q^{2} \left[U_{\phi}\dyad{\xi}{\xi}U_{\phi}^{\dagger} \otimes (\dyad{0}{0} + \dyad{1}{1}) + U_{-2\phi}\dyad{\xi}{\xi}U_{-2\phi}^{\dagger
	} \otimes \dyad{2}{2}\right] \\&+ \frac{1}{6} q(1-q) \big[U_{\phi}\dyad{\xi}{\xi}U_{\phi}^{\dagger} \otimes (\dyad{0}{0} + \dyad{1}{1}) + (a^{2}\dyad{0}{0} + b^{2}\dyad{1}{1}) \otimes (2\dyad{0}{0} + 2\dyad{1}{1} + 3\dyad{2}{2}) \\&+ a^{2}\dyad{0}{0} \otimes \dyad{1}{1} + b^{2}\dyad{1}{1} \otimes \dyad{0}{0} + \dyad{\xi}{\xi} \otimes 2\dyad{3}{3}\big] \\&+ \frac{1}{4} (1-q)^{2} \big[(a^{2}\dyad{0}{0} + b^{2}\dyad{1}{1}) \otimes (\dyad{0}{0} + \dyad{1}{1} + \dyad{2}{2}) + \dyad{\xi}{\xi} \otimes \dyad{3}{3}\big].
\end{align*}
Successful measurement corresponds to states $\ket{0}$ and $\ket{1}$, while the unsuccessful measurement corresponds to measuring states $\ket{2}$ and $\ket{3}$.

Let us now analyze the proportion of the entire implemented transformation that is constituted by a unitary transformation, i.e., by the transformation we wish to implement. We have also calculated the implementation of the depolarizing noisy channel ${\cal{E}}_{\psi}$ for the cases $N=3,4$. Let us write the former where, from the equation (\ref{ancImpSt}), the input state is $\ket{\psi} = \frac{1}{\sqrt{4}} \sum_{i=0}^{3} \ket{i}$:
\begin{align*}
	&{\cal{C}}_{\ominus}\big[\dyad{\xi}{\xi} \otimes {\cal{E}}_{\phi}^{\otimes3} (\dyad{\psi}{\psi})\big] = \\&\frac{1}{4}q^{3} \big[U_{\phi} \dyad{\xi}{\xi} U_{\phi}^{\dagger} \otimes (\dyad{0}{0} + \dyad{1}{1} + \dyad{2}{2}) + U_{-3\phi} \dyad{\xi}{\xi} U_{-3\phi}^{\dagger} \otimes \dyad{3}{3}\big] + \\& \frac{1}{8} q^{2}(1-q) \big[U_{\phi} \dyad{\xi}{\xi} U_{\phi}^{\dagger} \otimes 2(\dyad{0}{0} + \dyad{1}{1} + \dyad{2}{2}) + \\&(a^{2} \dyad{0}{0} + b^{2} \dyad{1}{1}) \otimes (2\dyad{0}{0} + 3\dyad{1}{1} + 2\dyad{2}{2} + 4\dyad{3}{3}) + \\&a^{2}\dyad{0}{0} \otimes \dyad{2}{2} + b^{2}\dyad{1}{1} \otimes \dyad{0}{0} + \dyad{\xi}{\xi} \otimes (2\dyad{4}{4} + \dyad{5}{5} + 2\dyad{6}{6} + \dyad{7}{7})\big] + \\& \frac{1}{16} q (1-q)^{2} \big[U_{\phi} \dyad{\xi}{\xi} U_{\phi}^{\dagger} \otimes (\dyad{0}{0} + \dyad{1}{1} + \dyad{2}{2}) + \\&(a^{2} \dyad{0}{0} + b^{2} \dyad{1}{1}) \otimes (5\dyad{0}{0} + 7\dyad{1}{1} + 5\dyad{2}{2} + 6\dyad{3}{3}) \\&a^{2}\dyad{0}{0} \otimes 2\dyad{2}{2} + b^{2}\dyad{1}{1} \otimes 2\dyad{0}{0} + \dyad{\xi}{\xi} \otimes (6\dyad{4}{4} + 4\dyad{5}{5} + 6\dyad{6}{6} + 4\dyad{7}{7})\big] + \\& \frac{1}{8} (1-q)^{3} \big[(a^{2} \dyad{0}{0} + b^{2} \dyad{1}{1}) \otimes (\dyad{0}{0} + \dyad{1}{1} + \dyad{2}{2} + \dyad{3}{3}) + \dyad{\xi}{\xi} \otimes \sum_{i=4}^{7}\ket{i} \big].
\end{align*}
And the result for $N=4$ with input state from the equation (\ref{ancImpSt}), $\ket{\psi} = \frac{1}{\sqrt{5}} \sum_{i=0}^{4} \ket{i}$:
\begin{align*}\label{VMC4}
	&{\cal{C}}_{\ominus}\big[\dyad{\xi}{\xi} \otimes {\cal{E}}_{\phi}^{\otimes4} (\dyad{\psi}{\psi})\big] = \\&\frac{1}{5} q^{4} \big[U_{\phi} \dyad{\xi}{\xi} U_{\phi}^{\dagger} \otimes (\dyad{0}{0} + \dyad{1}{1} + \dyad{2}{2} + \dyad{3}{3}) + U_{-4\phi} \dyad{\xi}{\xi} U_{-4\phi}^{\dagger} \otimes \dyad{4}{4}\big] + \\& \frac{1}{10} q^{3} (1-q) \big[U_{\phi} \dyad{\xi}{\xi} U_{\phi}^{\dagger} \otimes 3(\dyad{0}{0} + \dyad{1}{1} + \dyad{2}{2} + \dyad{3}{3}) + \\&(a^{2} \dyad{0}{0} + b^{2} \dyad{1}{1}) \otimes (2\dyad{0}{0} + 3\dyad{1}{1} + 3\dyad{2}{2} + 2\dyad{3}{3} + 5\dyad{4}{4}) + \\&a^{2}\dyad{0}{0} \otimes \dyad{3}{3} + b^{2}\dyad{1}{1} \otimes \dyad{0}{0} + \\&\dyad{\xi}{\xi} \otimes (2\dyad{5}{5} + \dyad{6}{6} + 2\dyad{7}{7} + \dyad{8}{8} + \dyad{9}{9} + \dyad{10}{10} + 2\dyad{12}{12} + \dyad{14}{14} + \dyad{15}{15}) \big] + \\& \frac{1}{20} q^{2} (1-q)^{2} \big[U_{\phi} \dyad{\xi}{\xi} U_{\phi}^{\dagger} \otimes 3(\dyad{0}{0} + \dyad{1}{1} + \dyad{2}{2} + \dyad{3}{3}) + \\&(a^{2} \dyad{0}{0} + b^{2} \dyad{1}{1}) \otimes (7\dyad{0}{0} + 10\dyad{1}{1} + 10\dyad{2}{2} + 7\dyad{3}{3} + 10\dyad{4}{4}) + \\&a^{2}\dyad{0}{0} \otimes (\dyad{2}{2} + 3\dyad{3}{3}) + b^{2}\dyad{1}{1} \otimes (3\dyad{0}{0} + \dyad{1}{1}) + \\&\dyad{\xi}{\xi} \otimes (8\dyad{5}{5} + 5\dyad{6}{6} + 9\dyad{7}{7} + 6\dyad{8}{8} + 4\dyad{9}{9} + 6\dyad{10}{10} + 3\dyad{11}{11} + 8\dyad{12}{12} + \\&2\dyad{13}{13} + 5\dyad{14}{14} + 4\dyad{15}{15})\big] + \\&\frac{1}{40} q (1-q)^{3} \big[U_{\phi} \dyad{\xi}{\xi} U_{\phi}^{\dagger} \otimes (\dyad{0}{0} + \dyad{1}{1} + \dyad{2}{2} + \dyad{3}{3}) + \\&(a^{2} \dyad{0}{0} + b^{2} \dyad{1}{1}) \otimes (9\dyad{0}{0} + 12\dyad{1}{1} + 12\dyad{2}{2} + 9\dyad{3}{3} + 10\dyad{4}{4}) + \\&a^{2}\dyad{0}{0} \otimes (\dyad{2}{2} + 3\dyad{3}{3}) + b^{2}\dyad{1}{1} \otimes (3\dyad{0}{0} + \dyad{1}{1}) + \\&\dyad{\xi}{\xi} \otimes (11\dyad{5}{5} + 9\dyad{6}{6} + 12\dyad{7}{7} + 10\dyad{8}{8} + 7\dyad{9}{9} + 10\dyad{10}{10} + 8\dyad{11}{11} \\&+ 11\dyad{12}{12} + 6\dyad{13}{13} + 9\dyad{14}{14} + 7\dyad{15}{15})\big] + \\&\frac{1}{16} (1-q)^{4} \big[(a^{2} \dyad{0}{0} + b^{2} \dyad{1}{1}) \otimes (\dyad{0}{0} + \dyad{1}{1} + \dyad{2}{2} + \dyad{3}{3} + \dyad{4}{4}) + \dyad{\xi}{\xi} \otimes \sum_{i=5}^{15} \dyad{i}{i}\big]. \numberthis
\end{align*}
If we take a look at the previous equations, we can see that the factors next to states which are in the tensor product with $U_{\phi} \dyad{\xi}{\xi} U_{\phi}^{\dagger}$ follow Pascal triangle (e.g. in case of $N=4$, in the term next to $q^{4}$ there is $\mathbf{1}\sum_{j=0}^{3}\dyad{j}{j}$, next to $q^{3}(1-q)$ and $q^{2}(1-q)^{2}$ there is $\mathbf{3}\sum_{j=0}^{3}\dyad{j}{j}$ and next to $q(1-q)^{3}$ there is again $\mathbf{1}\sum_{j=0}^{3}\dyad{j}{j}$. Together this makes $1,3,3,1$ which is a third row of the Pascal triangle). In addition, in the last term with $(1-q)^{N}$, there is no unitary transformation. This is due to implementation, where ${\cal{E}}_{\phi} = [q {\cal{U}}_{\phi} + (1-q) {\cal{C}}_{\mathbb{1}/2}]^{\otimes N}$ is a binomial distribution. We can also notice that the states next to $U_{\phi} \dyad{\xi}{\xi} U_{\phi}^{\dagger}$ are always the ones that form the successful measurement $\{\ket{j}\}_{j=0}^{N-1}$. And as the last point, we should also account for the factors $q^{\phi}(1-q)^{N-\phi}$ in front of the square brackets that come from the number of times $\phi$ unitary channel is mixed with contraction $N - \phi$ in the individual terms, where by $\phi$ we have denoted number of times unitary channel is mixed in the particular term. The number in front of every square bracket comes from normalization of input state $\ket{\psi}$ from equation (\ref{ancImpSt}) and then it is multiplied by $\frac{1}{2^{(N - \phi)}}$. For example, in equation (\ref{VMC4}) next to the $q(1-q)^{3}$ we have $\frac{1}{40} = \frac{1}{5 \times 2^{3}}$, where $5$ comes from normalization of the input state $\ket{\psi} = \frac{1}{\sqrt{N+1}} \sum_{i=0}^{N} \ket{i} = \frac{1}{\sqrt{5}} \sum_{i=0}^{4} \ket{i}$ and the power $3$ to which $2$ is raised comes from applying channel ${\cal{C}}_{\mathbb{1}/2}$ three times in this particular term. Together, the fraction $p_{U_{\phi}}$ of the resulting state for both the successful and the failed measurement that is taken by the unitary channel is:
\begin{align}\label{pufi}
	p_{U_{\phi}} = \sum_{\phi=1}^{N} \binom{N-1}{\phi-1} N \frac{q^{\phi}(1-q)^{N-\phi}}{(N+1)2^{N-\phi}} = \frac{N}{N+1} q \frac{(1+q)^{N-1}}{2^{N-1}},
\end{align}
where the number $N$ after the binomial coefficient comes from the fact that unitary channel applied on $\dyad{\xi}{\xi}$ is tied with $\dyad{j}{j}$ for $j=0,\dots,N-1$ through tensor product. Therefore, there exist $N$ possibilities to measure $U_{\phi}\dyad{\xi}{\xi}U_{\phi}^{\dagger}$. From the binomial coefficient $\binom{N-1}{\phi-1}$, we had to subtract one unitary channel because we are counting permutations in case of at least one unitary channel being applied (as there is no unitary transformation on the output if there is no on the input). The sum in the previous equation (\ref{pufi}) was evaluated using Wolfram Mathematica. The value we have obtained is the same as in equation (\ref{resWhiteNoise}) for the PSAR device and it is depicted in the figures \ref{compPSARbetweenPDandDep:b} and \ref{compVQandVMCandPSARforDep:b}. Share of unitary channel $p_{U_{\phi}}$ mixed in the entire channel goes to zero as $N$ goes to infinity, because $1+q \leq 2$. The only exception is $q=1$, when we recover $p_{U_{\phi}} \overset{q=1}{=} \frac{N}{N+1}$. This means, that with more noisy channels at our disposal, the performance is worsening, because our desire is to implement unitary channel with as little noise as possible (and we can see that this ability is decreasing with the increasing number $N$). 

Analyzing share that $a^{2} \dyad{0}{0} + b^{2} \dyad{1}{1}$ takes of the entire implemented channel is noticeably more challenging. Let us begin our operation with explicitly listing few examples of what exactly happens if we apply various number of unitary and depolarizing channels on the input state. Let us take as an example a case when $N=3$ and list all the possible combinations (for now, without the factors $q^{\phi}(1-q)^{N-\phi}$):
\begin{align*}
	({\cal{U}}_{\phi} \otimes {\cal{U}}_{\phi} \otimes {\cal{U}}_{\phi}) (\dyad{\psi}{\psi}) &= \frac{1}{4} \big[U_{\phi} \dyad{\xi}{\xi} U_{\phi}^{\dagger} \otimes (\dyad{0}{0} + \dyad{1}{1} + \dyad{2}{2}) + U_{-3\phi} \dyad{\xi}{\xi} U_{-3\phi}^{\dagger} \otimes \dyad{3}{3}\big], \\
	({\cal{C}}_{\mathbb{1}/2} \otimes {\cal{U}}_{\phi} \otimes {\cal{U}}_{\phi}) (\dyad{\psi}{\psi}) &= \frac{1}{8} \big[U_{\phi} \dyad{\xi}{\xi} U_{\phi}^{\dagger} \otimes (\dyad{0}{0} + \dyad{1}{1}) \\&+ (a^{2} \dyad{0}{0} + b^{2} \dyad{1}{1}) \otimes (2\dyad{2}{2} + \dyad{3}{3}) \\&+ a^{2} \dyad{0}{0} \otimes \dyad{3}{3} + b^{2} \dyad{1}{1} \otimes \dyad{1}{1} + \dyad{\xi}{\xi} \otimes (\dyad{5}{5} + \dyad{6}{6})\big], \\
	(\underset{2}{{\cal{U}}_{\phi}} \otimes \underset{1}{{\cal{C}}_{\mathbb{1}/2}} \otimes  \underset{0}{{\cal{U}}_{\phi}}) (\dyad{\psi}{\psi}) &= \frac{1}{8} \big[U_{\phi} \dyad{\xi}{\xi} U_{\phi}^{\dagger} \otimes (\dyad{0}{0} + \dyad{2}{2}) \\&+ (a^{2} \dyad{0}{0} + b^{2} \dyad{1}{1}) \otimes (2\dyad{1}{1} + \dyad{3}{3}) \\&+ a^{2} \dyad{0}{0} \otimes \dyad{2}{2} + b^{2} \dyad{1}{1} \otimes \dyad{0}{0} + \dyad{\xi}{\xi} \otimes (\dyad{4}{4} + \dyad{6}{6})\big], \\
	({\cal{U}}_{\phi} \otimes {\cal{U}}_{\phi} \otimes {\cal{C}}_{\mathbb{1}/2}) (\dyad{\psi}{\psi}) &= \frac{1}{8} \big[U_{\phi} \dyad{\xi}{\xi} U_{\phi}^{\dagger} \otimes (\dyad{1}{1} + \dyad{2}{2}) \\&+ (a^{2} \dyad{0}{0} + b^{2} \dyad{1}{1}) \otimes (2\dyad{0}{0} + \dyad{3}{3}) \\&+ a^{2} \dyad{0}{0} \otimes \dyad{1}{1} + b^{2} \dyad{1}{1} \otimes \dyad{3}{3} + \dyad{\xi}{\xi} \otimes (\dyad{4}{4} + \dyad{7}{7})\big], \\
	(\underset{2}{{\cal{U}}_{\phi}} \otimes \underset{1}{{\cal{C}}_{\mathbb{1}/2}} \otimes \underset{0}{{\cal{C}}_{\mathbb{1}/2}}) (\dyad{\psi}{\psi}) &= \frac{1}{16} \big[U_{\phi} \dyad{\xi}{\xi} U_{\phi}^{\dagger} \otimes \dyad{2}{2} \\&+ (a^{2} \dyad{0}{0} + b^{2} \dyad{1}{1}) \otimes (3\dyad{0}{0} + 3\dyad{1}{1} + \dyad{3}{3}) \\&+ a^{2} \dyad{0}{0} \otimes 2\dyad{2} + b^{2} \dyad{1}{1} \otimes 2\dyad{3}{3} \\&+ \dyad{\xi}{\xi} \otimes (3\dyad{4}{4} + \dyad{5}{5} + \dyad{6}{6} + \dyad{7}{7})\big], \\
	({\cal{C}}_{\mathbb{1}/2} \otimes {\cal{U}}_{\phi} \otimes {\cal{C}}_{\mathbb{1}/2}) (\dyad{\psi}{\psi}) &= \frac{1}{16} \big[U_{\phi} \dyad{\xi}{\xi} U_{\phi}^{\dagger} \otimes \dyad{1}{1} \\&+ (a^{2} \dyad{0}{0} + b^{2} \dyad{1}{1}) \otimes (2\dyad{0}{0} + \dyad{1}{1} + 2\dyad{2}{2} + 2\dyad{3}{3}) \\&+ \dyad{\xi}{\xi} \otimes (2\dyad{4}{4} + 2\dyad{5}{5} + 2\dyad{6}{6} + 2\dyad{7}{7})\big], \\
	({\cal{C}}_{\mathbb{1}/2} \otimes {\cal{C}}_{\mathbb{1}/2} \otimes {\cal{U}}_{\phi}) (\dyad{\psi}{\psi}) &= \frac{1}{16} \big[U_{\phi} \dyad{\xi}{\xi} U_{\phi}^{\dagger} \otimes \dyad{0}{0} \\&+ (a^{2} \dyad{0}{0} + b^{2} \dyad{1}{1}) \otimes (3\dyad{1}{1} + 3\dyad{2}{2} + \dyad{3}{3}) \\&+ a^{2} \dyad{0}{0} \otimes 2\dyad{3}{3} + b^{2} \dyad{1}{1} \otimes 2\dyad{0}{0} \\&+ \dyad{\xi}{\xi} \otimes (\dyad{4}{4} + \dyad{5}{5} + 3\dyad{6}{6} + \dyad{7}{7})\big], \\
	({\cal{C}}_{\mathbb{1}/2} \otimes {\cal{C}}_{\mathbb{1}/2} \otimes {\cal{C}}_{\mathbb{1}/2}) (\dyad{\psi}{\psi}) &= \frac{1}{8} \big[(a^{2} \dyad{0}{0} + b^{2} \dyad{1}{1}) \otimes (\dyad{0}{0} + \dyad{1}{1} + \dyad{2}{2} + \dyad{3}{3}) \\&+ \dyad{\xi}{\xi} \otimes (\dyad{4}{4} + \dyad{5}{5} + \dyad{6}{6} + \dyad{7}{7})\big].
\end{align*}
Channel on the leftmost corresponds to successful measurement of ket vector $\ket{2}$, the one in the middle corresponds to $\ket{1}$ and the one on the right corresponds to $\ket{0}$. Let us call these positions $2$, $1$, and $0$ respectively. It is a consequence of our dictionary (\ref{generalDictionary}), where first $N$ states correspond to number of ones present in a state which are all written from the right. Let us quickly return to unitary channels and provide reasoning why $U_{\phi}\dyad{\xi}{\xi}U_{\phi}^{\dagger}$ is always bound to states forming success measurement. It is, again, the artifact of our state labeling (\ref{generalDictionary}), where unitary channel sitting on position $j$ always acts on $\dyad{j}{j}$ and $\dyad{j+1}{j+1}$ as identity, while to states $\dyad{j+1}{j}$ and $\dyad{j}{j+1}$ it adds $e^{i\phi}$ and $e^{-i\phi}$ respectively. Afterwards, shift-down operator shifts $\ket{j+1}$ to $\ket{j}$ and creates $U_{\phi} \dyad{\xi}{\xi} U_{\phi}^{\dagger}$ next to the state corresponding to the original position $j$ of the unitary channel. After our little detour, we can return to the analyses of the depolarizing channels. If there is one channel ${\cal{C}}_{\mathbb{1}/2}$ acting on position $1$, which is also highlighted in the previous equations, then it "creates" factor $1+1=2$ in front of the state $\dyad{1}{1}$. It also adds $a^{2}\dyad{0}{0}$ next to $\dyad{1+1}{1+1} = \dyad{2}{2}$ and $b^{2}\dyad{1}{1}$ next to $\dyad{1-1}{1-1} = \dyad{0}{0}$. In case there are two depolarizing channels acting next to each other, say on positions $0$ and $1$, which is again denoted among the previous equations, they put factor $1+2=3$ next to states $\dyad{0}{0}$ and $\dyad{1}{1}$. If we had $V$ depolarizing channels neighboring each other the factor next to the measuring state is $1+V$. Then also, with $a^{2}\dyad{0}{0}$ there is a factor $1+1=2$ next to $\dyad{1+1}{1+1} = \dyad{2}{2}$ and with $b^{2}\dyad{1}{1}$ the factor $1+1=2$ next to $\dyad{0-1}{0-1} = \dyad{3}{3}$, where there is a modulo arithmetic and in general $\dyad{-1}{-1} = \dyad{N}{N}$. Similarly, as before, for $V$ neighboring depolarizing channels the factors next to $a^{2} \dyad{0}{0}$ and $b^{2} \dyad{1}{1}$ acquire value $V$. Every time, there is at least one depolarizing channel implemented, we also obtain a term $(a^{2}\dyad{0}{0} + b^{2}\dyad{1}{1}) \otimes \dyad{3}{3}$ (or in general case $(a^{2}\dyad{0}{0} + b^{2}\dyad{1}{1}) \otimes \dyad{N}{N}$).

\phantomsection\label{pageForV}Let us also closer analyze the reason behind the factors next to $a^{2}\dyad{0}{0} + b^{2}\dyad{1}{1}$. In case of $N=3$, the input state is $\ket{\psi} = \frac{1}{\sqrt{3}} (\ket{000} + \ket{001} + \ket{011} + \ket{111}) = \frac{1}{\sqrt{3}} (\ket{0} + \ket{1} + \ket{2} + \ket{3})$. In case of acting by ${\cal{U}}_{\phi} \otimes {\cal{C}}_{\mathbb{1}/2} \otimes {\cal{U}}_{\phi}$ on $\ket{\psi}$, the states $\ket{001} = \ket{1}$ and $\ket{011} = \ket{2}$ behave effectively as the same state because $({\cal{U}}_{\phi} \otimes {\cal{C}}_{\mathbb{1}/2} \otimes {\cal{U}}_{\phi}) (\dyad{001}{001}) = ({\cal{U}}_{\phi} \otimes {\cal{C}}_{\mathbb{1}/2} \otimes {\cal{U}}_{\phi}) (\dyad{011}{011}) = \dyad{0}{0} \otimes \frac{1}{2} (\dyad{0}{0} + \dyad{1}{1}) \otimes \dyad{1}{1}$. Therefore, we obtain factor $2$ next to states $\ket{1}$ and $\ket{2}$. But the shift-down operator slides the $b^{2} \dyad{1}{1}$ off of measuring $\ket{2}$ to measuring $\ket{1}$ and off of $\ket{1}$ to $\ket{0}$. This also creates a redundant $a^{2}\dyad{0}{0}$ next to $\dyad{2}{2}$ and $b^{2}\dyad{1}{1}$ next to $\dyad{0}{0}$ with factor $2$. But the unitary channel takes $1$ of $a^{2}\dyad{0}{0} \otimes \dyad{2}{2}$ and $1$ of $b^{2}\dyad{1}{1} \otimes \dyad{0}{0}$ to form $U_{\phi} \dyad{\xi}{\xi} U_{\phi}^{\dagger}$ (similar effect can be seen in the equation (\ref{VQC3})). In case of acting with ${\cal{U}}_{\phi} \otimes {\cal{C}}_{\mathbb{1}/2} \otimes {\cal{C}}_{\mathbb{1}/2}$, the argument is similar with the adjustment that now states $\ket{000}$, $\ket{001}$ and $\ket{011}$ behave effectively as the same state, therefore creating factor $3$ in the result. In general, for $V$ neighboring contractions we will obtain $(a^{2}\dyad{0}{0} + b^{2}\dyad{1}{1}) \otimes (1+V) \sum_{i=j}^{j+V-1}\dyad{i}{i}$, where $j$ is the position of the rightmost neighbor. We also obtain $a^{2} \dyad{0}{0} \otimes V \dyad{j+1}{j+1}$ and $b^{2} \dyad{1}{1} \otimes V \dyad{j-1}{j-1}$. We shall also discover at least one $a^{2} \dyad{0}{0} + b^{2} \dyad{1}{1}$ next to $\dyad{3}{3}$ every time at least one ${\cal{C}}_{\mathbb{1}/2}$ channel is applied. This is because states $\dyad{000}{000}$ and $\dyad{111}{111}$ will be preserved after applying whatever mixture of ${\cal{U}}_{\phi}$ and ${\cal{C}}_{\mathbb{1}/2}$ channels on the input. With the exception of adding normalization $\frac{1}{2^{I}}$, where $I$ is the number of used depolarizing channels. And then, operator ${\cal{C}}_{\mathbb{1}/2}$ will slide $b^{2} \dyad{1}{1}$ from $\dyad{000}{000} = \dyad{0}{0}$ to $\dyad{111}{111} = \dyad{3}{3}$. This holds for arbitrary value $N$, where states $\dyad{0\cdots0}{0\cdots0}$ and $\dyad{1\cdots1}{1\cdots1}$ will not be affected by applying arbitrary mixture of channels (with the exception of normalization). 

To obtain total success probability we have to derive the prescription for the fraction of $a^{2}\dyad{0}{0} + b^{2}\dyad{1}{1}$, corresponding only to a successful measurement, of the entire implemented transformation, i.e. to measuring vectors $\ket{j}$ for $j = 0, \cdots, N-1$. Let us first calculate the number of permutations of $a^{2} \dyad{0}{0} + b^{2} \dyad{1}{1}$ appearing in the resulting state. We shall start with an example. Let $N=4$ and the number of applied depolarizing channels be $I = 1$, which means that number of unitary channels is $\phi = N-I = 3$. We shall also consider number of neighboring depolarizing channels $V$. Of course, for the present case $V = 1$. Let us list all possible permutations also with the expression to evaluate the number of possible permutations for a given case:
\begin{align*}
	&j = 0: \qquad ? \otimes ? \otimes {\cal{U}}_{\phi} \otimes {\cal{C}}_{\mathbb{1}/2}  &&\frac{(N-2)!}{(I-1)!(\phi-1)!} = \frac{2!}{0!2!} = 1, \\
	&j = 1: \qquad ? \otimes {\cal{U}}_{\phi} \otimes {\cal{C}}_{\mathbb{1}/2} \otimes {\cal{U}}_{\phi}  &&\frac{(N-3)!}{(I-1)!(\phi-2)!} = \frac{1!}{0!1!} = 1, \\ 
	&j = 2: \qquad  {\cal{U}}_{\phi} \otimes {\cal{C}}_{\mathbb{1}/2} \otimes {\cal{U}}_{\phi} \otimes ?  &&\frac{(N-3)!}{(I-1)!(\phi-2)!} = \frac{1!}{0!1!} = 1, \\
	&j = 3: \qquad {\cal{C}}_{\mathbb{1}/2} \otimes {\cal{U}}_{\phi} \otimes ? \otimes ? &&\frac{(N-2)!}{(I-1)!(\phi-1)!} = \frac{2!}{0!2!} = 1,
\end{align*}
where $j$ denotes the position of the applied depolarizing channel and $?$ denotes that both channels on given position are, in principle, possible to occur. The procedure to derive all possible scenarios is as follows: at first, we put in one depolarizing channel on a given position and then we surround it with unitary channels so that we guarantee that there is exactly only one depolarizing channel without other neighboring depolarizing channel. This way we guarantee that we do not count the same permutation multiple times. The expression on the right side for calculating possible permutations is only a consequence of the previous procedure where we subtract the number of all allocated channels from $N$, the number of fixed depolarizing channels from $I$ and the number of anchored unitary channels from $\phi$. By going through all possible positions $j$ we cover all possible permutations for applying a given number $I$ and $\phi$ of depolarizing and unitary channels. Together we have four permutations for this particular case. We shall repeat this procedure for all possible number of neighboring $V$ depolarizing channels. It is because that way we find number of permutations corresponding to successful measurement to which we can assign the same value. In the just considered case of $\phi = 3$ and $V = 1$, we have found out that there shall be $4$ possibilities for $a^{2}\dyad{0}{0} + b^{2}\dyad{1}{1}$ with coefficient $V+1 = 2$ to appear next to a state corresponding to a successful measurement.

Now, let us assume that $N=4$ and $I=2$, therefore $\phi = 2$. We shall examine what happens if there are two neighboring depolarizing channels, which we shall denote with $V = 2$:
\begin{align*}
	&j = 0: \qquad ? \otimes {\cal{U}}_{\phi} \otimes {\cal{C}}_{\mathbb{1}/2} \otimes {\cal{C}}_{\mathbb{1}/2}  &&\frac{(N-3)!}{(I-2)!(\phi-1)!} = \frac{1!}{0!1!} = 1, \\
	&j = 1: \qquad {\cal{U}}_{\phi} \otimes {\cal{C}}_{\mathbb{1}/2} \otimes {\cal{C}}_{\mathbb{1}/2} \otimes {\cal{U}}_{\phi}  &&\frac{(N-4)!}{(I-2)!(\phi-2)!} = \frac{0!}{0!0!} = 1, \\ 
	&j = 2: \qquad  {\cal{C}}_{\mathbb{1}/2} \otimes {\cal{C}}_{\mathbb{1}/2} \otimes {\cal{U}}_{\phi} \otimes ?  &&\frac{(N-3)!}{(I-2)!(\phi-1)!} = \frac{1!}{0!1!} = 1,
\end{align*}
where $j$ now denotes the position of the rightmost of the neighboring depolarizing channels. Together we have $3$ orderings where there are two neighboring depolarizing channels $V = 2$ in this case. Thus, we know, that there are $3$ permutations with coefficient $3$ appearing with $a^{2}\dyad{0}{0} + b^{2}\dyad{1}{1}$ next to a state corresponding to a successful measurement. Now, we ask how many permutations there exist in situation with two depolarizing channels $I = 2$ and with no depolarizing channels $V = 1$:
\begin{align*}
	&j = 0: \qquad ? \otimes ? \otimes {\cal{U}}_{\phi} \otimes {\cal{C}}_{\mathbb{1}/2}  &&\frac{(N-2)!}{(I-1)!(\phi-1)!} = \frac{2!}{1!1!} = 2, \\
	&j = 1: \qquad ? \otimes {\cal{U}}_{\phi} \otimes {\cal{C}}_{\mathbb{1}/2} \otimes {\cal{U}}_{\phi}  &&\frac{(N-3)!}{(I-1)!(\phi-2)!} = \frac{1!}{1!0!} = 1, \\ 
	&j = 2: \qquad  {\cal{U}}_{\phi} \otimes {\cal{C}}_{\mathbb{1}/2} \otimes {\cal{U}}_{\phi} \otimes ?  &&\frac{(N-3)!}{(I-1)!(\phi-2)!} = \frac{1!}{1!0!} = 1, \\
	&j = 3: \qquad {\cal{C}}_{\mathbb{1}/2} \otimes {\cal{U}}_{\phi} \otimes ? \otimes ? &&\frac{(N-2)!}{(I-1)!(\phi-1)!} = \frac{2!}{1!1!} = 2.
\end{align*}
Together, we have $6$ permutations.

As the last case, let us analyze $N=4$, $I=3$ and $\phi=1$. Firstly, for $V = 1$, we obtain:
\begin{align*}
	&j = 0: \qquad ? \otimes ? \otimes {\cal{U}}_{\phi} \otimes {\cal{C}}_{\mathbb{1}/2}  &&\frac{(N-2)!}{(I-1)!(\phi-1)!} = \frac{2!}{2!0!} = 1, \\
	&j = 3: \qquad  {\cal{C}}_{\mathbb{1}/2} \otimes {\cal{U}}_{\phi} \otimes ?  \otimes ? &&\frac{(N-2)!}{(I-1)!(\phi-1)!} = \frac{2!}{2!0!} = 1.
\end{align*}
In this case, we only have "boundary" cases for $j = 1, 3$, because we only have one unitary channel, therefore we cannot surround the depolarizing channel from both sides. That together accounts for two permutations. Let us move to the case of $V = 2$:
\begin{align*}
	&j = 0: \qquad ? \otimes {\cal{U}}_{\phi} \otimes {\cal{C}}_{\mathbb{1}/2} \otimes {\cal{C}}_{\mathbb{1}/2} &&\frac{(N-3)!}{(I-2)!(\phi-1)!} = \frac{1!}{1!0!} = 1, \\
	&j = 2: \qquad  {\cal{C}}_{\mathbb{1}/2} \otimes {\cal{C}}_{\mathbb{1}/2} \otimes {\cal{U}}_{\phi} \otimes ? &&\frac{(N-3)!}{(I-2)!(\phi-1)!} = \frac{1!}{1!0!} = 1.
\end{align*}
Together that makes two different permutations. And finally, what if $V = 3$?
\begin{align*}
	&j = 0: \qquad {\cal{U}}_{\phi} \otimes {\cal{C}}_{\mathbb{1}/2} \otimes {\cal{C}}_{\mathbb{1}/2} \otimes {\cal{C}}_{\mathbb{1}/2} &&\frac{(N-4)!}{(I-3)!(\phi-1)!} = \frac{0!}{0!0!} = 1, \\
	&j = 1: \qquad  {\cal{C}}_{\mathbb{1}/2} \otimes {\cal{C}}_{\mathbb{1}/2} \otimes {\cal{C}}_{\mathbb{1}/2} \otimes {\cal{U}}_{\phi} &&\frac{(N-3)!}{(I-2)!(\phi-1)!} = \frac{0!}{0!0!} = 1.
\end{align*}
Again, we have two different possibilities. We skip the trivial cases of $\phi=0$ (there is only one permutation) and $\phi=N$ (no implemented depolarizing channels).

In general, for a given value of $N$ and $N-1 \geq \phi \geq 1$ the total number of permutations is obtained by summing over all number of neighbors we can have. Let us first calculate the number of boundary cases (where the $V$ neighbors are put on the rightmost or the leftmost positions):
\begin{align}\label{permutationsB}
	\sum_{V=1}^{N-\phi} 2 \frac{\left(N - V - 1\right)!}{(I-V)!(\phi-1)!},
\end{align}
where the factor $2$ appears because we have two boundary positions for $j = 0$ and $j = N-1-V$. As for the term counting permutations, in the numerator, we have to subtract the number of already used depolarizing, i.e. $V$, and unitary, i.e. $1$, channels and the same applies also for the denominator. The rest of the possible cases can be calculated as follows where the number of unitary channels is $N-1 \geq \phi \geq 2$:
\begin{align}\label{permutations}
	\sum_{V=1}^{N-\phi} \left[N-2-(V-1)\right] \frac{\left(N - V - 2\right)!}{(I-V)!(\phi-2)!}.
\end{align}
The first difference from previous equation, comes from the fact that now, we have to bound depolarizing channels from both sides, i.e. we subtract $2$ unitary channels in the fraction from both, $N$ and $\phi$. The factor in front of the fraction counts the number of possible positions $N-(V-1)$ one can put $V$ neighboring depolarizing channels on, while also discounting the $2$ boundary cases.


Let us evaluate the probability of successful measurement coming from the term $a^{2}\dyad{0}{0} + b^{2}\dyad{1}{1}$ for a general case of arbitrary $N$ and $N-1 \geq \phi \geq 1$. Now, we shall add factor $(V + 1)$ to equation (\ref{permutationsB}) as was explained on pages \pageref{pageForV} and 118. We shall also realize that there is $V$ neighboring positions on which the channel can sit:
\begin{align}\label{ab1}
	\sum_{V=1}^{N-\phi} 2 \frac{\left(N - V - 1\right)!}{(I-V)!(\phi-1)!} V (V+1),
\end{align}
For other cases, with $V$ neighbors being in the middle, where $N-1 \geq \phi \geq 2$, adding values to permutations from equation (\ref{permutations}) we get:
\begin{align}\label{ab2}
	\sum_{V=1}^{N-\phi} \left[N-2-(V-1)\right] \frac{\left(N - V - 2\right)!}{(I-V)!(\phi-2)!} V (V+1),
\end{align}
Putting previous equations (\ref{ab1}) and (\ref{ab2}) together, we obtain for $\phi \geq 2$:
\begin{align}\label{ab3}
	\sum_{V=1}^{N-\phi} \left[V (V+1)\right] \left\{2 \frac{\left(N - V - 1\right)!}{(I-V)!(\phi-1)!} + \left[N-2-(V-1)\right] \frac{\left(N - V - 2\right)!}{(I-V)!(\phi-2)!} \right\}.
\end{align}
In the case of $\phi = 1$, we have a simplified situation, where we only encounter the boundary cases:
\begin{align}\label{ab4}
	\sum_{V=1}^{N-\phi} 2 \frac{(N-V-1)!}{(I-V)!(\phi-1)!} V (V+1) \overset{\phi=1}{=} \sum_{V=1}^{N-1} 2 \frac{(N-V-1)!}{(N-1-V)!} V (V+1) = \sum_{V=1}^{N-1} 2 V(V+1).
\end{align}

We have to also take into account the terms corresponding only to $a^{2} \dyad{0}{0}$ and $b^{2} \dyad{1}{1}$ respectively. But that only means that for every permutation possible, we have to add another number of neighbors (see pages \pageref{pageForV} and 118), therefore modifying equations (\ref{ab3}) and (\ref{ab4}) in the following way:
\begin{align*}\label{ab5}
	&\phi = 1: &&\sum_{V=1}^{N-1} 2\left[V(V+1) + V\right],\\
	N-1 \geq &\phi \geq 2: &&\sum_{V=1}^{N-\phi} \left[V (V+1) + V\right] \left[2 \frac{\left(N - V - 1\right)!}{(I-V)!(\phi-1)!} + \left[N-2-(V-1)\right] \frac{\left(N - V - 2\right)!}{(I-V)!(\phi-2)!} \right]. \numberthis
\end{align*}

Unfortunately, we have to also subtract number of times $a^{2} \dyad{0}{0}$ or $b^{2} \dyad{1}{1}$ "leak" into the unsuccessful measurement, i.e. number of times the corresponding measurement is $\ket{N}$. Basically, we have to calculate how many boundary cases there are, which we have already done previously in the equation (\ref{permutationsB}) and the answer is $\sum_{V=1}^{N-\phi} 2\frac{(N-V-1)!}{(I-V)!(\phi-1)!}$. Now, we should multiply this number by either $a^{2} V$ or $b^{2} V$, but realizing, that they always sum up to $V = a^{2} V + b^{2} V$ (because there is always the same number of boundary cases on the left as on the right side) and realizing that we need two different permutations to come to this kind of summation (one from the left and one from the right), the resulting number we have to subtract is:
\begin{align}\label{ab6}
	\sum_{V=1}^{N-\phi} \frac{(N-V-1)!}{(I-V)!(\phi-1)!} V.
\end{align}
Using that $I=N-\phi$, for $\phi = 1$ we obtain $\frac{(N-V-1)!}{(N-\phi-V)!(\phi-1)!} V = \frac{(N-V-1)!}{(N-1-V)!(1-1)!} V = V$. Putting this result from equation (\ref{ab6}) together with equations (\ref{ab5}) and using that $V(V+1) + V = V^{2} + 2V = V(V+2)$, we obtain the fraction of $a^{2} \dyad{0}{0} + b^{2} \dyad{1}{1}$, $a^{2} \dyad{0}{0}$ and $b^{2} \dyad{1}{1}$ of the entire implemented channel but only for the successful measurement corresponding to term  for concrete values of $N$ and $I$:
\begin{align*}\label{ab7}
	&\phi = 1: &&\sum_{V=1}^{N-1} \left\{2 \left[V(V+1) + V\right] - V\right\} = \sum_{V=1}^{N-1} 2V^{2} + 3V = \sum_{V=1}^{N-1} V (2V + 3),\\
	N-1 \geq &\phi \geq 2: &&\sum_{V=1}^{N-\phi} \bigg\{ \left[V(V+2)\right] \bigg[2 \frac{\left(N - V - 1\right)!}{(I-V)!(\phi-1)!} + \left[N-2-(V-1)\right] \frac{\left(N - V - 2\right)!}{(I-V)!(\phi-2)!} \bigg] \\& &&- V \frac{(N-V-1)!}{(I-V)!(\phi-1)!}\bigg\}. \numberthis
\end{align*}

The share of $a^{2} \dyad{0}{0} + b^{2} \dyad{1}{1}$ corresponding to the successful measurement in case of $I=N$ is $\frac{(1-q)^{N}}{2^{N}} N$. So, by using the previous expression, equations (\ref{ab7}), adding factors $q^{\phi}(1-q)^{N-\phi}$ and normalizations $\frac{1}{(N+1) 2^{N-\phi}}$, the fraction $p_{ab}$ of the entire channel corresponding to successful measurement taken by $a^{2} \dyad{0}{0} + b^{2} \dyad{1}{1}$, $a^{2} \dyad{0}{0}$ and $b^{2} \dyad{1}{1}$ is:
\begin{align*}\label{psab}
	p_{ab} &= \frac{(1-q)^{N}}{2^{N}} N + \frac{q(1-q)^{N-1}}{(N+1)2^{N-1}} \sum_{V=1}^{N-1} \left[V (2V + 3)\right] \\&+ \sum_{\phi=2}^{N-1} \frac{q^{\phi}(1-q)^{(N-\phi)}}{(N+1)2^{(N-\phi)}} \bigg\{\sum_{V=1}^{N-\phi} \bigg[V(V+2) \left[2 \frac{(N-V-1)!}{(N-\phi-V)!(\phi-1)!} + \frac{(N-V-1)!}{(N-\phi-V)!(\phi-2)!}\right] \\&- V \frac{(N-V-1)!}{(N-\phi-V)!(\phi-1)!}\bigg]\bigg\} \\&= \frac{(1-q)^{N}}{2^{N}} N + \frac{q(1-q)^{N-1}}{(N+1)2^{N-1}} \sum_{V=1}^{N-1} \left[V (2V + 3)\right] \\&+ \sum_{\phi=2}^{N-1} \frac{q^{\phi}(1-q)^{(N-\phi)}}{(N+1)2^{(N-\phi)}} \bigg\{\sum_{V=1}^{N-\phi} \bigg[(2V(V+2) - V) \frac{(N-V-1)!}{(N-\phi-V)!(\phi-1)!} \\&+ V(V+2) \frac{(N-V-1)!}{(N-\phi-V)!(\phi-2)!} \bigg]\bigg\} \\&\overset{(i)}{=} \frac{(1-q)^{N}}{2^{N}} N + \sum_{\phi=1}^{N-1} \frac{q^{\phi}(1-q)^{N-\phi}}{(N+1)2^{N-\phi}} \sum_{V=1}^{N-\phi} \left[V (2V + 3) \frac{(N-V-1)!}{(N-\phi-V)!(\phi-1)!}\right] \\&+ \sum_{\phi=2}^{N-1} \frac{q^{\phi}(1-q)^{(N-\phi)}}{(N+1)2^{(N-\phi)}} \sum_{V=1}^{N-\phi} V(V+2) \frac{(N-V-1)!}{(N-\phi-V)!(\phi-2)!}, \numberthis
\end{align*}
where we have used that $I = N-\phi$ and in $(i)$ that $2V(V+2) - V = V(2V+3)$ and that for $\phi=1$ the expression $\frac{(N-V-1)!}{(N-\phi-V)!(\phi-1)!} = 1$.

By putting together previous shares of probabilities $p_{ab}$ (\ref{psab}) with the one we have gotten from unitary channels $p_{U_{\phi}}$ in equation (\ref{pufi}), we obtain the probability of successful measurement:
\begin{align*} \label{VQdep}
	p_{suc} &= \frac{(1-q)^{N}}{2^{N}} N + \frac{N}{N+1} q^{N} \\
	&+ \sum_{\phi=1}^{N-1} \frac{q^{\phi}(1-q)^{N-\phi}}{(N+1)2^{N-\phi}} \left\{ \sum_{V=1}^{N-\phi} \left[V (2V + 3) \frac{(N-V-1)!}{(N-\phi-V)!(\phi-1)!}\right] + \frac{N!}{(\phi-1)!(N-\phi)!}\right\} \\
	&+ \sum_{\phi=2}^{N-1} \frac{q^{\phi}(1-q)^{(N-\phi)}}{(N+1)2^{(N-\phi)}} \sum_{V=1}^{N-\phi} V(V+2) \frac{(N-V-1)!}{(N-\phi-V)!(\phi-2)!}, \numberthis
\end{align*}
where the term $\frac{N}{N+1} q^{N}$ was obtained from equation (\ref{pufi}) for $\phi = N$ while the term $\sum_{\phi=1}^{N-1} \frac{q(1-q)^{N}}{(N+1)2^{N-1}} \\ \frac{N!}{(\phi-1)!(N-\phi)!}$ comes directly from the same equation for the rest of possible values of $\phi$.

\paragraph{Phase Damping}
Let us repeat the calculation for the case of phase damping. Again, we shall make the calculation for the $N = 2$ case, when the input state is the same as in equation (\ref{impInSt}). Now, we are trying to implement channel ${\cal{F}}_{\phi}^{\otimes2}$ and we start with acting with this channel on the input state:
\begin{align*}
	&{\cal{F}}_{\phi}^{\otimes2} (\dyad{\psi}{\psi}) = \\&q^{2} {\cal{U}}_{\psi}^{\otimes2} (\dyad{\psi}{\psi}) + q(1-q) \big[({\cal{U}}_{\phi} \otimes {\cal{P}}) + ({\cal{P}} \otimes {\cal{U}}_{\phi})\big] (\dyad{\psi}{\psi}) + (1-q)^{2} {\cal{P}}^{\otimes2} (\dyad{\psi}{\psi}).
\end{align*}
Let us evaluate the previous equation term by term. First term ${\cal{U}}_{\psi}^{\otimes2} (\dyad{\psi}{\psi})$ was already calculated in (\ref{StInpStImpl}). Therefore, let us proceed with the following term:
\begin{align*}
	&({\cal{U}}_{\phi} \otimes {\cal{P}}) (\dyad{\psi}{\psi}) = \frac{1}{3} \\&\big[\dyad{0}{0} \otimes {\cal{P}}(\dyad{0}{0}) + \dyad{0}{0} \otimes {\cal{P}}(\dyad{0}{1}) + e^{-i\phi}\dyad{0}{1} \otimes {\cal{P}}(\dyad{0}{1}) + \dyad{0}{0} \otimes {\cal{P}}(\dyad{1}{0}) + \dyad{0}{0} \otimes {\cal{P}}(\dyad{1}{1}) \\&+ e^{-i\phi}\dyad{0}{1} \otimes {\cal{P}}(\dyad{1}{1}) + e^{i\phi}\dyad{1}{0} \otimes {\cal{P}}(\dyad{1}{0}) + e^{i\phi}\dyad{1}{0} \otimes {\cal{P}}(\dyad{1}{1}) + \dyad{1}{1} \otimes {\cal{P}}(\dyad{1}{1})\big] \overset{(i)}{=} \\&\frac{1}{3} \big(\dyad{00}{00} + \dyad{01}{01} + e^{-i\phi} \dyad{01}{11} + e^{i\phi} \dyad{11}{01} + \dyad{11}{11}\big) \overset{(\ref{dictionary})}{=} \\&\frac{1}{3} (\dyad{0}{0} + \dyad{1}{1} + e^{-i\phi} \dyad{1}{2} + e^{i\phi} \dyad{2}{1} + \dyad{2}{2}),
\end{align*}
where in $(i)$ we know that ${\cal{P}} (\dyad{0}{1}) + {\cal{P}} (\dyad{1}{0}) = 0$, ${\cal{P}}(\dyad{0}{0}) = \dyad{0}{0}$ and ${\cal{P}}(\dyad{1}{1}) = \dyad{1}{1}$. Analogously we can evaluate the next term:
\begin{align*}
	&({\cal{P}} \otimes {\cal{U}}_{\phi}) (\dyad{\psi}{\psi}) = \frac{1}{3} (\dyad{0}{0} + e^{-i\phi} \dyad{0}{1} + e^{i\phi} \dyad{1}{0} + \dyad{1}{1} + \dyad{2}{2}).
\end{align*}
And the last term:
\begin{align*}
	({\cal{P}} \otimes {\cal{P}}) (\dyad{\psi}{\psi}) = \frac{1}{3} (\dyad{00}{00} + \dyad{01}{01} + \dyad{11}{11}) = \frac{1}{3} (\dyad{0}{0} + \dyad{1}{1} + \dyad{2}{2}).
\end{align*}
Now, we calculate the product states with $\ket{\xi} = a \ket{0} + b \ket{1}$ (where $\dyad{\xi}{\xi} \otimes {\cal{U}}_{\phi}^{\otimes 2}$ was already calculated in the equation (\ref{xiuu})):
\begin{align*}
	&\dyad{\xi}{\xi} \otimes ({\cal{U}}_{\phi} \otimes {\cal{P}}) (\dyad{\psi}{\psi}) = \\&\frac{1}{3} \big[a^{2} (\dyad{00}{00} + \dyad{01}{01} + e^{-i\phi} \dyad{01}{02} + e^{i\phi} \dyad{02}{01} + \dyad{02}{02}) \\&+ ab^{\ast} (\dyad{00}{10} + \dyad{01}{11} + e^{-i\phi} \dyad{01}{12} + e^{i\phi} \dyad{02}{11} + \dyad{02}{12}) \\&+ a^{\ast}b (\dyad{10}{00} + \dyad{11}{01} + e^{-i\phi} \dyad{11}{02} + e^{i\phi} \dyad{12}{01} + \dyad{12}{02}) \\&+ b^{2} (\dyad{10}{10} + \dyad{11}{11} + e^{-i\phi} \dyad{11}{12} + e^{i\phi} \dyad{12}{11} + \dyad{12}{12})\big].
\end{align*}
The second term is:
\begin{align*}
	&\dyad{\xi}{\xi} \otimes ({\cal{P}} \otimes {\cal{U}}_{\phi}) (\dyad{\psi}{\psi}) = \\&\frac{1}{3} \big[a^{2} (\dyad{00}{00} + e^{-i\phi} \dyad{00}{01} + e^{i\phi} \dyad{01}{00} + \dyad{01}{01} + \dyad{02}{02}) \\
	&+ ab^{\ast} (\dyad{00}{10} + e^{-i\phi} \dyad{00}{11} + e^{i\phi} \dyad{01}{10} + \dyad{01}{11} + \dyad{02}{12}) \\
	&+ a^{\ast}b (\dyad{10}{00} + e^{-i\phi} \dyad{10}{01} + e^{i\phi} \dyad{11}{00} + \dyad{11}{01} + \dyad{12}{02}) \\
	&+ b^{2} (\dyad{10}{10} + e^{-i\phi} \dyad{10}{11} + e^{i\phi} \dyad{11}{10} + \dyad{11}{11} + \dyad{12}{12})\big].
\end{align*}
And the last term gives us:
\begin{align*}
	&\dyad{\xi}{\xi} \otimes {\cal{P}}^{\otimes2} (\dyad{\psi}{\psi}) = \\& \frac{1}{3} \big[a^{2} (\dyad{00}{00} + \dyad{01}{01} + \dyad{02}{02}) + ab^{\ast} (\dyad{00}{10} + \dyad{01}{11} + \dyad{02}{12}) \\&+ a^{\ast}b (\dyad{10}{00} + \dyad{11}{01} + \dyad{12}{02}) + b^{2} (\dyad{10}{10} + \dyad{11}{11} + \dyad{12}{12})\big].
\end{align*}
Let us apply shift-down operator on the previous states:
\begin{align*} \label{xifip}
	&{\cal{C}}_{\ominus} \big[\dyad{\xi}{\xi} \otimes ({\cal{U}}_{\phi} \otimes {\cal{P}}) (\dyad{\psi}{\psi})\big] = \\&\frac{1}{3} \big[a^{2} (\dyad{00}{00} + \dyad{01}{01} + e^{-i\phi} \dyad{01}{02} + e^{i\phi} \dyad{02}{01} + \dyad{02}{02}) \\&+ ab^{\ast} (\dyad{00}{12} + \dyad{01}{10} + e^{-i\phi} \dyad{01}{11} + e^{i\phi} \dyad{02}{10} + \dyad{02}{11}) \\&+ a^{\ast}b (\dyad{12}{00} + \dyad{10}{01} + e^{-i\phi} \dyad{10}{02} + e^{i\phi} \dyad{11}{01} + \dyad{11}{02}) \\&+ b^{2} (\dyad{12}{12} + \dyad{10}{10} + e^{-i\phi} \dyad{10}{11} + e^{i\phi} \dyad{11}{10} + \dyad{11}{11})\big] \\&\overset{(i)}{=} \frac{1}{3} \big[ (a^{2} \dyad{0}{0} + ab^{\ast}e^{-i\phi} \dyad{0}{1} + a^{\ast}be^{i\phi} \dyad{1}{0} + b^{2} \dyad{1}{1}) \otimes \dyad{0}{0} \\
	& (a^{2} \dyad{0}{0} + b^{2} \dyad{1}{1}) \otimes (\dyad{1}{1} + \dyad{2}{2}) \big]. \numberthis
\end{align*}
In $(i)$, we are again only explicitly writing out diagonal states in the second space due to measurement. Let us also write the result for the other term with one phase damping and one unitary channel applied:
\begin{align*} \label{xipfi}
	&{\cal{C}}_{\ominus} \left[\dyad{\xi}{\xi} \otimes ({\cal{P}} \otimes {\cal{U}}_{\phi}) (\dyad{\psi}{\psi})\right] = \\
	&\frac{1}{3} \big[a^{2} (\dyad{00}{00} + e^{-i\phi}\dyad{00}{01} + e^{i\phi}\dyad{01}{00} + \dyad{01}{01} + \dyad{02}{02}) \\
	&+ ab^{\ast} (\dyad{00}{12} + e^{-i\phi}\dyad{00}{10} + e^{i\phi}\dyad{01}{12} + \dyad{01}{10} + \dyad{02}{11}) \\
	&+ a^{\ast}b (\dyad{12}{00} + e^{-i\phi}\dyad{12}{01} + e^{i\phi}\dyad{10}{00} + \dyad{10}{01} + \dyad{11}{02}) \\
	&+ b^{2} (\dyad{12}{12} + e^{-i\phi}\dyad{12}{10} + e^{i\phi}\dyad{10}{12} + \dyad{10}{10} + \dyad{11}{11})\big]\\
	&= \frac{1}{3} \big[ (a^{2} \dyad{0}{0} + ab^{\ast} e^{-i\phi} \dyad{0}{1} + a^{\ast}b e^{i\phi} \dyad{1}{0} + b^{2} \dyad{1}{1}) \otimes \dyad{0}{0}\\ 
	&+ (a^{2} \dyad{0}{0} + b^{2} \dyad{1}{1}) \otimes (\dyad{1}{1} + \dyad{2}{2}) \big]\\
	&= \frac{1}{3} \big[ U_{\phi} \dyad{\xi}{\xi} U_{\phi}^{\dagger} \otimes \dyad{0}{0} + (a^{2} \dyad{0}{0} + b^{2} \dyad{1}{1}) \otimes (\dyad{1}{1} + \dyad{2}{2}) \big]. \numberthis
\end{align*}
And the last state:
\begin{align*} \label{xipp}
	&{\cal{C}}_{\ominus} \big[\dyad{\xi}{\xi} \otimes {\cal{P}}^{\otimes2} (\dyad{\psi}{\psi})\big] = \\& \frac{1}{3} \big[a^{2} (\dyad{00}{00} + \dyad{01}{01} + \dyad{02}{02}) + ab^{\ast} (\dyad{00}{12} + \dyad{01}{10} + \dyad{02}{11}) \\&+ a^{\ast}b (\dyad{12}{00} + \dyad{10}{01} + \dyad{11}{02}) + b^{2} (\dyad{12}{12} + \dyad{10}{10} + \dyad{11}{11})\big] = \\& \frac{1}{3} \big[(a^{2}\dyad{0}{0} + b^{2}\dyad{1}{1}) \otimes (\dyad{0}{0} + \dyad{1}{1} + \dyad{2}{2})\big]. \numberthis
\end{align*}
Putting equations (\ref{xiuu}), (\ref{xifip}), (\ref{xipfi}) and (\ref{xipp}) together, we obtain:
\begin{align*}\label{implPD}
	&{\cal{C}}_{\ominus} \big[\dyad{\xi}{\xi} \otimes {\cal{F}}_{\phi}^{\otimes2}(\dyad{\psi}{\psi})\big] = \\
	&\frac{1}{3} \big\{q^{2} \big[U_{\phi}\dyad{\xi}{\xi}U_{\phi}^{\dagger} \otimes (\dyad{0}{0} + \dyad{1}{1}) + U_{-2\phi}\dyad{\xi}{\xi}U_{-2\phi}^{\dagger} \otimes \dyad{2}{2}\big] \\
	&+ 2q(1-q) \big[U_{\phi}\dyad{\xi}{\xi}U_{\phi}^{\dagger} \otimes (\dyad{0}{0} + \dyad{1}{1}) + (a^{2}\dyad{0}{0} + b^{2}\dyad{1}{1}) \otimes (\dyad{0}{0} + \dyad{1}{1} + 2\dyad{2}{2}) \\
	&+ (1-q)^{2} (a^{2}\dyad{0}{0} + b^{2}\dyad{1}{1}) \otimes (\dyad{0}{0} + \dyad{1}{1} + \dyad{2}{2})\big] \big\}. \numberthis
\end{align*}
We have also calculated results for $N = 3$:
\begin{align*}
	&{\cal{C}}_{\ominus} \big[\dyad{\xi}{\xi} \otimes {\cal{F}}_{\phi}^{\otimes3}(\dyad{\psi}{\psi})\big] = \\
	&\frac{1}{4} \big\{ q^{3} \big[ U_{\phi} \dyad{\xi}{\xi} U_{\phi}^{\dagger} \otimes (\dyad{0}{0} + \dyad{1}{1} + \dyad{2}{2}) + U_{-3\phi} \dyad{\xi}{\xi} U_{-3\phi}^{\dagger} \otimes \dyad{3}{3}\big]\\
	& q^{2}(1-q) \big[ U_{\phi} \dyad{\xi}{\xi} U_{\phi}^{\dagger} \otimes 2(\dyad{0}{0} + \dyad{1}{1} + \dyad{2}{2})\\ 
	&+ (a^{2} \dyad{0}{0} + b^{2} \dyad{1}{1}) \otimes (\dyad{0}{0} + \dyad{1}{1} + \dyad{2}{2} + 3 \dyad{3}{3})\big]\\
	& q(1-q)^{2}  \big[ U_{\phi} \dyad{\xi}{\xi} U_{\phi}^{\dagger} \otimes (\dyad{0}{0} + \dyad{1}{1} + \dyad{2}{2})\\ 
	&+ (a^{2} \dyad{0}{0} + b^{2} \dyad{1}{1}) \otimes (2 (\dyad{0}{0} + \dyad{1}{1} + \dyad{2}{2}) + 3 \dyad{3}{3})\big]\\
	& (1-q)^{3} \big[ (a^{2} \dyad{0}{0} + b^{2} \dyad{1}{1}) \otimes (\dyad{0}{0} + \dyad{1}{1} + \dyad{2}{2} + \dyad{3}{3})\big] \big\}
\end{align*}
We can generalize the results using binomial distribution:
\begin{align*}\label{VQPDfinal}
	&{\cal{C}}_{\ominus} \big[\dyad{\xi}{\xi} \otimes {\cal{F}}_{\phi}^{\otimes N}(\dyad{\psi}{\psi})\big] = \\
	&\frac{1}{N+1} \biggl[ q^{N} U_{-N\phi} \dyad{\xi}{\xi} U_{-N\phi}^{\dagger} \otimes \dyad{N}{N} +
	 \sum_{\phi=1}^{N} \binom{N-1}{\phi-1} q^{\phi} (1-q)^{N-\phi}  U_{\phi} \dyad{\xi}{\xi} U_{\phi}^{\dagger} \otimes \sum_{j=0}^{N-1} \dyad{j}{j}\\
	& \sum_{P=1}^{N} \binom{N-1}{P-1} q^{N-P} (1-q)^{P} (a^{2} \dyad{0}{0} + b^{2} \dyad{1}{1}) \otimes \sum_{j=0}^{N-1} \dyad{j}{j}\\
	& \sum_{P=1}^{N} \binom{N}{P} q^{N-P} (1-q)^{P} (a^{2} \dyad{0}{0} + b^{2} \dyad{1}{1}) \otimes \dyad{N}{N} \biggr] \numberthis
\end{align*}
Now, we can look at the individual probabilities arising from previous equation (\ref{VQPDfinal}). Let us begin with probability of implementing unitary transformation:
\begin{align*}
	p_{U_{\phi}} = \frac{1}{N+1} \sum_{\phi=1}{N} \binom{N-1}{\phi-1} q^{\phi} (1-q)^{N-\phi} N = \frac{Nq}{N+1}
\end{align*}
where the additional factor $N$ comes from number of states corresponding to successful measurement. We can continue with evaluating success probability corresponding to measuring noisy part of implemented channel:
\begin{align*}
	p_{ab,suc} = \frac{N}{N+1} \sum_{P=1}^{N} q^{N-P} (1-q)^{P} = \frac{N}{N+1} (1-q).
\end{align*}
Putting the two previous equations together, we recover total success probability:
\begin{align}\label{VQpd}
	p_{suc} = \frac{N}{N+1} \big[ \sum_{\phi=1}^{N} \binom{N-1}{\phi-1} q^{\phi} (1-q)^{N-\phi} + \sum_{P=1}^{N} \binom{N-1}{P-1} q^{N-P} (1-q)^{P}\big] = \frac{N}{N+1} (q+1-q) = \frac{N}{N+1}
\end{align}

\subsubsection{Comparison of Implementations}

Let us compare the implementation of noisy depolarizing channel through the virtual qudit (equation (\ref{VQdep})), Vidal-Masanes-Cirac scheme (\ref{genVMC}) and the retrieval of channels through the PSAR device (\ref{resWhiteNoise}). This comparison for success probability $p_{suc}$ can be seen in the figure \ref{compVQandVMCandPSARforDep:a} for $N=1$, $N=3$, $N=7$ and $N=15$, where PSAR is depicted with solid lines, Vidal-Masanes-Cirac with dotted lines and virtual qudit with dashed lines. We can see that the implementation through Vidal-Masanes-Cirac scheme gives the highest probability of success. It does not depend on mixing parameter $q$ and always performs better than virtual qudit and PSAR apart from trivial cases ($N=1$ and $q=1$). PSAR device perfoms very similarly to the virtual qudit implementation, albeit with slightly worse probabilities in general except for $q=1$. Success probability for PSAR and virtual qudit goes to $0$ with the number $N$ going to infinity. However, for every $N$, there is an interval of high value of $q$ (which is shrinking for growing $N$) that improves the probability of success $p_{suc}$ compared to lower value of $N$.
\begin{figure}[h]
	\begin{subfigure}{0.65\textwidth}
		\hspace*{-0.5cm}
		\includegraphics[width=0.87\textwidth]{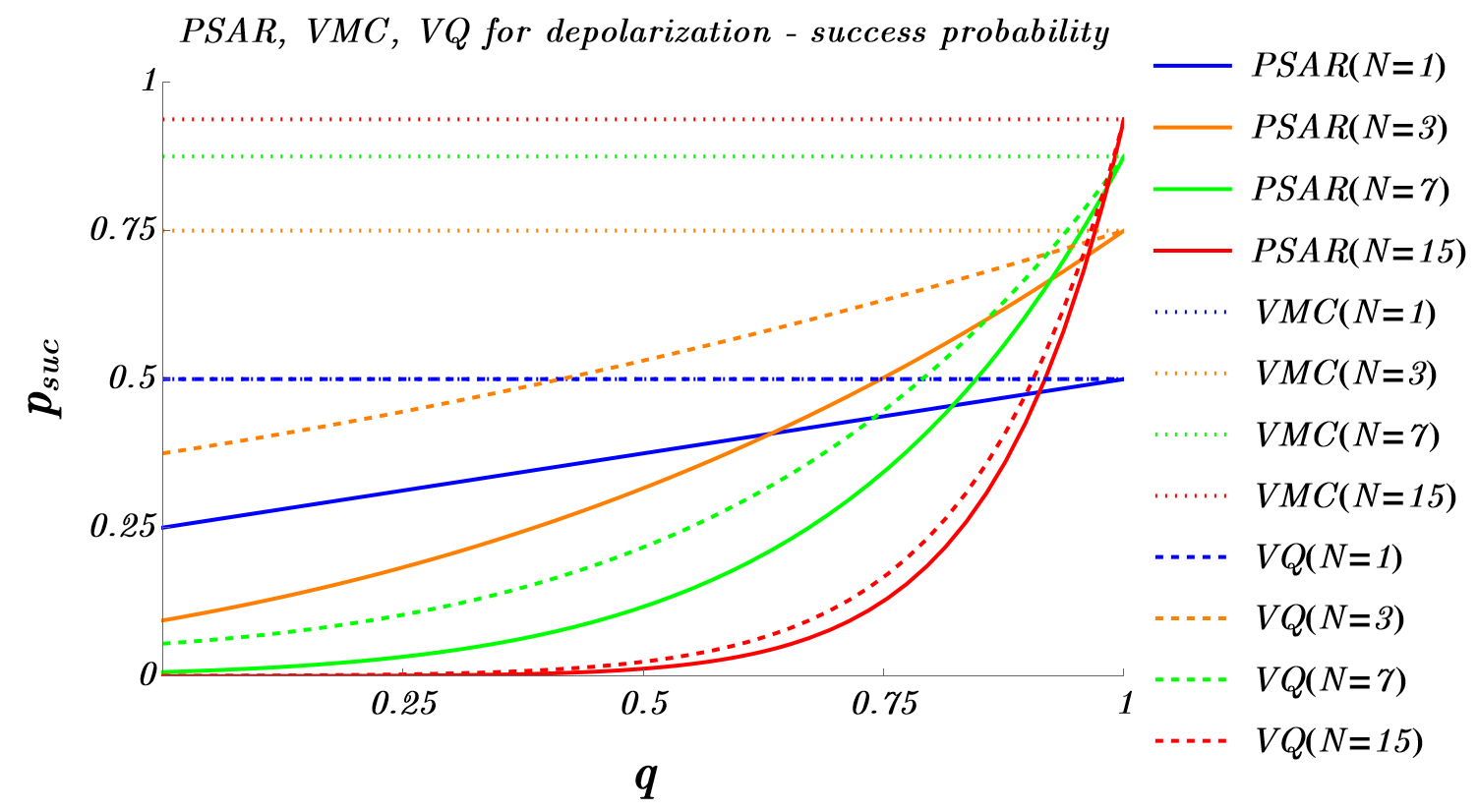}
		\vspace*{0cm}
		\caption{Comparison of success probability for implementing \\noisy depolarizing channel through PSAR device (solid \\lines), Vidal-Masanes-Cirac scheme (VMC, dotted lines) \\and virtual qudit (VQ, dashed lines) for $N=1$, $N=3$, \\$N=7$, and $N=15$.}
		\label{compVQandVMCandPSARforDep:a}
	\end{subfigure}
	\hspace*{-1.5cm}
	\begin{subfigure}{0.65\textwidth}
		\centering
		\hspace*{-3.5cm}
		\vspace*{0.5cm}
		\includegraphics[width=0.68\textwidth]{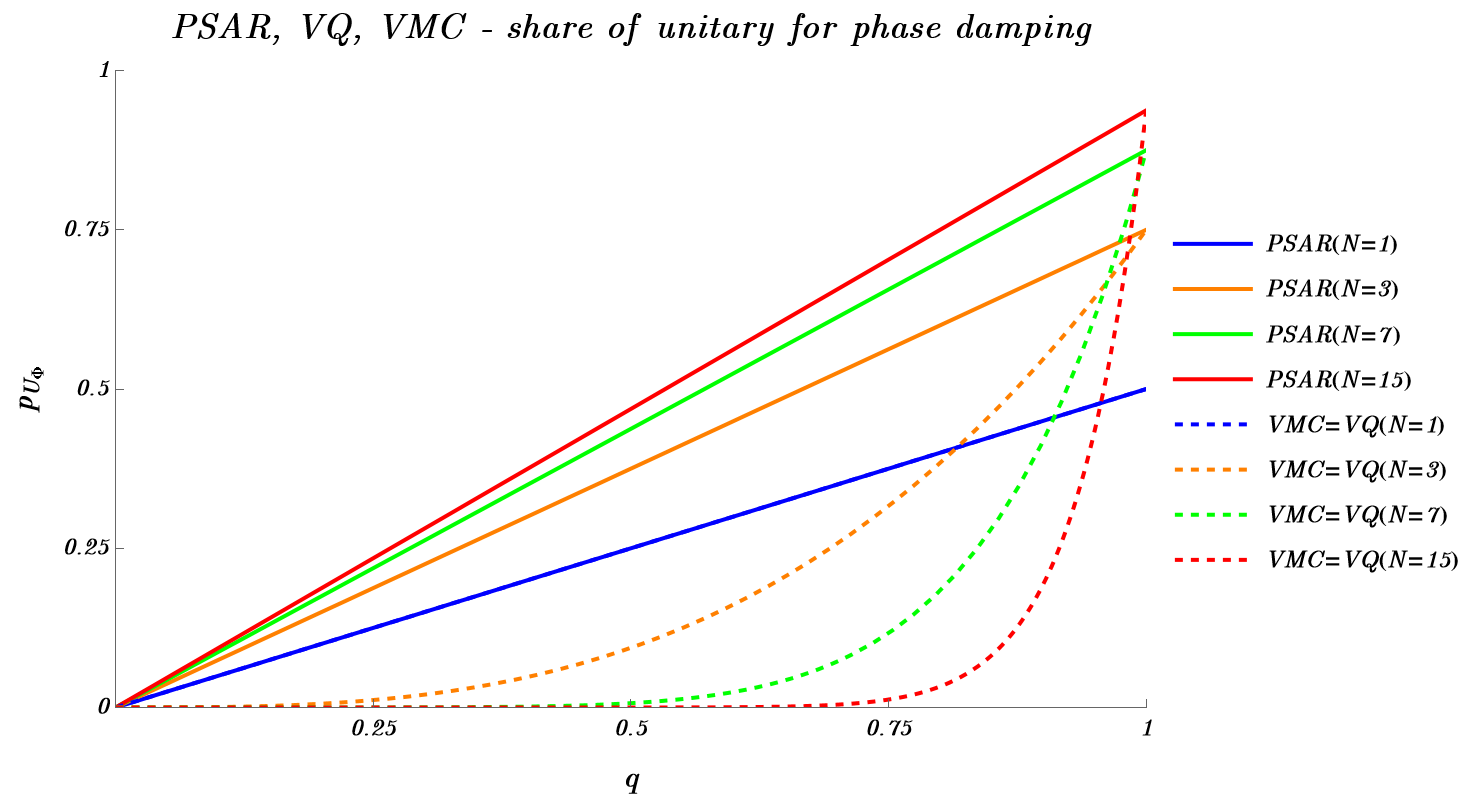}
		\vspace*{-0.5cm}
		\caption{Comparison of share that unitary channel $p_{U_{\phi}}$ \\takes from the whole implemented channel for \\PSAR (solid lines), virtual qudit (solid lines) and \\Vidal-Masanes-Cirac (dotted lines) in case of \\depolarizing channel.}
		\label{compVQandVMCandPSARforDep:b}
	\end{subfigure}
	\vspace*{0.2cm}
	\caption{Comparison of respective implementations of depolarizing noisy channel.}
	\label{compVQandVMCandPSARforDep}
\end{figure}

In the figure \ref{compVQandVMCandPSARforDep:b} there is depicted in-mixture of unitary channel $p_{U_{\phi}}$ in the entire output state for PSAR, Vidal-Masanes-Cirac and virtual qudit. In this case, values for PSAR (\ref{resWhiteNoise}) and virtual qudit (\ref{pufi}) are identical $p_{U_{\phi}}(PSAR) = p_{U_{\phi}}(VQ) = \frac{Nq}{N+1}(\frac{1+q}{2})^{N-1}$ and they are depicted with solid lines, while value $p_{U_{\phi}}(VMC) = \frac{Nq^{N}}{N+1}$ for Vidal-Masanes-Cirac (\ref{genVMC}) is depicted with dashed lines. We can see, that PSAR and virtual qudit are able to preserve unitary transformation in the resulting state better than implementations through Vidal-Masanes-Cirac scheme.

\begin{figure}[h!]
	\centering
	\includegraphics[width=.75\linewidth]{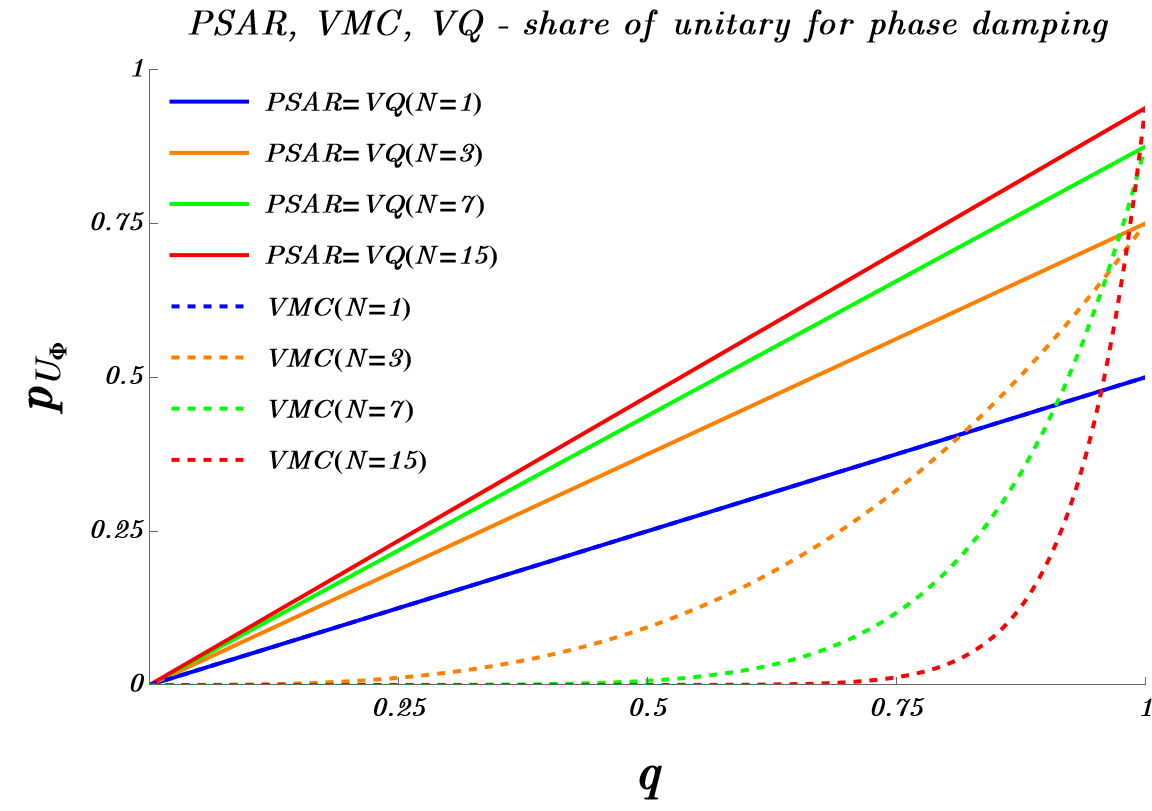}
	\vspace*{0.3cm}
	\caption{Comparison of fraction of unitary channel $p_{U_{\phi}}$ mixed in the resulting channel corresponding to successful measurement for phase damping via PSAR (solid lines), Vidal-Masanes-Cirac (dotted lines) and virtual qudit (dotted lines).}
	\label{compPDBetweenVMCandVQ}
\end{figure}

For phase damping, probability of successful implementation stays the same for PSAR, Vidal-Masanes-Cirac and also for virtual qudit. This probability was already shown in the figure \ref{compPSARbetweenPDandDep:a}. In the figure \ref{compPDBetweenVMCandVQ} there is depicted a comparison between share of unitary transformation $p_{U_{\phi}}$ of the entire implemented channel in case of implementing noisy phase damping channel through virtual qudit (\ref{VQpd}), Vidal-Masanes-Cirac scheme (\ref{genVMC}) and PSAR (\ref{redPhaseDamp}). Performance for virtual qudit implementation and PSAR $p_{U_{\phi}}(VQ) = p_{U_{\phi}}(PSAR) = \frac{Nq}{N+1}$ are identical and their performance is better in comparison with PSAR $p_{U_{\phi}}(PSAR) = \frac{Nq^{N}}{N+1}$. For VQ and PSAR fraction of unitary is linearly dependent on mixing factor $q$ and also rises with number of uses $N$. For VMC the fraction of unitary transformations goes to $0$ with growing $N$. We can also see that there is a small region of $q$ for which the performance of VMC implementation is improving with growing $N$. Unfortunately, this region is diminishing with the rising number $N$.

In case of Vidal-Masanes-Cirac implementation, the result is the same for both channels (\ref{genVMC}). Thus, let us now compare performance of virtual qudit implementation for both noisy channels. Comparison of success probability $p_{suc}$ of implementation of depolarizing channel (\ref{VQdep}) with the implementation of phase damping (\ref{VQpd}) through virtual qudit is shown in the figure \ref{compForVQbetweenDepAndPD:suc}. Depolarizing channel is depicted with solid lines, while phase damping with dashed lines. For all cases, with the exception of $q=1$ and $N=1$, virtual qudit si more successful in implementing noisy phase damping channel. Probability of success for depolarizing channel is dependent on the number of uses $N$ and goes to $0$ with $N \rightarrow \infty$.
\begin{figure}[h!]
	\hspace*{-1.2cm}
	\begin{subfigure}{.65\textwidth}
		\centering
		\vspace*{-1.5cm}
		\hspace*{0cm}
		\includegraphics[width=.85\linewidth]{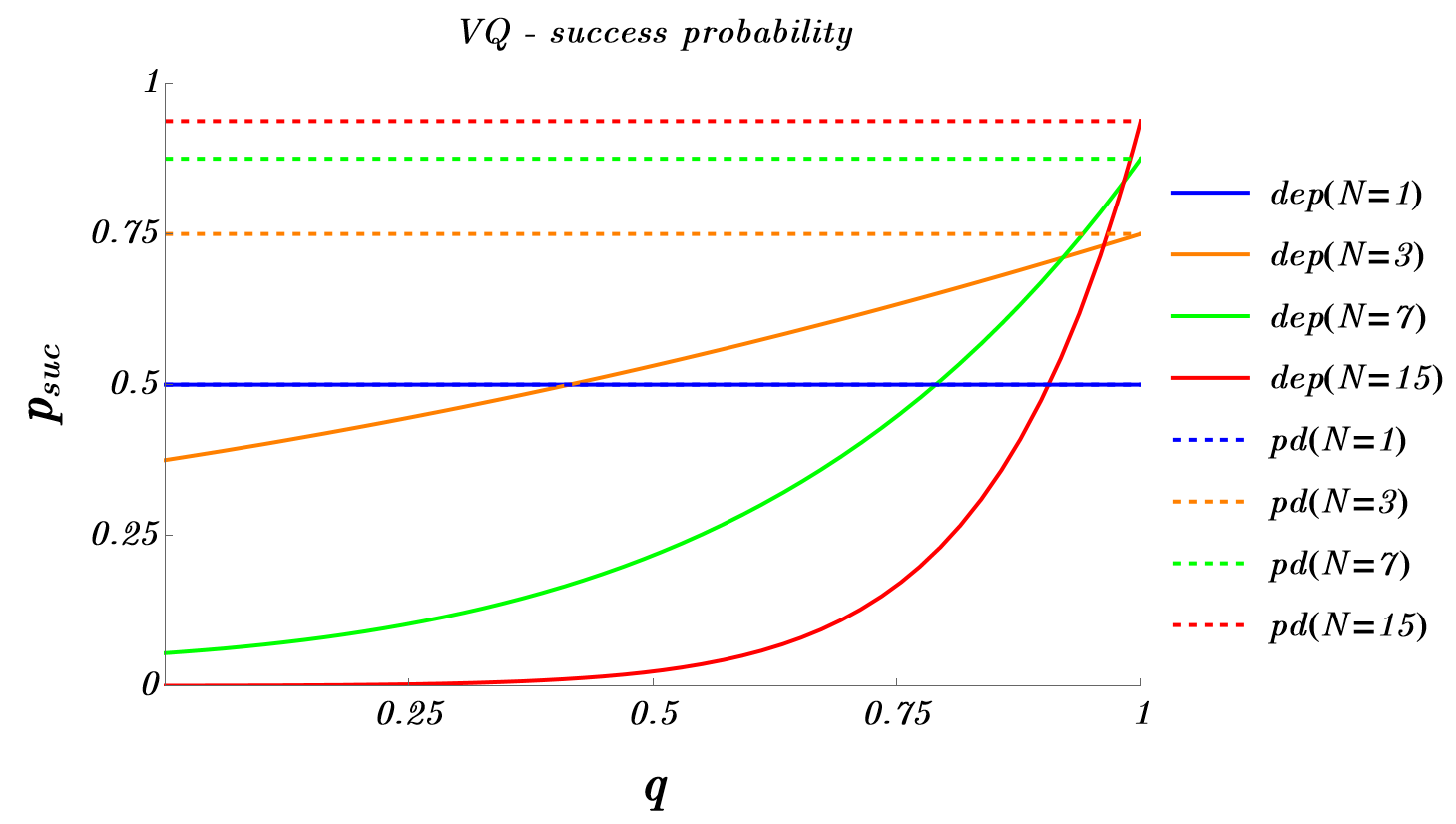}
		\caption{Comparison of success probability $p_{suc}$.}
		\label{compForVQbetweenDepAndPD:suc}
	\end{subfigure}
	\hspace*{-0.5cm}
	\begin{subfigure}{.6\textwidth}
		\centering
		\hspace*{-2.7cm}
		\vspace*{-1.8cm}
		\includegraphics[width=.73\linewidth]{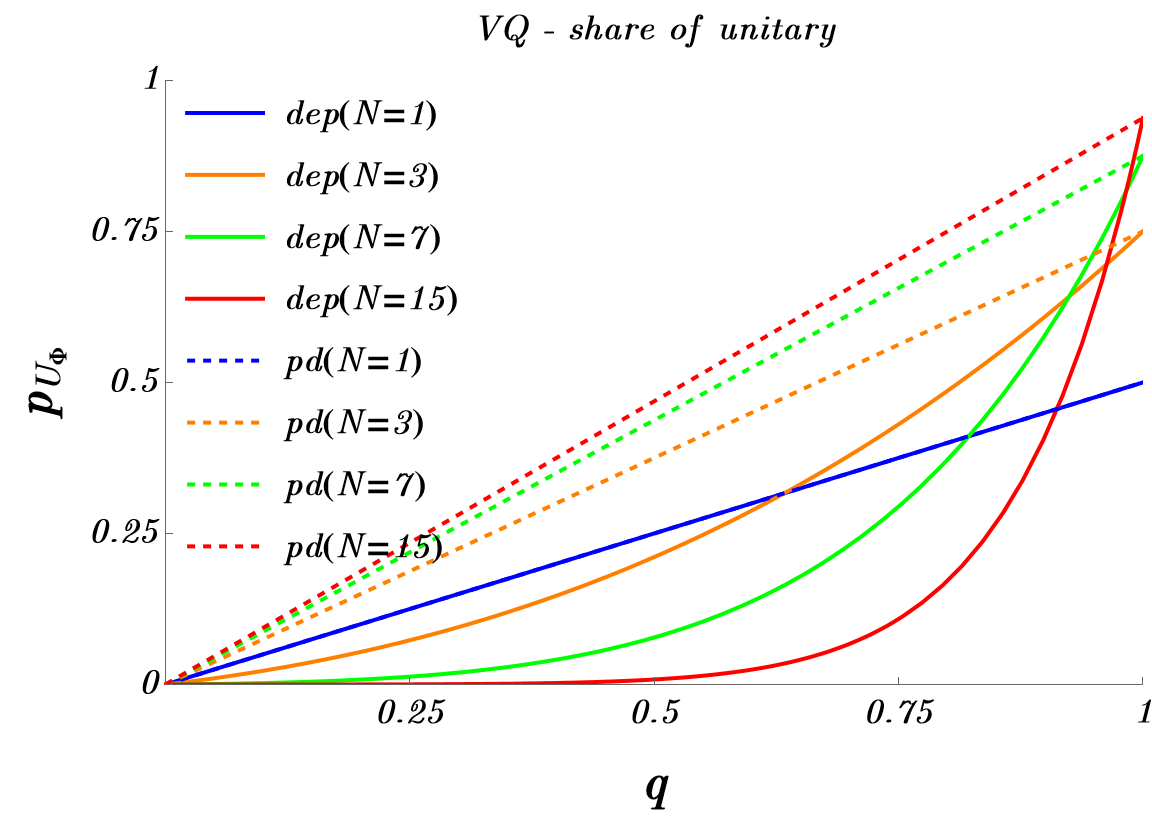}
		\vspace*{1.6cm}
		\caption{Comparison of share of unitary channel $p_{U_{\phi}}$ of \\the entire implemented channel.}
		\label{compForVQbetweenDepAndPD:ufi}
	\end{subfigure}
	\vspace*{0.2cm}
	\caption{Comparisons between implementation of depolarizing channel (dep, solid lines) and phase damping (pd, dashed lines) through virtual qudit for $N=1$, $N=3$, $N=7$ and $N=15$.}
	\label{compForVQbetweenDepAndPD}
\end{figure}

In the figure \ref{compForVQbetweenDepAndPD:ufi} we have depicted a comparison of share of unitary channel $p_{U_{\phi}}$ of the entire implemented channel for implementation of both noisy channels through virtual qudit. Value for depolarizing channel is $p_{U_{\phi}}(dep) = \frac{Nq}{N+1}(\frac{1+q}{2})^{N-1}$ as can be seen from equation (\ref{pufi}), while value for phase damping is $p_{U_{\phi}}(dp) = \frac{Nq}{N+1}$ as can be seen in (\ref{VQpd}). Virtual qudit is more successful with protecting unitary channel in case of noisy phase damping channel. For depolarizing channel, the share of unitary channel is going to zero as $N$ goes to infinity, while for phase damping, performance of VQ is improving. Both figures \ref{compForVQbetweenDepAndPD:suc} and \ref{compForVQbetweenDepAndPD:ufi} show regions of high value $q$ for which the success probability and representation of unitary channel of the result, respectively, are improving with growing $N$ for depolarization. But again, this region is shrinking with the growing number $N$.

\chapter{Conclusions}
We have examined the equivalence conditions for deterministic and probabilistic processors. At first, deterministic equivalence was defined in definition \ref{detEqv}, as well as three types of probabilistic equivalences were defined - strong in \ref{strongEqv}, weak in \ref{weakEqv} and structural in \ref{structEqv}. In theorem \ref{SufNecCondDetProc}, sufficient and necessary condition for deterministic equivalence of unitarily bonded processors was presented, followed by some concrete solutions. Further, conditions for deterministically equivalent processors with dimensions of data and program spaces $D = P = 2$ were derived. Equivalence of SWAP processor $S$ with the same dimensions was solved. It was discovered that $S$ is equivalent with $\left(U \otimes V\right) W \left(U^{\prime} \otimes V^{\prime}\right)$, where $U$, $V$, $U^{\prime}$ and $V^{\prime}$ are $2$-dimensional unitaries and $W = \exp(i\sum_{\alpha} \alpha \sigma_{\alpha} \otimes \sigma_{\alpha})$ with $\alpha = \{x, y, z\}$ and values for $x$, $y$ and $z$ given in equation (\ref{xyzxyzxyz}). Furthermore, necessary and sufficient conditions for structurally equivalent processors, with unitary relations between them, was given in theorem \ref{necAndSucStruct}. Relations between operators of structurally equivalent processors were also investigated and it was discovered that their spans must be identical (theorem \ref{spansOfOperators}). Specific co-isometric relation was given for processors with orthogonal operators. In theorem \ref{sufNecDimProbWeak}, sufficient condition for weakly (strongly) equivalent processors was stated. In the end of the chapter, relations between individual equivalences were examined. It was discovered, that only structural equivalence implies additional type of equivalence, i.e., processors that are structurally equivalent are either weakly or strongly equivalent.

Robustness of optimal probabilistic storage and retrieval device for phase gates to noises - depolarization and phase damping - was also investigated. In the case of implementing noisy channel composed of convex combination of unitary channel with the depolarizing one, it was found that the retrieved channel is noisy with probability that is decreasing with growing number of implemented original noisy channels (equation (\ref{resWhiteNoise})). For the implementation of unitary channel combined with phase damping, PSAR device is again implementing noisy channel in the case of successful measurement. However, the probability of success has not changed compared to having access to phase gates without any noise and is increasing with the rising number $N$ (equation \ref{redPhaseDamp}). Two concrete implementations were examined - Vidal-Masanes-Cirac and virtual qudit. Both implementations were shown to perform better than what was calculated through PSAR device for implementing noisy depolarizing channel as can be seen in the figure \ref{compVQandVMCandPSARforDep}. By comparing implementations of noisy depolarizing channel through virtual qudit and Vidal-Masanes-Cirac it was discovered that success probability is higher for Vidal-Masanes-Cirac, albeit the less noisy resulting channel is recovered via virtual qudit as is depicted in the figures \ref{compVQandVMCandPSARforDep}. For noisy phase damping channel, virtual qudit and Vidal-Masanes-Cirac (and even PSAR) gave the same results. However, with respect to preserving unitary channel, both implementations performed equally well, albeit worse than PSAR as can be seen in figure \ref{compPDBetweenVMCandVQ}. Vidal-Masanes-Cirac scheme performed equally well for both noisy channels. Comparison of implementations of depolarizing and phase damping channels through virtual qudit revealed that the success probability is higher for phase damping, but the less noisy implemented channel is given for depolarizing channel as can be seen in figures \ref{compForVQbetweenDepAndPD}.



\end{document}